\documentclass[a4paper,notitlepage,10pt]{article}
\usepackage[sectionbib]{natbib}
\usepackage{amsmath,mathtools}
\usepackage{amssymb,amsthm} 
\usepackage{geometry}
\usepackage{setspace}
\usepackage{enumitem}
\usepackage{dsfont}
\usepackage{algorithm,algpseudocode,float}
\usepackage{lipsum}
\usepackage{graphicx}
\usepackage{multirow}
\usepackage{booktabs}
\usepackage{caption}
\usepackage{xr}
\usepackage{longtable,tabularx}
\usepackage{threeparttable}
\usepackage{adjustbox}
\usepackage{siunitx}
\usepackage[mathscr]{euscript}
\usepackage{array}
\usepackage{makecell}
\usepackage{url}

\usepackage{dcolumn}
\newcolumntype{d}{D{.}{.}{-1}} 

\allowdisplaybreaks[1]     
\usepackage{array}
\usepackage{xr}

\makeatletter
\newcommand*{\addFileDependency}[1]{%
  \typeout{(#1)}
  \@addtofilelist{#1}
  \IfFileExists{#1}{}{\typeout{No file #1.}}
}
\makeatother

\newcommand{\bm}{\mathbf}
\newcommand{\bbm}{\boldsymbol}

\newcommand{\op}{\mathrm{op}}
\newcommand{\F}{\mathrm{F}}
\newcommand{\tr}{\mathrm{tr}}

\newcommand{\vect}{\mathrm{vec}}
\newcommand{\argmin}{\mathrm{argmin}}
\newcommand{\argmax}{\mathrm{argmax}}
\newcommand{\rank}{\mathrm{rank}}
\newcommand{\tucrank}{\mathrm{Tucrank}}
\newcommand{\cm}[1]{\mbox{\boldmath$\mathscr{#1}$}}

\makeatletter
\newenvironment{breakablealgorithm}
{
		\begin{center}
			\refstepcounter{algorithm}
			\hrule height.8pt depth0pt \kern2pt
			\renewcommand{\caption}[2][\relax]{
				{\raggedright\textbf{\ALG@name~\thealgorithm} ##2\par}%
				\ifx\relax##1\relax 
				\addcontentsline{loa}{algorithm}{\protect\numberline{\thealgorithm}##2}%
				\else 
				\addcontentsline{loa}{algorithm}{\protect\numberline{\thealgorithm}##1}%
				\fi
				\kern2pt\hrule\kern2pt
			}
		}{
		\kern2pt\hrule\relax
	\end{center}
}
\makeatother

\usepackage[
  colorlinks,
  linkcolor=blue,
  citecolor=blue,
  filecolor=magenta,
  urlcolor=cyan,
  pageanchor=true   
]{hyperref}

\makeatletter
\renewcommand*{\l@section}[2]{%
  \addpenalty{-\@highpenalty}%
  \addvspace{0.7em}%
  \@dottedtocline{1}{1em}{2em}{\bfseries #1}{#2}}
\renewcommand*{\l@subsection}[2]{%
  \@dottedtocline{2}{3em}{2em}{#1}{#2}}
\makeatother

\newtheorem{assumption}{Assumption}
\newtheorem{definition}{Definition}

\newtheorem{theorem}{Theorem}
\newtheorem{lemma}{Lemma}

\title{\vspace{-2cm}Personalized Federated Learning for Tensor Regression}
\author{Kejun Chen, Xianqi Wei and Qianqian Zhu\thanks{Address for correspondence: Qianqian Zhu, School of Statistics and Data Science, Institute of Big Data Research, Shanghai University of Finance and Economics, Shanghai, China. Email: zhu.qianqian@mail.shufe.edu.cn}  \\ \vspace{-0.3cm}\textit{Shanghai University of Finance and Economics}}

\begin{document}

\setlength{\parindent}{16pt}

\maketitle

\begin{abstract}
	The growing availability of tensor-valued data across multiple institutions creates opportunities for collaborative analysis, but also raises challenges related to data privacy, high dimensionality, and client heterogeneity. This paper introduces a personalized federated tensor regression framework that addresses all three simultaneously. Each client's coefficient tensor is decomposed into a globally shared low-Tucker-rank component and a locally sparse deviation, estimated via a two-stage privacy-preserving procedure. We establish finite-sample upper bounds and minimax lower bounds that quantify the privacy-accuracy trade-off, and prove the consistency of the supporting initialization and rank-selection steps. Simulation studies confirm that the federated approach improves estimation and prediction over purely local methods, especially when per-client data are scarce, and an MRI-based ADHD study illustrates its strong performance under real privacy constraints.
\end{abstract}


\newpage



\section{Introduction}\label{sec:Introduction}

Recent advances in data collection and scientific measurement have made rich and high-dimensional datasets increasingly common in modern applications. In many fields, such as neuroimaging and medical imaging, observations are naturally recorded as multiway arrays indexed by spatial locations, imaging modalities, or other structured features. These tensor-valued data contain useful information for scientific discovery and clinical decision-making. For example, imaging data can be used to predict disease severity or other clinical outcomes \citep{hou2015hierarchical,guhaniyogi2017bayesian}. Tensor regression provides a natural framework for modeling the relationship between tensor-valued covariates and responses.

Tensor-valued data contain rich multiway information that is typically lost if the data are vectorized for standard regression analysis \citep{seber2003linear,negahban2009unified}. To fully exploit this structure, tensor regression directly models the relationship between the tensor-valued response and predictors, leading to several well-studied sub-frameworks: scalar-on-tensor regression \citep{zhou2013tensor,li2018tucker}, vector-on-tensor regression \citep{ding2018matrix,pfeiffer2021least}, and tensor-on-tensor regression \citep{lock2018tensor,luo2024tensor}. However, the high dimensionality of tensor data makes unconstrained estimation impractical. Consequently, a central theme in the literature is to impose low-dimensional structural assumptions on the coefficient tensor, with the most common choices being low-rankness, sparsity, or a combination of both \citep{li2017parsimonious,sun2017store,garvesh2019convex}.

In practice, however, tensor-valued data are often collected across multiple sites, and the sample size available at each individual site can be limited. In the federated learning framework, we regard each participating site as a client. For instance, a single hospital or imaging center may only have a small number of subjects with available medical images. When the tensor dimension is large relative to the local sample size, relying solely on site-specific data may be insufficient for stable estimation and reliable prediction. A natural remedy is to borrow information across related sites. Federated learning provides a useful framework for this purpose by allowing clients to collaboratively estimate a model while retaining their raw data locally \citep{konecny2015federated,mcmahan2017communication}. In a standard federated procedure, each client computes model parameters or gradients from its local data and communicates these updates to a central server, which aggregates them to update the global model \citep{kairouz2021advances}.

Although federated learning avoids direct raw-data pooling, the communicated model updates may still reveal sensitive information about local records. Decentralization alone therefore does not provide a formal privacy guarantee \citep{geyer2017differentially,kaissis2020secure}. Several mechanisms have been developed to protect information during distributed computation, including secure multiparty computation \citep{yao1982protocols}, homomorphic encryption \citep{gentry2009fully}, and secure aggregation \citep{bonawitz2017practical}. These cryptographic mechanisms protect local inputs or client updates during computation and aggregation, but generally do not limit the information that may be revealed by the final model or other released outputs. Differential privacy (DP) provides a complementary form of protection by formally and quantitatively limiting the influence of any individual record on released results \citep{dwork2006calibrating,dwork2014algorithmic}. In this paper, we adopt differential privacy by adding carefully calibrated random noise to the communicated updates during federated learning \citep{abadi2016deep,dong2022gaussian}.

Beyond privacy concerns, multi-site tensor data are often highly heterogeneous. Differences in scanners, acquisition protocols, populations, or measurement procedures can introduce systematic site effects into imaging variables and their relationships with clinical outcomes \citep{fortin2018harmonization,yamashita2019harmonization}. Because the coefficient tensor captures this association, such site effects naturally lead to client-specific coefficient tensors, making a single shared tensor overly restrictive in federated settings \citep{smith2017federated,li2020federated}.

Recent studies have extended scalar-on-tensor and tensor-on-tensor regression to federated settings \citep{konyar2024federated,zhang2024federated}. These methods facilitate decentralized tensor analysis under data-sharing constraints, but they do not jointly provide formal privacy protection, heterogeneity accommodation, and theoretical guarantees. To address these gaps, we develop a personalized federated learning framework for tensor regression that integrates formal differential privacy, personalized modeling of client heterogeneity, and finite-sample theoretical guarantees within a unified approach.
Our main contributions are threefold.
\begin{enumerate}
    \item[(i)] We introduce a personalized federated framework for tensor regression that simultaneously addresses high-dimensionality, differential privacy and client heterogeneity. By decomposing each coefficient tensor into a shared low Tucker-rank component and a client-specific weakly sparse component, the framework captures common low-dimensional structure across clients while preserving local deviations. This construction provides a flexible and general approach for modeling tensor-valued responses and/or predictors in distributed environments.

    \item[(ii)] We develop two-stage estimation algorithms for both the federated setting and a single-client benchmark. In the federated regime, the shared component is first learned via differentially private gradient updates; subsequently, each client refines its own sparse deviation locally. The single-client procedure follows an analogous two-stage scheme without privacy protection. To enable stable implementation, we further provide an initialization estimator and a Tucker-rank selection strategy.

    \item[(iii)] We establish theoretical guarantees for both federated and single-client settings. For the federated procedure, we derive non-asymptotic upper bounds and matching minimax lower bounds, confirming rate optimality. Parallel upper and lower bounds are provided for the single-client case, allowing us to characterize when the statistical benefit of cross-client information sharing outweighs the additional cost of privacy protection. In addition, we establish convergence of the initialization step and consistency of the rank selection criterion.
\end{enumerate}

The rest of this paper is organized as follows. Section~\ref{sec:Preliminaries-and-Problem-setup} introduces notation, tensor algebra, and the model setup. Section~\ref{sec:two-stage-est} presents the proposed two-stage estimation procedures for federated and single-client learning. Section~\ref{sec:Implementary_issues} discusses implementation issues, including initialization, Tucker-rank selection, and tuning-parameter selection. Section~\ref{subsec:Theory} provides theoretical analysis for the proposed estimators. Sections~\ref{sec:Simulation} and \ref{sec:RealData} contain simulation studies and an empirical application, respectively. Additional notation, algorithmic details, implementation procedures, auxiliary theoretical results, and proofs are provided in the Supplementary Material. The dataset and computer programs for the analysis in Section \ref{sec:RealData} are available at \url{https://github.com/weishenqi/PFTR}.

\section{Preliminaries and Problem Setup}\label{sec:Preliminaries-and-Problem-setup}

\subsection{Notation and Tensor Algebra}

Vectors and matrices are denoted by boldface lowercase and uppercase letters, respectively (e.g., $\bbm a$, $\bbm \xi$, $\bm A$, $\bbm \Delta$). 
For a matrix $\bm A \in \mathbb R^{m \times n}$, denote its transpose, Frobenius norm, and nuclear norm by $\bm A^\top$, $\|\bm A\|_\F$, and $\|\bm A\|_*$, respectively. For $\nu \in [1,\infty)$, the $\ell_\nu$ norm of $\bm A$ is defined as $\|\bm A\|_\nu =( \sum_{i=1}^m \sum_{j=1}^n |A_{ij}|^\nu)^{1/\nu}$; the same formula for $\nu \in (0,1)$ defines the entrywise $\ell_\nu$ quasi-norm. When $\nu=0$, $\|\bm A\|_0=\#\{(i,j): A_{ij} \neq 0\}$ is the number of nonzero entries of $\bm A$.
For sequences $\{x_n\}$ and $\{y_n\}$, we write $x_n\gtrsim y_n$ if there exists a constant $C>0$ such that $x_n\geq C y_n$ for all $n$, and $x_n\asymp y_n$ if both $x_n\gtrsim y_n$ and $y_n\gtrsim x_n$ hold. For any positive integer $n$, let $[n] = \{1,2,\dots,n\}$.

For mode-$\ell$ tensor $\cm{M}\in\mathbb R^{p_1\times p_2\times\cdots\times p_\ell}$ and $S \subset [\ell]$, its multi-mode matricization $\cm{M}_{[S]}$ is a $\prod_{j\in S}p_j\times \prod_{j\notin S}p_j$ matrix with the $(i,j)$-th entry mapped from the $(i_1,i_2,\ldots,i_\ell)$-th entry of $\cm{M}$, where
\[
  i = 1 + \sum_{s\in S}(i_s-1)I_s \ \text{and} \ j = 1 + \sum_{s\notin S}(i_s-1)J_s \ \text{with}\ I_s = \prod_{\substack{m\in S\\m<s}}p_m \ \text{and} \ J_s = \prod_{\substack{m\notin S\\m<s}}p_m.
\]
For notational simplicity, write $\cm{M}_{[s]}=\cm{M}_{[\{s\}]}$ for the mode-$s$ unfolding of $\cm{M}$. The Tucker rank of $\cm{M}$ is defined as $\tucrank(\cm{M})=(r_1,\ldots,r_\ell)$ with $r_s=\mathrm{rank}(\cm{M}_{[s]})$ for $s\in[\ell]$. 
Under the above indexing convention, it holds that $\cm{M}_{[S]} = \cm{M}_{[S^c]}^\top$, where $S^c=[\ell]\setminus S$.
For any two tensors $\cm{X}\in\mathbb R^{q_1\times\cdots\times q_m\times p_1\times\cdots\times p_d}$ and $\cm{Y}\in\mathbb R^{p_1\times\cdots\times p_d}$, their generalized inner product $\langle \cm{X},\cm{Y}\rangle$ is defined as a mode-$m$ tensor in $\mathbb R^{q_1\times\cdots\times q_m}$ with entries $\langle \cm{X},\cm{Y}\rangle_{i_1,\ldots,i_m} = \sum_{\ell_1=1}^{p_1}\cdots\sum_{\ell_d=1}^{p_d} \cm{X}_{i_1,\ldots,i_m,\ell_1,\ldots,\ell_d}\,\cm{Y}_{\ell_1,\ldots,\ell_d}$ for $(i_1,\ldots,i_m)\in[q_1]\times \cdots \times [q_m]$.
Similarly, for any two tensors $\cm{X}\in\mathbb R^{q_1\times\cdots\times q_m}$ and $\cm{Y}\in\mathbb R^{p_1\times\cdots\times p_d}$, their tensor outer product $\cm X\circ\cm Y$ is defined as a mode-$(d+m)$ tensor in $\mathbb R^{q_1\times\cdots\times q_m\times p_1\times\cdots\times p_d}$ with entries $(\cm X\circ\cm Y)_{i_1,\ldots,i_m,\ell_1,\ldots,\ell_d} = \cm X_{i_1,\ldots,i_m}\,\cm Y_{\ell_1,\ldots,\ell_d}$ for $(i_1,\ldots,i_m)\in[q_1]\times\cdots\times[q_m]$ and $(\ell_1,\ldots,\ell_d)\in[p_1]\times\cdots\times[p_d]$.
For any tensor $\cm{M}$, its Frobenius norm is defined as $\|\cm{M}\|_{\F} = \sqrt{\langle \cm{M},\cm{M}\rangle}$. For any $\nu\in[0,\infty)$, the entrywise $\ell_\nu$ norm (or quasi-norm when $\nu\in[0,1)$) of $\cm M$ is identical to the $\ell_\nu$ norm of any matricization of $\cm M$. That is, $\|\cm M\|_\nu = \|\cm M_{[S]}\|_\nu$ for any nonempty subset $S\subset[\ell]$.

\subsection{Model Settings}\label{sec:Setting}
Suppose there are $K$ clients, indexed by $k\in[K]$. Client $k$ holds a local dataset $\mathsf D_k=\{(\cm{X}_{k,i},\cm{Y}_{k,i})\}_{i=1}^{n_k}$ with sample size $n_k$, where $\cm{X}_{k,i}\in\mathbb R^{p_1\times\cdots\times p_d}$ is a mode-$d$ covariate tensor and $\cm{Y}_{k,i}\in\mathbb R^{q_1\times\cdots\times q_m}$ is a mode-$m$ response tensor. For each fixed client $k\in[K]$, we assume that $(\cm X_{k,i},\cm Y_{k,i})$ are independent and identically distributed ($i.i.d.$) realizations generated from the following client-specific tensor regression model: 
\begin{align}\label{eq:model}
  \cm{Y}_{k,i} = \bigl\langle \cm{A}_k,\cm{X}_{k,i}\bigr\rangle + \cm{E}_{k,i},\quad k\in[K],\ \ i\in[n_k], 
\end{align}
where $\cm{A}_k\in \mathbb{R}^{q_1\times\cdots\times q_m\times p_1\times\cdots\times p_d}$ is the unknown coefficient tensor for client $k$, and $\cm{E}_{k,i}\in\mathbb R^{q_1\times\cdots\times q_m}$ is a mean-zero random noise tensor, which is $i.i.d.$ across indices $i$. Furthermore, the noise sequences $\{\cm{E}_{k,i}\}_{i=1}^{n_k}$ are assumed to be independent across different clients $k\in[K]$.

To characterize the structural similarity among heterogeneous clients, we assume that each coefficient tensor $\cm{A}_k$ can be decomposed as follows:
\begin{align}\label{eq:decomposition}
  \cm{A}_k = \cm{A}_0 + \cm{B}_k, \quad k\in[K],
\end{align}
where $\cm{A}_0$ is a global component shared by all clients, and $\cm{B}_k$ is a client-specific deviation. To address the high dimensionality of $\cm{A}_k$, we further assume that $\cm{A}_0$ admits a low Tucker-rank structure and each $\cm{B}_k$ is sparse. 
It is worth noting that $\cm A_0$ and $\cm B_k$ are not uniquely identifiable. To alleviate this issue, we introduce the weak identifiability condition in Assumption~\ref{assump:weak-identifiability} to ensure the decomposition remains meaningful without requiring strict uniqueness.

The model in \eqref{eq:model} defines a unified regression framework whose specific form is determined by the mode dimensions of the predictor and response tensors. Assuming the nondegenerate setting where $p_j>1$ for all $j\in[d]$ and $q_j>1$ for all $j\in[m]$, the framework subsumes several classical regression models. In particular, it reduces to regression with multivariate responses and covariates \citep{seber2003linear,montgomery2021introduction} when $(d,m)=(1,1)$, to regression with multivariate responses and matrix-valued covariates \citep{ding2018matrix,pfeiffer2021least} when $(d,m)=(2,1)$, and to matrix-on-matrix regression \citep{lock2018tensor} when $(d,m)=(2,2)$.
The framework also naturally accommodates scalar response scenarios. Specifically, it includes scalar-on-vector regression, which corresponds to the ordinary linear model \citep{seber2003linear,montgomery2021introduction}; scalar-on-matrix regression, which reduces to trace regression and is fundamental to low-rank matrix estimation and matrix completion \citep{koltchinskii2011nuclear,negahban2012restricted}; and higher-order scalar-on-tensor regression, a model widely adopted in neuroimaging and medical imaging studies \citep{zhou2013tensor,li2018tucker}.

\section{Two-Stage Estimation Procedure}\label{sec:two-stage-est}

To fit model \eqref{eq:model} with the decomposition \eqref{eq:decomposition}, we proceed in two stages. First, we estimate the shared low Tucker-rank tensor $\cm A_0$ as a global representation. 
Given this shared component, we then estimate the sparse deviation tensor $\cm{B}_k$ for each client and obtain an estimate for the personalized coefficient tensor $\cm{A}_k$.
Sections \ref{sec:fed-learning} and \ref{sec:single-client-learning} describe how this two-stage strategy is tailored to federated and single-client settings, respectively.

\subsection{Federated Learning}\label{sec:fed-learning}
For federated learning, Stage-I learns the shared component via differentially private federated procedure, whereas Stage-II performs local refinement using the shared estimate from Stage-I. 

\textbf{Stage-I: Differentially Private Representation Learning.} To separate the shared component $\cm A_0$ from client-specific deviations $\cm B_k$ in the decomposition \eqref{eq:decomposition}, we temporarily treat $\cm B_k$ as nuisance parameters. 
Under \eqref{eq:decomposition}, the original model \eqref{eq:model} can be rewritten as
\begin{align}\label{eq:reformulated-model}
  \cm Y_{k,i} = \langle \cm A_0,\cm X_{k,i}\rangle + \cm E^{\dagger}_{k,i},
  \ \text{with} \
  \cm E^{\dagger}_{k,i} = \langle \cm B_k,\cm X_{k,i}\rangle + \cm E_{k,i}.
\end{align}
Based on \eqref{eq:reformulated-model}, we define the local least-squares loss
$\ell_k(\cm A) = (2n_k)^{-1}\sum_{i=1}^{n_k}\|\cm Y_{k,i}-\langle \cm A,\cm X_{k,i}\rangle\|_{\F}^2$.
For privacy protection and sensitivity control, we construct a truncated version of the local gradient by clipping both the covariate tensors and the residual tensors. For client $k$ at iteration $t$, the truncated gradient of $\ell_k(\cm A)$ is given by
\[
\cm G_{\mathscr A,k}^{\vee}(\cm A) = \frac{1}{n_k}\sum_{i=1}^{n_k}\cm R_{k,i}^{\vee}(\cm A;\tau_{\mathscr E,k}^{(t)})\circ \cm X_{k,i}^{\vee}(\tau_{\mathscr X,k}),
\]
where the clipped covariate tensor and the clipped residual tensor are respectively defined by
\[
\cm X_{k,i}^{\vee}(\tau_{\mathscr X,k}) = \frac{\tau_{\mathscr X,k}\wedge\|\cm X_{k,i}\|_\F}{\|\cm X_{k,i}\|_\F}\cm X_{k,i}, \;\;\text{and}\;\;
  \cm R_{k,i}^{\vee}(\cm A;\tau_{\mathscr E,k}^{(t)})
  =
  \frac{\tau_{\mathscr E,k}^{(t)}\wedge\|\langle\cm A,\cm X_{k,i}\rangle-\cm Y_{k,i}\|_\F}
  {\|\langle\cm A,\cm X_{k,i}\rangle-\cm Y_{k,i}\|_\F}
  (\langle\cm A,\cm X_{k,i}\rangle-\cm Y_{k,i}).
\]
Here, the truncation levels $\tau_{\mathscr X,k}>0$ and $\tau_{\mathscr E,k}^{(t)}>0$ control the maximal allowed Frobenius norm of the covariate tensor and the residual tensor, respectively, and are crucial in calibrating the sensitivity of the gradient for subsequent privacy perturbation; see Section \ref{sec:tuning-param-selection} for further discussion.

Denote by $\cm A_0^{(t)}$ the estimate of $\cm A_0$ at iteration $t$. The initialization $\cm A_0^{(0)}$ is specified in Section~\ref{sec:alg-ini}.
At each iteration $t$, client $k$ computes the truncated local gradient $\cm G_{\mathscr A,k}^{\vee}(\cm A_0^{(t)})$ and adds an $i.i.d.$ Gaussian noise tensor $\cm W_{\mathscr A,k}^{(t)}$ to protect the local data. The noisy truncated gradient is then projected onto the tangent space $\mathcal T_{\bbm r}(\cm A_0^{(t)})$ of the Tucker-rank manifold, which keeps the update direction compatible with the prescribed low-rank structure:
\begin{align}\label{eq:private-gradient}
  \widetilde{\cm G}_{\mathscr A,k}^{(t)}
  =
  \mathcal P_{\mathcal T_{\bbm r}(\boldsymbol{\mathscr{A}}_0^{(t)})}
  \bigl(
  \cm G_{\mathscr A,k}^{\vee}(\cm A_0^{(t)})+\cm W_{\mathscr A,k}^{(t)}
  \bigr),
  \ \text{with} \ 
  [\cm W_{\mathscr A,k}^{(t)}]_{i_1,\ldots,i_{d+m}}\overset{i.i.d.}{\sim}\mathscr N\bigl(0,(\sigma_k^{(t)})^2\bigr).
\end{align}
Here, the projection operator $\mathcal P_{\mathcal T_{\bbm r}(\boldsymbol{\mathscr{A}}_0^{(t)})}(\cdot)$ denotes the orthogonal projection onto this tangent space; its explicit form is provided in Section~\ref{sec:basic-notation} of the Supplementary Material. The noise variance $(\sigma_k^{(t)})^2$ is calibrated according to the desired privacy parameters $(\varepsilon,\delta)$; see Theorem~\ref{thm:federated_representation_error} for details. The Tucker rank $\bbm r$ of $\cm A_0$ is a key structural parameter, and its selection is discussed in Section~\ref{sec:tucker-rank-selection}.
The server then collects $\widetilde{\cm G}_{\mathscr A,k}^{(t)}$ from all clients and aggregates them using sample-size weights $n_k/n$ with $n=\sum_{k=1}^K n_k$, which ensures that clients with larger local datasets contribute proportionally more to the global update.
Based on the aggregated gradient, the server performs a gradient descent step followed by a retraction onto the Tucker-rank-$\bbm r$ manifold:
\begin{align}\label{eq:rgd-update}
  \cm A_0^{(t+1)}
  =
  \mathcal R_{\bbm r}\Bigl(
  \cm A_0^{(t)}-\eta_{\mathscr A}\sum_{k=1}^K\frac{n_k}{n}\widetilde{\cm G}_{\mathscr A,k}^{(t)}
  \Bigr),
\end{align}
with step size $\eta_{\mathscr A}>0$. The retraction operator $\mathcal R_{\bbm r}(\cdot)$ ensures that the updated iterate remains on the Tucker-rank-$\bbm r$ manifold, thereby enabling efficient Riemannian-style updates. Following \citet{luo2024tensor}, the operator $\mathcal R_{\bbm r}(\cdot)$ can be implemented by either truncated high-order singular value decomposition (T-HOSVD) or sequentially truncated high-order singular value decomposition (ST-HOSVD); see Algorithms~\ref{alg:thosvd} and \ref{alg:sthosvd} in Section~\ref{sec:THOSVD_STHOSVD} of the Supplementary Material for details.
After $T_g$ iterations, we set $\widehat{\cm A}_0=\cm A_0^{(T_g)}$ as the estimated shared tensor.

\textbf{Stage-II: Personalized Refinement.} Given the shared low Tucker-rank tensor estimate $\widehat{\cm A}_0$ from Stage-I, Stage-II recovers client-specific deviations $\cm B_k$ and the corresponding personalized coefficient tensors $\cm A_k$. The refinement is performed on each client using only its own local data. For each client $k\in[K]$, a natural estimator for the sparse deviation $\cm B_k$ is obtained by solving
\begin{align}\label{eq:stage2-opt}
  \widehat{\cm B}_k^{\mathrm{opt}}
  \in
  \argmin_{\boldsymbol{\mathscr B}_k}
  \left\{
  \frac{1}{2n_k}\sum_{i=1}^{n_k}
  \bigl\|\cm Y_{k,i}-\langle \widehat{\cm A}_0+\cm B_k,\cm X_{k,i}\rangle\bigr\|_{\F}^2
  +\omega_k\|\cm B_k\|_1
  \right\},
  \ \text{for} \ k\in[K].
\end{align}
Here, $\omega_k>0$ are client-specific regularization parameters that control the sparsity level. Note that the optimization problem in \eqref{eq:stage2-opt} is convex and can therefore be solved efficiently by first-order methods, we adopt a FISTA-based algorithm to compute the estimator.
Conditioned on $\widehat{\cm A}_0$, define the empirical loss for client $k\in[K]$ as $\mathcal L_k(\cm B_k) =(2n_k)^{-1}\sum_{i=1}^{n_k}
\bigl\|\cm Y_{k,i}-\langle \widehat{\cm A}_0+\cm B_k,\cm X_{k,i}\rangle\bigr\|_{\F}^2$, with gradient tensor
\[
  \cm G_{\mathscr B,k}(\cm B_k)
  =
  \frac{1}{n_k}\sum_{i=1}^{n_k}
  \Bigl(
  \langle \widehat{\cm A}_0+\cm B_k,\cm X_{k,i}\rangle-\cm Y_{k,i}
  \Bigr)\circ \cm X_{k,i}.
\]
Following the standard FISTA algorithm \citep{beck2009fast}, each iteration of the optimization proceeds in two steps. First, starting from the extrapolated point $\cm U_k^{(t)}$, we take a gradient step with respect to $\mathcal L_k(\cm B_k)$ and apply the entrywise soft-thresholding operator to enforce sparsity:
\begin{align}\label{eq:personalized-estimator}
  \cm B_k^{(t+1)}
  =
  \operatorname{Soft}_{\eta_{\mathscr B,k}\omega_k}
  \Bigl(
  \cm U_k^{(t)}-\eta_{\mathscr B,k}\cm G_{\mathscr B,k}(\cm U_k^{(t)})
  \Bigr),
\end{align}
where $[\operatorname{Soft}_{\tau}(\cm T)]_{i_1,\ldots,i_{d+m}} = \operatorname{sign}([\cm T]_{i_1,\ldots,i_{d+m}})(|[\cm T]_{i_1,\ldots,i_{d+m}}|-\tau \vee 0)$. The momentum parameter and the extrapolated iterate are then updated following the standard FISTA extrapolation rule:
\begin{align}\label{eq:momentum-param-extra-iter}
  q_{k,t+1}
  =
  \frac{1+\sqrt{1+4q_{k,t}^2}}{2},
  \ \text{and} \ 
  \cm U_k^{(t+1)}
  =
  \cm B_k^{(t+1)}
  + 
  \frac{q_{k,t}-1}{q_{k,t+1}}
  \bigl(
  \cm B_k^{(t+1)}-\cm B_k^{(t)}
  \bigr).
\end{align}
After $T_l^{(k)}$ local iterations, the personalized deviation estimator is given by $\widehat{\cm B}_k=\cm B_k^{(T_l^{(k)})}$, and the personalized coefficient tensor is constructed as $\widehat{\cm A}_k = \widehat{\cm A}_0+\widehat{\cm B}_k$ for $k\in[K]$. 

The proposed two-stage estimation procedure is summarized in Algorithm~\ref{alg:fed_two_stage}. The detailed configuration of the inputs is provided in Section \ref{sec:Implementary_issues}.

\vspace{4mm}

\begin{breakablealgorithm}
\caption{Federated Learning for Personalized Tensor Regression}
\label{alg:fed_two_stage}
\begingroup
\begin{algorithmic}[1]
  \State \textbf{Input:} Local datasets $\mathsf D_k=\{(\cm X_{k,i},\cm Y_{k,i})\}_{i=1}^{n_k}$ for $k\in[K]$; global and local iteration numbers $T_g$ and $\{T_l^{(k)}\}_{k=1}^K$; Tucker rank $\bbm r$; Stage-I step size $\eta_{\mathscr A}$ and Stage-II step sizes $\{\eta_{\mathscr B,k}\}_{k=1}^K$; Gaussian noise levels $\{\sigma_k^{(t)}\}$; sparsity penalties $\{\omega_k\}_{k=1}^K$; initialization $\cm A_0^{(0)}$.

  \State \textbf{Stage-I:} For $t=0,\cdots,T_g-1$, each client computes, privatizes, and projects its truncated gradient as in \eqref{eq:private-gradient}.
  \State The server aggregates the privatized gradients and updates $\cm A_0^{(t+1)}$ by \eqref{eq:rgd-update}.
  \State Set $\widehat{\cm A}_0\gets\cm A_0^{(T_g)}$.

  \State \textbf{Stage-II:} For each client $k\in[K]$ in parallel, initialize $\cm B_k^{(0)}=\cm U_k^{(0)}=\bm 0$ and $q_{k,0}=1$.
  \State For $t=0,\cdots,T_l^{(k)}-1$, update $\cm B_k^{(t+1)}$ by \eqref{eq:personalized-estimator}, and update $q_{k,t+1}$ and $\cm U_k^{(t+1)}$ by \eqref{eq:momentum-param-extra-iter}.
  \State Set $\widehat{\cm B}_k\gets\cm B_k^{(T_l^{(k)})}$ and $\widehat{\cm A}_k\gets\widehat{\cm A}_0+\widehat{\cm B}_k$ for all $k\in[K]$.

  \State \textbf{Output:} $\widehat{\cm A}_0$ and $\{\widehat{\cm A}_k\}_{k=1}^K$.
\end{algorithmic}
\endgroup
\end{breakablealgorithm}

\subsection{Single-Client Learning}\label{sec:single-client-learning}
As a benchmark, we also consider a single-client learning strategy that relies solely on local data. In this setting, we adopt the same two-stage procedure, with the key differences that all computations are performed locally and no privacy-preserving mechanisms are required.

\textbf{Stage-I: Local Representation Learning.} For a fixed client $k\in[K]$, Stage-I estimates the shared low Tucker-rank component using only its own dataset $\mathsf D_k$. As in the federated setting, treating $\cm B_k$ in \eqref{eq:decomposition} as a nuisance parameter leads to the local least-squares loss $\ell_k(\cm A) = (2n_k)^{-1}\sum_{i=1}^{n_k}\|\cm Y_{k,i}-\langle \cm A,\cm X_{k,i}\rangle\|_{\F}^2$,
and its gradient $\cm G_{\mathscr A,k}(\cm A) = n_k^{-1}\sum_{i=1}^{n_k}(\langle \cm A,\cm X_{k,i}\rangle-\cm Y_{k,i})\circ \cm X_{k,i}$ for each client $k\in[K]$.
Starting from an initial value $\cm A_{0,k}^{\mathrm{loc},(0)}$, whose selection is specified in Section~\ref{sec:alg-ini}, client $k$ updates the iterate by Riemannian gradient descent over the Tucker-rank-$\bbm r$ manifold,
\begin{align}\label{eq:local-rgd-update}
  \cm A_{0,k}^{\mathrm{loc},(t+1)}
  =
  \mathcal R_{\bbm r}\Bigl(
  \cm A_{0,k}^{\mathrm{loc},(t)}
  -
  \eta_{\mathscr A}\,\mathcal P_{\mathcal T_{\bbm r}(\boldsymbol{\mathscr A}_{0,k}^{\mathrm{loc},(t)})}
  \bigl(\cm G_{\mathscr A,k}^{\mathrm{loc},(t)}\bigr)
  \Bigr),
\end{align}
and sets $\widehat{\cm A}_{0,k}^{\mathrm{loc}}=\cm A_{0,k}^{\mathrm{loc},(T_g^{(k)})}$ after $T_g^{(k)}$ iterations.

\textbf{Stage-II: Personalized Refinement.}  
Given the locally estimated shared component $\widehat{\cm A}_{0,k}^{\mathrm{loc}}$, Stage-II proceeds exactly as its federated counterpart. Specifically, one replaces the federated shared component $\widehat{\cm A}_{0}$ in \eqref{eq:personalized-estimator} with $\widehat{\cm A}_{0,k}^{\mathrm{loc}}$ and solves for the sparse deviation $\widehat{\cm B}_k^{\mathrm{loc}}$, yielding the personalized estimator $\widehat{\cm A}_k^{\mathrm{loc}} = \widehat{\cm A}_{0,k}^{\mathrm{loc}} + \widehat{\cm B}_k^{\mathrm{loc}}$. Algorithm~\ref{alg:single_two_stage} in Section~\ref{sec:two-stage-single-est} of the Supplementary Material summarizes the complete two-stage single-client procedure.

\section{Implementary Issues}\label{sec:Implementary_issues}


\subsection{Algorithm Initialization}\label{sec:alg-ini}
Both the federated and single-client two-stage procedures require a suitable initial value for the low Tucker-rank component $\cm A_0$. We begin by constructing a preliminary one-stage estimator for each client using its local dataset.

Recall that each $\cm A_k$ in \eqref{eq:decomposition} decomposes into a low Tucker-rank shared part and a weakly sparse deviation. A natural approach is to estimate both components jointly via a penalized least-squares formulation. While the sum of mode-wise nuclear norms, $\sum_{s=1}^{d+m}\|\cm A_{[s]}\|_*$, is a classical convex surrogate for the Tucker rank, it is known to perform poorly in recovering the underlying Tucker-rank structure \citep{mu2014square,garvesh2019convex}. To circumvent this issue, we exploit an equivalent vectorized representation of model \eqref{eq:model}:
\begin{align}\label{eq:model-vectorized}
  \vect(\cm Y_{k,i}) = \left((\cm A_0)_{[S_X]} + (\cm B_k)_{{[S_X]}}\right) \vect(\cm X_{k,i}) + \vect(\cm E_{k,i}),
\end{align}
where $S_X=[d+m]\setminus [m]$ denotes the set of modes corresponding to the predictor tensors $\{\cm X_{k,i}\}$.
Under this representation, $(\cm A_0)_{[S_X]}$ is of low rank because $\cm A_0$ has low Tucker rank $\bbm r$ and $\rank((\cm A_0)_{[S_X]}) \leq \min(r_{\bbm q},r_{\bbm p})$ with $r_{\bbm q}=\prod_{\ell=1}^{m}r_\ell$ and $r_{\bbm p}=\prod_{\ell=1}^{d}r_{m+\ell}$ (see Lemma~4 in \cite{mu2014square}). Meanwhile, $(\cm B_k)_{[S_X]}$ inherits the weak sparsity of $\cm B_k$, as matricization preserves sparsity patterns.
Consequently, it is natural to regularize the nuclear norm of $(\cm A_0)_{[S_X]}$ and the $\ell_1$ norm of $\cm B_k$ (i.e., $(\cm B_k)_{{[S_X]}}$). This leads to the preliminary one-stage estimator for client $k$:
\begin{align}\label{eq:alg-ini}
  (\widetilde{\cm A}_{0,k},\widetilde{\cm B}_k)
  =
  \argmin_{\boldsymbol{\mathscr A}_{0,k},\boldsymbol{\mathscr B}_k}
  \Biggl\{
  \frac{1}{2n_k}\sum_{i=1}^{n_k}
  \bigl\|\cm Y_{k,i}-\langle \cm A_{0,k}+\cm B_k,\cm X_{k,i}\rangle\bigr\|_{\F}^2
  + \lambda_k\|(\cm A_{0,k})_{[S_X]}\|_*
  + \varpi_k\|\cm B_k\|_1
  \Biggr\},
\end{align}
subject to $\|(\cm B_{k})_{[S_X]}\|_{\op}\leq \zeta$, where $\lambda_k$ and $\varpi_k$ are tuning parameters that control the respective regularization strengths. The operator-norm bound imposes a mild weak-identifiability condition that prevents the sparse deviation from spuriously acting as a low-rank component \citep{agarwal2012noisy}.
This convex optimization problem can be solved efficiently via the alternating direction method of multipliers \citep{neal2011distributed}; see Algorithm \ref{alg:admm-single} in Section~\ref{sec:ADMM} of the Supplementary Material for details.

For the single-client setting, the $\widetilde{\cm A}_{0,k}$ directly serves as the local initialization, i.e., $\cm A_{0,k}^{\mathrm{loc},(0)}=\widetilde{\cm A}_{0,k}$.
In the federated setting, we initialize the shared component $\cm A_0$ with the estimate from the client having the largest sample size, i.e., $\cm A_0^{(0)}=\widetilde{\cm A}_{0,k^\star}$ with $k^\star=\argmax_{k\in[K]} n_k$.

\subsection{Tucker-Rank Selection}\label{sec:tucker-rank-selection}
The Tucker rank of $\cm A_0\in\mathbb R^{q_1\times\cdots\times q_m\times p_1\times\cdots\times p_d}$ is $\bbm r=(r_1,\ldots,r_{d+m})$ with $r_s=\rank((\cm A_0)_{[s]})$ and $(\cm A_0)_{[s]}$ being the mode-$s$ matricization for $s\in[d+m]$. Hence, selecting the Tucker rank of $\cm A_0$ reduces to estimating the ranks of its mode-wise matricizations. We adopt the ridge-type ratio criterion \citep{xia2015consistently} to determine these matrix ranks.

For the single-client case, the Tucker rank $\bbm r$ is selected based on the preliminary one-stage estimator $\widetilde{\cm A}_{0,k}$ in Section~\ref{sec:alg-ini}. Let $\widetilde{\sigma}_{k,s,r}$ denote the $r$-th singular value of the mode-$s$ matricization of $\widetilde{\cm A}_{0,k}$.
Using the local dataset $\mathsf D_k$ of client $k$, we estimate the rank along mode $s$ (i.e. $r_s$) by the following ridge-type ratio criterion 
\begin{align}\label{eq:Ridge-type estimator}
  \widehat r_{k,s}
  =
  \argmin_{1\leq r\leq \bar r_s-1}
  \frac{\widetilde{\sigma}_{k,s,r+1}+c_s(\bbm p,\bbm q,n_k)}
       {\widetilde{\sigma}_{k,s,r}+c_s(\bbm p,\bbm q,n_k)},
  \ \text{for}\  s\in[d+m],
\end{align}
where $\bar r_s$ is a prescribed upper bound for $r_s$, and $c_s(\bbm p,\bbm q,n_k)$ is a ridge-type penalty that stabilizes the ratio, with $\bbm p=(p_1,\ldots,p_d)$ and $\bbm q=(q_1,\ldots,q_m)$ for simplicity. Motivated by the consistency conditions for ridge-type ratio estimators in Theorem~\ref{thm:rank_selection}, we recommend setting $c_s(\bbm p,\bbm q,n_k)=10^{-2}\sqrt{pq/n_k}$, where $p=\prod_{\ell=1}^d p_\ell$ and $q=\prod_{\ell=1}^m q_\ell$. The resulting Tucker-rank estimator for client $k$ is given by $\widehat{\bbm r}_k^{\mathrm{loc}}=(\widehat r_{k,1},\ldots,\widehat r_{k,d+m})$.

In the federated setting, we apply \eqref{eq:Ridge-type estimator} to each preliminary estimator $\widetilde{\cm A}_{0,k}$, thereby obtaining client-specific candidate ranks $\{\widehat r_{k,s}\}_{k=1}^K$ for each mode $s\in[d+m]$. To robustify the rank selection against unstable local estimates, we aggregate these candidates by taking the mode across clients: $\widehat r_s^{\mathrm{fed}} = \operatorname{mode}(\{\widehat r_{1,s},\ldots,\widehat r_{K,s}\})$ for $s\in[d+m]$.
The Tucker-rank estimate for the federated procedure is given by $\widehat{\bbm r}^{\mathrm{fed}} = (\widehat r_1^{\mathrm{fed}},\ldots,\widehat r_{d+m}^{\mathrm{fed}})$.

\subsection{Tuning Parameter Selection}\label{sec:tuning-param-selection}
The proposed federated two-stage procedure needs to choose the tuning parameters including the step size $\eta_{\mathscr A}$ in Stage-I, and the regularization parameters $\{\omega_k\}_{k=1}^K$ in Stage-II.
Before selecting these parameters, we first fix hyper-parameters such as the truncation levels $\tau_{\mathscr X,k}$ and $\tau_{\mathscr E,k}^{(t)}$, the standard deviation $\sigma_k^{(t)}$ of the Gaussian noise for privacy, the iteration numbers $T_g, T_l^{(k)}$, and the step sizes $\{\eta_{\mathscr B,k}\}_{k=1}^K$ in Stage-II, based on theoretical guidance and empirical stability. 

For the truncation levels, we set $\tau_{\mathscr X,k}$ to the maximum of $\{\|\cm X_{k,i}\|_{\F}\}_{i=1}^{n_k}$ for each client $k$. At global iteration $t$, we similarly set $\tau_{\mathscr E,k}^{(t)}$ to the maximum of the current residual norms $\{\|\langle\cm A_0^{(t)},\cm X_{k,i}\rangle-\cm Y_{k,i}\|_{\F}\}_{i=1}^{n_k}$. 
We also consider the empirical $0.90$th, $0.95$th, and $0.99$th quantiles as alternative truncation levels and find that the results are relatively insensitive to these choices. 
We therefore use the empirical maxima to avoid information loss.
To ensure that the Stage-I procedure satisfies the desired $(\varepsilon,\delta)$-DP guarantee, we set the standard deviation of the Gaussian noise as
\begin{align}\label{eq:noise-scale}
  \sigma_k^{(t)} = \frac{2\tau_{\mathscr E,k}^{(t)}\tau_{\mathscr X,k}}{n_k}\frac{T_g\sqrt{2\log(1.25T_g/\delta)}}{\varepsilon}.
\end{align}
Guided by Theorem~\ref{thm:federated_representation_error}, we set the number of global iterations to $T_g=\lceil 10\log(n)\rceil$ with $n=\sum_{k=1}^{K}n_k$, beyond which further iterations provide negligible improvement. Since FISTA iterations typically stabilize within about 100 steps, we set the number of local iterations to $T_l^{(k)}=200$.
Following standard practice for FISTA \citep{beck2009fast}, the step size $\eta_{\mathscr B,k}$ for client $k$ is taken as $\eta_{\mathscr B,k}=(2n_k^{-1}\|\sum_{i=1}^{n_k}\vect(\cm X_{k,i})\vect^\top(\cm X_{k,i})\|_{\op})^{-1}$.

With the truncation levels, privacy noise scale, iteration numbers, and local step sizes fixed, we select the step size $\eta_{\mathscr A}$ and regularization parameters $\{\omega_k\}_{k=1}^K$ by $V$-fold cross validation. For each client $k\in[K]$, the local dataset $\mathsf D_k$ is randomly split into training and validation folds $\mathsf D_{k,v}^{\mathrm{tr}}$ and $\mathsf D_{k,v}^{\mathrm{val}}$.
Given each candidate combination $(\eta_{\mathscr A},\{\omega_k\}_{k=1}^K)$ over a grid ranging from $0.001$ to $0.05$, we employ Algorithm \ref{alg:fed_two_stage} using the training sets $\{\mathsf D_{k,v}^{\mathrm{tr}}\}_{k=1}^K$ and evaluate the prediction performance via validation sets $\{\mathsf D_{k,v}^{\mathrm{val}}\}_{k=1}^K$. 
Specifically, we run the federated Stage-I to obtain the shared estimator $\widehat{\cm A}_{0}^{(-v)}(\eta_{\mathscr A})$ and perform the personalized Stage-II to obtain the client-specific deviation estimators $\{\widehat{\cm B}_{k}^{(-v)}(\widehat{\cm A}_{0}^{(-v)}(\eta_{\mathscr A}),\omega_k)\}_{k=1}^K$. 
Then the performance of each candidate is evaluated by the average cross-validated squared Frobenius prediction error $\mathrm{Error}_{\mathrm{cv}}=V^{-1}\sum_{v=1}^V\mathrm{Error}_{\mathrm{cv}}^{(v)}$, where
\[
  \mathrm{Error}_{\mathrm{cv}}^{(v)}
  =
  \frac{1}{\sum_{k=1}^K |\mathsf D_{k,v}^{\mathrm{val}}|}
  \sum_{k=1}^K
  \sum_{(\boldsymbol{\mathscr X}_{k,i},\boldsymbol{\mathscr Y}_{k,i})\in \mathsf D_{k,v}^{\mathrm{val}}}
  \left\|
  \cm Y_{k,i}
  -
  \left\langle
  \widehat{\cm A}_{0}^{(-v)}(\eta_{\mathscr A})+
  \widehat{\cm B}_{k}^{(-v)}(\widehat{\cm A}_{0}^{(-v)}(\eta_{\mathscr A}),\omega_k),
  \cm X_{k,i}
  \right\rangle
  \right\|_{\F}^2.
\]
The optimal tuning parameters are those that minimize $\mathrm{Error}_{\mathrm{cv}}$ over the prescribed grid. In our numerical studies, we adopt $V=5$ as a practical trade-off between bias, variance, and computational cost \citep{james2013introduction}.

The tuning of step sizes and regularization parameters for the single-client procedure in Section~\ref{sec:single-client-learning}, as well as the tuning involved in initialization in Section~\ref{sec:alg-ini}, follows the same validation-based principle. Hence detailed descriptions are deferred to Section~\ref{sec:tuning-single} of the Supplementary Material.

\section{Theoretical Analysis}\label{subsec:Theory}

We first specify the structural parameter space induced by the decomposition \eqref{eq:decomposition}. The shared component $\cm A_0^*$ is assumed to have Tucker rank $\bbm r=(r_1,\cdots,r_{d+m})$ with all $r_s$ fixed and independent of the tensor dimensions. The client-specific sparse deviations satisfy $\cm B_k^*\in\mathbb B_\nu$ for all $k\in[K]$, where, for $\nu\in[0,1)$,
\[
  \mathbb B_\nu = \left\{\cm B \in \mathbb{R}^{q_1\times\cdots\times q_m\times p_1\times\cdots\times p_d}: \mathds 1_{\{\nu=0\}}\|\cm B\|_0/s_0 + \mathds 1_{\{0<\nu<1\}}\|\cm B\|_\nu^\nu/s_\nu\leq 1\right\}.
\]
Under this sparse condition on the client-specific deviations, we define $h_{\mathscr B,k} = \|\cm B_k^*\|_{\infty}^{1-\nu/2}s_\nu^{1/2}$ for $k\in[K]$, which bounds the Frobenius norm via $\|\cm B_k^*\|_{\F}\leq h_{\mathscr B,k}$ for each client $k$. 
The largest deviation bound across clients is denoted by $\bar h_{\mathscr B}=\max_{k\in[K]}h_{\mathscr B,k}$.

The decomposition \eqref{eq:decomposition} is not automatically identifiable, because a sparse tensor may also have low-rank structure, and conversely a low-rank tensor may contain sparse patterns. Motivated by the equivalent vectorized representation \eqref{eq:model-vectorized}, we introduce the following weak-identifiability condition to resolve this ambiguity.

\begin{assumption}[Weak identifiability]\label{assump:weak-identifiability}
There exists a constant $\zeta>0$ such that
\[
  \max_{k\in[K]}\|(\cm B_k^*)_{[S_X]}\|_{\op}\leq \zeta.
\]
\end{assumption}

Assumption \ref{assump:weak-identifiability} restricts the sparse deviation from displaying a spuriously strong low-rank component along the matricization; see \cite{agarwal2012noisy} for analogous discussions.

Next, we impose regularity conditions on the covariates and noise tensors. These assumptions underpin the non-asymptotic error bounds developed in the subsequent analysis.

\begin{assumption}[Sub-Gaussian design]\label{assump:subg-design}
  The collections $\{\cm X_{k,i}\}_{i=1}^{n_k}$ are independent across clients $k\in[K]$, and for each client the observations $\{\cm X_{k,i}\}_{i=1}^{n_k}$ are i.i.d. across $i$. Moreover, 
  the vectorized design tensor admits the representation $\vect(\cm X_{k,i}) = \bbm\Sigma_{\mathscr X,k}^{1/2}\bbm\xi_{k,i}$, where $\bbm\Sigma_{\mathscr X,k}\in\mathbb R^{p\times p}$ with $p=\prod_{\ell=1}^d p_\ell$, and $\bbm\xi_{k,i}\in\mathbb R^p$ has $i.i.d.$ standardized sub-Gaussian entries with parameter $\sigma_{\mathscr X,k}$. That is, $\mathbb E(\bbm\xi_{k,i})=\bm 0$, $\mathrm{Cov}(\bbm\xi_{k,i})=\bm I_p$, and for every $j\in[p]$ and $\mu\in\mathbb R$, $\mathbb E\exp(\mu\xi_{k,i,j}) \leq \exp(\mu^2\sigma_{\mathscr X,k}^2/2)$.
\end{assumption}

\begin{assumption}[Conditionally sub-Gaussian noise]\label{assump:subg-noise}
  The collections $\{\cm{E}_{k,i}\}_{i=1}^{n_k}$ are independent across clients $k\in[K]$, and for each client the noise tensors $\{\cm{E}_{k,i}\}_{i=1}^{n_k}$ are $i.i.d.$ across $i$. Conditional on $\cm X_{k,i}$, the random vector $\vect(\cm E_{k,i})$ is mean-zero and sub-Gaussian with covariance proxy $\bbm\Sigma_{\mathscr E,k}$ and parameter $\sigma_{\mathscr E,k}$. That is, $\mathbb E[\vect(\cm E_{k,i})\mid \cm X_{k,i}]=\bm 0$ and, for any $\bbm u\in\mathbb R^q$ with $q=\prod_{\ell=1}^m q_\ell$ and $\mu\in\mathbb R$, $\mathbb E\left[\exp\left(\mu\langle \bbm u,\vect(\cm E_{k,i})\rangle\right)\mid \cm X_{k,i}\right] \leq \exp(\mu^2\sigma_{\mathscr E,k}^2\bbm u^\top\bbm\Sigma_{\mathscr E,k}\bbm u/2)$, where $\bbm\Sigma_{\mathscr E,k}\in\mathbb R^{q\times q}$.
\end{assumption}

Assumption~\ref{assump:subg-design} requires the design tensor $\cm X_{k,i}$ to satisfy a sub-Gaussian condition with client-specific covariance structure, a standard requirement in high-dimensional regression analysis \citep{negahban2009unified,cai2020semisupervised}.
Assumption~\ref{assump:subg-noise} further assumes that the noise tensor $\cm E_{k,i}$ is conditionally mean-zero and conditionally sub-Gaussian given the design tensor. Notably, this is weaker than the assumption of full independence between the design and noise adopted in many high-dimensional regression settings \citep{negahban2009unified,cai2020semisupervised}.

\subsection{Theory for Stage-I Representation Learning}\label{sec:stage1-theory}

We begin by recalling the standard notion of differential privacy, which quantifies the information leakage of the federated Stage-I procedure. Two local datasets $\mathsf D_k=\{(\cm X_{k,i},\cm Y_{k,i})\}_{i=1}^{n_k}$ and $\mathsf D_k'=\{(\cm X_{k,i}',\cm Y_{k,i}')\}_{i=1}^{n_k}$ are called \emph{adjacent} if they differ in exactly one observation pair. 

\begin{definition}[$(\varepsilon,\delta)$-differential privacy \citep{dwork2006calibrating}]\label{def:dp}
A randomized mechanism $\mathcal M$ satisfies $(\varepsilon,\delta)$-differential privacy if, for any two adjacent datasets $\mathsf D_k$ and $\mathsf D_k'$ and any measurable set $\mathcal O$ in the output space of $\mathcal M$,
\[
  \mathbb P(\mathcal M(\mathsf D_k)\in\mathcal O)
  \leq
  \exp(\varepsilon)\mathbb P(\mathcal M(\mathsf D_k')\in\mathcal O)+\delta.
\]
\end{definition}

In Definition~\ref{def:dp}, the datasets $\mathsf D_k$ and $\mathsf D_k'$ are treated as fixed inputs, and the probabilities are taken only over the randomness of the mechanism $\mathcal M$ \citep{lu2026versatile}. In our setting, this randomness stems exclusively from the Gaussian perturbations added during the federated updates, rather than from the sampling randomness of the observed data.

We now present the theoretical guarantees for Stage-I representation learning in Theorems \ref{thm:federated_representation_error} and \ref{thm:representation_single_client}. To simplify the exposition, we collect the relevant model and algorithmic quantities that will appear throughout the analysis. For each client $k$, let $\lambda_{\mathscr X,k}^{\min}$ and $\lambda_{\mathscr X,k}^{\max}$ denote the minimum and maximum eigenvalues of $\bbm\Sigma_{\mathscr X,k}$, respectively, and let $\lambda_{\mathscr E,k}^{\max}$ denote the maximum eigenvalue of $\bbm\Sigma_{\mathscr E,k}$. The corresponding global envelope quantities are defined as $\underline\lambda_{\mathscr X}^{\min}=\min_{k\in[K]}\lambda_{\mathscr X,k}^{\min}$ and $\bar\lambda_{\mathscr X}^{\max}=\max_{k\in[K]}\lambda_{\mathscr X,k}^{\max}$. We further introduce $\bar\kappa_{\mathscr X}=\bar\lambda_{\mathscr X}^{\max}/\underline\lambda_{\mathscr X}^{\min}$, which serves as a uniform upper bound on the client-specific design condition numbers. In addition, set $\bar\sigma_{\mathscr X}=\max_{k\in[K]}\sigma_{\mathscr X,k}$, $\bar\sigma_{\mathscr E}=\max_{k\in[K]}\sigma_{\mathscr E,k}$, and $\bar\lambda_{\mathscr E}^{\max}=\max_{k\in[K]}\lambda_{\mathscr E,k}^{\max}$. Let $\mu_{\mathscr A}^{\min}=\min_{s\in[d+m]}\sigma_{r_s}((\cm A_0^*)_{[s]})$ be the smallest nonzero singular value of the true low-rank component across all matricizations. Write $r_{\bbm q}=\prod_{\ell=1}^{m}r_\ell$ and $r_{\bbm p}=\prod_{\ell=1}^{d}r_{m+\ell}$ for the Tucker-rank-induced upper bounds on the response-side and predictor-side matricization ranks, respectively, and $df_{\bbm r} = \prod_{s=1}^{d+m}r_s + \sum_{s=1}^{m} (q_s-r_s)r_s + \sum_{s=1}^{d} (p_s-r_{m+s})r_{m+s}$ for the effective degrees of freedom of the Tucker-rank-$\bbm r$ component.
On the algorithmic side, recall that $\tau_{\mathscr X,k}$ and $\tau_{\mathscr E,k}^{(t)}$ are the truncation levels for the covariate tensor and the current residual tensor, respectively. We also fix the constant $c_{d,m}=1/\{4(\sqrt{d+m}+1)\}$. For clarity of presentation, the admissible Stage-I step-size interval $\mathcal I_{\mathscr A}^{\mathrm{fed}}$, the initialization radius $R_{\mathscr A}$, and the auxiliary constants $\mathsf M_{\mathscr A}^{\mathrm{fed}}$ and $\Gamma_{\mathscr A}^{\mathrm{fed}}$ are specified in Section~\ref{sec:theorems-const} of the Supplementary Material.

\begin{theorem}[Federated Stage-I representation learning]\label{thm:federated_representation_error}
For each client $k$ and iteration $t$, suppose Gaussian noise level $\sigma_k^{(t)}$ added to the transmitted gradient is calibrated according to \eqref{eq:noise-scale}. Then the federated Stage-I procedure satisfies $(\varepsilon,\delta)$-differential privacy for each client.
Moreover, suppose Assumptions~\ref{assump:weak-identifiability}--\ref{assump:subg-noise} hold for all clients, and $\bar\kappa_{\mathscr X}<(1+c_{d,m})/(1-c_{d,m})$. Take the Stage-I step size $\eta_{\mathscr A}\in\mathcal I_{\mathscr A}^{\mathrm{fed}}$. Assume the initialization satisfies $\|\cm A_0^{(0)}-\cm A_0^*\|_\F\leq R_{\mathscr A}$, the weak-identifiability level satisfies $\zeta \leq c_{d,m}R_{\mathscr A}/(\eta_{\mathscr A}\bar\lambda_{\mathscr X}^{\max}\sqrt{r_{\bbm q}+r_{\bbm p}})$, and the pooled sample size satisfies
\[
  n \gtrsim \mathsf M_{\mathscr A}^{\mathrm{fed}}df_{\bbm r} \vee \frac{\Gamma_{\mathscr A}^{\mathrm{fed}}}{\mu_{\mathscr A}^{\min}\bar\lambda_{\mathscr X}^{\max}}\frac{T_g\sqrt{\log(T_g/\delta)}}{\varepsilon}\sqrt{K(df_{\bbm r}+\log T_g)}\max_{0\leq t<T_g}\max_{k\in[K]}\tau_{\mathscr E,k}^{(t)}\tau_{\mathscr X,k}.
\]
Suppose $T_g\asymp\log n$, $\tau_{\mathscr X,k} \asymp \sigma_{\mathscr X,k}\sqrt{\lambda_{\mathscr X,k}^{\max}}\sqrt{p+\log n}$ and $\tau_{\mathscr E,k}^{(t)} \asymp (R_{\mathscr A}+h_{\mathscr B,k})\sigma_{\mathscr X,k}\sqrt{\lambda_{\mathscr X,k}^{\max}}\sqrt{p+\log n} + \sigma_{\mathscr E,k}\sqrt{\lambda_{\mathscr E,k}^{\max}}\sqrt{q+\log n}$. Then, with probability at least $1 - C\exp(-Cdf_{2\bbm r}) - C\exp(-C(df_{\bbm r}+\log \log n)) - C\exp(-C\log n)$,
\begin{align*}
  \|\widehat{\cm A}_0-\cm A_0^*\|_\F &\lesssim \left\{\bar\sigma_{\mathscr X}^2\bar h_B \vee \bar\sigma_{\mathscr X}\bar\sigma_{\mathscr E}\left(\frac{\bar\lambda_{\mathscr E}^{\max}}{\bar\lambda_{\mathscr X}^{\max}}\right)^{1/2}\right\}\sqrt{\frac{df_{\bbm r}}{n}} + \sqrt{r_{\bbm q}+r_{\bbm p}}\,\zeta\\
  &\quad+ \frac{\log n\sqrt{\log(\log n/\delta)}\sqrt{K(df_{\bbm r}+\log\log n)}\max_{0\leq t<T_g}\max_{k\in[K]}\tau_{\mathscr E,k}^{(t)}\tau_{\mathscr X,k}}{\bar\lambda_{\mathscr X}^{\max}\varepsilon n}.
\end{align*}
\end{theorem}

We now interpret the conditions and the error bound in Theorem~\ref{thm:federated_representation_error}. The Gaussian noise level $\sigma_k^{(t)}$ is chosen to ensure that each transmitted gradient satisfies the Gaussian-mechanism privacy guarantee, and composition over $T_g$ iterations yields the claimed $(\varepsilon,\delta)$-differential privacy for each client. The Stage-I step-size condition $\eta_{\mathscr A}\in\mathcal I_{\mathscr A}^{\mathrm{fed}}$ and the upper bound on the federated design condition number $\bar\kappa_{\mathscr X}$ guarantee sufficient local curvature of the population loss around $\cm A_0^*$ and contractivity of the Riemannian gradient update. The initialization requirement $\|\cm A_0^{(0)}-\cm A_0^*\|_\F\leq R_{\mathscr A}$ places the initial iterate inside this region of contraction. 
The upper bound on $\zeta$ controls the bias induced by the non-identifiable sparse deviations, preventing them from dominating the low-rank update and pushing the iterates outside the contraction region. The pooled sample-size condition has two parts: the term $\mathsf M_{\mathscr A}^{\mathrm{fed}}df_{\bbm r}$ is the standard requirement in high-dimensional tensor regression, reflecting that the low Tucker-rank structure reduces the effective complexity from the ambient tensor dimension to the degrees of freedom $df_{\bbm r}$. The remainder ensures that the privacy perturbations accumulated over $T_g$ iterations are sufficiently small for every noisy Stage-I iterate to remain within the contraction region.

Under these conditions, Theorem~\ref{thm:federated_representation_error} establishes the estimation guarantee for the federated Stage-I estimator.
The resulting error bound has three components. The first term is the statistical error for estimating the shared low Tucker-rank component from the pooled sample. The second term $\sqrt{r_{\bbm q}+r_{\bbm p}}\,\zeta$ is the unavoidable cost of weak identifiability: without a strict separation between low-rank and sparse structures, a bias of this order is generally unavoidable, as confirmed by the minimax lower bound in Theorem~\ref{thm:joint-minimax-lower-bound}.
The third term is the privacy cost induced by the Gaussian perturbations injected into the transmitted gradients. This term scales as $1/\varepsilon$ with respect to the privacy budget $\varepsilon$ and only as $\sqrt{\log(\log n/\delta)}$ with respect to the failure probability $\delta$, reflecting the standard privacy--utility trade-off with a mild logarithmic dependence on $1/\delta$. Plugging in the theoretical orders of $\tau_{\mathscr X,k}$ and $\tau_{\mathscr E,k}^{(t)}$, the privacy cost is, up to logarithmic factors and model-dependent constants, of order $\sqrt{Kdf_{\bbm r}}(p+\sqrt{pq})/(\varepsilon n)$. Therefore, the sample size required to make the privacy overhead negligible needs only to exceed the dimension-reduced quantity $\varepsilon^{-1}\sqrt{Kdf_{\bm r}}(p+\sqrt{pq})$, rather than the full ambient dimension $pq$. This demonstrates that privacy can be achieved at a modest sample-size premium.

Before stating the single-client analogue, we introduce the corresponding client-specific notation. The same rank and signal-strength symbols $r_{\bbm q}$, $r_{\bbm p}$, $df_{\bbm r}$, and $\mu_{\mathscr A}^{\min}$ are reused. 
For a fixed client $k$, let $\kappa_{\mathscr X,k}=\lambda_{\mathscr X,k}^{\max}/\lambda_{\mathscr X,k}^{\min}$ denote the client-specific design condition number. The admissible Stage-I step-size interval $\mathcal I_{\mathscr A,k}$, the initialization radius $R_{\mathscr A,k}$, and the sample-size constant $\mathsf M_{\mathscr A,k}^{\mathrm{single}}$ are collected in Section~\ref{sec:theorems-const} of the Supplementary Material for readability.

\begin{theorem}[Single-client Stage-I representation learning]\label{thm:representation_single_client}
For a fixed client $k$, suppose Assumptions~\ref{assump:weak-identifiability}--\ref{assump:subg-noise} hold and $\kappa_{\mathscr X,k}<(1+c_{d,m})/(1-c_{d,m})$. Take the Stage-I step size $\eta_{\mathscr A,k}\in\mathcal I_{\mathscr A,k}$. Assume the initialization satisfies $\|\cm A_{0,k}^{\mathrm{loc},(0)}-\cm A_0^*\|_\F\leq R_{\mathscr A,k}$, the weak-identifiability level satisfies $\zeta \leq c_{d,m}R_{\mathscr A,k}/(\eta_{\mathscr A,k}\lambda_{\mathscr X,k}^{\max}\sqrt{r_{\bbm q}+r_{\bbm p}})$, and the local sample size satisfies $n_k \gtrsim \mathsf M_{\mathscr A,k}^{\mathrm{single}}df_{\bbm r}$. Suppose the number of local Stage-I iterations is $T_g^{(k)}\asymp\log n_k$. Then, with probability at least $1-C\exp(-Cdf_{\bbm r})$,
\[
  \|\widehat{\cm A}_{0,k}^{\mathrm{loc}}-\cm A_0^*\|_\F \lesssim \left\{\sigma_{\mathscr X,k}^2 h_{\mathscr B,k} \vee \sigma_{\mathscr E,k}\left(\frac{\lambda_{\mathscr E,k}^{\max}}{\lambda_{\mathscr X,k}^{\max}}\right)^{1/2}\right\}\sqrt{\frac{df_{\bbm r}}{n_k}} + \sqrt{r_{\bbm q}+r_{\bbm p}}\,\zeta.
\]
\end{theorem}

Theorem~\ref{thm:representation_single_client} provides the corresponding Stage-I guarantee when estimation is performed solely on the local data of client $k$. The conditions on the design condition number $\kappa_{\mathscr X,k}$, step size $\eta_{\mathscr A,k}$, initialization radius $R_{\mathscr A,k}$, and weak-identifiability level $\zeta$ play exactly the same roles as their federated counterparts, now expressed in client-specific form. Likewise, the requirement on the local sample size $n_k$ is the single-client analogue of the pooled sample size condition in Theorem~\ref{thm:federated_representation_error}.
Since the single-client procedure does not communicate any information across clients, it requires no privacy perturbation. Consequently, the error bound contains only the local statistical error (scaled by $n_k$) and the same weak-identifiability cost $\sqrt{r_{\bbm q}+r_{\bbm p}}\,\zeta$, with no privacy term.

\subsection{Theory for Stage-II Personalized Refinement}\label{sec:stage2-theory}

Given a Stage-I estimator $\overline{\cm A}_0$ of $\cm A_0^*$ (i.e., $\widehat{\cm A}_0$ or $\widehat{\cm A}_{0,k}^{\mathrm{loc}}$), client $k$ estimates its sparse deviation $\cm B_k^*$ via the $\ell_1$-penalized problem. Theorem~\ref{thm:personalized_error} provides the error bound for the FISTA solution $\overline{\cm B}_k(\overline{\cm A}_0)$. 

\begin{theorem}[Stage-II personalized refinement]\label{thm:personalized_error}
Suppose Assumptions~\ref{assump:weak-identifiability}--\ref{assump:subg-noise} hold. For a client $k\in[K]$, the Stage-II local step size is set to $\eta_{\mathscr B,k}=(2n_k^{-1}\|\sum_{i=1}^{n_k}\vect(\cm X_{k,i})\vect^\top(\cm X_{k,i})\|_{\op})^{-1}$, the sparsity regularization parameter is chosen as $\omega_k\asymp\vartheta_k(\mathds 1_{\{\nu=0\}}\sqrt{\log(pq/s_0)/n_k} + \mathds 1_{\{0<\nu<1\}}\sqrt{\log(pq)/n_k})$ with $\vartheta_k = \sigma_{\mathscr X,k}\sigma_{\mathscr E,k}(\lambda_{\mathscr X,k}^{\max}\lambda_{\mathscr E,k}^{\max})^{1/2}$, 
and the number of local FISTA iterations is $T_l^{(k)}\asymp \sqrt{\kappa_{\mathscr X,k}}(1+h_{\mathscr B,k}/\mathsf{Error}_{\mathscr B}^{(k)}(\overline{\cm A}_0))$, where 
\[\mathsf{Error}_{\mathscr B}^{(k)}(\overline{\cm A}_0) = \kappa_{\mathscr X,k}\|\overline{\cm A}_0-\cm A_0^*\|_\F + \sqrt{s_\nu}(\omega_k/\lambda_{\mathscr X,k}^{\min})^{1-\nu/2}.\]
If the local sample size satisfies $n_k\gtrsim \mathsf L_\nu \vee (\sigma_{\mathscr X,k}^{4}\kappa_{\mathscr X,k}^{2}\vee\sigma_{\mathscr X,k}^{2}\kappa_{\mathscr X,k}\vee1)p$ with $\mathsf L_\nu = \mathds 1_{\{\nu=0\}}s_0\log(pq/s_0) + \mathds 1_{\{0<\nu<1\}}\log(pq)$, then, with probability at least $1-C\exp(-C\mathsf L_\nu)-C\exp(-Cp)$, we have $\|\overline{\cm B}_k(\overline{\cm A}_0)-\cm B_k^*\|_\F \lesssim \mathsf{Error}_{\mathscr B}^{(k)}(\overline{\cm A}_0)$, and consequently, $\|\overline{\cm A}_0-\cm A_0^*\|_\F + \|\overline{\cm B}_k(\overline{\cm A}_0)-\cm B_k^*\|_\F \lesssim \mathsf{Error}_{\mathscr B}^{(k)}(\overline{\cm A}_0)$.
\end{theorem}

The tuning parameters and sample-size condition in Theorem~\ref{thm:personalized_error} follow standard prescriptions for sparse regression. Particularly, the step size $\eta_{\mathscr B,k}$ is the inverse of the Lipschitz constant of the quadratic loss, the regularization parameter $\omega_k$ is chosen at the noise level to guarantee optimal sparse recovery, and the number of FISTA iterations $T_l^{(k)}$ ensures sufficient convergence.

The error bound $\mathsf{Error}_{\mathscr B}^{(k)}(\overline{\cm A}_0)$ consists of two parts. The first term, $\kappa_{\mathscr X,k}\|\overline{\cm A}_0-\cm A_0^*\|_\F$, reflects the propagation of the Stage-I estimation error: an inaccurate shared estimator directly degrades the personalized refinement. The second term, $\sqrt{s_\nu}(\omega_k/\lambda_{\mathscr X,k}^{\min})^{1-\nu/2}$, is the local sparse estimation error for estimating $\cm B_k^*$. This term adapts to the sparsity regime, with a sharper rate when $\nu=0$ (exact sparsity) compared to the weak sparsity case $\nu\in(0,1)$.

We now instantiate the general Stage-II bound with the specific Stage-I estimators.
Substituting the federated Stage-I estimator $\widehat{\cm A}_0$ into Theorem~\ref{thm:personalized_error} yields $\widehat{\cm B}_k=\overline{\cm B}_k(\widehat{\cm A}_0)$. Then by Theorem~\ref{thm:federated_representation_error} and the condition number bound $\kappa_{\mathscr X,k}\leq\bar\kappa_{\mathscr X}<(1+c_{d,m})/(1-c_{d,m})$, we obtain, with high probability, 
\begin{align*}
  \|\widehat{\cm A}_0-\cm A_0^*\|_\F + \|\widehat{\cm B}_k-\cm B_k^*\|_\F
  &\lesssim
  \left(\bar\sigma_{\mathscr X}^2\bar h_{\mathscr B} \vee \frac{\bar\vartheta}{\lambda_{\mathscr X,k}^{\min}}\right)\sqrt{\frac{df_{\bbm r}}{n}}
  + \sqrt{r_{\bbm q}+r_{\bbm p}}\,\zeta
  + \sqrt{s_\nu}\left(\frac{\omega_k}{\lambda_{\mathscr X,k}^{\min}}\right)^{1-\nu/2}\\
  &+
  \frac{\log n\sqrt{\log(\log n/\delta)}\sqrt{K(df_{\bbm r}+\log\log n)}}
       {\bar\lambda_{\mathscr X}^{\max}\varepsilon n}
  \max_{0\leq t<T_g}\max_{k\in[K]}\tau_{\mathscr E,k}^{(t)}\tau_{\mathscr X,k},
\end{align*}
where $\bar\vartheta=\bar\sigma_{\mathscr X}\bar\sigma_{\mathscr E}(\bar\lambda_{\mathscr X}^{\max}\bar\lambda_{\mathscr E}^{\max})^{1/2}$. Similarly, using the single-client Stage-I estimator $\widehat{\cm A}_{0,k}^{\mathrm{loc}}$ gives $\widehat{\cm B}_k^{\mathrm{loc}}=\overline{\cm B}_k(\widehat{\cm A}_{0,k}^{\mathrm{loc}})$, which combined with Theorem~\ref{thm:representation_single_client} leads to, with high probability,
\[
  \|\widehat{\cm A}_{0,k}^{\mathrm{loc}}-\cm A_0^*\|_\F + \|\widehat{\cm B}_k^{\mathrm{loc}}-\cm B_k^*\|_\F
  \lesssim
  \left(\sigma_{\mathscr X,k}^2 h_{\mathscr B,k} \vee \frac{\vartheta_k}{\lambda_{\mathscr X,k}^{\min}}\right)\sqrt{\frac{df_{\bbm r}}{n_k}}
  + \sqrt{r_{\bbm q}+r_{\bbm p}}\,\zeta
  + \sqrt{s_\nu}\left(\frac{\omega_k}{\lambda_{\mathscr X,k}^{\min}}\right)^{1-\nu/2}.
\]
The minimax lower bound in Theorem~\ref{thm:joint-minimax-lower-bound} shows that both the federated and single-client personalized error bounds are rate-optimal in their respective regimes. Hence, a direct comparison of the two bounds provides a fair way to evaluate the benefit of federation.

The federated procedure improves the Stage-I statistical term by replacing the local sample size $n_k$ with the pooled size $n$, but it also incurs an additional privacy cost. 
When the client-specific quantities $h_{\mathscr B,k}$, $\sigma_{\mathscr X,k}$, $\sigma_{\mathscr E,k}$, and the eigenvalue bounds are of the same order across clients, the federated bound enjoys a smaller dominant statistical error. 
Thus, the federated estimator is preferable when this pooled-sample gain dominates the privacy cost. A sufficient condition for this to hold for client $k$ is 
\begin{align*}
  n \gtrsim
  \frac{\sqrt{K n_k\log(\log n/\delta)}}{\bar\lambda_{\mathscr X}^{\max}\varepsilon}
  \frac{\max_{0\leq t<T_g}\max_{k\in[K]}\tau_{\mathscr E,k}^{(t)}\tau_{\mathscr X,k}}
  {\sigma_{\mathscr X,k}^2 h_{\mathscr B,k}
  \vee
  \sigma_{\mathscr X,k}\sigma_{\mathscr E,k}(\lambda_{\mathscr X,k}^{\max}\lambda_{\mathscr E,k}^{\max})^{1/2}/\lambda_{\mathscr X,k}^{\min}}.
\end{align*}
This condition is mild in high-dimensional settings and transparently captures the privacy-utility trade-off. A smaller privacy budget $\varepsilon$ (stronger privacy) inflates the required pooled sample size as $1/\varepsilon$, while the failure probability $\delta$ exerts only a logarithmic influence through $\sqrt{\log(\log n/\delta)}$.

\subsection{Theory for Initialization and Tucker-Rank Selection}\label{sec:one-stage-est-theory}

We next justify the initialization and Tucker-rank selection procedures in Sections~\ref{sec:alg-ini}--\ref{sec:tucker-rank-selection}.

\begin{theorem}[One-stage initialization error]\label{thm:one-stage-initialization-error}
Suppose Assumptions~\ref{assump:weak-identifiability}--\ref{assump:subg-noise} hold for client $k$, and assume the local sample size satisfies $n_k \gtrsim (p+q) \vee \mathsf L_\nu \vee (\sigma_{\mathscr X,k}^{4}\kappa_{\mathscr X,k}^{2} \vee \sigma_{\mathscr X,k}^{2}\kappa_{\mathscr X,k})p$, where $\mathsf L_\nu$ is defined in Theorem~\ref{thm:personalized_error}. Consider the one-stage initialization estimator $(\widetilde{\cm A}_{0,k},\widetilde{\cm B}_{k})$ defined in \eqref{eq:alg-ini} with tuning parameters chosen as 
\[
	\lambda_k \asymp \vartheta_{k}\sqrt{\frac{p+q}{n_k}} + \lambda_{\mathscr X,k}^{\min}\zeta
	\ \text{and} \ 
	\varpi_k \asymp \vartheta_{k}\left(\mathds 1_{\{\nu=0\}}\sqrt{\frac{\log(pq/s_0)}{n_k}}+\mathds 1_{\{0<\nu<1\}}\sqrt{\frac{\log(pq)}{n_k}}\right),
\]
where $\vartheta_{k} = \sigma_{\mathscr X,k}\sigma_{\mathscr E,k}(\lambda_{\mathscr X,k}^{\max}\lambda_{\mathscr E,k}^{\max})^{1/2}$. Then, with probability at least $1-C\exp(-C(p+q))-C\exp(-C\mathsf L_\nu)-C\exp(-Cp)$,
\[
	\|\widetilde{\cm A}_{0,k}-\cm A_0^*\|_{\F} + \|\widetilde{\cm B}_k-\cm B_k^*\|_{\F}
	\lesssim
	\sqrt{r_{\bbm q} + r_{\bbm p}}\frac{\lambda_k}{\lambda_{\mathscr X,k}^{\min}}
	+
	\sqrt{s_\nu}\left(\frac{\varpi_k}{\lambda_{\mathscr X,k}^{\min}}\right)^{1-\nu/2}.
\]
\end{theorem}

Theorem~\ref{thm:one-stage-initialization-error} characterizes the statistical accuracy for the optimizer of convex problem in \eqref{eq:alg-ini}, which serves as the preliminary estimator for initializing the two-stage procedures. The optimization error of the ADMM solver can be made arbitrarily small relative to the statistical error by running sufficient iterations, so the theorem focuses on the statistical performance. The low-rank error component in the bound is driven by $\lambda_k\asymp \vartheta_k\sqrt{(p+q)/n_k}$ and is therefore of a larger order than the single-client Stage-I rate established in Theorem~\ref{thm:representation_single_client}. This is expected because the convex formulation regularizes only the $S_X$-matricization in \eqref{eq:model-vectorized} through its nuclear norm, rather than directly exploiting the full low Tucker-rank structure of $\cm A_0^*$. Nevertheless, this rate is sufficient for initialization purposes, and both simulation and empirical results suggest that the resulting preliminary estimator provides a stable warm start for the subsequent two-stage procedures.

\begin{theorem}[Tucker-rank selection consistency]\label{thm:rank_selection}
Suppose $r_s^*\leq \bar r_s$ for each $s\in[d+m]$ and the conditions of Theorem~\ref{thm:one-stage-initialization-error} hold. For each mode $s\in[d+m]$, assume that  the ridge-type penalty $c_s(\bbm p,\bbm q,n_k)$ satisfies
\[
  \sqrt{r_{\bbm q}+r_{\bbm p}}\frac{\lambda_k}{\lambda_{\mathscr X,k}^{\min}}
  + 
  \sqrt{s_\nu}\left(\frac{\varpi_k}{\lambda_{\mathscr X,k}^{\min}}\right)^{1-\nu/2}
  \ll
  c_s(\bbm p,\bbm q,n_k)
  \ll
  \sigma_{r_s^*}((\cm A_0^*)_{[s]})
  \min_{1\leq j\leq r_s^*-1}
  \frac{\sigma_{j+1}((\cm A_0^*)_{[s]})}{\sigma_j((\cm A_0^*)_{[s]})}.
\]
Then, $\mathbb P\left(\widehat{\bbm r}_k^{\mathrm{loc}}=\bbm r^*\right)\to1$ as $n_k\to\infty$ and the tensor dimensions $(\bbm p,\bbm q)$ diverge.
\end{theorem}

Theorem~\ref{thm:rank_selection} provides a consistency guarantee for the Tucker-rank selection step based on the preliminary estimator $\widetilde{\cm A}_{0,k}$. The lower bound on the ridge penalty $c_s(\bbm p,\bbm q,n_k)$ ensures that the estimation error from Theorem~\ref{thm:one-stage-initialization-error} is dominated by the ridge term, which stabilizes the spurious small singular values beyond the true rank. The upper bound requires the ridge penalty to be smaller than the signal gap among the nonzero singular values of $(\cm A_0^*)_{[s]}$, so that the true rank location is not masked. Under this separation, the ridge-type ratio criterion in \eqref{eq:Ridge-type estimator} consistently selects the true Tucker rank $\bbm r^*$ with probability tending to one.

\subsection{Minimax Lower Bounds}\label{sec:minimax-lower-bound}

We finally establish minimax lower bounds for both the federated and single-client regimes under model \eqref{eq:model} and decomposition \eqref{eq:decomposition}. 
To compare these two regimes on a common footing, we specialize to the isotropic Gaussian submodel with $\vect(\cm X_{k,i})\sim N(\bm 0,\bm I_p)$ and $\vect(\cm E_{k,i})\sim N(\bm 0,\bm I_q)$, independent across clients and samples. Let the federated parameter space be 
\[
  \Theta^{\mathrm{fed}}
  =
  \left\{
  (\cm A_0,\{\cm B_k\}_{k=1}^K):
  \tucrank(\cm A_0)\leq \bbm r,\ 
  \cm B_k\in\mathbb B_\nu,\ 
  \|(\cm B_k)_{[S_X]}\|_{\op}\leq\zeta
  \ \text{for all }k\in[K]
  \right\},
\]
and for a fixed client $k$, define the single-client class $\Theta_k^{\mathrm{single}}$ by restricting to the pair $(\cm A_0,\cm B_k)$ with the same constraints.
The joint minimax risk for estimating the common component $\cm A_0^*$ and the client's deviation $\cm B_k^*$ is given by $\mathcal R_k(\Theta)=\inf_{\widehat{\boldsymbol{\mathscr A}}_0,\widehat{\boldsymbol{\mathscr B}}_k}\sup_{(\boldsymbol{\mathscr A}_0^*,\boldsymbol{\mathscr B}_k^*)\in\Theta}\mathbb E_{\boldsymbol{\mathscr A}_0^*,\boldsymbol{\mathscr B}_k^*}\left[\|\widehat{\cm A}_0-\cm A_0^*\|_{\F}+\|\widehat{\cm B}_k-\cm B_k^*\|_{\F}\right]$.
Set $\mathfrak s_{\nu,k}=\mathds 1_{\{\nu=0\}}\sqrt{s_0\log(pq/s_0)/n_k}+\mathds 1_{\{0<\nu<1\}}\{s_\nu(\log(pq)/n_k)^{1-\nu/2}\}^{1/2}$. 

\begin{theorem}[Joint minimax lower bounds]\label{thm:joint-minimax-lower-bound}
Under the isotropic Gaussian submodel stated above, assume $q_\ell\geq2r_\ell$ for $\ell\in[m]$, $p_\ell\geq2r_{m+\ell}$ for $\ell\in[d]$, and $\zeta\geq \mathfrak s_{\nu,k}$. For exact sparsity $\nu=0$, suppose $2(r_{\bbm q}\wedge r_{\bbm p})\leq s_0\leq pq/2$; for weak sparsity $\nu\in(0,1)$, assume $s_\nu\geq2(r_{\bbm q}\wedge r_{\bbm p})\zeta^\nu$. Then, for the single-client and federated settings, the following lower bounds hold, respectively,
\[
  \mathcal R_k(\Theta) \gtrsim
  \begin{cases}
  \displaystyle
  \sqrt{\frac{df_{\bbm r}}{n_k}} + \mathfrak s_{\nu,k} + \sqrt{r_{\bbm q}\wedge r_{\bbm p}}\,\zeta,
  & \Theta=\Theta_k^{\mathrm{single}},\\[3mm]
  \displaystyle
  \sqrt{\frac{df_{\bbm r}}{n}} + \mathfrak s_{\nu,k} + \sqrt{r_{\bbm q}\wedge r_{\bbm p}}\,\zeta,
  & \Theta=\Theta^{\mathrm{fed}}.
  \end{cases}
\]
\end{theorem}

Theorem~\ref{thm:joint-minimax-lower-bound} shows that, excluding the additional error induced by a finite privacy budget, the upper bounds derived above are rate-optimal under the current minimax framework.
Specifically, the low Tucker-rank term of the lower bound scales as $\sqrt{df_{\bbm r}/n}$ in the federated regime and $\sqrt{df_{\bbm r}/n_k}$ in the single-client regime, matching the corresponding terms in the upper bounds up to constants that depend on the design and noise parameters. The sparse term $\mathfrak s_{\nu,k}$ also coincides with the local sparse estimation rate obtained in the upper bounds. Finally, the weak-identifiability term in the lower bound is $\sqrt{r_{\bbm q}\wedge r_{\bbm p}}\,\zeta$, whereas the upper bounds contain $\sqrt{r_{\bbm q}+r_{\bbm p}}\,\zeta$. Since each mode rank $r_s$ is fixed and does not grow with the tensor dimensions, both $r_{\bbm q}$ and $r_{\bbm p}$ are finite constants, and hence these two terms are of the same order.

\section{Simulation Studies}\label{sec:Simulation}

We conduct simulation studies to evaluate the finite-sample performance of the proposed personalized federated tensor regression method. We first compare its estimation accuracy with several benchmark methods across varying per-client sample sizes, and then examine the privacy–utility trade-off under different privacy parameters. 
In all experiments, the federated and single‑client algorithms are initialized using the procedure in Section~\ref{sec:alg-ini}, the Tucker rank is selected according to Section~\ref{sec:tucker-rank-selection}, and remaining tuning parameters are chosen as described in Section~\ref{sec:tuning-param-selection}. 
ADMM‑based initialization tuning follows Section~\ref{sec:tuning-single} of the Supplementary Material.

We consider $K=5$ clients, with coefficient tensors $\cm A_k=\cm A_0+\cm B_k\in\mathbb R^{5\times5\times10\times5}$. The generation of the common and client-specific components proceeds as follows.
The shared component $\cm A_0$ has Tucker rank $\bbm r=(3,3,3,3)$. We first draw a core tensor $\cm G\in\mathbb R^{3\times3\times3\times3}$ with $i.i.d.$ standard normal entries and rescale it to satisfy $\|\cm G\|_{\F}=1$. We then independently generate random Gaussian matrices of sizes $5\times3$, $5\times3$, $10\times3$, and $5\times3$, and orthonormalize their columns to obtain factor matrices $\bm U_s$ such that $\bm U_s^\top\bm U_s=\bm I_{r_s}$ for $s\in[4]$. The shared component is formed as $\cm A_0 = \cm G \times_{s=1}^4 \bm U_s$. For each client $k$, the sparse deviation $\cm B_k$ is constructed by selecting $50$ entries uniformly at random, drawing their values independently from $\mathscr N(0,1)$, and setting all remaining entries to zero. The resulting tensor is then rescaled so that $\|\cm B_k\|_{\F}/\|\cm A_0\|_{\F}=0.5$.  
Given the true coefficient tensors $\cm A_k=\cm A_0+\cm B_k$, we generate the predictor and response tensors $\{(\cm X_{k,i},\cm Y_{k,i})\}$ from model \eqref{eq:model}. 
The entries of $\cm X_{k,i}$ and the noise tensor $\cm E_{k,i}$ are $i.i.d.$ from $\mathscr N(0,1)$ and $\mathscr N(0,0.01)$, respectively. The response $\cm Y_{k,i}$ is then computed according to \eqref{eq:model}.

To evaluate estimation accuracy, we compute the Frobenius-norm error of the shared component as $\|\overline{\cm A}_0-\cm A_0\|_\F$, and the average errors for the client-specific deviations and personalized coefficient tensors as $5^{-1}\sum_{k=1}^5\|\overline{\cm B}_k-\cm B_k\|_\F$ and $5^{-1}\sum_{k=1}^5\|\overline{\cm A}_k-\cm A_k\|_\F$, respectively. Here $\overline{\cm A}_0$, $\overline{\cm B}_k$, and $\overline{\cm A}_k$ denote the corresponding estimators. 
We compare four estimation strategies: (i) the non-private personalized federated estimator with $(\varepsilon,\delta)=(\infty,1)$, (ii) the private personalized federated estimator with $(\varepsilon,\delta)=(20,0.1)$, (iii) the single-client two-stage estimator, and (iv) the ADMM initialization obtained by averaging the five client-specific initial estimators.

Figure~\ref{fig:sim-error-boxplot} displays the estimation errors for $\cm A_0$, $\{\cm B_k\}_{k=1}^5$, and $\{\cm A_k\}_{k=1}^5$ via boxplots over 1,000 replications at different per-client sample sizes. 
For all methods, the errors decrease as the per-client sample size grows, corroborating the theoretical results. At each sample size, the non-private federated estimator has the smallest error, followed by the private federated estimator, the single-client two-stage estimator, and the ADMM initialization. The gap between the two federated estimators reflects the cost of privacy protection, while their improvement over the single-client estimator demonstrates the benefit of cross-client information sharing. The advantage of the single-client two-stage estimator over the ADMM initialization further highlights the value of the two-stage refinement. Moreover, the gap between the private and non-private federated estimators shrinks as the sample size increases, indicating that the effect of the added Gaussian noise becomes less pronounced as more observations become available at each client.

\begin{figure}[H]
\centering
\includegraphics[width=0.8\textwidth]{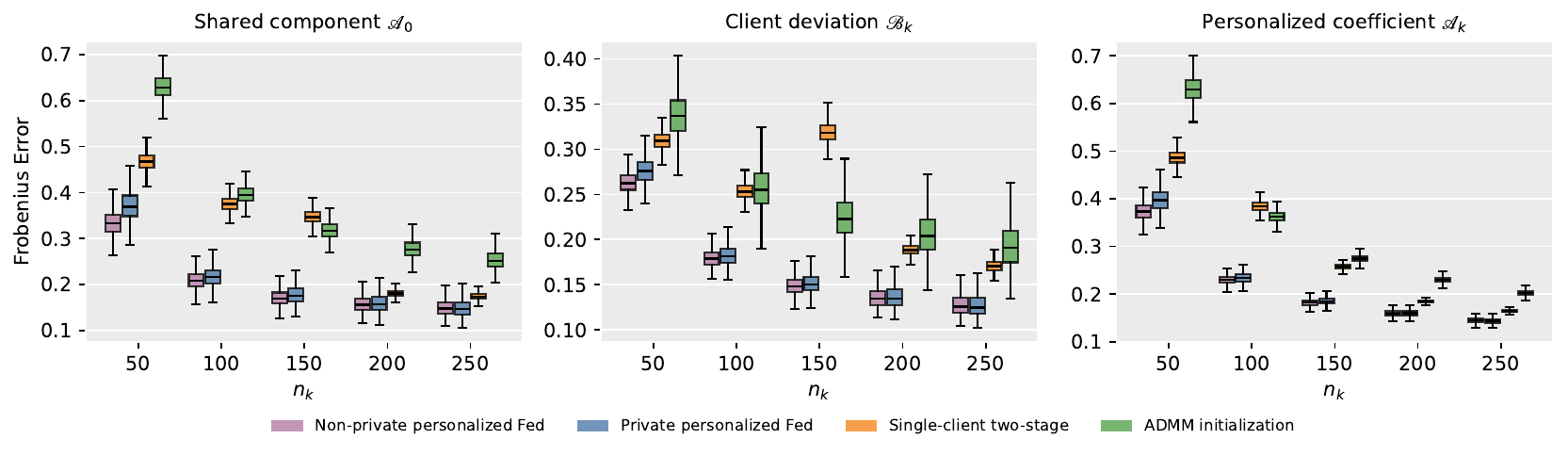}
\caption{Frobenius-norm estimation errors for $\cm A_0$, $\{\cm B_k\}_{k=1}^5$, and $\{\cm A_k\}_{k=1}^5$.}
\label{fig:sim-error-boxplot}
\end{figure}

We also investigate the privacy-utility trade-off by varying the privacy parameters $(\varepsilon,\delta)$ for the federated personalized estimators in Figure~\ref{fig:privacy-sensitivity}, keeping $n_k=250$ per client and averaging over 1,000 replications. 
The errors increase as either $\varepsilon$ or $\delta$ decreases, confirming the expected cost of stronger privacy protection. The impact of $\varepsilon$ is more pronounced than that of $\delta$, consistent with the dependence structure of the theoretical error bound in Theorem~\ref{thm:federated_representation_error}. 
The three estimation targets exhibit similar privacy-utility patterns because the privacy perturbation enters solely through Stage-I estimation of $\cm A_0$ and subsequently propagates to the estimates of $\{\cm B_k\}_{k=1}^5$ and $\{\cm A_k\}_{k=1}^5$.

\begin{figure}[H]
\centering
\includegraphics[width=0.6\textwidth]{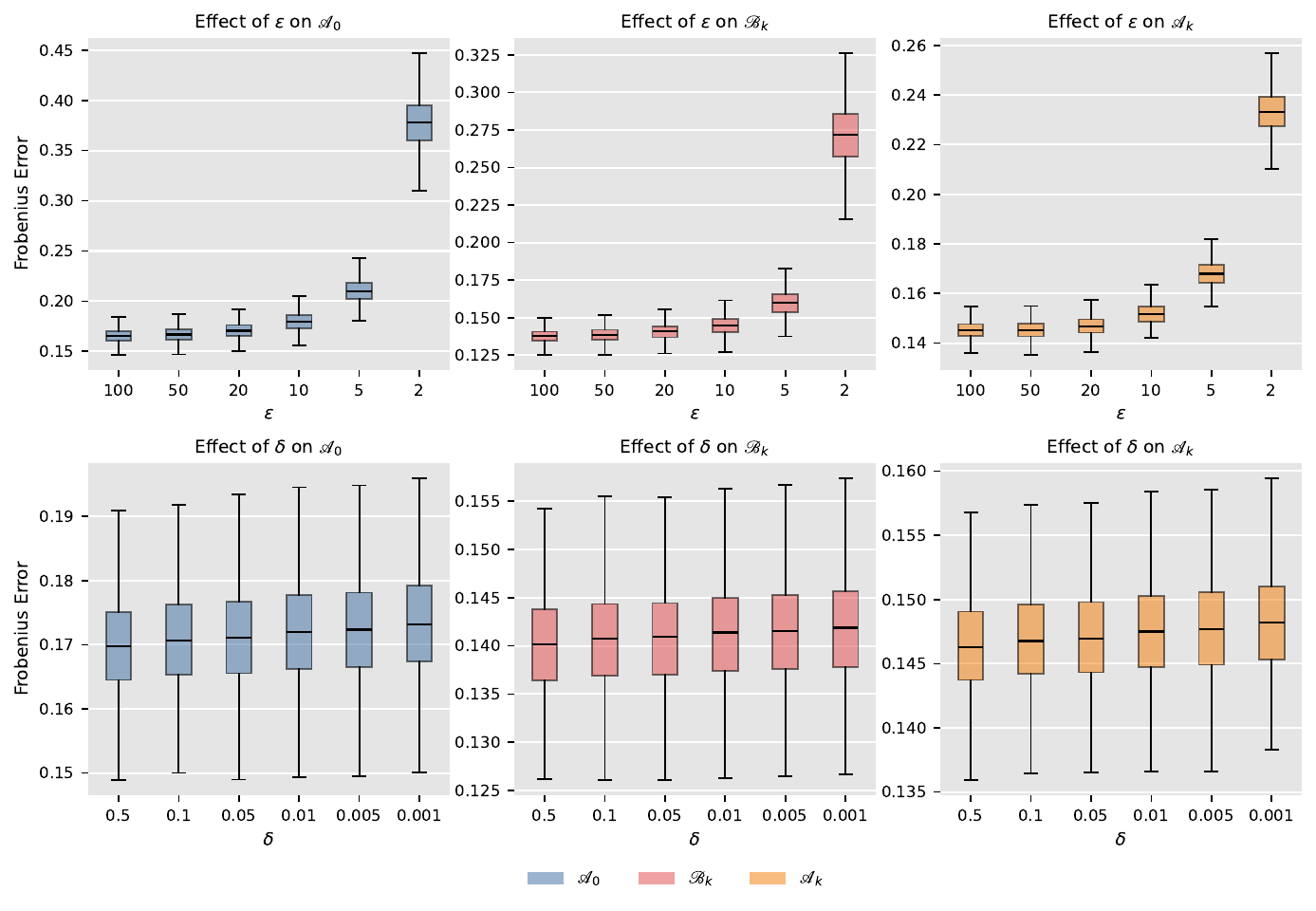}
\caption{Privacy-utility tradeoff for $\cm A_0$, $\{\cm B_k\}_{k=1}^5$, and $\{\cm A_k\}_{k=1}^5$. The upper panel varies $\varepsilon$ with $\delta=0.1$ fixed, whereas the lower panel varies $\delta$ with $\varepsilon=20$ fixed.}
\label{fig:privacy-sensitivity}
\end{figure}

\section{Empirical Analysis}\label{sec:RealData}

ADHD symptoms are clinically meaningful and important for early identification, monitoring, and individualized intervention. Empirical evidence shows that structural MRI features are associated with ADHD symptoms \citep{zhou2013tensor}. This section therefore investigates the role of MRI in analyzing and predicting ADHD symptom severity. Because the participating sites collect comparable imaging and clinical variables, borrowing information across sites may improve analysis and prediction. However, privacy and data-sharing restrictions prevent direct pooling of individual-level records. Moreover, MRI information is typically recorded as a tensor. These considerations make our DP federated learning framework of tensor regression well suited for this application.

We use data from three sites participating in the ADHD-200 Sample Initiative: the Kennedy Krieger Institute (KKI), Peking University (Peking), and New York University (NYU). The data are publicly available at \url{http://fcon_1000.projects.nitrc.org/indi/adhd200/}. After excluding subjects with missing or invalid measurements, the final sample comprises 556 subjects: 80 from KKI, 221 from Peking, and 255 from NYU. Following the Daubechies D4 wavelet-based image-reduction procedure used in \cite{hou2015hierarchical}, we transform each preprocessed structural MRI scan into a $12\times14\times12$ tensor and denote the resulting covariate for subject $i$ at site $k$ by $\cm X_{k,i}\in \mathbb R^{12\times 14 \times 12}$. Let $Y_{k,i}\in\mathbb R$ denote the ADHD symptom score, a quantitative index in the dataset that measures ADHD symptom severity. We consider the scalar-on-tensor regression
\begin{align}\label{eq:adhd-model}
  Y_{k,i} = \left\langle \cm A_k,\cm X_{k,i}\right\rangle+E_{k,i},
  \quad \text{with} \ k\in\{\text{KKI},\text{NYU},\text{Peking}\},
\end{align}
where $E_{k,i}\in \mathbb R$ is mean-zero random noise and $\cm A_k\in\mathbb R^{12\times14\times12}$ is the site-specific coefficient tensor. Because structural MRI measurements exhibit strong spatial dependence and common neuroanatomical organization, the imaging-symptom association shared across sites is expected to admit a low-dimensional Tucker-rank structure. In contrast, differences in scanners, acquisition protocols, or subject populations are expected to affect only a limited subset of imaging features, motivating sparse site-specific deviations. Therefore, we decompose $\cm A_k$ into a shared low Tucker-rank component $\cm A_0$ and a weakly sparse site-specific deviation $\cm B_k$ as in \eqref{eq:decomposition}.

To implement the proposed federated procedure in Section~\ref{sec:fed-learning}, we first compute the local one-stage estimators in Section~\ref{sec:alg-ini} and use the preliminary estimate from NYU, the site with the largest sample size, to initialize Stage-I. Applying the ridge-type ratio criterion in Section~\ref{sec:tucker-rank-selection} to the preliminary estimators $\widetilde{\cm A}_{0,k}$ and aggregating the site-specific rank candidates yields the federated Tucker-rank estimate $\widehat{\bbm r}^{\mathrm{fed}}=(1,3,1)$. We set the privacy parameters to $(\varepsilon,\delta)=(20,0.1)$. Following Section~\ref{sec:tuning-param-selection}, the remaining tuning parameters are selected by five-fold cross-validation, and the iteration numbers are set according to the implementation guidelines therein.

Figure~\ref{fig:adhd-heatmaps} visualizes the mode-3 unfoldings of the estimated shared component $\widehat{\cm A}_0$ and the three site-specific deviations $\{\widehat{\cm B}_k\}$. Consistent with the structural decomposition in \eqref{eq:decomposition}, $\widehat{\cm A}_0$ exhibits low Tucker-rank structure and captures the imaging--symptom association shared across sites through a few multilinear spatial factors. By projecting the high-magnitude entries of $\widehat{\cm A}_0$ back to the standard brain space, we identify approximate correspondence with the left middle and superior frontal gyri, cerebellar Crus regions, temporal/fusiform areas, and the left temporal pole. These frontal, cerebellar, and temporal regions are well-documented in the ADHD neuroimaging literature as sites of structural alteration \citep{castellanos1996quantitative,berquin1998cerebellum,kobel2010temporal}. In contrast, the nonzero coefficients of each $\widehat{\cm B}_k$ constitute localized adjustments to the shared pattern, allowing the model to flexibly absorb site-specific effects induced by differences in acquisition protocols or scanner characteristics.

\begin{figure}[!htbp]
\centering
\includegraphics[width=\textwidth]{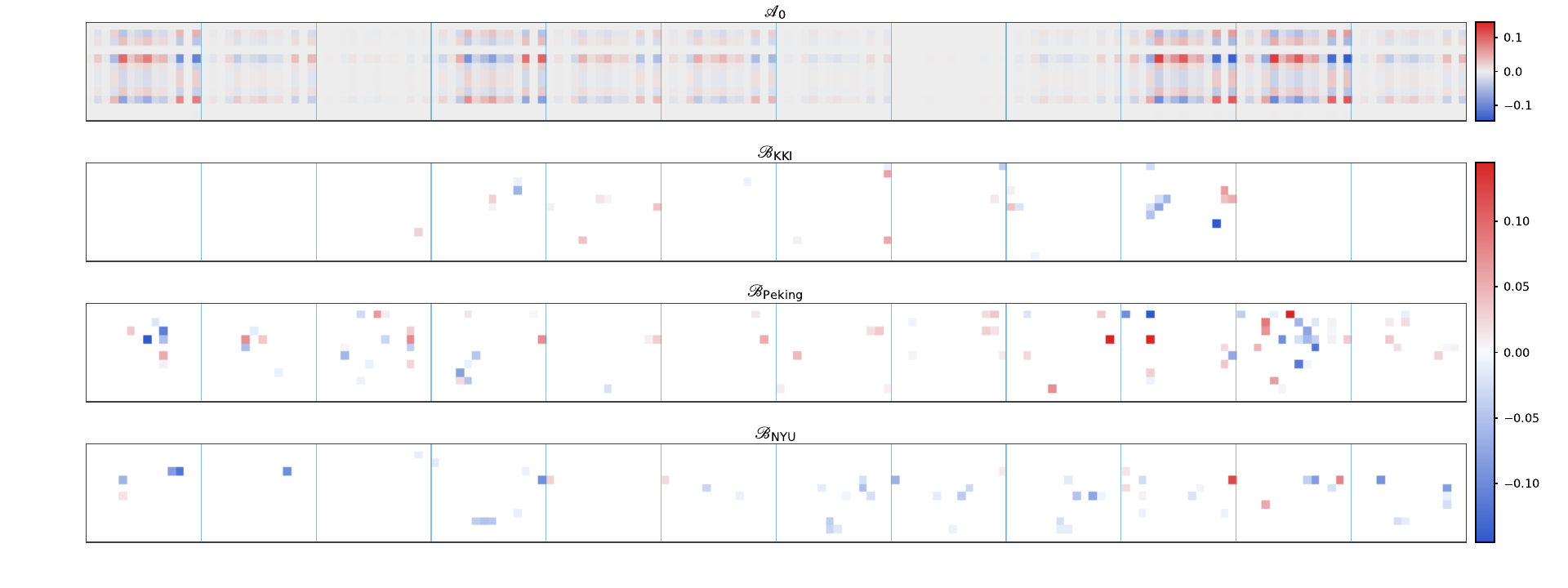}
\caption{Mode-3 unfoldings of $\widehat{\cm A}_0$ and $\{\widehat{\cm B}_k\}$ in the ADHD application.}
\label{fig:adhd-heatmaps}
\end{figure}

We next evaluate the prediction of ADHD symptom score under model~\eqref{eq:adhd-model} by comparing the proposed personalized federated two-stage method (\textsf{Fed-2Stage}) with three benchmarks: 
(i) federated averaging method (\textsf{Fed-Avg}), which aggregates the local estimates by simple averaging, and uses the resulting tensor $\widehat{\cm A}^{\mathrm{avg}}=K^{-1}\sum_{k=1}^K\widehat{\cm A}_k^{\mathrm{loc}}$ as a common predictor \citep{mcmahan2017communication};
(ii) federated common method (\textsf{Fed-Common}), which omits the site-specific deviations and uses the shared low Tucker-rank estimator from Stage-I of the federated procedure in Section~\ref{sec:fed-learning} for all sites; 
and (iii) single-client two-stage method (\textsf{Single-2Stage}) in Section~\ref{sec:single-client-learning}, which is fitted at each site using only local data. For a fair comparison, all four methods use the same fold partitions, and their tuning parameters are selected according to the validation procedure in Section~\ref{sec:tuning-param-selection}.

Prediction performance is evaluated by five-fold cross-validation. Let $\mathcal I_{k,v}^{\mathrm{val}}$ be the validation set for site $k$, and let $\widehat{\cm A}_k^{(-v)}$ be the coefficient tensor estimator for site $k$ fitted without fold $v$. The validation mean squared error is calculated by $\mathrm{VE}_{k,v} = |\mathcal I_{k,v}^{\mathrm{val}}|^{-1}\sum_{i\in\mathcal I_{k,v}^{\mathrm{val}}}\{Y_{k,i} - \langle\widehat{\cm A}_k^{(-v)},\cm X_{k,i}\rangle\}^2$ for site $k$ without fold $v$.
Then for site $k$, the mean validation error is defined by the five-fold mean $\overline{\mathrm{VE}}_k=5^{-1}\sum_{v=1}^5\mathrm{VE}_{k,v}$, and the fold-to-fold standard deviation is $\mathrm{SD}_k=\{\sum_{v=1}^5(\mathrm{VE}_{k,v}-\overline{\mathrm{VE}}_k)^2/4\}^{1/2}$. To summarize the overall prediction performance across the three sites, we also consider an equal-weight aggregate measure. For each fold $v\in[5]$, define $\mathrm{VE}_{\mathrm{avg},v}=(\mathrm{VE}_{\mathrm{KKI},v}+\mathrm{VE}_{\mathrm{Peking},v}+\mathrm{VE}_{\mathrm{NYU},v})/3$.
The mean and fold-to-fold standard deviation of this aggregate measure are $\overline{\mathrm{VE}}_{\mathrm{avg}}=5^{-1}\sum_{v=1}^5\mathrm{VE}_{\mathrm{avg},v}$ and $\mathrm{SD}_{\mathrm{avg}}=\{\sum_{v=1}^5(\mathrm{VE}_{\mathrm{avg},v}-\overline{\mathrm{VE}}_{\mathrm{avg}})^2/4\}^{1/2}$, respectively.

Table~\ref{tab:adhd_baseline_validation} reports the site-specific and aggregate validation measures defined above for all four methods. The proposed method \textsf{Fed-2Stage} achieves the smallest site-specific validation error at all three sites and the smallest equal-weight aggregate validation error. Its aggregate validation error is 0.805, representing reductions of 23.9\%, 42.3\%, and 27.5\% relative to \textsf{Fed-Avg}, \textsf{Fed-Common}, and \textsf{Single-2Stage}, respectively. The comparison with \textsf{Fed-Avg} further suggests that jointly estimating a shared low Tucker-rank component and sparse site-specific deviations is more effective than simply averaging independently fitted local coefficients. The improvement over \textsf{Fed-Common} indicates that a common federated component alone does not adequately account for site heterogeneity, while the improvement over \textsf{Single-2Stage} demonstrates the benefit of borrowing information across sites. In addition, \textsf{Fed-2Stage} has the smallest fold-to-fold standard deviations for NYU and the equal-weight aggregate measure and the second smallest for KKI, indicating generally stable predictive performance. Overall, these results support combining differentially private cross-site estimation of the shared structure with local personalization under $(\varepsilon,\delta) = (20,0.1)$.

\begin{table}[t]
\centering
\caption[Validation errors in the ADHD application.]{Five-fold site-specific and equal-weight aggregate validation errors for four methods, with fold-to-fold standard deviations in parentheses. 
Bold and underlining indicate the smallest and second smallest values, respectively.}
\label{tab:adhd_baseline_validation}
\begingroup
\setlength{\tabcolsep}{5pt}
\renewcommand{\arraystretch}{1.2}
\begin{tabular}{lcccc}
\toprule
Method & $\overline{\mathrm{VE}}_{\mathrm{KKI}}\;(\mathrm{SD}_{\mathrm{KKI}})$ & $\overline{\mathrm{VE}}_{\mathrm{Peking}}\;(\mathrm{SD}_{\mathrm{Peking}})$ & $\overline{\mathrm{VE}}_{\mathrm{NYU}}\;(\mathrm{SD}_{\mathrm{NYU}})$ & $\overline{\mathrm{VE}}_{\mathrm{avg}}\;(\mathrm{SD}_{\mathrm{avg}})$ \\
\midrule
\textsf{Fed-2Stage}
& \textbf{0.903} (\underline{0.340}) & \textbf{0.738} (0.151) & \textbf{0.773} (\textbf{0.091}) & \textbf{0.805} (\textbf{0.105}) \\
\textsf{Fed-Avg}
& \underline{1.360} (\textbf{0.225}) & \underline{0.791} (\underline{0.149}) & \underline{1.022} (\underline{0.160}) & \underline{1.058} (\underline{0.124}) \\
\textsf{Fed-Common}
& 1.504 (1.043) & 1.372 (0.326) & 1.308 (0.295) & 1.395 (0.528) \\
\textsf{Single-2Stage}
& 1.384 (0.584) & 0.811 (\textbf{0.137}) & 1.138 (0.170) & 1.111 (0.159) \\
\bottomrule
\end{tabular}
\endgroup
\end{table}

\section{Conclusion and Discussion}\label{sec:Conclusions}
This paper develops a privacy-preserving personalized federated learning framework for high-dimensional tensor regression. The proposed method decomposes each client-specific coefficient tensor into a low Tucker-rank component shared across clients and a sparse client-specific deviation, thereby capturing common multilinear structure while accommodating client heterogeneity. Estimation proceeds in two stages: differentially private federated representation learning for the shared component through noisy gradient aggregation, followed by local personalized estimation of the client-specific deviations. We establish non-asymptotic error bounds and matching minimax lower bounds for both the federated and single-client estimators, together with theoretical guarantees for the initialization estimator and Tucker-rank selection procedure. A direct comparison of the federated and single-client rates reveals that federated learning improves upon single-client estimation when the reduction in statistical error from cross-client information sharing outweighs the additional cost of privacy protection. Simulation studies corroborate the theoretical findings and demonstrate the advantages of the proposed method across varying levels of local sample sizes and privacy budgets. An empirical analysis of ADHD symptom index further shows that the federated method achieves desirable predictive performance while respecting privacy constraints, highlighting its practical value for privacy-sensitive and heterogeneous tensor data.

Several promising directions remain for future research. First, the proposed framework can be extended to accommodate mixed covariates of different tensor orders, including scalars, vectors, matrices, and higher-order tensors. This would allow the inclusion of standard control variables and demographic factors, thereby yielding a more comprehensive characterization of the association with the response. Second, it is of interest to generalize the methodology to settings with dependent data. Both the estimation procedures and the theoretical guarantees should be adapted to handle temporal or spatial dependence, which is common in longitudinal and spatially organized neuroimaging studies. Third, the current linear specification can be extended to nonlinear tensor regression models, such as additive, kernel-based, and neural network models. Such extensions would capture more flexible relationships between tensor-valued covariates and the response, broadening the scope of potential applications.





\appendix
\setcounter{section}{0}
\setcounter{subsection}{0}
\setcounter{subsubsection}{0}

\renewcommand\thesection{S.\arabic{section}}
\renewcommand\thesubsection{S.\arabic{section}.\arabic{subsection}}
\renewcommand\thesubsubsection{S.\arabic{section}.\arabic{subsection}.\arabic{subsubsection}}

\makeatletter
\renewcommand{\theHsection}{appendix.\arabic{section}}
\renewcommand{\theHsubsection}{appendix.\arabic{section}.\arabic{subsection}}
\renewcommand{\theHsubsubsection}{appendix.\arabic{section}.\arabic{subsection}.\arabic{subsubsection}}

\section*{Appendix}
\makeatother
	This supplementary material provides supporting technical details for the main theoretical and algorithmic results. Section~\ref{sec:notation} introduces the notation and theoretical constants used throughout the main paper and the supplement. Section~\ref{append:Primary lemmas} collects the auxiliary results required for proving the main theorems. Section~\ref{append:proofs-of-upper-bounds-and-consistency} establishes the upper bounds for the federated and single-client procedures, as well as the consistency of Tucker-rank selection. Section~\ref{sec:minimax-lower bound} proves the corresponding minimax lower bounds. Section~\ref{append:proofs of primary lemmas} gives the detailed proofs of the auxiliary lemmas. Section~\ref{append:technical lemmas} presents additional technical tools used in the theoretical analysis. Finally, Section~\ref{sec:Algorithmic_Details} describes the implementation of the proposed estimation procedures and the tuning-parameter selection strategy.

	\setcounter{lemma}{0}
	\setcounter{equation}{0}

\section{Notation}\label{sec:notation}

\subsection{Basic Notation}\label{sec:basic-notation}
	Throughout the paper, vectors and matrices are denoted by boldface lowercase and uppercase letters, respectively, such as $\bbm a$, $\bbm \xi$, $\bm A$, and $\bbm \Delta$. For a matrix $\bm A\in\mathbb R^{m\times n}$, $\bm A^\top$, $\|\bm A\|_\F$, and $\|\bm A\|_*$ denote its transpose, Frobenius norm, and nuclear norm. For $\nu\in[1,\infty)$, define $\|\bm A\|_\nu=(\sum_{i=1}^m\sum_{j=1}^n |A_{ij}|^\nu)^{1/\nu}$; for $\nu\in(0,1)$, the same expression is used as the entrywise $\ell_\nu$ quasi-norm. We also write $\|\bm A\|_0=\#\{(i,j):A_{ij}\neq0\}$. For positive sequences $x_n$ and $y_n$, $x_n\gtrsim y_n$ means $x_n\geq Cy_n$ for a universal constant $C>0$, and $x_n\asymp y_n$ means both $x_n\gtrsim y_n$ and $y_n\gtrsim x_n$. For a positive integer $n$, let $[n]=\{1,2,\cdots,n\}$.

	We next recall the tensor operations used in the proofs. Let $\cm M\in\mathbb R^{p_1\times\cdots\times p_\ell}$ and $S\subset[\ell]$. The multi-mode matricization $\cm M_{[S]}$ is the $\prod_{j\in S}p_j\times\prod_{j\notin S}p_j$ matrix formed by grouping the modes in $S$ as rows and the remaining modes as columns. The entry indexed by $(i_1,\cdots,i_\ell)$ is mapped to $(i,j)$ with $i=1+\sum_{s\in S}(i_s-1)I_s$ and $j=1+\sum_{s\notin S}(i_s-1)J_s$, where $I_s=\prod_{m\in S,m<s}p_m$ and $J_s=\prod_{m\notin S,m<s}p_m$. Under this convention, $\cm M_{[S]}=\cm M_{[S^c]}^\top$, where $S^c=[\ell]\setminus S$. Conversely, for any matrix $\bm Z$ conformable with the $S$-matricization, $\mathrm{Fold}_{[S]}(\bm Z)$ denotes the inverse tensorization under the same indexing convention; equivalently, if $\bm Z=\cm Z_{[S]}$, then $\mathrm{Fold}_{[S]}(\bm Z)=\cm Z$. For a single index $s\in[\ell]$, we write $\cm M_{[s]}=\cm M_{[\{s\}]}$ for the mode-$s$ unfolding and $\mathrm{Fold}_{[s]}(\cdot)=\mathrm{Fold}_{[\{s\}]}(\cdot)$ for the corresponding inverse folding operation. 

	For $\cm X\in\mathbb R^{q_1\times\cdots\times q_m\times p_1\times\cdots\times p_d}$ and $\cm Y\in\mathbb R^{p_1\times\cdots\times p_d}$, their generalized inner product is the tensor $\langle\cm X,\cm Y\rangle\in\mathbb R^{q_1\times\cdots\times q_m}$ whose $(i_1,\cdots,i_m)$ entry is $\sum_{\ell_1=1}^{p_1}\cdots\sum_{\ell_d=1}^{p_d}\cm X_{i_1,\cdots,i_m,\ell_1,\cdots,\ell_d}\cm Y_{\ell_1,\cdots,\ell_d}$. For $\cm X\in\mathbb R^{q_1\times\cdots\times q_m}$ and $\cm Y\in\mathbb R^{p_1\times\cdots\times p_d}$, their outer product $\cm X\circ\cm Y$ is the tensor in $\mathbb R^{q_1\times\cdots\times q_m\times p_1\times\cdots\times p_d}$ with entries $(\cm X\circ\cm Y)_{i_1,\cdots,i_m,\ell_1,\cdots,
	\ell_d}=\cm X_{i_1,\cdots,i_m}\cm Y_{\ell_1,\cdots,\ell_d}$. For any tensor $\cm M$, $\|\cm M\|_\F$ denotes the Frobenius norm. Entrywise $\ell_\nu$ norms are invariant under matricization, i.e., $\|\cm M\|_\nu=\|\cm M_{[S]}\|_\nu$ for any nonempty $S\subset[\ell]$ and any $\nu\in[0,\infty]$.

	We now summarize the Tucker-rank notation. If $\cm M\in\mathbb R^{p_1\times\cdots\times p_\ell}$ has Tucker rank $\bbm r=(r_1,\cdots,r_\ell)$, then $r_s=\rank(\cm M_{[s]})$ and $\tucrank(\cm M)=\bbm r$. Equivalently, $\cm M$ admits a Tucker decomposition $\cm M=\cm S\times_1\bm U_1\times_2\cdots\times_\ell\bm U_\ell=\cm S\times_{s=1}^{\ell}\bm U_s$, where $\cm S\in\mathbb R^{r_1\times\cdots\times r_\ell}$ is the core tensor and $\bm U_s\in\mathbb R^{p_s\times r_s}$ has orthonormal columns. Under the matricization convention above, $\cm M_{[s]}=\bm U_s\cm S_{[s]}(\otimes_{i\neq s}\bm U_i)^\top$ and $\cm M_{[S]}=(\otimes_{i\in S}\bm U_i)\cm S_{[S]}(\otimes_{i\notin S}\bm U_i)^\top$, where the Kronecker products follow the reverse order induced by the matricization convention.

	For tensors in $\mathbb R^{q_1\times\cdots\times q_m\times p_1\times\cdots\times p_d}$, let $S_X=[d+m]\setminus[m]$, so that the output modes form the rows and the predictor modes form the columns of the $S_X$-matricization. Thus, $\cm A_{[S_X]}\in\mathbb R^{q\times p}$, where $q=\prod_{a=1}^{m}q_a$ and $p=\prod_{b=1}^{d}p_b$, and $\vect(\langle\cm A,\cm X\rangle)=\cm A_{[S_X]}\vect(\cm X)$. For a Tucker-rank vector $\bbm r=(r_1,\cdots,r_{d+m})$, define $r_{\bbm q}=\prod_{\ell=1}^{m}r_\ell$ and $r_{\bbm p}=\prod_{\ell=1}^{d}r_{m+\ell}$. The Tucker effective dimension is $df_{\bbm r}=\prod_{s=1}^{d+m}r_s+\sum_{s=1}^{m}r_s(q_s-r_s)+\sum_{s=1}^{d}r_{m+s}(p_s-r_{m+s})$.

	Finally, we recall the tangent-space projector for the Tucker-rank manifold. Let $\cm X=\cm S\times_{k=1}^{d+m}\bm U_k$ be a Tucker-rank-$\bbm r$ tensor with orthonormal factor matrices $\{\bm U_k\}_{k=1}^{d+m}$. For each $k\in[d+m]$, set $P_{\bm U_k}=\bm U_k\bm U_k^\top$ and $P_{\bm U_k^\perp}=\bm I_{d_k}-\bm U_k\bm U_k^\top$. Let $\bm V_k\in\mathbb R^{(\prod_{j\neq k}r_j)\times r_k}$ be an orthonormal basis for the row space of the mode-$k$ unfolding $\cm S_{[k]}$, and define $\bm W_k=(\bm U_{d+m}\otimes\cdots\otimes\bm U_{k+1}\otimes\bm U_{k-1}\otimes\cdots\otimes\bm U_1)\bm V_k$ and $P_{\bm W_k}=\bm W_k\bm W_k^\top$. Thus, $P_{\bm W_k}$ projects onto the $r_k$-dimensional row space induced jointly by the core tensor and the factor matrices in all modes other than $k$.

	Let $\mathcal M_{\bbm r}$ denote the smooth manifold of tensors with Tucker rank $\bbm r$, and let $T_{\boldsymbol{\mathscr X}}\mathcal M_{\bbm r}$ be its tangent space at $\cm X$. With the matricization and folding notation introduced above, the orthogonal projection onto $T_{\boldsymbol{\mathscr X}}\mathcal M_{\bbm r}$ is
	\[
	P_{T_{\boldsymbol{\mathscr X}}\mathcal M_{\bbm r}}(\cm Z)
	=
	\cm Z\times_{k=1}^{d+m}P_{\bm U_k}
	+
	\sum_{k=1}^{d+m}
	\mathrm{Fold}_{[k]}\left(P_{\bm U_k^\perp}\,\cm Z_{[k]}\,P_{\bm W_k}\right).
	\]

	\subsection{Theoretical Constants}\label{sec:theorems-const}

	We collect the auxiliary constants used in the single-client and federated theoretical results. These constants are introduced only to state the high probability recursion and sample size requirements compactly: the quantities $\bar\delta_{\mathscr A,k}$ and $\bar\delta_{\mathscr A}$ describe the admissible curvature perturbations, the intervals $\mathcal I_{\mathscr A,k}$ and $\mathcal I_{\mathscr A}^{\mathrm{fed}}$ specify the allowable step sizes, and $R_{\mathscr A,k}$ and $R_{\mathscr A}$ give the local radii in which the contraction arguments are carried out.

	For each client $k\in[K]$, define
	\begin{align*}
		\bar\delta_{\mathscr A,k} = \frac{\lambda_{\mathscr X,k}^{\min}}{2}\left[c_{d,m}(1+\kappa_{\mathscr X,k})-(\kappa_{\mathscr X,k}-1)\right]
		\ \text{and} \ 
		\mathcal I_{\mathscr A,k} = \left[\frac{1-c_{d,m}}{\lambda_{\mathscr X,k}^{\min}-\bar\delta_{\mathscr A,k}},\frac{1+c_{d,m}}{\lambda_{\mathscr X,k}^{\max}+\bar\delta_{\mathscr A,k}}\right].
	\end{align*}
	Here $\lambda_{\mathscr X,k}^{\min}$, $\lambda_{\mathscr X,k}^{\max}$, and $\kappa_{\mathscr X,k}$ are given in Section~\ref{subsec:Theory}, and $c_{d,m}=(\sqrt{d+m}-1)/(\sqrt{d+m}+1)$ is a dimension-dependent constant. The curvature tolerance $\bar\delta_{\mathscr A,k}$ is positive whenever the client-specific condition number is controlled whenever $\kappa_{\mathscr X,k}<(1+c_{d,m})/(1-c_{d,m})$; in that case, the interval $\mathcal I_{\mathscr A,k}$ is nonempty and contains step sizes compatible with the local curvature bounds.
	For any $\eta_{\mathscr A,k}\in\mathcal I_{\mathscr A,k}$, set
	\begin{align*}
		\Gamma_{\mathscr A,k} = (\sqrt{d+m}+1)(d+m)\left[2 + \eta_{\mathscr A,k}\left(\lambda_{\mathscr X,k}^{\max}-\lambda_{\mathscr X,k}^{\min}+2\bar\delta_{\mathscr A,k}\right)\right]
		\ \text{and} \ 
		R_{\mathscr A,k}=\frac{\mu_{\mathscr A}^{\min}}{4\Gamma_{\mathscr A,k}}.
	\end{align*}
	The client-specific sample size multiplier is
	\begin{align*}
		\mathsf M_{\mathscr A,k}^{\mathrm{single}} = \max\left\{1, \frac{\sigma_{\mathscr X,k}^{4}(\lambda_{\mathscr X,k}^{\max})^{2}}{1\wedge \bar\delta_{\mathscr A,k}^{2}},\frac{\sigma_{\mathscr X,k}^{4}h_{\mathscr B,k}^{2} \vee \sigma_{\mathscr X,k}^{2}\sigma_{\mathscr E,k}^{2}\lambda_{\mathscr E,k}^{\max}/\lambda_{\mathscr X,k}^{\max}}{(\mu_{\mathscr A}^{\min})^2}\right\}.
	\end{align*}
	This multiplier gathers the constants that enter the single-client sample size lower bound.

	Define the pooled curvature tolerance $\overline \delta_{\mathscr A}$ and the federated admissible step-size interval by
	\begin{align*}
		\bar\delta_{\mathscr A} = \frac{\underline\lambda_{\mathscr X}^{\min}}{2}\left[c_{d,m}(1+\bar\kappa_{\mathscr X})-(\bar\kappa_{\mathscr X}-1)\right]
		\ \text{and} \ 
		\mathcal I_{\mathscr A}^{\mathrm{fed}} = \left[\frac{1-c_{d,m}}{\underline\lambda_{\mathscr X}^{\min}-\bar\delta_{\mathscr A}},\frac{1+c_{d,m}}{\bar\lambda_{\mathscr X}^{\max}+\bar\delta_{\mathscr A}}\right].
	\end{align*}
	To ensure $\bar\delta_{\mathscr A}>0$ and the pooled step-size interval is well defined, we require the pooled condition number $\bar\kappa_{\mathscr X}=\overline\lambda_{\mathscr X}^{\max}/\underline\lambda_{\mathscr X}^{\min}<(1+c_{d,m})/(1-c_{d,m})$. For any $\eta_{\mathscr A}\in\mathcal I_{\mathscr A}^{\mathrm{fed}}$, set
	\begin{align*}
		\Gamma_{\mathscr A}^{\mathrm{fed}} = (\sqrt{d+m}+1)(d+m)\left[2 + \eta_{\mathscr A}\left(\bar\lambda_{\mathscr X}^{\max}-\underline\lambda_{\mathscr X}^{\min}+2\bar\delta_{\mathscr A}\right)\right]
		\ \text{and} \ 
		R_{\mathscr A}=\frac{\mu_{\mathscr A}^{\min}}{4\Gamma_{\mathscr A}^{\mathrm{fed}}}.
	\end{align*}
	These are the federated analogues of $\Gamma_{\mathscr A,k}$ and $R_{\mathscr A,k}$, obtained by replacing client-specific quantities with their uniform envelopes.
	The federated sample size multiplier is
	\begin{align*}
		\mathsf M_{\mathscr A}^{\mathrm{fed}} = \max\left\{1, \frac{\bar\sigma_{\mathscr X}^{4}(\bar\lambda_{\mathscr X}^{\max})^{2}}{1\wedge \bar\delta_{\mathscr A}^{2}},\frac{\bar\sigma_{\mathscr X}^{4}\bar h_{\mathscr B}^{2} \vee \bar\sigma_{\mathscr X}^{2}\bar\sigma_{\mathscr E}^{2}\bar\lambda_{\mathscr E}^{\max}/\bar\lambda_{\mathscr X}^{\max}}{(\mu_{\mathscr A}^{\min})^2}\right\}.
	\end{align*}
	It plays the same role as $\mathsf M_{\mathscr A,k}^{\mathrm{single}}$, but uniformly over clients.

\section{Primary Lemmas}\label{append:Primary lemmas}
	\renewcommand{\theequation}{B.\arabic{equation}}
	\renewcommand{\theHequation}{B.\arabic{equation}}

	\renewcommand{\thelemma}{B.\arabic{lemma}}
	\renewcommand{\theHlemma}{B.\arabic{lemma}}

	\setcounter{lemma}{0}
	\setcounter{equation}{0}

	\begin{lemma}[Uniform empirical RSC]\label{lem:rsc-condition}
		Suppose Assumption~\ref{assump:subg-design} holds. If $n_k \gtrsim\{\sigma_{\mathscr X,k}^{4}\kappa_{\mathscr X,k}^{2}\vee \sigma_{\mathscr X,k}^{2}\kappa_{\mathscr X,k}\vee 1\}p$ with $p=\prod_{j=1}^{d}p_j$, then, with probability at least $1-\exp(-Cp)$, for all $\cm T\in\mathbb R^{q_1\times\cdots\times q_m\times p_1\times\cdots\times p_d}$,
		\[
			\frac{1}{n_k}\sum_{i=1}^{n_k}\left\|\langle\cm T,\cm X_{k,i}\rangle\right\|_\F^2 \geq \frac{1}{2}\lambda_{\mathscr X,k}^{\min}\|\cm T\|_\F^2.
		\]
	\end{lemma}

	\begin{lemma}[Operator-norm deviation bound]\label{lem:deviation-condition-op}
	Suppose Assumptions~\ref{assump:subg-design} and \ref{assump:subg-noise} hold. If $n_k \gtrsim p+q$ with $p=\prod_{j=1}^{d}p_j$ and $q=\prod_{j=1}^{m}q_j$, then, with probability at least $1-\exp(-C(p+q))$,
	\[
		\left\|\frac{1}{n_k}\sum_{i=1}^{n_k}(\cm E_{k,i}\circ \cm X_{k,i})_{[S_X]}\right\|_{\op}\lesssim \sigma_{\mathscr X,k}\sigma_{\mathscr E,k}(\lambda_{\mathscr X,k}^{\max}\lambda_{\mathscr E,k}^{\max})^{1/2}\sqrt{\frac{p+q}{n_k}}.
	\]
	\end{lemma}

	\begin{lemma}[Entrywise deviation bound]\label{lem:deviation-condition-infty}
	Suppose Assumptions~\ref{assump:subg-design} and \ref{assump:subg-noise} hold. If $n_k\gtrsim \log(pq)$ with $p=\prod_{j=1}^{d}p_j$ and $q=\prod_{j=1}^{m}q_j$, then with probability at least $1-\exp(-C\log(pq))$,
	\[
		\left\|\frac{1}{n_k}\sum_{i=1}^{n_k}\cm E_{k,i}\circ \cm X_{k,i}\right\|_\infty \lesssim \sigma_{\mathscr X,k}\sigma_{\mathscr E,k}(\lambda_{\mathscr X,k}^{\max}\lambda_{\mathscr E,k}^{\max})^{1/2}\sqrt{\frac{\log(pq)}{n_k}}.
	\]
	\end{lemma}

	\begin{lemma}[Refined sparse deviation bound]\label{lem:deviation-condition-sparse}
		Suppose Assumptions~\ref{assump:subg-design} and \ref{assump:subg-noise} hold. For any integer $1\leq s\leq pq$ with $p=\prod_{j=1}^{d}p_j$ and $q=\prod_{j=1}^{m}q_j$, if $n_k\gtrsim s\log(pq/s)$, then, with probability at least $1-\exp(-Cs\log(pq/s))$, the following bounds hold uniformly over any $\cm H\in\mathbb R^{q_1\times\cdots\times q_m\times p_1\times\cdots\times p_d}$:
		\[
			\left|\left\langle \frac{1}{n_k}\sum_{i=1}^{n_k}\cm E_{k,i}\circ\cm X_{k,i},\cm H\right\rangle\right| \lesssim \vartheta_k
			\begin{cases}
				\displaystyle
				\sqrt{\dfrac{s\log(pq/s)}{n_k}}\|\cm H\|_\F,
				& \text{if } \|\cm H\|_1\leq\sqrt{s}\|\cm H\|_\F,\\[1.2ex]
				\displaystyle
				\sqrt{\dfrac{\log(pq/s)}{n_k}}\|\cm H\|_1,
				& \text{if } \|\cm H\|_1>\sqrt{s}\|\cm H\|_\F.
			\end{cases}
		\]
	\end{lemma}

	\begin{lemma}[Restricted quadratic concentration over Tucker-rank tensors]\label{lem:rsc-condition-tucker-rank-r}
		Suppose Assumption~\ref{assump:subg-design} holds. Set $\delta_{\bbm r} = C\sigma_{\mathscr X,k}^{2}\lambda_{\mathscr X,k}^{\max}\sqrt{df_{\bbm r}/n_k}$. If $n_k \gtrsim df_{\bbm r}$ with $df_{\bbm r}=\prod_{s=1}^{d+m}r_s+\sum_{s=1}^{m}r_s(q_s-r_s)+\sum_{s=1}^{d}r_{m+s}(p_s-r_{m+s})$, then, with probability at least $1-\exp(-Cdf_{\bbm r})$, the following restricted empirical quadratic bound holds uniformly over all $\cm T\in\mathbb R^{q_1\times\cdots\times q_m\times p_1\times\cdots\times p_d}$ satisfying $\tucrank(\cm T)\leq\bbm r$:
		\[
			\left|\frac{1}{n_k}\sum_{i=1}^{n_k}\left\|\left\langle\cm T,\cm X_{k,i}\right\rangle\right\|_\F^2 - \mathbb E\left\|\left\langle\cm T,\cm X_{k,i}\right\rangle\right\|_\F^2\right| \leq \delta_{\bbm r}\|\cm T\|_\F^2,
		\]
	\end{lemma}

	\begin{lemma}[Deviation bound over Tucker-rank tensors]\label{lem:deviation-condition-tucker-rank-r}
		Suppose Assumptions~\ref{assump:subg-design} and \ref{assump:subg-noise} hold. Let $\mathbb S_{\bbm r} =\{\cm T \in \mathbb R^{q_1\times\cdots\times q_m\times p_1\times\cdots\times p_d}:\tucrank(\cm T)\leq\bbm r,\ \|\cm T\|_\F=1\}$. If $n_k \gtrsim df_{\bbm r}$ with $df_{\bbm r}=\prod_{s=1}^{d+m}r_s+\sum_{s=1}^{m}r_s(q_s-r_s)+\sum_{s=1}^{d}r_{m+s}(p_s-r_{m+s})$, with probability at least $1-\exp(-Cdf_{\bbm r})$,
		\[
			\sup_{\boldsymbol{\mathscr T} \in\mathbb S_{\bbm r}}\left|\left\langle\frac{1}{n_k}\sum_{i=1}^{n_k}\cm E_{k,i}\circ \cm X_{k,i},\cm T\right\rangle\right| \lesssim \sigma_{\mathscr X,k}\sigma_{\mathscr E,k}(\lambda_{\mathscr X,k}^{\max}\lambda_{\mathscr E,k}^{\max})^{1/2}\sqrt{\frac{df_{\bbm r}}{n_k}}.
		\]
	\end{lemma}

		\begin{lemma}[Restricted bilinear empirical-process bound]\label{lem:restricted-bilinear-event-2r}
			Suppose Assumption~\ref{assump:subg-design} holds. Define $\mathbb S_{\bbm r} =\{\cm T:\tucrank(\cm T)\leq \bbm r,\ \|\cm T\|_\F=1\}$. If $n_k\gtrsim \sigma_{\mathscr X,k}^{4}df_{\bbm r}$ with $df_{\bbm r}=\prod_{s=1}^{d+m}r_s+\sum_{s=1}^{m}r_s(q_s-r_s)+\sum_{s=1}^{d}r_{m+s}(p_s-r_{m+s})$, then, with probability at least $1-\exp(-Cdf_{\bbm r})$,
			\[
				\sup_{\boldsymbol{\mathscr U},\boldsymbol{\mathscr V}\in\mathbb S_{\bbm r}}\left|\frac{1}{n_k}\sum_{i=1}^{n_k}\left\langle\langle\cm U,\cm X_{k,i}\rangle,\langle\cm V,\cm X_{k,i}\rangle\right\rangle - \mathbb E\!\left[\left\langle\langle\cm U,\cm X_{k,i}\rangle,\langle\cm V,\cm X_{k,i}\rangle\right\rangle\right]\right| \lesssim \sigma_{\mathscr X,k}^{2}\lambda_{\mathscr X,k}^{\max}\sqrt{\frac{df_{\bbm r}}{n_k}}.
			\]
		\end{lemma}

		\begin{lemma}[Tangent-space bilinear deviation bound]\label{lem:tangent-deviation-bilinear-bound}
		Suppose Assumption~\ref{assump:subg-design} holds and recall that $h_{\mathscr B,k} = \|\cm B_k^*\|_{\infty}^{1-\nu/2}s_\nu^{1/2}$ such that $\|\cm B_k^*\|_\F\leq h_{\mathscr B,k}$ for any $\nu\in[0,1)$. If $n_k\gtrsim df_{\bbm r}$, then with probability at least $1-\exp(-Cdf_{\bbm r})$,
			\begin{align}\label{eq:tangent-deviation-bilinear-bound}
				\sup_{\substack{\boldsymbol{\mathscr U}\in\mathcal T_{\bbm r}(\boldsymbol{\mathscr A}_{0,k}^{(t)})\\ \|\boldsymbol{\mathscr U}\|_\F=1}}\left|\frac{1}{n_k}\sum_{i=1}^{n_k}\left\langle\langle \cm B_k^*,\cm X_{k,i}\rangle,\langle \cm U,\cm X_{k,i}\rangle\right\rangle - \mathbb E\bigl[\left\langle\langle \cm B_k^*,\cm X_{k,i}\rangle,\langle \cm U,\cm X_{k,i}\rangle\right\rangle\bigr]\right| \lesssim \sigma_{\mathscr X,k}^2\lambda_{\mathscr X,k}^{\max}h_{\mathscr B,k}\sqrt{\frac{df_{\bbm r}}{n_k}}.
			\end{align}
		\end{lemma}

	\begin{lemma}[high probability truncation-level bound]\label{lem:truncation-level-choice}
		Suppose Assumptions~\ref{assump:subg-design} and~\ref{assump:subg-noise} hold. Then, with probability at least $1-C\exp(-C\log n)$, uniformly over $i\in[n_k]$, $\|\cm X_{k,i}\|_{\F}\lesssim \sigma_{\mathscr X,k}\sqrt{\lambda_{\mathscr X,k}^{\max}}\sqrt{p+\log n}$ and 
		\begin{align*}
			\|\langle \cm A,\cm X_{k,i}\rangle-\cm Y_{k,i}\|_{\F} \lesssim \left((R_{\mathscr A,k}+h_{\mathscr B,k})\sigma_{\mathscr X,k}\sqrt{\lambda_{\mathscr X,k}^{\max}}\sqrt{p+\log n} + \sigma_{\mathscr E,k}\sqrt{\lambda_{\mathscr E,k}^{\max}}\sqrt{q+\log n}\right),
		\end{align*}
		for all $\cm A$ satisfying $\|\cm A-\cm A_0^*\|_{\F}\leq R_{\mathscr A,k}$ where $R_{\mathscr A,k}$ is defined in Section~\ref{sec:theorems-const} and $h_{\mathscr B,k} = \|\cm B_k^*\|_{\infty}^{1-\nu/2}s_\nu^{1/2}$.
	\end{lemma}

		\begin{lemma}[Contraction on the Tucker tangent space]\label{lem:tangent-contraction-tensor-general}
			Let $\cm X_1,\cdots,\cm X_n\in\mathbb R^{p_1\times\cdots\times p_d}$, and define $\mathcal H_n(\cm Z) = n^{-1}\sum_{i=1}^n \langle \cm Z,\cm X_i\rangle\circ\cm X_i$ for $\cm Z\in\mathbb R^{q_1\times\cdots\times q_m\times p_1\times\cdots\times p_d}$.
			Fix any Tucker-rank-$\bbm r$ tensor $\cm A\in\mathbb R^{q_1\times\cdots\times q_m\times p_1\times\cdots\times p_d}$. Assume that, for any $\cm U\in \mathcal T_{\bbm r}(\cm A)$,
			\[
			C_{\mathrm{tan}}^{\min}\|\cm U\|_\F^2 \leq \frac{1}{n}\sum_{i=1}^n \|\langle \cm U,\cm X_i\rangle\|_\F^2 \leq C_{\mathrm{tan}}^{\max}\|\cm U\|_\F^2.
			\]
			Then, for any $\cm Z\in\mathcal T_{\bbm r}(\cm A)$ and $\eta>0$,
			\[
				\Big\|\cm Z-\eta\,\mathcal P_{\mathcal T_{\bbm r}(\boldsymbol{\mathscr A})}\bigl(\mathcal H_n(\cm Z)\bigr)\Big\|_\F \leq \max\left\{|1-\eta C_{\mathrm{tan}}^{\min}|,|1-\eta C_{\mathrm{tan}}^{\max}|\right\}\|\cm Z\|_\F.
			\]
		\end{lemma}

	\begin{lemma}[Nuclear-norm bound on tangent vectors]\label{lem:tangent-space-nuclear-bound}
		Let $\cm A\in\mathbb R^{q_1\times\cdots\times q_m\times p_1\times\cdots\times p_d}$ with $\tucrank(\cm A)\leq\bbm r$. Then, for any $\cm U\in\mathcal T_{\bbm r}(\cm A)$, $\|\cm U_{[S_X]}\|_* \leq \sqrt{r_{\bbm q}+r_{\bbm p}}\,\|\cm U\|_\F$.
	\end{lemma}

	\begin{lemma}[Tucker rank of the normal projection]\label{lem:normal-projection-tucker-rank-2r}
		Let $\cm A,\cm A^*\in\mathbb R^{q_1\times\cdots\times q_m\times p_1\times\cdots\times p_d}$ be two Tucker-rank-$\bbm r$ tensors. Then
		$
			\tucrank\left(\mathcal P_{\mathcal T_{\bbm r}^{\perp}(\boldsymbol{\mathscr A})}(\cm A-\cm A^*)\right) \leq 2\bbm r.
		$
	\end{lemma}

	\begin{lemma}[Tucker rank of tangent-space elements]\label{lem:tangent-space-tucker-rank-2r}
		Let $\cm A\in\mathbb R^{q_1\times\cdots\times q_m\times p_1\times\cdots\times p_d}$ be a Tucker-rank-$\bbm r$ tensor. Then, for any $\cm T\in\mathcal T_{\bbm r}(\cm A)$, $\tucrank(\cm T)\leq 2\bbm r$.
	\end{lemma}

\section{Proofs of Upper Bounds and Consistency Results}\label{append:proofs-of-upper-bounds-and-consistency}

	\renewcommand{\theequation}{C.\arabic{equation}}
	\renewcommand{\theHequation}{C.\arabic{equation}}

	\renewcommand{\thetheorem}{C.\arabic{theorem}}
	\renewcommand{\theHtheorem}{C.\arabic{theorem}}

	\setcounter{theorem}{0}
	\setcounter{equation}{0}

	We organize this section as follows. We first prove the single-client Stage-I representation error bound in Theorem~\ref{thm:representation_single_client} and then extend the same arguments to the federated setting in Theorem~\ref{thm:federated_representation_error}. Next, we prove the personalized estimation error bound in Theorem~\ref{thm:personalized_error}, which refines the shared representation estimator by estimating the client-specific sparse deviation. Finally, we establish the one-stage initialization error bound in Theorem~\ref{thm:one-stage-initialization-error} and the Tucker-rank selection consistency result in Theorem~\ref{thm:rank_selection}.

	\begin{proof}[\textbf{Proof of Theorem~\ref{thm:representation_single_client}}]
		We organize the proof into four steps. Step~1 introduces the deterministic notation used throughout the arguments. Step~2 defines the high probability event on which all empirical processes are controlled uniformly over deterministic rank classes. Conditional on these events, Step~3 derives the one-step recursion directly for the $t$-th projected-gradient iterate. Step~4 uses this recursion to derive the final error bound.
		
		\noindent
		\textbf{\emph{Step 1: notation and auxiliary constants.}} We first write the Euclidean gradient at the $t$-th local Stage-I iterate as
		\[
			\cm G_{\mathscr A,k}^{loc,(t)} = \frac{1}{n_k}\sum_{i=1}^{n_k}\bigl\langle\cm A_0^{(t)}-\cm A_0^*-\cm B_k^*,\cm X_{k,i}\bigr\rangle\circ\cm X_{k,i} - \frac{1}{n_k}\sum_{i=1}^{n_k}\cm E_{k,i}\circ\cm X_{k,i}.
		\]
		For completeness, we recall the client-specific constants used in the recursion. Recall that $S_X=[d+m]\setminus[m]$, $r_{\bbm q}=\prod_{\ell=1}^{m}r_\ell$, $r_{\bbm p}=\prod_{\ell=1}^{d}r_{m+\ell}$, and the Tucker degrees of freedom at rank $\bbm r$ is $df_{\bbm r} = \prod_{s=1}^{d+m}r_s + \sum_{s=1}^{m}r_s(q_s-r_s) + \sum_{s=1}^{d}r_{m+s}(p_s-r_{m+s})$. Also recall that $\mu_{\mathscr A}^{\min} = \min_{s\in[d+m]}\sigma_{r_s}\bigl((\cm A_0^*)_{[s]}\bigr)$,
		\[
			\delta_{2\bbm r} = C\sigma_{\mathscr X,k}^{2}\lambda_{\mathscr X,k}^{\max}\sqrt{\frac{df_{2\bbm r}}{n_k}}
			\ \text{and} \ 
			\bar\delta_{\mathscr A,k} = \frac{\lambda_{\mathscr X,k}^{\min}}{2}\left[c_{d,m}(1+\kappa_{\mathscr X,k})-(\kappa_{\mathscr X,k}-1)\right].
		\]
		Under $\kappa_{\mathscr X,k}<(1+c_{d,m})/(1-c_{d,m})$, $\overline \delta_{\mathscr A,k}>0$ and the admissible step-size interval is
		\[
			\mathcal I_{\mathscr A,k} = \left[\frac{1-c_{d,m}}{\lambda_{\mathscr X,k}^{\min}-\bar\delta_{\mathscr A,k}},\frac{1+c_{d,m}}{\lambda_{\mathscr X,k}^{\max}+\bar\delta_{\mathscr A,k}}\right].
		\]
		For any $\eta_{\mathscr A,k}\in\mathcal I_{\mathscr A,k}$, recall that $\Gamma_{\mathscr A,k} = (\sqrt{d+m}+1)(d+m)[2+\eta_{\mathscr A,k}(\lambda_{\mathscr X,k}^{\max}-\lambda_{\mathscr X,k}^{\min}+2\bar\delta_{\mathscr A,k})]$ and $R_{\mathscr A,k} = \mu_{\mathscr A}^{\min}/(4\Gamma_{\mathscr A,k})$. Finally, under the weak sparsity condition on the client-specific deviation, $\|\cm B_k^*\|_{\F}\leq h_{\mathscr B,k}$ with $h_{\mathscr B,k}=\|\cm B_k^*\|_{\infty}^{1-\nu/2}s_{\nu}^{1/2}$.

		\noindent
		\textbf{\emph{Step 2: the high probability event.}} Since $df_{2\bbm r}\asymp df_{\bbm r}$, the theorem's sample size condition implies that
		\begin{align}\label{eq:sample size-condition-tucker-rank-2r}
			n_k \gtrsim \max\left\{1,\sigma_{\mathscr X,k}^{4},\sigma_{\mathscr X,k}^{4}(\lambda_{\mathscr X,k}^{\max})^2,\frac{\sigma_{\mathscr X,k}^{4}(\lambda_{\mathscr X,k}^{\max})^2}{\bar\delta_{\mathscr A,k}^2},\sigma_{\mathscr X,k}^{4}h_{\mathscr B,k}^{2},\sigma_{\mathscr E,k}^{2}\lambda_{\mathscr X,k}^{\max}\lambda_{\mathscr E,k}^{\max}\right\}df_{2\bbm r}.
		\end{align}
		We next define the following four events, which control the empirical processes over the Tucker-rank-$2\bbm r$ class:
		\begin{align*}
			&\mathcal E_{\mathscr A,k}^{(1)} = \Biggl\{\sup_{\boldsymbol{\mathscr U}\in\mathbb S_{2\bbm r}}\left|\frac{1}{n_k}\sum_{i=1}^{n_k}\left\langle\langle \cm B_k^*,\cm X_{k,i}\rangle,\langle \cm U,\cm X_{k,i}\rangle\right\rangle - \mathbb E\!\left[\left\langle\langle \cm B_k^*,\cm X_{k,i}\rangle,\langle \cm U,\cm X_{k,i}\rangle\right\rangle\right]\right| \leq C\sigma_{\mathscr X,k}^{2}\lambda_{\mathscr X,k}^{\max}h_{\mathscr B,k}\sqrt{\frac{df_{2\bbm r}}{n_k}}\Biggr\},\\
			&\mathcal E_{\mathscr A,k}^{(2)} = \Biggl\{\sup_{\boldsymbol{\mathscr U}\in\mathbb S_{2\bbm r}}\left|\left\langle\frac{1}{n_k}\sum_{i=1}^{n_k}\cm E_{k,i}\circ\cm X_{k,i},\cm U\right\rangle\right| \leq C\sigma_{\mathscr X,k}\sigma_{\mathscr E,k}(\lambda_{\mathscr X,k}^{\max}\lambda_{\mathscr E,k}^{\max})^{1/2}\sqrt{\frac{df_{2\bbm r}}{n_k}}\Biggr\},\\
			&\mathcal E_{\mathscr A,k}^{(3)} = \Biggl\{\left|\frac{1}{n_k}\sum_{i=1}^{n_k}\left\|\left\langle\cm U,\cm X_{k,i}\right\rangle\right\|_\F^2 - \mathbb E\left\|\left\langle\cm U,\cm X_{k,i}\right\rangle\right\|_\F^2\right| \leq \delta_{2\bbm r}\|\cm U\|_\F^2,\forall\, \tucrank(\cm U)\leq 2\bbm r\Biggr\},\\
			&\mathcal E_{\mathscr A,k}^{(4)} = \Biggl\{\sup_{\boldsymbol{\mathscr U},\boldsymbol{\mathscr V}\in\mathbb S_{2\bbm r}}\left|\frac{1}{n_k}\sum_{i=1}^{n_k}\left\langle\langle\cm U,\cm X_{k,i}\rangle,\langle\cm V,\cm X_{k,i}\rangle\right\rangle - \mathbb E\!\left[\left\langle\langle\cm U,\cm X_{k,i}\rangle,\langle\cm V,\cm X_{k,i}\rangle\right\rangle\right]\right| \leq \delta_{2\bbm r}\Biggr\}.
		\end{align*}

		Given the sample size condition in \eqref{eq:sample size-condition-tucker-rank-2r}, the four events hold with probability at least $1-\exp(-Cdf_{2\bbm r})$, respectively, after applying Lemmas~\ref{lem:tangent-deviation-bilinear-bound}, \ref{lem:deviation-condition-tucker-rank-r}, \ref{lem:rsc-condition-tucker-rank-r}, and \ref{lem:restricted-bilinear-event-2r} to the Tucker-rank-$2\bbm r$ class. Define $\mathcal E_{\mathscr A,k} = \mathcal E_{\mathscr A,k}^{(1)} \cap \mathcal E_{\mathscr A,k}^{(2)} \cap \mathcal E_{\mathscr A,k}^{(3)} \cap \mathcal E_{\mathscr A,k}^{(4)}$. By the union bound, after adjusting the universal constants, $\mathbb P(\mathcal E_{\mathscr A,k}) \geq 1-C\exp(-Cdf_{2\bbm r})$. The rest of the proof is conditional on $\mathcal E_{\mathscr A,k}$.

		\noindent
		\textbf{\emph{Step 3: one-step recursion on $\mathcal E_{\mathscr A,k}$.}} Conditional on $\mathcal E_{\mathscr A,k}$, then for any $k\in[K]$, we have
		\begin{align*}
			\|\cm A_{0,k}^{(t+1)} - \cm A_0^*\|_\F & = \Bigl\|\mathcal R_{\bbm r}\Bigl(\cm A_{0,k}^{(t)} - \eta_{\mathscr A,k} \mathcal P_{\mathcal T_{\bbm r}(\boldsymbol{\mathscr A}_{0,k}^{(t)})}(\cm G_{\mathscr A,k}^{loc,(t)})\Bigr) - \cm A_0^*\Bigr\|_\F \\
			& \leq \Bigl\|\mathcal R_{\bbm r}\Bigl(\cm A_{0,k}^{(t)} - \eta_{\mathscr A,k} \mathcal P_{\mathcal T_{\bbm r}(\boldsymbol{\mathscr A}_{0,k}^{(t)})}(\cm G_{\mathscr A,k}^{loc,(t)})\Bigr) - \Bigl(\cm A_{0,k}^{(t)} - \eta_{\mathscr A,k} \mathcal P_{\mathcal T_{\bbm r}(\boldsymbol{\mathscr A}_{0,k}^{(t)})}(\cm G_{\mathscr A,k}^{loc,(t)})\Bigr)\Bigr\|_\F \\
			&\quad + \Bigl\|\cm A_{0,k}^{(t)} - \eta_{\mathscr A,k} \mathcal P_{\mathcal T_{\bbm r}(\boldsymbol{\mathscr A}_{0,k}^{(t)})}(\cm G_{\mathscr A,k}^{loc,(t)}) - \cm A_0^*\Bigr\|_\F \\
			&\overset{(i)}{\leq}\sqrt{d+m}\,\Bigl\|\mathcal P_{\mathcal M_{\bbm r}}\Bigl(
			\cm A_{0,k}^{(t)} - \eta_{\mathscr A,k} \mathcal P_{\mathcal T_{\bbm r}(\boldsymbol{\mathscr A}_{0,k}^{(t)})}(\cm G_{\mathscr A,k}^{loc,(t)})\Bigr) - \Bigl(\cm A_{0,k}^{(t)} - \eta_{\mathscr A,k} \mathcal P_{\mathcal T_{\bbm r}(\boldsymbol{\mathscr A}_{0,k}^{(t)})}(\cm G_{\mathscr A,k}^{loc,(t)})\Bigr)\Bigr\|_\F \\
			&\quad + \Bigl\|\cm A_{0,k}^{(t)} - \eta_{\mathscr A,k} \mathcal P_{\mathcal T_{\bbm r}(\boldsymbol{\mathscr A}_{0,k}^{(t)})}(\cm G_{\mathscr A,k}^{loc,(t)}) - \cm A_0^*\Bigr\|_\F \\
			&\overset{(ii)}{\leq}(\sqrt{d+m}+1)\Bigl\|\cm A_{0,k}^{(t)} - \eta_{\mathscr A,k} \mathcal P_{\mathcal T_{\bbm r}(\boldsymbol{\mathscr A}_{0,k}^{(t)})}(\cm G_{\mathscr A,k}^{loc,(t)}) - \cm A_0^*\Bigr\|_\F \\
			& = (\sqrt{d+m}+1)\Biggl\|\mathcal P_{\mathcal T_{\bbm r}(\boldsymbol{\mathscr A}_{0,k}^{(t)})}(\cm A_{0,k}^{(t)}-\cm A_0^*) + \mathcal P_{\mathcal T_{\bbm r}^{\perp}(\boldsymbol{\mathscr A}_{0,k}^{(t)})}(\cm A_{0,k}^{(t)}-\cm A_0^*) \\
			&\qquad - \eta_{\mathscr A,k}\mathcal P_{\mathcal T_{\bbm r}(\boldsymbol{\mathscr A}_{0,k}^{(t)})}\Biggl(\frac{1}{n_k}\sum_{i=1}^{n_k}\bigl\langle\cm A_{0,k}^{(t)}-\cm A_0^*-\cm B_k^*,\cm X_{k,i}\bigr\rangle\circ \cm X_{k,i} - \frac{1}{n_k}\sum_{i=1}^{n_k}\cm E_{k,i}\circ \cm X_{k,i}\Biggr)\Biggr\|_\F \\
			& \leq (\sqrt{d+m}+1)\Biggl(\underbrace{\Bigl\|\mathcal P_{\mathcal T_{\bbm r}^{\perp}(\boldsymbol{\mathscr A}_{0,k}^{(t)})}(\cm A_{0,k}^{(t)}-\cm A_0^*)\Bigr\|_\F}_{\mathrm{Term\ I}} \\
			&\qquad + \underbrace{\eta_{\mathscr A,k}\Bigl\|\mathcal P_{\mathcal T_{\bbm r}(\boldsymbol{\mathscr A}_{0,k}^{(t)})}\Bigl(\frac{1}{n_k}\sum_{i=1}^{n_k}\langle \cm B_k^*,\cm X_{k,i}\rangle \circ \cm X_{k,i}\Bigr)\Bigr\|_\F}_{\mathrm{Term\ II}} \\
			&\qquad + \underbrace{\eta_{\mathscr A,k}\Bigl\|\mathcal P_{\mathcal T_{\bbm r}(\boldsymbol{\mathscr A}_{0,k}^{(t)})}\Bigl(\frac{1}{n_k}\sum_{i=1}^{n_k}\cm E_{k,i}\circ \cm X_{k,i}\Bigr)\Bigr\|_\F}_{\mathrm{Term\ III}} \\
			&\qquad + \underbrace{\Bigl\|\mathcal P_{\mathcal T_{\bbm r}(\boldsymbol{\mathscr A}_{0,k}^{(t)})}\Bigl(\cm A_{0,k}^{(t)}-\cm A_0^* - \eta_{\mathscr A,k} \frac{1}{n_k}\sum_{i=1}^{n_k}\bigl\langle\mathcal P_{\mathcal T_{\bbm r}(\boldsymbol{\mathscr A}_{0,k}^{(t)})}(\cm A_{0,k}^{(t)}-\cm A_0^*),\cm X_{k,i}\bigr\rangle \circ \cm X_{k,i}\Bigr)\Bigr\|_\F}_{\mathrm{Term\ IV}} \\
			&\qquad + \underbrace{\eta_{\mathscr A,k}\Bigl\|\mathcal P_{\mathcal T_{\bbm r}(\boldsymbol{\mathscr A}_{0,k}^{(t)})}\Bigl(\frac{1}{n_k}\sum_{i=1}^{n_k}\bigl\langle\mathcal P_{\mathcal T_{\bbm r}^{\perp}(\boldsymbol{\mathscr A}_{0,k}^{(t)})}(\cm A_{0,k}^{(t)}-\cm A_0^*),\cm X_{k,i}\bigr\rangle\circ \cm X_{k,i}\Bigr)\Bigr\|_\F}_{\mathrm{Term\ V}}\Biggr).
		\end{align*}
		In inequality $(i)$, $\mathcal M_{\bbm r}=\{\cm A:\tucrank(\cm A)\leq\bbm r\}$ denotes the Tucker-rank-$\bbm r$ tensor set, and $\mathcal P_{\mathcal M_{\bbm r}}(\cm Z)$ denotes a best Frobenius-norm projection of $\cm Z$ onto $\mathcal M_{\bbm r}$. Inequality $(i)$ follows from the quasi-projection property in Lemma~\ref{lem:quasi-projection-thosvd-sthosvd}. Inequality $(ii)$ follows from Lemma~\ref{lem:quasi-projection-thosvd-sthosvd} by using $\cm A_0^*\in\mathcal M_{\bbm r}$ as a feasible Tucker-rank-$\bbm r$ approximation of $\cm A_{0,k}^{(t)} - \eta_{\mathscr A,k}\mathcal P_{\mathcal T_{\bbm r}(\boldsymbol{\mathscr A}_{0,k}^{(t)})}(\cm G_{\mathscr A,k}^{loc,(t)})$. Next, we derive upper bounds for \textbf{Term I}--\textbf{Term V} separately.

		For \textbf{Term I}, by Lemma \ref{lem:normal-component-bound}, we have
		\begin{align}\label{eq:term-I-bound}
			\Bigl\|\mathcal P_{\mathcal T_{\bbm r}^{\perp}(\boldsymbol{\mathscr A}_{0,k}^{(t)})}\bigl(\cm A_{0,k}^{(t)}-\cm A_0^*\bigr)\Bigr\|_\F \leq \frac{2(d+m)}{\mu_{\mathscr A}^{\min}}\|\cm A_{0,k}^{(t)}-\cm A_0^*\|_\F^2.
		\end{align}

		For \textbf{Term II}, by the variational characterization of the Frobenius norm on the tangent space,
		\begin{align*}
			&\Bigl\|\mathcal P_{\mathcal T_{\bbm r}(\boldsymbol{\mathscr A}_{0,k}^{(t)})}\Bigl(\frac{1}{n_k}\sum_{i=1}^{n_k}\langle \cm B_k^*,\cm X_{k,i}\rangle\circ \cm X_{k,i}\Bigr)\Bigr\|_\F = \sup_{\substack{\boldsymbol{\mathscr U}\in\mathcal T_{\bbm r}(\boldsymbol{\mathscr A}_{0,k}^{(t)})\\ \|\boldsymbol{\mathscr U}\|_\F=1}}\left|\frac{1}{n_k}\sum_{i=1}^{n_k}\left\langle\langle \cm B_k^*,\cm X_{k,i}\rangle,\langle \cm U,\cm X_{k,i}\rangle\right\rangle\right| \\
			&\quad\leq \sup_{\substack{\boldsymbol{\mathscr U}\in\mathcal T_{\bbm r}(\boldsymbol{\mathscr A}_{0,k}^{(t)})\\ \|\boldsymbol{\mathscr U}\|_\F=1}}\left|\mathbb E\bigl[\left\langle\langle \cm B_k^*,\cm X_{k,i}\rangle,\langle \cm U,\cm X_{k,i}\rangle\right\rangle\bigr]\right| + \sup_{\substack{\boldsymbol{\mathscr U}\in\mathcal T_{\bbm r}(\boldsymbol{\mathscr A}_{0,k}^{(t)})\\ \|\boldsymbol{\mathscr U}\|_\F=1}}\left|\frac{1}{n_k}\sum_{i=1}^{n_k}\left\langle\langle \cm B_k^*,\cm X_{k,i}\rangle,\langle \cm U,\cm X_{k,i}\rangle\right\rangle - \mathbb E\bigl[\left\langle\langle \cm B_k^*,\cm X_{k,i}\rangle,\langle \cm U,\cm X_{k,i}\rangle\right\rangle\bigr]\right|.
		\end{align*}
		Let $\bbm x_{k,i}=\vect(\cm X_{k,i})$. Under the $S_X$-matricization convention $\vect(\langle\cm A,\cm X\rangle)=\cm A_{[S_X]}\vect(\cm X)$, we have
		\begin{align}\label{eq:expectation-B-k-U-X}
			\left|\mathbb E\left[\left\langle\langle\cm B_k^*,\cm X_{k,i}\rangle,\langle\cm U,\cm X_{k,i}\rangle\right\rangle\right]\right| &= \left|\mathbb E\left[\left((\cm B_k^*)_{[S_X]}\bbm x_{k,i}\right)^{\top}\left(\cm U_{[S_X]}\bbm x_{k,i}\right)\right]\right|\notag \\
			&= \left|\mathbb E\left[\bbm x_{k,i}^{\top}(\cm B_k^*)_{[S_X]}^{\top}\cm U_{[S_X]}\bbm x_{k,i}\right]\right|\notag \\
			&= \left|\mathbb E\left[\tr\left((\cm B_k^*)_{[S_X]}^{\top}\cm U_{[S_X]}\bbm x_{k,i}\bbm x_{k,i}^{\top}\right)\right]\right|\notag \\
			&= \left|\tr\left((\cm B_k^*)_{[S_X]}^{\top}\cm U_{[S_X]}\mathbb E\left(\bbm x_{k,i}\bbm x_{k,i}^{\top}\right)\right)\right|\notag \\
			&= \left|\tr\left((\cm B_k^*)_{[S_X]}^{\top}\cm U_{[S_X]}\bbm\Sigma_{\mathscr X,k}\right)\right|\notag \\
			&= \left|\tr\left((\cm B_k^*)_{[S_X]}\bbm\Sigma_{\mathscr X,k}\cm U_{[S_X]}^{\top}\right)\right|\notag \\
			&\leq \left\|(\cm B_k^*)_{[S_X]}\bbm\Sigma_{\mathscr X,k}\right\|_{\op}\left\|\cm U_{[S_X]}\right\|_*\notag \\
			&\leq \lambda_{\mathscr X,k}^{\max}\left\|(\cm B_k^*)_{[S_X]}\right\|_{\op}\left\|\cm U_{[S_X]}\right\|_*.
		\end{align}
		Since $\cm U\in\mathcal T_{\bbm r}(\boldsymbol{\mathscr A}_{0,k}^{(t)})$, Lemma~\ref{lem:tangent-space-nuclear-bound} gives $\|\cm U_{[S_X]}\|_*\leq \sqrt{r_{\bbm q}+r_{\bbm p}}\|\cm U\|_\F$. Together with the weak-identifiability condition in Assumption~\ref{assump:weak-identifiability}, namely $\|(\cm B_k^*)_{[S_X]}\|_{\op}\leq\zeta$, this implies $\left|\mathbb E\bigl[\left\langle\langle \cm B_k^*,\cm X_{k,i}\rangle,\langle \cm U,\cm X_{k,i}\rangle\right\rangle\bigr]\right| \leq\sqrt{r_{\bbm q}+r_{\bbm p}} \lambda_{\mathscr X,k}^{\max}\zeta$ for every $\cm U\in\mathcal T_{\bbm r}(\boldsymbol{\mathscr A}_{0,k}^{(t)})$ with $\|\cm U\|_\F=1$. Therefore, on the event $\mathcal E_{\mathscr A,k}^{(1)}$, we have
		\begin{align}\label{eq:term-II-bound}
			\textbf{Term II} \leq \eta_{\mathscr A,k}\sqrt{r_{\bbm q}+r_{\bbm p}}\,\lambda_{\mathscr X,k}^{\max}\zeta + C\eta_{\mathscr A,k}\sigma_{\mathscr X,k}^2\lambda_{\mathscr X,k}^{\max}h_{\mathscr B,k}\sqrt{\frac{df_{2\bbm r}}{n_k}}.
		\end{align}

		For \textbf{Term III}, by the variational characterization of the Frobenius norm on the tangent space, we have
		\begin{align*}
			\textbf{Term III} = \eta_{\mathscr A,k}\left\|\mathcal P_{\mathcal T_{\bbm r}(\boldsymbol{\mathscr A}_{0,k}^{(t)})}\left(\frac{1}{n_k}\sum_{i=1}^{n_k}\cm E_{k,i}\circ \cm X_{k,i}\right)\right\|_\F = \eta_{\mathscr A,k}\sup_{\substack{\boldsymbol{\mathscr T}\in\mathcal T_{\bbm r}(\boldsymbol{\mathscr A}_{0,k}^{(t)})\\ \|\boldsymbol{\mathscr T}\|_\F\leq 1}}\left|\left\langle\frac{1}{n_k}\sum_{i=1}^{n_k}\cm E_{k,i}\circ \cm X_{k,i},\cm T\right\rangle\right|.
		\end{align*}
		By Lemma~\ref{lem:tangent-space-tucker-rank-2r}, every element in $\mathcal T_{\bbm r}(\boldsymbol{\mathscr A}_{0,k}^{(t)})$ has Tucker rank at most $2\bbm r$. Moreover, since the functional inside the inner product is linear in $\cm T$, the supremum over the Tucker-rank-$2\bbm r$ unit ball is the same as the supremum over $\mathbb S_{2\bbm r}$. Therefore, on the event $\mathcal E_{\mathscr A,k}^{(2)}$,
		\begin{align}\label{eq:term-III-bound}
			\textbf{Term III} \leq C\eta_{\mathscr A,k}\sigma_{\mathscr X,k}\sigma_{\mathscr E,k}(\lambda_{\mathscr X,k}^{\max}\lambda_{\mathscr E,k}^{\max})^{1/2}\sqrt{\frac{df_{2\bbm r}}{n_k}}.
		\end{align}

		For \textbf{Term IV}, note that it can be written as
		\[
			\Biggl\|\mathcal P_{\mathcal T_{\bbm r}(\boldsymbol{\mathscr A}_{0,k}^{(t)})}(\cm A_{0,k}^{(t)}-\cm A_0^*) - \eta_{\mathscr A,k}\mathcal P_{\mathcal T_{\bbm r}(\boldsymbol{\mathscr A}_{0,k}^{(t)})}\Biggl(\frac{1}{n_k}\sum_{i=1}^{n_k}\bigl\langle\mathcal P_{\mathcal T_{\bbm r}(\boldsymbol{\mathscr A}_{0,k}^{(t)})}(\cm A_{0,k}^{(t)}-\cm A_0^*),\cm X_{k,i}\bigr\rangle\circ \cm X_{k,i}\Biggr)\Biggr\|_\F.
		\]

		By Lemma~\ref{lem:tangent-space-tucker-rank-2r}, every element of $\mathcal T_{\bbm r}(\boldsymbol{\mathscr A}_{0,k}^{(t)})$ has Tucker rank at most $2\bbm r$. Thus, on the event $\mathcal E_{\mathscr A,k}^{(3)}$, the following restricted empirical quadratic bound holds uniformly over all $\cm T\in\mathcal T_{\bbm r}(\boldsymbol{\mathscr A}_{0,k}^{(t)})$:
		\begin{align}\label{eq:restricted-empirical-quadratic-bound}
			\left|\frac{1}{n_k}\sum_{i=1}^{n_k}\left\|\left\langle\cm T,\cm X_{k,i}\right\rangle\right\|_\F^2 - \mathbb E\left\|\left\langle\cm T,\cm X_{k,i}\right\rangle\right\|_\F^2\right| \leq \delta_{2\bbm r}\|\cm T\|_\F^2.
		\end{align}
		Moreover, $\mathbb E\left\|\left\langle\cm T,\cm X_{k,i}\right\rangle\right\|_\F^2 = \mathbb E[\vect^\top(\cm X_{k,i})\cm T_{[S_X]}^\top\cm T_{[S_X]}\vect(\cm X_{k,i})] = \tr(\cm T_{[S_X]}^\top\cm T_{[S_X]}\bbm\Sigma_{\mathscr X,k})$.
		Since the eigenvalues of $\bbm\Sigma_{\mathscr X,k}$ are between $\lambda_{\mathscr X,k}^{\min}$ and $\lambda_{\mathscr X,k}^{\max}$, we have
		\begin{align}\label{eq:single-expectation-quadratic-bound}
			\lambda_{\mathscr X,k}^{\min}\|\cm T\|_\F^2 \leq \mathbb E\left\|\left\langle\cm T,\cm X_{k,i}\right\rangle\right\|_\F^2 \leq \lambda_{\mathscr X,k}^{\max}\|\cm T\|_\F^2.
		\end{align}
		Combining \eqref{eq:restricted-empirical-quadratic-bound} and \eqref{eq:single-expectation-quadratic-bound} gives $(\lambda_{\mathscr X,k}^{\min}-\delta_{2\bbm r})\|\cm T\|_\F^2 \leq n_k^{-1}\sum_{i=1}^{n_k}\left\|\left\langle\cm T,\cm X_{k,i}\right\rangle\right\|_\F^2 \leq (\lambda_{\mathscr X,k}^{\max}+\delta_{2\bbm r})\|\cm T\|_\F^2$,
		where $\delta_{2\bbm r}$ is as defined above. Therefore, applying Lemma~\ref{lem:tangent-contraction-tensor-general} with $\cm T = \mathcal P_{\mathcal T_{\bbm r}(\boldsymbol{\mathscr A}_{0,k}^{(t)})}(\cm A_{0,k}^{(t)}-\cm A_0^*)$, we have
		\begin{align}\label{eq:term-IV-bound}
			\textbf{Term IV} \leq \rho_{\mathscr A,k}\left\|\mathcal P_{\mathcal T_{\bbm r}(\boldsymbol{\mathscr A}_{0,k}^{(t)})}(\cm A_{0,k}^{(t)}-\cm A_0^*)\right\|_\F \leq \rho_{\mathscr A,k}\|\cm A_{0,k}^{(t)}-\cm A_0^*\|_\F,
		\end{align}
		where $\rho_{\mathscr A,k} = \max\{|1-\eta_{\mathscr A,k}(\lambda_{\mathscr X,k}^{\min}-\delta_{2\bbm r})|,|1-\eta_{\mathscr A,k}(\lambda_{\mathscr X,k}^{\max}+\delta_{2\bbm r})|\}$.

		For \textbf{Term V}, by the variational characterization of the Frobenius norm, we have
		\begin{align*}
			\textbf{Term V} & = \eta_{\mathscr A,k}\Biggl\|\mathcal P_{\mathcal T_{\bbm r}(\boldsymbol{\mathscr A}_{0,k}^{(t)})}\Biggl(\frac{1}{n_k}\sum_{i=1}^{n_k}\bigl\langle\mathcal P_{\mathcal T_{\bbm r}^{\perp}(\boldsymbol{\mathscr A}_{0,k}^{(t)})}(\cm A_{0,k}^{(t)}-\cm A_0^*),\cm X_{k,i}\bigr\rangle \circ \cm X_{k,i}\Biggr)\Biggr\|_\F  \\
			& = \eta_{\mathscr A,k}\sup_{\substack{\boldsymbol{\mathscr T}\in\mathcal T_{\bbm r}(\boldsymbol{\mathscr A}_{0,k}^{(t)})\\ \|\boldsymbol{\mathscr T}\|_\F\leq 1}}\left|\frac{1}{n_k}\sum_{i=1}^{n_k}\left\langle\bigl\langle\mathcal P_{\mathcal T_{\bbm r}^{\perp}(\boldsymbol{\mathscr A}_{0,k}^{(t)})}(\cm A_{0,k}^{(t)}-\cm A_0^*),\cm X_{k,i}\bigr\rangle,\langle \cm T,\cm X_{k,i}\rangle\right\rangle\right|.
		\end{align*}

		Let $\cm N_k^{(t)}=\mathcal P_{\mathcal T_{\bbm r}^{\perp}(\boldsymbol{\mathscr A}_{0,k}^{(t)})}(\cm A_{0,k}^{(t)}-\cm A_0^*)$. For any $\cm T\in\mathcal T_{\bbm r}(\boldsymbol{\mathscr A}_{0,k}^{(t)})$ with $\|\cm T\|_\F\leq 1$, decompose the preceding bilinear form into its population and empirical fluctuation parts:
		\begin{align*}
			\left|\frac{1}{n_k}\sum_{i=1}^{n_k}\left\langle\langle\cm N_k^{(t)},\cm X_{k,i}\rangle,\langle \cm T,\cm X_{k,i}\rangle\right\rangle\right| &\leq \left|\frac{1}{n_k}\sum_{i=1}^{n_k}\left\langle\langle\cm N_k^{(t)},\cm X_{k,i}\rangle,\langle \cm T,\cm X_{k,i}\rangle\right\rangle - \mathbb E\left[\left\langle\langle\cm N_k^{(t)},\cm X_{k,i}\rangle,\langle \cm T,\cm X_{k,i}\rangle\right\rangle\right]\right|\\
			&\quad + \left|\mathbb E\left[\left\langle\langle\cm N_k^{(t)},\cm X_{k,i}\rangle,\langle \cm T,\cm X_{k,i}\rangle\right\rangle\right]\right|.
		\end{align*}

		For the population term, by the matricization convention associated with $S_X$,
		\begin{align*}
			\mathbb E\left[\left\langle\langle\cm N_k^{(t)},\cm X_{k,i}\rangle,\langle \cm T,\cm X_{k,i}\rangle\right\rangle\right]
			=\tr\left((\cm N_k^{(t)})_{[S_X]}\bbm\Sigma_{\mathscr X,k}\cm T_{[S_X]}^\top\right).
		\end{align*}
		By the orthogonality of the projection operator, we have $\langle\cm N_k^{(t)},\cm T\rangle=\tr((\cm N_k^{(t)})_{[S_X]}\cm T_{[S_X]}^\top)=0$. Then, for any scalar $c$,
		\begin{align*}
			\mathbb E\left[\left\langle\langle\cm N_k^{(t)},\cm X_{k,i}\rangle,\langle \cm T,\cm X_{k,i}\rangle\right\rangle\right]
			=\tr\left((\cm N_k^{(t)})_{[S_X]}(\bbm\Sigma_{\mathscr X,k}-c\bm I)\cm T_{[S_X]}^\top\right).
		\end{align*}
		Taking $c=(\lambda_{\mathscr X,k}^{\max}+\lambda_{\mathscr X,k}^{\min})/2$, we have $\|\bbm\Sigma_{\mathscr X,k}-c\bm I\|_{\op} \leq (\lambda_{\mathscr X,k}^{\max}-\lambda_{\mathscr X,k}^{\min})/2$. Since $\|\cm T\|_\F \leq 1$, thus
		\begin{align}\label{eq:population-term-V-bound}
			\left|\mathbb E\left[\left\langle\langle\cm N_k^{(t)},\cm X_{k,i}\rangle,\langle \cm T,\cm X_{k,i}\rangle\right\rangle\right]\right| \leq \frac{\lambda_{\mathscr X,k}^{\max}-\lambda_{\mathscr X,k}^{\min}}{2}\|\cm N_k^{(t)}\|_\F.
		\end{align}

		By Lemma~\ref{lem:normal-projection-tucker-rank-2r}, we have $\tucrank(\mathcal P_{\mathcal T_{\bbm r}^{\perp}(\boldsymbol{\mathscr A}_{0,k}^{(t)})}(\cm A_{0,k}^{(t)}-\cm A_0^*))\leq 2\bbm r$. Moreover, every unit-Frobenius element of $\mathcal T_{\bbm r}(\boldsymbol{\mathscr A}_{0,k}^{(t)})$ belongs to $\mathbb S_{2\bbm r}$. Hence, the event $\mathcal E_{\mathscr A,k}^{(4)}$ gives
		\begin{align}\label{eq:empirical-term-V-bound}
			&\sup_{\substack{\boldsymbol{\mathscr T}\in\mathcal T_{\bbm r}(\boldsymbol{\mathscr A}_{0,k}^{(t)})\\ \|\boldsymbol{\mathscr T}\|_\F\leq 1}}\left|\frac{1}{n_k}\sum_{i=1}^{n_k}\left\langle\langle\cm N_k^{(t)},\cm X_{k,i}\rangle,\langle \cm T,\cm X_{k,i}\rangle\right\rangle - \mathbb E\left[\left\langle\langle\cm N_k^{(t)},\cm X_{k,i}\rangle,\langle \cm T,\cm X_{k,i}\rangle\right\rangle\right]\right| \leq C \sigma_{\mathscr X,k}^{2}\lambda_{\mathscr X,k}^{\max}\sqrt{\frac{df_{2\bbm r}}{n_k}}\|\cm N_k^{(t)}\|_\F.
		\end{align}

		Combining \eqref{eq:population-term-V-bound}, \eqref{eq:empirical-term-V-bound} and Lemma~\ref{lem:normal-component-bound}, we have
		\begin{align}\label{eq:term-V-bound}
			\textbf{Term V} & \leq \eta_{\mathscr A,k}\left(\frac{\lambda_{\mathscr X,k}^{\max}-\lambda_{\mathscr X,k}^{\min}}{2} + \delta_{2\bbm r}\right)\left\|\mathcal P_{\mathcal T_{\bbm r}^{\perp}(\boldsymbol{\mathscr A}_{0,k}^{(t)})}(\cm A_{0,k}^{(t)}-\cm A_0^*)\right\|_\F\notag\\
			& \leq \eta_{\mathscr A,k}\left(\frac{\lambda_{\mathscr X,k}^{\max}-\lambda_{\mathscr X,k}^{\min}}{2} + \delta_{2\bbm r}\right)\frac{2(d+m)}{\mu_{\mathscr A}^{\min}}\|\cm A_{0,k}^{(t)}-\cm A_0^*\|_\F^2\notag\\
			&= \eta_{\mathscr A,k}(\lambda_{\mathscr X,k}^{\max}-\lambda_{\mathscr X,k}^{\min} + 2\delta_{2\bbm r})\frac{d+m}{\mu_{\mathscr A}^{\min}}\|\cm A_{0,k}^{(t)}-\cm A_0^*\|_\F^2. 
		\end{align}
		where $\delta_{2\bbm r} = C\sigma_{\mathscr X,k}^2\lambda_{\mathscr X,k}^{\max}\sqrt{df_{2\bbm r}/n_k}$ and $\mu_{\mathscr A}^{\min} = \min_{s\in[d+m]}\sigma_{r_s}((\cm A_0^*)_{[s]})$. Combining \eqref{eq:term-I-bound}, \eqref{eq:term-II-bound}, \eqref{eq:term-III-bound}, \eqref{eq:term-IV-bound} and \eqref{eq:term-V-bound}, and using $\delta_{2\bbm r}\leq\bar\delta_{\mathscr A,k}$ given the sample size requirement in \eqref{eq:sample size-condition-tucker-rank-2r}, the following bound holds:
		\begin{align}\label{eq:A0-one-step-pre-eta}
			\|\cm A_{0,k}^{(t+1)}-\cm A_0^*\|_\F & \leq \left((\sqrt{d+m}+1)\rho_{\mathscr A,k} + \frac{\Gamma_{\mathscr A,k}}{\mu_{\mathscr A}^{\min}}\|\cm A_{0,k}^{(t)}-\cm A_0^*\|_\F\right)\|\cm A_{0,k}^{(t)}-\cm A_0^*\|_\F + (\sqrt{d+m}+1)\sum_{\ell=1}^{3}\mathsf{Rem}_{\mathscr A,k}^{(\ell)},
		\end{align}
		where $\Gamma_{\mathscr A,k} = (\sqrt{d+m}+1)(d+m)[2+\eta_{\mathscr A,k}(\lambda_{\mathscr X,k}^{\max} - \lambda_{\mathscr X,k}^{\min} + 2\bar\delta_{\mathscr A,k})]$,
		\[
			\begin{aligned}
				\mathsf{Rem}_{\mathscr A,k}^{(1)} &= \eta_{\mathscr A,k}\sqrt{r_{\bbm q}+r_{\bbm p}}\,\lambda_{\mathscr X,k}^{\max}\zeta,\\
				\mathsf{Rem}_{\mathscr A,k}^{(2)} &= C\eta_{\mathscr A,k}\sigma_{\mathscr X,k}^2\lambda_{\mathscr X,k}^{\max}h_{\mathscr B,k}\sqrt{\frac{df_{2\bbm r}}{n_k}},\\
				\mathsf{Rem}_{\mathscr A,k}^{(3)} &= C\eta_{\mathscr A,k}\sigma_{\mathscr X,k}\sigma_{\mathscr E,k}(\lambda_{\mathscr X,k}^{\max}\lambda_{\mathscr E,k}^{\max})^{1/2}\sqrt{\frac{df_{2\bbm r}}{n_k}}.
			\end{aligned}
		\]

		\noindent
		\textbf{\emph{Step 4: induction and final error bound.}} We now carry out the induction using the bound in \eqref{eq:A0-one-step-pre-eta}. Suppose that the induction hypothesis $\|\cm A_{0,k}^{(t)}-\cm A_0^*\|_\F \leq R_{\mathscr A,k}$ holds, with $R_{\mathscr A,k}=\mu_{\mathscr A}^{\min}/(4\Gamma_{\mathscr A,k})$. Recall that $c_{d,m} = 1/\{4(\sqrt{d+m}+1)\}$ and the step size $\eta_{\mathscr A,k} \in \mathcal I_{\mathscr A,k} = [\underline\eta_{\mathscr A,k},\overline\eta_{\mathscr A,k}]$, where
		\[
			\underline{\eta}_{\mathscr A,k} = \frac{1-c_{d,m}}{\lambda_{\mathscr X,k}^{\min}-\bar\delta_{\mathscr A,k}} = \frac{4(\sqrt{d+m}+1)-1}{4(\sqrt{d+m}+1)(\lambda_{\mathscr X,k}^{\min}-\bar\delta_{\mathscr A,k})}
			\ \text{and} \ 
			\overline{\eta}_{\mathscr A,k} = \frac{1+c_{d,m}}{\lambda_{\mathscr X,k}^{\max}+\bar\delta_{\mathscr A,k}} = \frac{4(\sqrt{d+m}+1)+1}{4(\sqrt{d+m}+1)(\lambda_{\mathscr X,k}^{\max}+\bar\delta_{\mathscr A,k})}.
		\]
		The nonemptiness of $\mathcal I_{\mathscr A,k}$ follows from $2\bar\delta_{\mathscr A,k} \leq c_{d,m}(\lambda_{\mathscr X,k}^{\min}+\lambda_{\mathscr X,k}^{\max})-(\lambda_{\mathscr X,k}^{\max}-\lambda_{\mathscr X,k}^{\min})$, which is exactly guaranteed by the definition of $\bar\delta_{\mathscr A,k}$ and the condition $\kappa_{\mathscr X,k}<(1+c_{d,m})/(1-c_{d,m})$. Let $a_k = \lambda_{\mathscr X,k}^{\min}-\delta_{2\bbm r}$ and $b_k = \lambda_{\mathscr X,k}^{\max}+\delta_{2\bbm r}$. Since $\delta_{2\bbm r}\leq\bar\delta_{\mathscr A,k}$ given the sample size requirement in \eqref{eq:sample size-condition-tucker-rank-2r}, we have $\lambda_{\mathscr X,k}^{\min}-\delta_{2\bbm r} \geq \lambda_{\mathscr X,k}^{\min}-\bar\delta_{\mathscr A,k}$ and $\lambda_{\mathscr X,k}^{\max}+\delta_{2\bbm r} \leq \lambda_{\mathscr X,k}^{\max}+\bar\delta_{\mathscr A,k}$. Thus, for every $\eta_{\mathscr A,k}\in\mathcal I_{\mathscr A,k}$,
		\[
			\frac{1-c_{d,m}}{a_k} \leq \frac{1-c_{d,m}}{\lambda_{\mathscr X,k}^{\min}-\bar\delta_{\mathscr A,k}} \leq \eta_{\mathscr A,k} \leq \frac{1+c_{d,m}}{\lambda_{\mathscr X,k}^{\max}+\bar\delta_{\mathscr A,k}} \leq \frac{1+c_{d,m}}{b_k}.
		\]
		Therefore, $\eta_{\mathscr A,k}a_k \geq 1-c_{d,m}$. Moreover, using $a_k\leq b_k$, we also have $\eta_{\mathscr A,k}a_k \leq \eta_{\mathscr A,k}b_k \leq 1+c_{d,m}$. Therefore, $1-c_{d,m} \leq \eta_{\mathscr A,k}a_k \leq 1+c_{d,m}$, which implies $|1-\eta_{\mathscr A,k}a_k| \leq c_{d,m}$. Similarly, the upper bound on $\eta_{\mathscr A,k}$ gives $\eta_{\mathscr A,k}b_k \leq 1+c_{d,m}$, while $b_k\geq a_k$ and the preceding lower bound imply $\eta_{\mathscr A,k}b_k \geq \eta_{\mathscr A,k}a_k \geq 1-c_{d,m}$. Thus, $1-c_{d,m} \leq \eta_{\mathscr A,k}b_k \leq 1+c_{d,m}$, and hence $|1-\eta_{\mathscr A,k}b_k| \leq c_{d,m}$. Note that $\rho_{\mathscr A,k} = \max\{|1-\eta_{\mathscr A,k}a_k|,|1-\eta_{\mathscr A,k}b_k|\}$. Then, combining the two bounds gives
		\[
			\rho_{\mathscr A,k} \leq c_{d,m} = \frac{1}{4(\sqrt{d+m}+1)}.
		\]
		Consequently, $(\sqrt{d+m}+1)\rho_{\mathscr A,k} \leq 1/4$. Moreover, by the induction hypothesis $\|\cm A_{0,k}^{(t)}-\cm A_0^*\|_\F \leq R_{\mathscr A,k}$ and $R_{\mathscr A,k} = \mu_{\mathscr A}^{\min}/(4\Gamma_{\mathscr A,k})$, we have
		\[
			\frac{\Gamma_{\mathscr A,k}}{\mu_{\mathscr A}^{\min}}\|\cm A_{0,k}^{(t)}-\cm A_0^*\|_\F \leq \frac{\Gamma_{\mathscr A,k}}{\mu_{\mathscr A}^{\min}}R_{\mathscr A,k} = \frac14.
		\]
		Substituting the preceding two bounds into \eqref{eq:A0-one-step-pre-eta} yields
		\begin{align}\label{eq:A0-clean-recursion}
			\|\cm A_0^{(t+1)}-\cm A_0^*\|_\F \leq \frac12 \|\cm A_{0,k}^{(t)}-\cm A_0^*\|_\F + (\sqrt{d+m}+1)\sum_{\ell=1}^{3}\mathsf{Rem}_{\mathscr A,k}^{(\ell)}.
		\end{align}
		The deviation condition $\zeta \leq c_{d,m}R_{\mathscr A,k}/(\eta_{\mathscr A,k}\lambda_{\mathscr X,k}^{\max}\sqrt{r_{\bbm q}+r_{\bbm p}})$ implies that $(\sqrt{d+m}+1)\mathsf{Rem}_{\mathscr A,k}^{(1)} \leq R_{\mathscr A,k}/4$.
		We next verify that $\mathsf{Rem}_{\mathscr A,k}^{(2)}$ and $\mathsf{Rem}_{\mathscr A,k}^{(3)}$ are also small given the sample size requirement in \eqref{eq:sample size-condition-tucker-rank-2r}. Using the upper bound $\eta_{\mathscr A,k}\leq\bar\eta_{\mathscr A,k}$ and $R_{\mathscr A,k} = \mu_{\mathscr A}^{\min}/(4\Gamma_{\mathscr A,k})$, we have, after absorbing constants depending only on $d$ and $m$,
		\[
			\eta_{\mathscr A,k} \leq \frac{C_{d,m}}{\lambda_{\mathscr X,k}^{\max}+\bar\delta_{\mathscr A,k}}, \qquad \frac{\lambda_{\mathscr X,k}^{\max}-\lambda_{\mathscr X,k}^{\min}+2\bar\delta_{\mathscr A,k}}{\lambda_{\mathscr X,k}^{\max}+\bar\delta_{\mathscr A,k}}\leq 2,
		\]
		and hence $\eta_{\mathscr A,k}/R_{\mathscr A,k} \lesssim 1/(\mu_{\mathscr A}^{\min}\lambda_{\mathscr X,k}^{\max})$. Since the theorem's sample size condition \eqref{eq:sample size-condition-tucker-rank-2r} implies
		\[
			n_k \gtrsim \frac{\sigma_{\mathscr X,k}^{4}h_{\mathscr B,k}^{2}}{(\mu_{\mathscr A}^{\min})^2} df_{2\bbm r}
			\ \text{and} \
			n_k \gtrsim \frac{\sigma_{\mathscr X,k}^{2}\sigma_{\mathscr E,k}^{2}\lambda_{\mathscr E,k}^{\max}/\lambda_{\mathscr X,k}^{\max}}{(\mu_{\mathscr A}^{\min})^2}df_{2\bbm r}.
		\]
		Consequently, we have $(\sqrt{d+m}+1)\mathsf{Rem}_{\mathscr A,k}^{(2)} \leq R_{\mathscr A,k}/8$, and $(\sqrt{d+m}+1)\mathsf{Rem}_{\mathscr A,k}^{(3)} \leq R_{\mathscr A,k}/8$. Together with $(\sqrt{d+m}+1)\mathsf{Rem}_{\mathscr A,k}^{(1)}\leq R_{\mathscr A,k}/4$, we have
		\begin{align}\label{eq:reminder-bound}
			(\sqrt{d+m}+1)\sum_{\ell=1}^{3}\mathsf{Rem}_{\mathscr A,k}^{(\ell)} \leq \frac14R_{\mathscr A,k}+\frac18R_{\mathscr A,k}+\frac18R_{\mathscr A,k} = \frac12R_{\mathscr A,k}.
		\end{align}
		Then combining \eqref{eq:A0-clean-recursion}, \eqref{eq:reminder-bound} and the induction hypothesis $\|\cm A_0^{(t)}-\cm A_0^*\|_\F \leq R_{\mathscr A,k}$, we have $\|\cm A_0^{(t+1)}-\cm A_0^*\|_\F \leq R_{\mathscr A,k}$. This concludes the induction step. By the initialization condition in the theorem, $\|\cm A_0^{(0)}-\cm A_0^*\|_\F \leq R_{\mathscr A,k}$. Therefore, by induction, all iterates remain in the local neighborhood $\|\cm A_{0,k}^{(t)}-\cm A_0^*\|_\F \leq R_{\mathscr A,k}$ for  $0\leq t\leq T_g^{(k)}$.

		Iterating the recursion \eqref{eq:A0-clean-recursion}, we have
		\[
			\|\widehat{\cm A}_{0,k}-\cm A_0^*\|_\F = \|\cm A_0^{(T_g^{(k)})}-\cm A_0^*\|_\F \leq 2^{-T_g^{(k)}}\|\cm A_0^{(0)}-\cm A_0^*\|_\F + 2(\sqrt{d+m}+1)\sum_{\ell=1}^{3}\mathsf{Rem}_{\mathscr A,k}^{(\ell)}.
		\]
		Equivalently, substituting the three remainder terms gives
		\begin{align*}
			\|\widehat{\cm A}_{0,k}-\cm A_0^*\|_\F & \leq 2^{-T_g^{(k)}}\|\cm A_0^{(0)}-\cm A_0^*\|_\F + 2(\sqrt{d+m}+1)\eta_{\mathscr A,k}\sqrt{r_{\bbm q}+r_{\bbm p}}\,\lambda_{\mathscr X,k}^{\max}\zeta \\
			&\quad + C(\sqrt{d+m}+1)\eta_{\mathscr A,k}\sigma_{\mathscr X,k}^{2}\lambda_{\mathscr X,k}^{\max}h_{\mathscr B,k}\sqrt{\frac{df_{2\bbm r}}{n_k}} \\
			&\quad + C(\sqrt{d+m}+1)\eta_{\mathscr A,k}\sigma_{\mathscr X,k}\sigma_{\mathscr E,k}(\lambda_{\mathscr X,k}^{\max}\lambda_{\mathscr E,k}^{\max})^{1/2}\sqrt{\frac{df_{2\bbm r}}{n_k}}.
		\end{align*}
		Taking $T_g^{(k)}\asymp\log n_k$ makes $2^{-T_g^{(k)}}\|\cm A_0^{(0)}-\cm A_0^*\|_\F$ negligible.
		Moreover, since
		\[
			\eta_{\mathscr A,k} \leq \frac{1+c_{d,m}}{\lambda_{\mathscr X,k}^{\max}+\bar\delta_{\mathscr A,k}}
			\leq \frac{1+c_{d,m}}{\lambda_{\mathscr X,k}^{\max}},
		\]
		and $df_{2\bbm r}\asymp df_{\bbm r}$, conditional on $\mathcal E_{\mathscr A,k}$, we have
		\begin{align*}
			\|\widehat{\cm A}_{0,k}-\cm A_0^*\|_\F \lesssim \left\{\sigma_{\mathscr X,k}^2h_{\mathscr B,k} \vee \sigma_{\mathscr X,k}\sigma_{\mathscr E,k}\left(\frac{\lambda_{\mathscr E,k}^{\max}}{\lambda_{\mathscr X,k}^{\max}}\right)^{1/2}\right\}\sqrt{\frac{df_{\bbm r}}{n_k}} + \sqrt{r_{\bbm q}+r_{\bbm p}}\zeta.
		\end{align*}
		Since $\mathbb P(\mathcal E_{\mathscr A,k})\geq 1-C\exp(-Cdf_{\bbm r})$ after adjusting constants, the claimed high probability bound follows. This completes the proof.
	\end{proof}

	\begin{proof}[\textbf{Proof of Theorem~\ref{thm:federated_representation_error}}]
		The proof is similar to that of Theorem~\ref{thm:representation_single_client}, but with the local empirical gradient replaced by the weighted federated gradient and with an additional Gaussian privacy-noise term. We organize the proof into five steps. Step~1 introduces the deterministic notation and auxiliary constants used throughout the argument. Step~2 derives the sensitivity of the truncated local gradient. Step~3 defines the high probability event on which the pooled empirical processes, truncation event, and aggregated privacy noise are controlled uniformly over deterministic rank classes. Conditional on this event, Step~4 derives the one-step recursion for the federated projected-gradient iterate. Step~5 verifies the privacy guarantee and uses the recursion to establish the final error bound.

		\noindent
		\textbf{\emph{Step 1: notation and auxiliary constants.}}
		Recall $\mu_{\mathscr A}^{\min} = \min_{s\in[d+m]}\sigma_{r_s}\bigl((\cm A_0^*)_{[s]}\bigr)$. Let $\bar\sigma_{\mathscr X}=\max_{k\in[K]}\sigma_{\mathscr X,k}$, $\bar\sigma_{\mathscr E}=\max_{k\in[K]}\sigma_{\mathscr E,k}$, $\bar\lambda_{\mathscr X}^{\max}=\max_{k\in[K]}\lambda_{\mathscr X,k}^{\max}$, $\bar\lambda_{\mathscr E}^{\max}=\max_{k\in[K]}\lambda_{\mathscr E,k}^{\max}$, and $\bar h_{\mathscr B}=\max_{k\in[K]}h_{\mathscr B,k}$.
		The truncation levels are chosen at $\tau_{\mathscr X,k} \asymp \sigma_{\mathscr X,k}\sqrt{\lambda_{\mathscr X,k}^{\max}}\sqrt{p+\log n}$, and $\tau_{\mathscr E,k}^{(t)} \asymp (R_{\mathscr A}+h_{\mathscr B,k})\sigma_{\mathscr X,k}\sqrt{\lambda_{\mathscr X,k}^{\max}}\sqrt{p+\log n} + \sigma_{\mathscr E,k}\sqrt{\lambda_{\mathscr E,k}^{\max}}\sqrt{q+\log n}$.

		For the $t$-th federated Stage-I iterate, write
		\[
			\cm G_{\mathscr A,k}^{loc,(t)} = \frac{1}{n_k}\sum_{i=1}^{n_k}\left(\langle \cm A_0^{(t)},\cm X_{k,i}\rangle-\cm Y_{k,i}\right)\circ \cm X_{k,i} = \frac{1}{n_k}\sum_{i=1}^{n_k}\bigl\langle\cm A_0^{(t)}-\cm A_0^*-\cm B_k^*,\cm X_{k,i}\bigr\rangle\circ\cm X_{k,i} - \frac{1}{n_k}\sum_{i=1}^{n_k}\cm E_{k,i}\circ\cm X_{k,i}.
		\]
		Define $\delta_{2\bbm r}^{\mathrm{fed}} = C\bar\sigma_{\mathscr X}^2\bar\lambda_{\mathscr X}^{\max}\sqrt{df_{\bbm r}/n}$.
		Under the stated sample size condition in Theorem~\ref{thm:federated_representation_error}, we have $\delta_{2\bbm r}^{\mathrm{fed}} \leq \bar\delta_{\mathscr A}$. Set $R_{\mathscr A}=\mu_{\mathscr A}^{\min}/(4\Gamma_{\mathscr A}^{\mathrm{fed}})$ and let $\Gamma_{\mathscr A}^{\mathrm{fed}} = (\sqrt{d+m}+1)(d+m)[2+\eta_{\mathscr A}(\bar\lambda_{\mathscr X}^{\max} - \underline\lambda_{\mathscr X}^{\min} + 2\bar\delta_{\mathscr A})]$.

		\noindent
		\textbf{\emph{Step 2: sensitivity of the truncated local gradient.}}
		We first characterize the sensitivity of the transmitted local gradient induced by the truncation operation. For fixed truncation levels $\tau_{\mathscr X,k}>0$ and $\tau_{\mathscr E,k}^{(t)}>0$, define $\cm G_{\mathscr A,k}^{\vee}(\cm A;\mathsf D_k) = n_k^{-1}\sum_{i=1}^{n_k}\cm R_{k,i}^{\vee}(\cm A;\tau_{\mathscr E,k}^{(t)}) \circ \cm X_{k,i}^{\vee}(\tau_{\mathscr X,k}) = n_k^{-1}\sum_{i=1}^{n_k}\cm H_{k,i}^{\vee}(\cm A;\mathsf D_k)$. Here, the argument $\mathsf D_k$ makes explicit that both 
		$\cm G_{\mathscr A,k}^{\vee}$ and its summands $\cm H_{k,i}^{\vee}$ are computed from the local dataset $\mathsf D_k$.
		By construction, $\|\cm R_{k,i}^{\vee}(\cm A;\tau_{\mathscr E,k}^{(t)})\|_{\F}\leq \tau_{\mathscr E,k}^{(t)}$ and $\|\cm X_{k,i}^{\vee}(\tau_{\mathscr X,k})\|_{\F}\leq \tau_{\mathscr X,k}$.
		Since $\|\cm U\circ\cm V\|_{\F}=\|\cm U\|_{\F}\|\cm V\|_{\F}$, each summand in the truncated gradient satisfies $\|\cm H_{k,i}^{\vee}(\cm A;\mathsf D_k)\|_{\F} \leq \tau_{\mathscr E,k}^{(t)}\tau_{\mathscr X,k}$.
		Fix the $t$-th iteration. Suppose that two neighboring datasets $\mathsf D_k$ and $\mathsf D_k'$ differ only in the $j$-th observation. Then
		\begin{align}\label{eq:truncated-gradient-sensitivity-bound}
			\bigl\|\cm G_{\mathscr A,k}^{\vee}(\cm A;\mathsf D_k) - \cm G_{\mathscr A,k}^{\vee}(\cm A;\mathsf D_k')\bigr\|_{\F}
			&= \left\|\frac{1}{n_k}\sum_{i=1}^{n_k}\cm H_{k,i}^{\vee}(\cm A;\mathsf D_k) - \frac{1}{n_k}\sum_{i=1}^{n_k}\cm H_{k,i}^{\vee}(\cm A;\mathsf D_k')\right\|_{\F}\notag\\
			&= \frac{1}{n_k}\left\|\cm H_{k,j}^{\vee}(\cm A;\mathsf D_k) - \cm H_{k,j}^{\vee}(\cm A;\mathsf D_k')\right\|_{\F}\notag\\
			&\leq \frac{1}{n_k}\left(\bigl\|\cm H_{k,j}^{\vee}(\cm A;\mathsf D_k)\bigr\|_{\F} + \bigl\|\cm H_{k,j}^{\vee}(\cm A;\mathsf D_k')\bigr\|_{\F}\right)\notag\\
			&\leq\frac{1}{n_k}\left(\tau_{\mathscr E,k}^{(t)}\tau_{\mathscr X,k} + \tau_{\mathscr E,k}^{(t)}\tau_{\mathscr X,k}\right) = \frac{2\tau_{\mathscr E,k}^{(t)}\tau_{\mathscr X,k}}{n_k}.
		\end{align}
		This is the sensitivity of the truncated local gradient transmitted by client $k$ at the $t$-th iteration, and it is used to calibrate the Gaussian perturbation in the transmitted truncated gradient.

		\noindent
		\textbf{\emph{Step 3: the high probability event.}}
		We collect all high probability events $\mathcal E_{\mathscr A,\ell}^{\mathrm{fed}}$ for $\ell\in[6]$ used in the subsequent deterministic recursion. Let $\mathcal E_{\mathscr A}^{\mathrm{fed}} = \bigcap_{\ell=1}^{6}\mathcal E_{\mathscr A,\ell}^{\mathrm{fed}}$.
		\begin{align*}
			\mathcal E_{\mathscr A,1}^{\mathrm{fed}} &= \Biggl\{\sup_{\boldsymbol{\mathscr U}\in\mathbb S_{2\bbm r}}\left|\sum_{k=1}^K\frac{n_k}{n}\left[\frac{1}{n_k}\sum_{i=1}^{n_k}\left\langle\langle \cm B_k^*,\cm X_{k,i}\rangle,\langle \cm U,\cm X_{k,i}\rangle\right\rangle-\mathbb E\left[\left\langle\langle \cm B_k^*,\cm X_{k,i}\rangle,\langle \cm U,\cm X_{k,i}\rangle\right\rangle\right]\right]\right| \leq C\bar\sigma_{\mathscr X}^2\bar\lambda_{\mathscr X}^{\max}\bar h_{\mathscr B}\sqrt{\frac{df_{\bbm r}}{n}}\Biggr\},\\[1mm]
			\mathcal E_{\mathscr A,2}^{\mathrm{fed}} &= \Biggl\{\sup_{\boldsymbol{\mathscr U}\in\mathbb S_{2\bbm r}}\left|\left\langle\sum_{k=1}^K\frac{n_k}{n}\frac{1}{n_k}\sum_{i=1}^{n_k}\cm E_{k,i}\circ\cm X_{k,i},\cm U\right\rangle\right| \leq C\bar\sigma_{\mathscr X}\bar\sigma_{\mathscr E}\sqrt{\bar\lambda_{\mathscr X}^{\max}\bar\lambda_{\mathscr E}^{\max}}\sqrt{\frac{df_{\bbm r}}{n}}\Biggr\},\\[1mm]
			\mathcal E_{\mathscr A,3}^{\mathrm{fed}} &= \Biggl\{\sup_{\boldsymbol{\mathscr U}\in\mathbb S_{2\bbm r}}\left|\sum_{k=1}^K\frac{n_k}{n}\left[\frac{1}{n_k}\sum_{i=1}^{n_k}\|\langle \cm U,\cm X_{k,i}\rangle\|_\F^2 - \mathbb E\|\langle \cm U,\cm X_{k,i}\rangle\|_\F^2\right]\right| \leq C\bar\sigma_{\mathscr X}^2\bar\lambda_{\mathscr X}^{\max}\sqrt{\frac{df_{\bbm r}}{n}}\Biggr\},\\[1mm]
			\mathcal E_{\mathscr A,4}^{\mathrm{fed}} &= \Biggl\{\sup_{\boldsymbol{\mathscr U},\boldsymbol{\mathscr V}\in\mathbb S_{2\bbm r}}\left|\sum_{k=1}^K\frac{n_k}{n}\left[\frac{1}{n_k}\sum_{i=1}^{n_k}\left\langle\langle \cm U,\cm X_{k,i}\rangle,\langle \cm V,\cm X_{k,i}\rangle\right\rangle - \mathbb E\left\{\left\langle\langle \cm U,\cm X_{k,i}\rangle,\langle \cm V,\cm X_{k,i}\rangle\right\rangle\right\}\right]\right| \leq C\bar\sigma_{\mathscr X}^2\bar\lambda_{\mathscr X}^{\max}\sqrt{\frac{df_{\bbm r}}{n}}\Biggr\},\\[1mm]
			\mathcal E_{\mathscr A,5}^{\mathrm{fed}} &= \bigcap_{k=1}^K\left\{\max_{i\in[n_k]}\|\cm X_{k,i}\|_{\F}\leq\tau_{\mathscr X,k},\quad\sup_{\|\boldsymbol{\mathscr A}-\boldsymbol{\mathscr A}_0^*\|_{\F}\leq R_{\mathscr A}}\max_{i\in[n_k]}\|\langle\cm A,\cm X_{k,i}\rangle-\cm Y_{k,i}\|_{\F}\leq\min_{0\leq t<T_g}\tau_{\mathscr E,k}^{(t)}\right\},\\[1mm]
			\mathcal E_{\mathscr A,6}^{\mathrm{fed}} &= \Biggl\{\max_{0\leq t<T_g}\left\|\mathcal P_{\mathcal T_{\bbm r}(\boldsymbol{\mathscr A}_0^{(t)})}\left(\sum_{k=1}^K \frac{n_k}{n}\cm W_{\mathscr A,k}^{(t)}\right)\right\|_\F \lesssim \frac{T_g\sqrt{\log(T_g/\delta)}}{\varepsilon n}\sqrt{K(df_{\bbm r}+\log T_g)}\max_{0\leq t<T_g}\max_{k\in[K]}\tau_{\mathscr E,k}^{(t)}\tau_{\mathscr X,k}\Biggr\}.
		\end{align*}

		The high probability bounds in $\mathcal E_{\mathscr A,1}^{\mathrm{fed}}$--$\mathcal E_{\mathscr A,4}^{\mathrm{fed}}$ follow from their weighted pooled form. We spell out the arguments for $\mathcal E_{\mathscr A,1}^{\mathrm{fed}}$. For any fixed $\cm U\in\mathbb S_{2\bbm r}$, the process in $\mathcal E_{\mathscr A,1}^{\mathrm{fed}}$ can be rewritten as
		\begin{align*}
			Z(\cm U) = \frac{1}{n}\sum_{k=1}^K\sum_{i=1}^{n_k}\left[\left\langle\langle \cm B_k^*,\cm X_{k,i}\rangle,\langle \cm U,\cm X_{k,i}\rangle\right\rangle - \mathbb E\left\{\left\langle\langle \cm B_k^*,\cm X_{k,i}\rangle,\langle \cm U,\cm X_{k,i}\rangle\right\rangle\right\}\right].
		\end{align*}
		Here $\cm B_k^*$ varies across clients, and hence $\left[\left\langle\langle \cm B_k^*,\cm X_{k,i}\rangle,\langle \cm U,\cm X_{k,i}\rangle\right\rangle - \mathbb E\left\{\left\langle\langle \cm B_k^*,\cm X_{k,i}\rangle,\langle \cm U,\cm X_{k,i}\rangle\right\rangle\right\}\right]$ is an independent but non-identically distributed centered empirical process with client-specific fixed directions $\cm B_k^*$. Using the $S_X$-matricization convention, we have
		\begin{align*}
			\left\langle\langle \cm B_k^*,\cm X_{k,i}\rangle,\langle \cm U,\cm X_{k,i}\rangle\right\rangle
			&= \vect^\top(\cm X_{k,i})\bm A_{k,\boldsymbol{\mathscr U}}\vect(\cm X_{k,i}),
		\end{align*}
		where
		\[
			\bm A_{k,\boldsymbol{\mathscr U}} = \frac{(\cm B_k^*)_{[S_X]}^\top\cm U_{[S_X]}+\cm U_{[S_X]}^\top(\cm B_k^*)_{[S_X]}}{2}.
		\]
		By Assumption~\ref{assump:subg-design}, $\vect(\cm X_{k,i}) = \bbm\Sigma_{\mathscr X,k}^{1/2}\bbm\xi_{k,i}$, where $\bbm\xi_{k,i}$ has uniformly sub-Gaussian coordinates. Thus, for fixed $\cm U$, each summand is a centered quadratic form $Q_{k,i}(\cm U) = \bbm\xi_{k,i}^\top\bbm\Sigma_{\mathscr X,k}^{1/2}\bm A_{k,\boldsymbol{\mathscr U}}\bbm\Sigma_{\mathscr X,k}^{1/2}\bbm\xi_{k,i} - \mathbb E[\bbm\xi_{k,i}^\top\bbm\Sigma_{\mathscr X,k}^{1/2}\bm A_{k,\boldsymbol{\mathscr U}}\bbm\Sigma_{\mathscr X,k}^{1/2}\bbm\xi_{k,i}]$. Moreover, since $\|\cm U\|_\F=1$ and $\|\cm B_k^*\|_\F\leq h_{\mathscr B,k} \leq \bar h_{\mathscr B}$,
		\begin{align*}
			\left\|\bbm\Sigma_{\mathscr X,k}^{1/2}\bm A_{k,\boldsymbol{\mathscr U}}\bbm\Sigma_{\mathscr X,k}^{1/2}\right\|_{\op} \leq \lambda_{\mathscr X,k}^{\max}\|\bm A_{k,\boldsymbol{\mathscr U}}\|_{\op} \leq \bar\lambda_{\mathscr X}^{\max}\bar h_{\mathscr B}\ \text{and} \ \left\|\bbm\Sigma_{\mathscr X,k}^{1/2}\bm A_{k,\boldsymbol{\mathscr U}}\bbm\Sigma_{\mathscr X,k}^{1/2}\right\|_{\F} \leq \lambda_{\mathscr X,k}^{\max}\|\bm A_{k,\boldsymbol{\mathscr U}}\|_{\F} \leq \bar\lambda_{\mathscr X}^{\max}\bar h_{\mathscr B}.
		\end{align*}
		By the Hanson-Wright inequality in Lemma~\ref{lem:hanson-wright} and the bounds above on the operator and Frobenius norms of $\bbm\Sigma_{\mathscr X,k}^{1/2}\bm A_{k,\boldsymbol{\mathscr U}}\bbm\Sigma_{\mathscr X,k}^{1/2}$, for any $u>0$,
		\[
			\mathbb P\left(|Q_{k,i}(\cm U)|\geq u\right)\leq 2\exp\left(-C\min\left\{\frac{u^2}{\bar\sigma_{\mathscr X}^4(\bar\lambda_{\mathscr X}^{\max})^2\bar h_{\mathscr B}^2},\frac{u}{\bar\sigma_{\mathscr X}^2\bar\lambda_{\mathscr X}^{\max}\bar h_{\mathscr B}}\right\}\right) =:2\exp\left(-C\min\left\{\frac{u^2}{M^2},\frac{u}{M}\right\}\right).
		\]
		where $M=\bar\sigma_{\mathscr X}^2\bar\lambda_{\mathscr X}^{\max}\bar h_{\mathscr B}$.
		For any integer $m\geq2$, using the tail-integration formula,
		\[
			\mathbb E|Q_{k,i}(\cm U)|^m = m\int_0^\infty s^{m-1}\mathbb P\{|Q_{k,i}(\cm U)|\geq s\}\,ds \leq 2m\int_0^M s^{m-1}\,ds + 2m\int_M^\infty s^{m-1}\exp\left(-C\frac{s}{M}\right)\,ds \leq C^m m! M^m.
		\]
		Since $Q_{k,i}(\cm U)$ is centered, for $|\lambda|<1/(CM)$,
		\[
			\mathbb E\exp[\lambda Q_{k,i}(\cm U)] = 1+\sum_{m=2}^\infty\frac{\lambda^m\mathbb E[Q_{k,i}(\cm U)^m]}{m!} \leq 1+\sum_{m=2}^\infty(C|\lambda|M)^m \leq \exp\left(CM^2\lambda^2\right).
		\]
		Therefore, by the moment-generating-function characterization of sub-exponential random variables, $Q_{k,i}(\cm U)$ is sub-exponential with parameters $\nu_{k,i}^2 \leq \bar\sigma_{\mathscr X}^4(\bar\lambda_{\mathscr X}^{\max})^2\bar h_{\mathscr B}^2$ and $\alpha_{k,i} \leq \bar\sigma_{\mathscr X}^2\bar\lambda_{\mathscr X}^{\max}\bar h_{\mathscr B}$. Since the variables $\{Q_{k,i}(\cm U)\}_{k,i}$ are independent and centered, the extension of \citet[Equation~(2.18)]{wainwright2019high} to independent sub-exponential variables yields, for every $u>0$,
		\begin{align*}
			\mathbb P\left(|Z(\cm U)|\geq u\right) \leq 2\exp\left(-Cn\min\left\{\frac{u^2}{\bar\sigma_{\mathscr X}^4(\bar\lambda_{\mathscr X}^{\max})^2\bar h_{\mathscr B}^2},\frac{u}{\bar\sigma_{\mathscr X}^2\bar\lambda_{\mathscr X}^{\max}\bar h_{\mathscr B}}\right\}\right).
		\end{align*}
		Therefore, taking $u=C\bar\sigma_{\mathscr X}^2\bar\lambda_{\mathscr X}^{\max}\bar h_{\mathscr B}\sqrt{df_{\bbm r}/n}$ and using the pooled sample size condition, we have
		\[
			\mathbb P\left(|Z(\cm U)| \geq C\bar\sigma_{\mathscr X}^2\bar\lambda_{\mathscr X}^{\max}\bar h_{\mathscr B}\sqrt{\frac{df_{\bbm r}+u}{n}}\right)\leq 2\exp(-Cdf_{\bbm r}).
		\]
		Finally, let $\mathcal N_{2\bbm r}$ be the product net of $\mathbb S_{2\bbm r}$ constructed in the proof of Lemma~\ref{lem:tangent-deviation-bilinear-bound}, satisfying $|\mathcal N_{2\bbm r}|\leq\exp(Cdf_{\bbm r})$. Applying a union bound over $\mathcal N_{2\bbm r}$, together with the same linear approximation arguments used in the proof of Lemma~\ref{lem:tangent-deviation-bilinear-bound}, gives
		\[
			\sup_{\boldsymbol{\mathscr U}\in\mathbb S_{2\bbm r}}|Z(\cm U)| \leq C\bar\sigma_{\mathscr X}^2\bar\lambda_{\mathscr X}^{\max}\bar h_{\mathscr B}\sqrt{\frac{df_{\bbm r}}{n}}
		\]
		with probability at least $1-C\exp(-Cdf_{\bbm r})$. This proves the high probability bound in $\mathcal E_{\mathscr A,1}^{\mathrm{fed}}$.

		The events $\mathcal E_{\mathscr A,2}^{\mathrm{fed}}$--$\mathcal E_{\mathscr A,4}^{\mathrm{fed}}$ are handled analogously: the weighted averages are written as
		\[
			\sum_{k=1}^K\frac{n_k}{n}\frac{1}{n_k}\sum_{i=1}^{n_k}(\cdot) = \frac{1}{n}\sum_{k=1}^K\sum_{i=1}^{n_k}(\cdot),
		\]
		and the summands are independent across $(k,i)$ and uniformly controlled by the envelope parameters $\bar\sigma_{\mathscr X}$, $\bar\sigma_{\mathscr E}$, $\bar\lambda_{\mathscr X}^{\max}$ and $\bar\lambda_{\mathscr E}^{\max}$. Therefore, the same covering and concentration arguments used in the single-client proof apply to the pooled independent non-identically distributed processes, with client-specific parameters replaced by their uniform envelopes.

		More precisely, the pooled analogues of the single-client empirical-process lemmas imply that, under the pooled empirical-process sample size requirement
		\[
			n \gtrsim \max\Biggl\{1, \frac{\bar\sigma_{\mathscr X}^{4}(\bar\lambda_{\mathscr X}^{\max})^{2}}{1\wedge \bar\delta_{\mathscr A}^{2}}\Biggr\}df_{\bbm r},
		\]
		the restricted quadratic and bilinear empirical processes are controlled at the pooled rate $\sqrt{df_{\bbm r}/n}$ and, in particular, $\delta_{2\bbm r}^{\mathrm{fed}}\leq\bar\delta_{\mathscr A}$. Therefore, the intersection of $\bigcap_{\ell=1}^{4}\mathcal E_{\mathscr A,\ell}^{\mathrm{fed}}$ holds with probability at least $1-C\exp(-Cdf_{\bbm r})$, after using $df_{2\bbm r}\asymp df_{\bbm r}$.

		The truncation event $\mathcal E_{\mathscr A,5}^{\mathrm{fed}}$ occurs with probability at least $1-C\exp(-C\log n)$. This follows from Lemma~\ref{lem:truncation-level-choice} and the choices $\tau_{\mathscr X,k} \asymp \sigma_{\mathscr X,k}\sqrt{\lambda_{\mathscr X,k}^{\max}}\sqrt{p+\log n}$ and $\tau_{\mathscr E,k}^{(t)} \asymp (R_{\mathscr A}+h_{\mathscr B,k})\sigma_{\mathscr X,k}\sqrt{\lambda_{\mathscr X,k}^{\max}}\sqrt{p+\log n} + \sigma_{\mathscr E,k}\sqrt{\lambda_{\mathscr E,k}^{\max}}\sqrt{q+\log n}$ uniformly over $0\leq t<T_g$.

		The privacy-noise event $\mathcal E_{\mathscr A,6}^{\mathrm{fed}}$ follows directly from Gaussian concentration. Fix an iteration $t$ and given the current iterate $\cm A_0^{(t)}$ and the noise scales $\{\sigma_k^{(t)}\}_{k=1}^K$. For each index $\bbm i$, the added noises $[\cm W_{\mathscr A,k}^{(t)}]_{\bbm i}$ are independent Gaussian random variables with variances $(\sigma_k^{(t)})^2$. Therefore, the weighted sum $\sum_{k=1}^K(n_k/n)\cm W_{\mathscr A,k}^{(t)}$ is again an entrywise Gaussian tensor. Moreover, let $\bbm w_t = \vect(\sum_{k=1}^K(n_k/n)\cm W_{\mathscr A,k}^{(t)})$. Then, $\bbm w_t\sim N(\bm 0,\nu_t^2\bm I)$ where $\nu_t^2=\sum_{k=1}^K(n_k/n)^2(\sigma_k^{(t)})^2$. 

		Let $\cm S^{(t)}$ and $\{\bm U_a^{(t)}\}_{a=1}^{d+m}$ denote the core tensor and orthonormal factor matrices in the Tucker decomposition of the current iterate $\boldsymbol{\mathscr A}_0^{(t)}$. By the tangent-space characterization of the Tucker-rank manifold \citep{luo2024tensor}, any tensor of the form  
		\[
		\cm T = \cm F \times_{a=1}^{d+m} \bm U_a^{(t)} + \sum_{j=1}^{d+m} \cm S^{(t)} \times_j \bigl( \bm U_{j,\perp}^{(t)} \bm D_j \bigr) \times_{a \neq j} \bm U_a^{(t)}
		\]
		belongs to the tangent space $\mathcal T_{\bbm r}(\boldsymbol{\mathscr A}_0^{(t)})$. Here, $\bm U_{j,\perp}^{(t)} \in \mathbb R^{d_j \times (d_j - r_j)}$ is an orthonormal complement of $\bm U_j^{(t)}$, $\cm F \in \mathbb R^{r_1 \times \cdots \times r_{d+m}}$, and $\bm D_j \in \mathbb R^{(d_j - r_j) \times r_j}$. The dimensions satisfy $d_j = q_j$ for $j \in [m]$ and $d_{m+s} = p_s$ for $s \in [d]$. The matrices $\bm U_a^{(t)}$ determine the orientation of the tangent space, while its dimension is governed by the total number of free parameters in $\cm F$ and $\{\bm D_j\}_{j=1}^{d+m}$.
		Thus the above representation is parametrized by $\prod_{s=1}^{d+m}r_s$ entries in $\cm F$ and by $r_j(d_j-r_j)$ entries in each $\bm D_j$. Therefore,
		\[
			\dim\bigl(\mathcal T_{\bbm r}(\boldsymbol{\mathscr A}_0^{(t)})\bigr) \leq \prod_{s=1}^{d+m}r_s + \sum_{s=1}^{m}r_s(q_s-r_s) + \sum_{s=1}^{d}r_{m+s}(p_s-r_{m+s}) = df_{\bbm r}.
		\]

		Let $\mathcal V_t = \{\vect(\cm Z):\cm Z\in\mathcal T_{\bbm r}(\boldsymbol{\mathscr A}_0^{(t)})\}$ be the vectorized tangent space. Since $\vect(\cdot)$ is a linear bijection from the ambient tensor space to its vectorized Euclidean space, it preserves dimension. Thus $r_t = \dim(\mathcal V_t) = \dim\bigl(\mathcal T_{\bbm r}(\boldsymbol{\mathscr A}_0^{(t)})\bigr)
		\leq df_{\bbm r}$. Let $\bbm q_{t,1},\cdots,\bbm q_{t,r_t}$ be an orthonormal basis of $\mathcal V_t$, and define $\bm Q_t=(\bbm q_{t,1},\cdots,\bbm q_{t,r_t})\in\mathbb R^{(\prod_{j=1}^{m}q_j\times \prod_{j=1}^{d}p_j)\times r_t}$. Then $\bm Q_t^\top\bm Q_t=\bm I_{r_t}$ and $\operatorname{span}(\bm Q_t)=\mathcal V_t$.
		Since the columns of $\bm Q_t$ form an orthonormal basis of $\mathcal V_t$, the Euclidean orthogonal projection of any vector $\vect(\cm Z)$ onto $\mathcal V_t$ is $\bm P_t\vect(\cm Z)$, where $\bm P_t=\bm Q_t\bm Q_t^{\top}$ is the orthogonal projection matrix onto an $r_t$-dimensional subspace. Consequently, $\bm P_t$ is symmetric and idempotent, i.e., $\bm P_t^\top=\bm P_t$ and $\bm P_t^2 = \bm P_t$. Its eigenvalues are either $0$ or $1$, and therefore $\rank(\bm P_t)=r_t\leq df_{\bbm r}$, $\|\bm P_t\|_{\op}\leq 1$, and $\|\bm P_t\|_{\F}^2=\rank(\bm P_t)=r_t\leq df_{\bbm r}$.

		On the other hand, $\mathcal P_{\mathcal T_{\bbm r}(\boldsymbol{\mathscr A}_0^{(t)})}(\cm Z)$ is the orthogonal projection of $\cm Z$ onto the tangent space of $\cm A_0^{(t)}$. Since vectorization only stacks the entries of a tensor into a vector without changing inner products, Frobenius orthogonality in the tensor space is equivalent to Euclidean orthogonality after vectorization. Therefore, vectorizing the Frobenius-orthogonal projection of $\cm Z$ onto $\mathcal T_{\bbm r}(\boldsymbol{\mathscr A}_0^{(t)})$ is the same as first vectorizing $\cm Z$ and then taking the Euclidean orthogonal projection onto $\mathcal V_t$, i.e., for every tensor $\cm Z$, $\vect(\mathcal P_{\mathcal T_{\bbm r}(\boldsymbol{\mathscr A}_0^{(t)})}(\cm Z)) = \bm P_t\vect(\cm Z)$. Hence, writing $\bbm w_t=\vect(\sum_{k=1}^K(n_k/n)\cm W_{\mathscr A,k}^{(t)})=\nu_t\bbm g_t$ with $\bbm g_t\sim N(\bm 0,\bm I)$, we have
		\[
			\left\|\mathcal P_{\mathcal T_{\bbm r}(\boldsymbol{\mathscr A}_0^{(t)})}\left(\sum_{k=1}^K\frac{n_k}{n}\cm W_{\mathscr A,k}^{(t)}\right)\right\|_{\F}^2 = \left\|\bm P_t\bbm w_t\right\|_2^2 = \bbm w_t^\top \bm P_t\bbm w_t = \nu_t^2\bbm g_t^\top\bm P_t\bbm g_t.
		\]
		Using the projection properties derived above, we have $\mathbb E(\bbm g_t^\top\bm P_t\bbm g_t) = \tr(\bm P_t) = \rank(\bm P_t) \leq df_{\bbm r}$, $\|\bm P_t\|_{\op}\leq 1$, and $\|\bm P_t\|_{\F}^2\leq df_{\bbm r}$. Applying the Hanson-Wright inequality in Lemma~\ref{lem:hanson-wright} to the quadratic form $\bbm g_t^\top\bm P_t\bbm g_t$ gives, for every $s>0$,
		\[
			\mathbb P\left(\left|\bbm g_t^\top\bm P_t\bbm g_t-\mathbb E(\bbm g_t^\top\bm P_t\bbm g_t)\right|>s\right)\leq 2\exp\left(-C\min\left\{\frac{s^2}{df_{\bbm r}},s\right\}\right).
		\]
		Taking $s=C(df_{\bbm r}+u)$, we have $\bbm g_t^\top\bm P_t\bbm g_t \leq C(df_{\bbm r}+u)$ with probability at least $1-\exp(-Cu)$. Therefore, for every $u>0$,
		\[
			\mathbb P\left(\left\|\mathcal P_{\mathcal T_{\bbm r}(\boldsymbol{\mathscr A}_0^{(t)})}\left(\sum_{k=1}^K\frac{n_k}{n}\cm W_{\mathscr A,k}^{(t)}\right)\right\|_{\F} > C\nu_t\sqrt{df_{\bbm r}+u}\right) \leq \exp(-Cu).
		\]

		Taking $u=df_{\bbm r}+\log T_g$, applying a union bound over $t=0,\cdots,T_g-1$, and using the Gaussian-mechanism calibration
		\[
			\sigma_k^{(t)} = \frac{2\tau_{\mathscr E,k}^{(t)}\tau_{\mathscr X,k}}{n_k}\cdot\frac{T_g\sqrt{2\log(1.25T_g/\delta)}}{\varepsilon},
		\]
		we have, with probability at least $1-C\exp(-C(df_{\bbm r}+\log T_g))$,
		\begin{align}\label{eq:fed-privacy-noise-summary-bound}
			\max_{0\leq t<T_g}\left\|\mathcal P_{\mathcal T_{\bbm r}(\boldsymbol{\mathscr A}_0^{(t)})}\left(\sum_{k=1}^K\frac{n_k}{n}\cm W_{\mathscr A,k}^{(t)}\right)\right\|_{\F} &\lesssim \sqrt{df_{\bbm r}+\log T_g}\max_{0\leq t<T_g}\left(\sum_{k=1}^K\left(\frac{n_k}{n}\right)^2(\sigma_k^{(t)})^2\right)^{1/2}\notag \\
			&\lesssim \frac{T_g\sqrt{\log(T_g/\delta)}}{\varepsilon n}\sqrt{K(df_{\bbm r}+\log T_g)}\max_{0\leq t<T_g}\max_{k\in[K]}\tau_{\mathscr E,k}^{(t)}\tau_{\mathscr X,k}.
		\end{align}
		This is precisely the bound required in the definition of $\mathcal E_{\mathscr A,6}^{\mathrm{fed}}$, and hence the Gaussian-noise event is included in the high probability event with probability at least $1-C\exp(-C(df_{\bbm r}+\log T_g))$.

		Combining the pooled empirical-process events $\mathcal E_{\mathscr A,1}^{\mathrm{fed}}$--$\mathcal E_{\mathscr A,4}^{\mathrm{fed}}$, the truncation event $\mathcal E_{\mathscr A,5}^{\mathrm{fed}}$, and the Gaussian-noise event $\mathcal E_{\mathscr A,6}^{\mathrm{fed}}$, a union bound yields 
		\[
			\mathbb P(\mathcal E_{\mathscr A}^{\mathrm{fed}}) \geq 1 - C\exp(-Cdf_{2\bbm r}) - C\exp(-C(df_{\bbm r}+\log T_g)) - C\exp(-C\log n). 
		\]
		The rest of the proof is conditional on $\mathcal E_{\mathscr A}^{\mathrm{fed}}$.

		\noindent
		\textbf{\emph{Step 4: one-step recursion on the high probability event.}}
		Conditional on $\mathcal E_{\mathscr A}^{\mathrm{fed}}$, we derive a one-step contraction inequality for the federated Stage-I update. By the definition of the privatized transmitted gradient in \eqref{eq:private-gradient}, client $k$ transmits
		\[
			\widetilde{\cm G}_{\mathscr A,k}^{(t)} = \mathcal P_{\mathcal T_{\bbm r}(\boldsymbol{\mathscr A}_0^{(t)})}\left(\cm G_{\mathscr A,k}^{\vee}(\cm A_0^{(t)}) + \cm W_{\mathscr A,k}^{(t)}\right).
		\]
		Moreover, on the truncation event $\mathcal E_{\mathscr A,5}^{\mathrm{fed}}$, the truncation is inactive for all iterates satisfying $\|\cm A_0^{(t)}-\cm A_0^*\|_{\F}\leq R_{\mathscr A}$. Hence, for every $k\in[K]$,
		\[
			\cm G_{\mathscr A,k}^{\vee}(\cm A_0^{(t)}) = \cm G_{\mathscr A,k}^{loc,(t)} = \frac{1}{n_k}\sum_{i=1}^{n_k}\bigl\langle\cm A_0^{(t)}-\cm A_0^*-\cm B_k^*,\cm X_{k,i}\bigr\rangle\circ\cm X_{k,i} - \frac{1}{n_k}\sum_{i=1}^{n_k}\cm E_{k,i}\circ\cm X_{k,i}.
		\]
		Therefore, using the linearity of the tangent-space projection, the aggregated noisy transmitted gradient can be written as
		\[
			\sum_{k=1}^K\frac{n_k}{n}\widetilde{\cm G}_{\mathscr A,k}^{(t)} = \mathcal P_{\mathcal T_{\bbm r}(\boldsymbol{\mathscr A}_0^{(t)})}\left(\sum_{k=1}^K\frac{n_k}{n}\cm G_{\mathscr A,k}^{loc,(t)} + \sum_{k=1}^K\frac{n_k}{n}\cm W_{\mathscr A,k}^{(t)}\right).
		\]
		Then,
		\begin{align*}
			&\|\cm A_0^{(t+1)}-\cm A_0^*\|_\F = \Bigl\|\mathcal R_{\bbm r}\Bigl(\cm A_0^{(t)} - \eta_{\mathscr A}\sum_{k=1}^K\frac{n_k}{n}\widetilde{\cm G}_{\mathscr A,k}^{(t)}\Bigr) - \cm A_0^*\Bigr\|_\F\\
			&\leq \Bigl\|\mathcal R_{\bbm r}\Bigl(\cm A_0^{(t)} - \eta_{\mathscr A}\sum_{k=1}^K\frac{n_k}{n}\widetilde{\cm G}_{\mathscr A,k}^{(t)}\Bigr) - \Bigl(\cm A_0^{(t)} - \eta_{\mathscr A}\sum_{k=1}^K\frac{n_k}{n}\widetilde{\cm G}_{\mathscr A,k}^{(t)}\Bigr)\Bigr\|_\F + \Bigl\|\cm A_0^{(t)} - \eta_{\mathscr A}\sum_{k=1}^K\frac{n_k}{n}\widetilde{\cm G}_{\mathscr A,k}^{(t)} - \cm A_0^*\Bigr\|_\F\\
			&\overset{(i)}{\leq}\sqrt{d+m}\,\Bigl\|\mathcal P_{\mathcal M_{\bbm r}}\Bigl(\cm A_0^{(t)} - \eta_{\mathscr A}\sum_{k=1}^K\frac{n_k}{n}\widetilde{\cm G}_{\mathscr A,k}^{(t)}\Bigr) - \Bigl(\cm A_0^{(t)} - \eta_{\mathscr A}\sum_{k=1}^K\frac{n_k}{n}\widetilde{\cm G}_{\mathscr A,k}^{(t)}\Bigr)\Bigr\|_\F + \Bigl\|\cm A_0^{(t)} - \eta_{\mathscr A}\sum_{k=1}^K\frac{n_k}{n}\widetilde{\cm G}_{\mathscr A,k}^{(t)} - \cm A_0^*\Bigr\|_\F\\
			&\overset{(ii)}{\leq}(\sqrt{d+m}+1)\Bigl\|\cm A_0^{(t)} - \eta_{\mathscr A}\sum_{k=1}^K\frac{n_k}{n}\widetilde{\cm G}_{\mathscr A,k}^{(t)} - \cm A_0^*\Bigr\|_\F\\
			&=(\sqrt{d+m}+1)\Biggl\|\mathcal P_{\mathcal T_{\bbm r}(\boldsymbol{\mathscr A}_0^{(t)})}(\cm A_0^{(t)}-\cm A_0^*) + \mathcal P_{\mathcal T_{\bbm r}^{\perp}(\boldsymbol{\mathscr A}_0^{(t)})}(\cm A_0^{(t)}-\cm A_0^*)\\
			&\quad -\eta_{\mathscr A}\mathcal P_{\mathcal T_{\bbm r}(\boldsymbol{\mathscr A}_0^{(t)})}\Biggl[\sum_{k=1}^K \frac{n_k}{n}\left\{\frac{1}{n_k}\sum_{i=1}^{n_k}\bigl\langle\cm A_0^{(t)}-\cm A_0^*-\cm B_k^*,\cm X_{k,i}\bigr\rangle\circ\cm X_{k,i} - \frac{1}{n_k}\sum_{i=1}^{n_k}\cm E_{k,i}\circ\cm X_{k,i}\right\} + \sum_{k=1}^K \frac{n_k}{n}\cm W_{\mathscr A,k}^{(t)}\Biggr]\Biggr\|_\F\\
			&\leq (\sqrt{d+m}+1)\Biggl(\underbrace{\Bigl\|\mathcal P_{\mathcal T_{\bbm r}^{\perp}(\boldsymbol{\mathscr A}_0^{(t)})}(\cm A_0^{(t)}-\cm A_0^*)\Bigr\|_\F}_{\mathrm{Term\ I}} 
			+ \underbrace{\eta_{\mathscr A}\Biggl\|\mathcal P_{\mathcal T_{\bbm r}(\boldsymbol{\mathscr A}_0^{(t)})}\left(\sum_{k=1}^K \frac{n_k}{n}\frac{1}{n_k}\sum_{i=1}^{n_k}\langle \cm B_k^*,\cm X_{k,i}\rangle\circ\cm X_{k,i}\right)\Biggr\|_\F}_{\mathrm{Term\ II}}\\
			&\quad+ \underbrace{\eta_{\mathscr A}\Biggl\|\mathcal P_{\mathcal T_{\bbm r}(\boldsymbol{\mathscr A}_0^{(t)})}\left(\sum_{k=1}^K \frac{n_k}{n}\frac{1}{n_k}\sum_{i=1}^{n_k}\cm E_{k,i}\circ\cm X_{k,i}\right)\Biggr\|_\F}_{\mathrm{Term\ III}}\\
			&\quad+ \underbrace{\Biggl\|\mathcal P_{\mathcal T_{\bbm r}(\boldsymbol{\mathscr A}_0^{(t)})}\Biggl(\cm A_0^{(t)}-\cm A_0^* - \eta_{\mathscr A}\sum_{k=1}^K \frac{n_k}{n}\frac{1}{n_k}\sum_{i=1}^{n_k}\Bigl\langle\mathcal P_{\mathcal T_{\bbm r}(\boldsymbol{\mathscr A}_0^{(t)})}(\cm A_0^{(t)}-\cm A_0^*),\cm X_{k,i}\Bigr\rangle\circ\cm X_{k,i}\Biggr)\Biggr\|_\F}_{\mathrm{Term\ IV}}\\
			&\quad+ \underbrace{\eta_{\mathscr A}\Biggl\|\mathcal P_{\mathcal T_{\bbm r}(\boldsymbol{\mathscr A}_0^{(t)})}\left(\sum_{k=1}^K \frac{n_k}{n}\frac{1}{n_k}\sum_{i=1}^{n_k}\Bigl\langle\mathcal P_{\mathcal T_{\bbm r}^{\perp}(\boldsymbol{\mathscr A}_0^{(t)})}(\cm A_0^{(t)}-\cm A_0^*),\cm X_{k,i}\Bigr\rangle\circ\cm X_{k,i}\right)\Biggr\|_\F}_{\mathrm{Term\ V}}
			+ \underbrace{\eta_{\mathscr A}\Biggl\|\mathcal P_{\mathcal T_{\bbm r}(\boldsymbol{\mathscr A}_0^{(t)})}\left(\sum_{k=1}^K \frac{n_k}{n}\cm W_{\mathscr A,k}^{(t)}\right)\Biggr\|_\F}_{\mathrm{Term\ VI}}\Biggr).
		\end{align*}
		In inequality $(i)$, $\mathcal M_{\bbm r}=\{\cm A:\tucrank(\cm A)\leq\bbm r\}$ denotes the Tucker-rank-$\bbm r$ tensor set, and $\mathcal P_{\mathcal M_{\bbm r}}(\cm Z)$ denotes the best Frobenius-norm projection of $\cm Z$ onto $\mathcal M_{\bbm r}$. Thus, inequality $(i)$ follows from the quasi-projection property in Lemma~\ref{lem:quasi-projection-thosvd-sthosvd}. Inequality $(ii)$ follows from Lemma~\ref{lem:quasi-projection-thosvd-sthosvd} by using $\cm A_0^*\in\mathcal M_{\bbm r}$ as a feasible Tucker-rank-$\bbm r$ approximation of $\cm A_0^{(t)} - \eta_{\mathscr A}\sum_{k=1}^K\frac{n_k}{n}\widetilde{\cm G}_{\mathscr A,k}^{(t)}$.

		For \textbf{Term I}, by Lemma~\ref{lem:normal-component-bound}, we have
		\begin{align}\label{eq:term-I-bound-fed}
			\textbf{Term\ I} = \Bigl\|\mathcal P_{\mathcal T_{\bbm r}^{\perp}(\boldsymbol{\mathscr A}_{0}^{(t)})}\bigl(\cm A_{0}^{(t)}-\cm A_0^*\bigr)\Bigr\|_\F \leq \frac{2(d+m)}{\mu_{\mathscr A}^{\min}}\|\cm A_{0}^{(t)}-\cm A_0^*\|_\F^2.
		\end{align}

		For \textbf{Term II}, by the variational characterization of the Frobenius norm on the tangent space,
		\begin{align*}
			&\Biggl\|\mathcal P_{\mathcal T_{\bbm r}(\boldsymbol{\mathscr A}_{0}^{(t)})}\left(\sum_{k=1}^K\frac{n_k}{n}\frac{1}{n_k}\sum_{i=1}^{n_k}\langle \cm B_k^*,\cm X_{k,i}\rangle\circ\cm X_{k,i}\right)\Biggr\|_\F \\
			=& \sup_{\substack{\boldsymbol{\mathscr U}\in\mathcal T_{\bbm r}(\boldsymbol{\mathscr A}_{0}^{(t)})\\\|\boldsymbol{\mathscr U}\|_\F=1}}\left|\sum_{k=1}^K\frac{n_k}{n}\frac{1}{n_k}\sum_{i=1}^{n_k}\left\langle\langle \cm B_k^*,\cm X_{k,i}\rangle,\langle \cm U,\cm X_{k,i}\rangle\right\rangle\right| \\
			\leq& \sup_{\substack{\boldsymbol{\mathscr U}\in\mathcal T_{\bbm r}(\boldsymbol{\mathscr A}_{0}^{(t)})\\\|\boldsymbol{\mathscr U}\|_\F=1}}\left(\left|\sum_{k=1}^K\frac{n_k}{n}\mathbb E\bigl[\left\langle\langle \cm B_k^*,\cm X_{k,i}\rangle,\langle \cm U,\cm X_{k,i}\rangle\right\rangle\bigr]\right| + \left|\sum_{k=1}^K\frac{n_k}{n}\left[\frac{1}{n_k}\sum_{i=1}^{n_k}\left\langle\langle \cm B_k^*,\cm X_{k,i}\rangle,\langle \cm U,\cm X_{k,i}\rangle\right\rangle -\mathbb E\bigl\{\left\langle\langle \cm B_k^*,\cm X_{k,i}\rangle,\langle \cm U,\cm X_{k,i}\rangle\right\rangle\bigr\}\right]\right|\right).
		\end{align*}
		By \eqref{eq:expectation-B-k-U-X}, we have $\left|\mathbb E\left[\left\langle\langle\cm B_k^*,\cm X_{k,i}\rangle,\langle\cm U,\cm X_{k,i}\rangle\right\rangle\right]\right| \leq \lambda_{\mathscr X,k}^{\max}\left\|(\cm B_k^*)_{[S_X]}\right\|_{\op}\left\|\cm U_{[S_X]}\right\|_*$.
		Since $\cm U\in\mathcal T_{\bbm r}(\boldsymbol{\mathscr A}_{0}^{(t)})$, Lemma~\ref{lem:tangent-space-nuclear-bound} gives $\|\cm U_{[S_X]}\|_* \leq \sqrt{r_{\bbm q}+r_{\bbm p}}\|\cm U\|_\F$. Together with the weak-identifiability condition in Assumption~\ref{assump:weak-identifiability}, namely $\max_{k\in[K]}\|(\cm B_k^*)_{[S_X]}\|_{\op}\leq \zeta$, for every $\cm U\in\mathcal T_{\bbm r}(\boldsymbol{\mathscr A}_{0}^{(t)})$ with $\|\cm U\|_\F=1$, we have
		\[
			\left|\sum_{k=1}^K\frac{n_k}{n}\mathbb E\bigl[\left\langle\langle \cm B_k^*,\cm X_{k,i}\rangle,\langle \cm U,\cm X_{k,i}\rangle\right\rangle\bigr]\right| \leq \sqrt{r_{\bbm q}+r_{\bbm p}}\sum_{k=1}^K\frac{n_k}{n}\lambda_{\mathscr X,k}^{\max}\zeta \leq \sqrt{r_{\bbm q}+r_{\bbm p}}\bar\lambda_{\mathscr X}^{\max}\zeta.
		\]
		Given $\mathcal E_{\mathscr A,1}^{\mathrm{fed}}$, we have
		\begin{align}\label{eq:term-II-bound-fed}
			\textbf{Term\ II} \leq \eta_{\mathscr A}\sqrt{r_{\bbm q}+r_{\bbm p}}\bar\lambda_{\mathscr X}^{\max}\zeta + C\eta_{\mathscr A}\bar\sigma_{\mathscr X}^2\bar\lambda_{\mathscr X}^{\max}\bar h_{\mathscr B}\sqrt{\frac{df_{\bbm r}}{n}}.
		\end{align}

		For \textbf{Term III}, by the variational characterization of the Frobenius norm on the tangent space,
		\[
			\mathrm{Term\ III} = \eta_{\mathscr A}\left\|\mathcal P_{\mathcal T_{\bbm r}(\boldsymbol{\mathscr A}_{0}^{(t)})}\left(\sum_{k=1}^K\frac{n_k}{n}\frac{1}{n_k}\sum_{i=1}^{n_k}\cm E_{k,i}\circ\cm X_{k,i}\right)\right\|_\F = \eta_{\mathscr A}\sup_{\substack{\boldsymbol{\mathscr U}\in\mathcal T_{\bbm r}(\boldsymbol{\mathscr A}_{0}^{(t)})\\ \|\boldsymbol{\mathscr U}\|_\F\leq1}}\left|\left\langle\sum_{k=1}^K\frac{n_k}{n}\frac{1}{n_k}\sum_{i=1}^{n_k}\cm E_{k,i}\circ\cm X_{k,i},\cm U\right\rangle\right|.
		\]
		By Lemma~\ref{lem:tangent-space-tucker-rank-2r}, every element in $\mathcal T_{\bbm r}(\boldsymbol{\mathscr A}_{0}^{(t)})$ has Tucker rank at most $2\bbm r$. Moreover, since the functional inside the inner product is linear in $\cm U$, the supremum over the Tucker-rank-$2\bbm r$ unit ball is the same as the supremum over $\mathbb S_{2\bbm r}$. Therefore, on the event $\mathcal E_{\mathscr A,2}^{\mathrm{fed}}$,
		\begin{align}\label{eq:term-III-bound-fed}
			\textbf{Term\ III} \leq C\eta_{\mathscr A}\bar\sigma_{\mathscr X}\bar\sigma_{\mathscr E}\sqrt{\bar\lambda_{\mathscr X}^{\max}\bar\lambda_{\mathscr E}^{\max}}\sqrt{\frac{df_{\bbm r}}{n}}.
		\end{align}

		For \textbf{Term IV}, note that it can be written as
		\[
			\Biggl\|\mathcal P_{\mathcal T_{\bbm r}(\boldsymbol{\mathscr A}_{0}^{(t)})}(\cm A_0^{(t)}-\cm A_0^*) - \eta_{\mathscr A}\mathcal P_{\mathcal T_{\bbm r}(\boldsymbol{\mathscr A}_{0}^{(t)})}\Biggl(\sum_{k=1}^K\frac{n_k}{n}\frac{1}{n_k}\sum_{i=1}^{n_k}\Bigl\langle\mathcal P_{\mathcal T_{\bbm r}(\boldsymbol{\mathscr A}_{0}^{(t)})}(\cm A_0^{(t)}-\cm A_0^*),\cm X_{k,i}\Bigr\rangle\circ\cm X_{k,i}\Biggr)\Biggr\|_\F.
		\]
		By Lemma~\ref{lem:tangent-space-tucker-rank-2r}, every element of $\mathcal T_{\bbm r}(\boldsymbol{\mathscr A}_{0}^{(t)})$ has Tucker rank at most $2\bbm r$. Thus, on the event $\mathcal E_{\mathscr A,3}^{\mathrm{fed}}$, the weighted restricted empirical quadratic bound holds uniformly over all $\cm T\in\mathcal T_{\bbm r}(\boldsymbol{\mathscr A}_{0}^{(t)})$:
		\begin{align}\label{eq:concentration-T-X-weighted}
			\left|\sum_{k=1}^K\frac{n_k}{n}\left[\frac{1}{n_k}\sum_{i=1}^{n_k}\|\langle \cm T,\cm X_{k,i}\rangle\|_\F^2 - \mathbb E\|\langle \cm T,\cm X_{k,i}\rangle\|_\F^2\right]\right| \leq \delta_{2\bbm r}^{\mathrm{fed}}\|\cm T\|_\F^2,
		\end{align}
		where $\delta_{2\bbm r}^{\mathrm{fed}} = C\bar\sigma_{\mathscr X}^2\bar\lambda_{\mathscr X}^{\max}\sqrt{df_{\bbm r}/n}$. 
		Moreover,
		\begin{align}\label{eq:T-X-exp-equivalent-form}
			\sum_{k=1}^K\frac{n_k}{n}\mathbb E\|\langle \cm T,\cm X_{k,i}\rangle\|_\F^2 = \sum_{k=1}^K\frac{n_k}{n}\tr\left(\cm T_{[S_X]}\bbm\Sigma_{\mathscr X,k}\cm T_{[S_X]}^\top\right).
		\end{align}
		Since the weights $n_k/n$ are nonnegative and sum to one, and the eigenvalues of $\bbm\Sigma_{\mathscr X,k}$ are bounded between $\lambda_{\mathscr X,k}^{\min}$ and $\lambda_{\mathscr X,k}^{\max}$, we have $\underline\lambda_{\mathscr X}^{\min}\|\cm T\|_\F^2 \leq \sum_{k=1}^K(n_k/n)\mathbb E\|\langle \cm T,\cm X_{k,i}\rangle\|_\F^2 \leq \bar\lambda_{\mathscr X}^{\max}\|\cm T\|_\F^2$, where $\underline\lambda_{\mathscr X}^{\min}=\min_{k\in[K]}\lambda_{\mathscr X,k}^{\min}$ and $\bar\lambda_{\mathscr X}^{\max}=\max_{k\in[K]}\lambda_{\mathscr X,k}^{\max}$. Combining \eqref{eq:concentration-T-X-weighted} and \eqref{eq:T-X-exp-equivalent-form} gives, for all $\cm T\in\mathcal T_{\bbm r}(\boldsymbol{\mathscr A}_{0}^{(t)})$,
		\[
			(\underline\lambda_{\mathscr X}^{\min}-\delta_{2\bbm r}^{\mathrm{fed}})\|\cm T\|_\F^2 \leq \sum_{k=1}^K\frac{n_k}{n}\frac{1}{n_k}\sum_{i=1}^{n_k}\|\langle \cm T,\cm X_{k,i}\rangle\|_\F^2 \leq (\bar\lambda_{\mathscr X}^{\max}+\delta_{2\bbm r}^{\mathrm{fed}})\|\cm T\|_\F^2.
		\]
		Therefore, applying Lemma~\ref{lem:tangent-contraction-tensor-general}, we have
		\begin{align}\label{eq:term-IV-bound-fed}
			\textbf{Term\ IV} \leq \rho_{\mathscr A}^{\mathrm{fed}}\left\|\mathcal P_{\mathcal T_{\bbm r}(\boldsymbol{\mathscr A}_{0}^{(t)})}(\cm A_0^{(t)}-\cm A_0^*)\right\|_\F \leq \rho_{\mathscr A}^{\mathrm{fed}}\|\cm A_0^{(t)}-\cm A_0^*\|_\F,
		\end{align}
		where $\rho_{\mathscr A}^{\mathrm{fed}} = \max\left\{|1-\eta_{\mathscr A}(\underline\lambda_{\mathscr X}^{\min}-\delta_{2\bbm r}^{\mathrm{fed}})|,|1-\eta_{\mathscr A}(\bar\lambda_{\mathscr X}^{\max}+\delta_{2\bbm r}^{\mathrm{fed}})|\right\}$.

		For \textbf{Term V}, by the variational characterization of the Frobenius norm, we have
		\begin{align*}
			\mathrm{Term\ V} &= \eta_{\mathscr A}\Biggl\|\mathcal P_{\mathcal T_{\bbm r}(\boldsymbol{\mathscr A}_{0}^{(t)})}\left(\sum_{k=1}^K\frac{n_k}{n}\frac{1}{n_k}\sum_{i=1}^{n_k}\Bigl\langle\mathcal P_{\mathcal T_{\bbm r}^{\perp}(\boldsymbol{\mathscr A}_{0}^{(t)})}(\cm A_0^{(t)}-\cm A_0^*),\cm X_{k,i}\Bigr\rangle\circ\cm X_{k,i}\right)\Biggr\|_\F \\
			&= \eta_{\mathscr A}\sup_{\substack{\boldsymbol{\mathscr T}\in\mathcal T_{\bbm r}(\boldsymbol{\mathscr A}_{0}^{(t)})\\ \|\boldsymbol{\mathscr T}\|_\F\leq1}}\left|\sum_{k=1}^K\frac{n_k}{n}\frac{1}{n_k}\sum_{i=1}^{n_k}\left\langle\Bigl\langle\mathcal P_{\mathcal T_{\bbm r}^{\perp}(\boldsymbol{\mathscr A}_{0}^{(t)})}(\cm A_0^{(t)}-\cm A_0^*),\cm X_{k,i}\Bigr\rangle,\langle\cm T,\cm X_{k,i}\rangle\right\rangle\right|.
		\end{align*}

		Let $\cm N^{(t)}=\mathcal P_{\mathcal T_{\bbm r}^{\perp}(\boldsymbol{\mathscr A}_{0}^{(t)})}(\cm A_0^{(t)}-\cm A_0^*)$. For any $\cm T\in\mathcal T_{\bbm r}(\boldsymbol{\mathscr A}_{0}^{(t)})$ with $\|\cm T\|_\F\leq1$, decompose the preceding bilinear form into its population and empirical fluctuation parts. For the population part, by the matricization convention associated with $S_X$,
		\[
			\sum_{k=1}^K\frac{n_k}{n}\mathbb E\left[\left\langle\langle \cm N^{(t)},\cm X_{k,i}\rangle,\langle \cm T,\cm X_{k,i}\rangle\right\rangle\right]
			= \sum_{k=1}^K\frac{n_k}{n}\tr\left((\cm N^{(t)})_{[S_X]}\bbm\Sigma_{\mathscr X,k}\cm T_{[S_X]}^\top\right).
		\]
		By the orthogonality of the projection operator, we have $\langle\cm N^{(t)},\cm T\rangle=\tr((\cm N^{(t)})_{[S_X]}\cm T_{[S_X]}^\top)=0$. Then, for any scalar $c$,
		\[
			\sum_{k=1}^K\frac{n_k}{n}\mathbb E\left[\left\langle\langle \cm N^{(t)},\cm X_{k,i}\rangle,\langle \cm T,\cm X_{k,i}\rangle\right\rangle\right]
			= \sum_{k=1}^K\frac{n_k}{n}\tr\left((\cm N^{(t)})_{[S_X]}(\bbm\Sigma_{\mathscr X,k}-c\bm I)\cm T_{[S_X]}^\top\right).
		\]
		Taking $c=(\bar\lambda_{\mathscr X}^{\max}+\underline\lambda_{\mathscr X}^{\min})/2$, we have $\|\bbm\Sigma_{\mathscr X,k}-c\bm I\|_{\op} \leq (\bar\lambda_{\mathscr X}^{\max}-\underline\lambda_{\mathscr X}^{\min})/2$ for $k\in[K]$. Thus, since $\|\cm T\|_\F\leq1$,
		\begin{align}\label{eq:population-part-term-v}
			\left|\sum_{k=1}^K\frac{n_k}{n}\mathbb E\left[\left\langle\langle \cm N^{(t)},\cm X_{k,i}\rangle,\langle \cm T,\cm X_{k,i}\rangle\right\rangle\right]\right| \leq \frac{\bar\lambda_{\mathscr X}^{\max}-\underline\lambda_{\mathscr X}^{\min}}{2}\|\cm N^{(t)}\|_\F.
		\end{align}

		For the empirical fluctuation part, Lemma~\ref{lem:normal-projection-tucker-rank-2r} gives $\tucrank(\mathcal P_{\mathcal T_{\bbm r}^{\perp}(\boldsymbol{\mathscr A}_{0}^{(t)})}(\cm A_0^{(t)}-\cm A_0^*))\leq 2\bbm r$, and every unit-Frobenius element of $\mathcal T_{\bbm r}(\boldsymbol{\mathscr A}_{0}^{(t)})$ belongs to $\mathbb S_{2\bbm r}$. Hence, on the event $\mathcal E_{\mathscr A,4}^{\mathrm{fed}}$,
		\begin{align}\label{eq:empirical-part-term-v}
			&\sup_{\substack{\boldsymbol{\mathscr T}\in\mathcal T_{\bbm r}(\boldsymbol{\mathscr A}_{0}^{(t)})\\ \|\boldsymbol{\mathscr T}\|_\F\leq1}}\left|\sum_{k=1}^K\frac{n_k}{n}\left[\frac{1}{n_k}\sum_{i=1}^{n_k}\left\langle\langle \cm N^{(t)},\cm X_{k,i}\rangle,\langle\cm T,\cm X_{k,i}\rangle\right\rangle - \mathbb E\left[\left\langle\langle \cm N^{(t)},\cm X_{k,i}\rangle,\langle\cm T,\cm X_{k,i}\rangle\right\rangle\right]\right]\right| \leq \delta_{2\bbm r}^{\mathrm{fed}}\|\cm N^{(t)}\|_\F.
		\end{align}
		Combining \eqref{eq:population-part-term-v}, \eqref{eq:empirical-part-term-v} and Lemma~\ref{lem:normal-component-bound}, we have
		\begin{align}\label{eq:term-V-bound-fed}
			\textbf{Term\ V} &\leq \eta_{\mathscr A}\left(\frac{\bar\lambda_{\mathscr X}^{\max}-\underline\lambda_{\mathscr X}^{\min}}{2} + \delta_{2\bbm r}^{\mathrm{fed}}\right)\|\mathcal P_{\mathcal T_{\bbm r}^{\perp}(\boldsymbol{\mathscr A}_{0}^{(t)})}(\cm A_0^{(t)}-\cm A_0^*)\|_\F \notag\\
			&\leq \eta_{\mathscr A}\left(\frac{\bar\lambda_{\mathscr X}^{\max}-\underline\lambda_{\mathscr X}^{\min}}{2} + \delta_{2\bbm r}^{\mathrm{fed}}\right)\frac{2(d+m)}{\mu_{\mathscr A}^{\min}}\|\cm A_0^{(t)}-\cm A_0^*\|_\F^2\notag \\
			&= \eta_{\mathscr A}(\bar\lambda_{\mathscr X}^{\max}-\underline\lambda_{\mathscr X}^{\min}+2\delta_{2\bbm r}^{\mathrm{fed}})\frac{d+m}{\mu_{\mathscr A}^{\min}}\|\cm A_0^{(t)}-\cm A_0^*\|_\F^2.
		\end{align}

		For \textbf{Term VI}, on the event $\mathcal E_{\mathscr A,6}^{\mathrm{fed}}$, we have
		\begin{align}\label{eq:term-VI-bound-fed}
			\textbf{Term\ VI} = \eta_{\mathscr A}\left\|\mathcal P_{\mathcal T_{\bbm r}(\boldsymbol{\mathscr A}_0^{(t)})}\left(\sum_{k=1}^K \frac{n_k}{n}\cm W_{\mathscr A,k}^{(t)}\right)\right\|_\F \lesssim \eta_{\mathscr A}\frac{T_g\sqrt{\log(T_g/\delta)}}{\varepsilon n}\sqrt{K(df_{\bbm r}+\log T_g)}\max_{0\leq t<T_g}\max_{k\in[K]}\tau_{\mathscr E,k}^{(t)}\tau_{\mathscr X,k}.
		\end{align}

		Combining \eqref{eq:term-I-bound-fed}, \eqref{eq:term-II-bound-fed}, \eqref{eq:term-III-bound-fed}, \eqref{eq:term-IV-bound-fed}, \eqref{eq:term-IV-bound-fed} and \eqref{eq:term-VI-bound-fed}, and using $\delta_{2\bbm r}^{\mathrm{fed}}\leq\bar\delta_{\mathscr A}$ given the sample size requirement in Theorem \ref{thm:federated_representation_error}, we have
		\begin{align}\label{eq:federated-one-step-recursion}
			\|\cm A_0^{(t+1)}-\cm A_0^*\|_\F \leq &\left\{(\sqrt{d+m}+1)\rho_{\mathscr A}^{\mathrm{fed}} + \frac{\Gamma_{\mathscr A}^{\mathrm{fed}}}{\mu_{\mathscr A}^{\min}}\|\cm A_0^{(t)}-\cm A_0^*\|_\F\right\}\|\cm A_0^{(t)}-\cm A_0^*\|_\F+(\sqrt{d+m}+1)\sum_{\ell=1}^4\mathsf{Rem}_{\mathscr A}^{\mathrm{fed},(\ell)},
		\end{align}
		where
		\[
		\begin{aligned}
			\mathsf{Rem}_{\mathscr A}^{\mathrm{fed},(1)} &= \eta_{\mathscr A}\sqrt{r_{\bbm q}+r_{\bbm p}}\bar\lambda_{\mathscr X}^{\max}\zeta,\\
			\mathsf{Rem}_{\mathscr A}^{\mathrm{fed},(2)} &= C\eta_{\mathscr A}\bar\sigma_{\mathscr X}^2\bar\lambda_{\mathscr X}^{\max}\bar h_{\mathscr B}\sqrt{\frac{df_{\bbm r}}{n}},\\
			\mathsf{Rem}_{\mathscr A}^{\mathrm{fed},(3)} &= C\eta_{\mathscr A}\bar\sigma_{\mathscr X}\bar\sigma_{\mathscr E}\sqrt{\bar\lambda_{\mathscr X}^{\max}\bar\lambda_{\mathscr E}^{\max}}\sqrt{\frac{df_{\bbm r}}{n}},\\
			\mathsf{Rem}_{\mathscr A}^{\mathrm{fed},(4)} &= C\eta_{\mathscr A} \frac{T_g\sqrt{\log(T_g/\delta)}}{\varepsilon n}\sqrt{K(df_{\bbm r}+\log T_g)}\max_{0\leq t<T_g}\max_{k\in[K]}\tau_{\mathscr E,k}^{(t)}\tau_{\mathscr X,k}.
		\end{aligned}
		\]

		\noindent
		\textbf{\emph{Step 5: privacy guarantee, induction, and final error bound.}} 
		We first verify the privacy guarantee. By \eqref{eq:truncated-gradient-sensitivity-bound}, the sensitivity of the truncated local gradient $\cm G_{\mathscr A,k}^{\vee}(\cm A_0^{(t)};\mathsf D_k)$ at iteration $t$ is deterministically bounded by $\Delta_{\mathscr A,k}^{\vee,(t)} \leq 2\tau_{\mathscr E,k}^{(t)}\tau_{\mathscr X,k}/n_k$ for every pair of neighboring local datasets. Given the Gaussian perturbation scale
		\begin{align}\label{eq:per-round-gaussian-noise-scale}
			\sigma_k^{(t)} = \frac{2\tau_{\mathscr E,k}^{(t)}\tau_{\mathscr X,k}}{n_k}\cdot \frac{T_g\sqrt{2\log(1.25T_g/\delta)}}{\varepsilon},
		\end{align}
		the Gaussian mechanism, e.g., \citet[Theorem~3.22]{dwork2014algorithmic}, ensures that the released message $\cm G_{\mathscr A,k}^{\vee}(\cm A_0^{(t)};\mathsf D_k)+\cm W_{\mathscr A,k}^{(t)}$ is $(\varepsilon/T_g,\delta/T_g)$-DP for client $k$ at iteration $t$. The subsequent server-side aggregation, tangent-space projection, and update are post-processing operations that do not access the raw local data and therefore incur no additional privacy loss. By the basic composition theorem, e.g., \citet[Theorem~3.1.6]{dwork2014algorithmic}, the sequence of messages released by client $k$ over $T_g$ iterations is $(\varepsilon,\delta)$-DP. Thus, the federated Stage-I procedure satisfies $(\varepsilon,\delta)$-DP for each client.

		We now carry out the induction using the recursion in \eqref{eq:federated-one-step-recursion}. Suppose that the induction hypothesis $\|\cm A_0^{(t)}-\cm A_0^*\|_\F\leq R_{\mathscr A}$ holds, with $R_{\mathscr A}=\mu_{\mathscr A}^{\min}/(4\Gamma_{\mathscr A}^{\mathrm{fed}})$. Recall that $c_{d,m}=1/\{4(\sqrt{d+m}+1)\}$ and the step size $\eta_{\mathscr A}\in\mathcal I_{\mathscr A}^{\mathrm{fed}}=[\underline\eta_{\mathscr A}^{\mathrm{fed}},\bar\eta_{\mathscr A}^{\mathrm{fed}}]$, where
		\[
			\underline\eta_{\mathscr A}^{\mathrm{fed}} = \frac{1-c_{d,m}}{\underline\lambda_{\mathscr X}^{\min}-\bar\delta_{\mathscr A}} = \frac{4(\sqrt{d+m}+1)-1}{4(\sqrt{d+m}+1)(\underline\lambda_{\mathscr X}^{\min}-\bar\delta_{\mathscr A})} 
			\ \text{and} \ 
			\bar\eta_{\mathscr A}^{\mathrm{fed}}= \frac{1+c_{d,m}}{\bar\lambda_{\mathscr X}^{\max}+\bar\delta_{\mathscr A}} = \frac{4(\sqrt{d+m}+1)+1}{4(\sqrt{d+m}+1)(\bar\lambda_{\mathscr X}^{\max}+\bar\delta_{\mathscr A})}.
		\]
		The nonemptiness of $\mathcal I_{\mathscr A}^{\mathrm{fed}}$ follows from $2\bar\delta_{\mathscr A} \leq c_{d,m}(\underline\lambda_{\mathscr X}^{\min}+\bar\lambda_{\mathscr X}^{\max}) - (\bar\lambda_{\mathscr X}^{\max}-\underline\lambda_{\mathscr X}^{\min})$, which is guaranteed by the definition of $\bar\delta_{\mathscr A}$ and the condition $\bar\kappa_{\mathscr X}<(1+c_{d,m})/(1-c_{d,m})$. Let $\bar a=\underline\lambda_{\mathscr X}^{\min}-\delta_{2\bbm r}^{\mathrm{fed}}$ and $\bar b=\bar\lambda_{\mathscr X}^{\max}+\delta_{2\bbm r}^{\mathrm{fed}}$. Since $\delta_{2\bbm r}^{\mathrm{fed}}\leq\bar\delta_{\mathscr A}$, we have $\underline\lambda_{\mathscr X}^{\min}-\delta_{2\bbm r}^{\mathrm{fed}}\geq\underline\lambda_{\mathscr X}^{\min}-\bar\delta_{\mathscr A}$ and $\bar\lambda_{\mathscr X}^{\max}+\delta_{2\bbm r}^{\mathrm{fed}}\leq\bar\lambda_{\mathscr X}^{\max}+\bar\delta_{\mathscr A}$. Thus, for every $\eta_{\mathscr A}\in\mathcal I_{\mathscr A}^{\mathrm{fed}}$,
		\[
			\frac{1-c_{d,m}}{\bar a} \leq \frac{1-c_{d,m}}{\underline\lambda_{\mathscr X}^{\min}-\bar\delta_{\mathscr A}} \leq \eta_{\mathscr A} \leq \frac{1+c_{d,m}}{\bar\lambda_{\mathscr X}^{\max}+\bar\delta_{\mathscr A}} \leq \frac{1+c_{d,m}}{\bar b}.
		\]
		Therefore, $\eta_{\mathscr A}\bar a\geq1-c_{d,m}$. Moreover, using $\bar a\leq\bar b$, we also have $\eta_{\mathscr A}\bar a\leq\eta_{\mathscr A}\bar b\leq1+c_{d,m}$. Therefore, $1-c_{d,m}\leq\eta_{\mathscr A}\bar a\leq1+c_{d,m}$, which implies $|1-\eta_{\mathscr A}\bar a|\leq c_{d,m}$. Similarly, the upper bound on $\eta_{\mathscr A}$ gives $\eta_{\mathscr A}\bar b\leq1+c_{d,m}$, while $\bar b\geq\bar a$ and the preceding lower bound imply $\eta_{\mathscr A}\bar b\geq\eta_{\mathscr A}\bar a\geq1-c_{d,m}$. Thus, $1-c_{d,m}\leq\eta_{\mathscr A}\bar b\leq1+c_{d,m}$, and hence $|1-\eta_{\mathscr A}\bar b|\leq c_{d,m}$. Note that $\rho_{\mathscr A}^{\mathrm{fed}}=\max\{|1-\eta_{\mathscr A}\bar a|,|1-\eta_{\mathscr A}\bar b|\}$. Then, combining the two bounds gives
		\begin{align}\label{eq:rho-fed-upper-bound}
			\rho_{\mathscr A}^{\mathrm{fed}} \leq c_{d,m} = \frac{1}{4(\sqrt{d+m}+1)}.
		\end{align}
		Consequently, $(\sqrt{d+m}+1)\rho_{\mathscr A}^{\mathrm{fed}}\leq1/4$. Moreover, by the induction hypothesis $\|\cm A_0^{(t)}-\cm A_0^*\|_\F \leq R_{\mathscr A}$ and $R_{\mathscr A}=\mu_{\mathscr A}^{\min}/(4\Gamma_{\mathscr A}^{\mathrm{fed}})$, we have
		\begin{align}\label{eq:gamma-mu-A-upper-bound}
			\frac{\Gamma_{\mathscr A}^{\mathrm{fed}}}{\mu_{\mathscr A}^{\min}}\|\cm A_0^{(t)}-\cm A_0^*\|_\F \leq \frac{\Gamma_{\mathscr A}^{\mathrm{fed}}}{\mu_{\mathscr A}^{\min}}R_{\mathscr A} = \frac14.
		\end{align}
		Substituting \eqref{eq:rho-fed-upper-bound} and \eqref{eq:gamma-mu-A-upper-bound} into \eqref{eq:federated-one-step-recursion} yields
		\begin{align}\label{eq:federated-clean-recursion}
		\|\cm A_0^{(t+1)}-\cm A_0^*\|_\F \leq \frac12\|\cm A_0^{(t)}-\cm A_0^*\|_\F + (\sqrt{d+m}+1)\sum_{\ell=1}^{4}\mathsf{Rem}_{\mathscr A}^{\mathrm{fed},(\ell)}.
		\end{align}
		The weak-identifiability condition $\eta_{\mathscr A}\sqrt{r_{\bbm q}+r_{\bbm p}}\,\bar\lambda_{\mathscr X}^{\max}\zeta\leq c_{d,m}R_{\mathscr A}$ implies that
		\begin{align}\label{eq:fed-rem1-small}
			(\sqrt{d+m}+1)\mathsf{Rem}_{\mathscr A}^{\mathrm{fed},(1)} \leq \frac14 R_{\mathscr A}.
		\end{align}
		We next verify that $\mathsf{Rem}_{\mathscr A}^{\mathrm{fed},(2)}$, $\mathsf{Rem}_{\mathscr A}^{\mathrm{fed},(3)}$, and $\mathsf{Rem}_{\mathscr A}^{\mathrm{fed},(4)}$ are also small under the sample size requirement in Theorem~\ref{thm:federated_representation_error}. Using the upper bound $\eta_{\mathscr A}\leq\bar\eta_{\mathscr A}^{\mathrm{fed}}$ and $R_{\mathscr A}=\mu_{\mathscr A}^{\min}/(4\Gamma_{\mathscr A}^{\mathrm{fed}})$, we have, after absorbing constants depending only on $d$ and $m$,
		\[
			\eta_{\mathscr A} \leq \frac{C_{d,m}}{\bar\lambda_{\mathscr X}^{\max}+\bar\delta_{\mathscr A}}, 
			\quad
			\frac{\bar\lambda_{\mathscr X}^{\max}-\underline\lambda_{\mathscr X}^{\min}+2\bar\delta_{\mathscr A}}{\bar\lambda_{\mathscr X}^{\max}+\bar\delta_{\mathscr A}}\leq 2,
		\]
		and hence $\eta_{\mathscr A}/R_{\mathscr A}\lesssim1/(\mu_{\mathscr A}^{\min}\bar\lambda_{\mathscr X}^{\max})$. Since the theorem's sample size condition implies
		\[
			n \gtrsim \frac{\bar\sigma_{\mathscr X}^{4}\bar h_{\mathscr B}^{2} \vee \bar\sigma_{\mathscr X}^{2}\bar\sigma_{\mathscr E}^{2}\bar\lambda_{\mathscr E}^{\max}/\bar\lambda_{\mathscr X}^{\max}}{(\mu_{\mathscr A}^{\min})^2}df_{\bbm r}.
		\]
		Consequently, we have
		\begin{align}\label{eq:fed-rem23-small}
			(\sqrt{d+m}+1)\mathsf{Rem}_{\mathscr A}^{\mathrm{fed},(2)} \leq \frac{1}{12}R_{\mathscr A} 
			\ \text{and} \ 
			(\sqrt{d+m}+1)\mathsf{Rem}_{\mathscr A}^{\mathrm{fed},(3)} \leq \frac{1}{12}R_{\mathscr A}.
		\end{align}
		Moreover, the same sample size condition in Theorem~\ref{thm:federated_representation_error}:
		\[
			n \gtrsim \frac{\Gamma_{\mathscr A}^{\mathrm{fed}}}{\mu_{\mathscr A}^{\min}\bar\lambda_{\mathscr X}^{\max}}\frac{T_g\sqrt{\log(T_g/\delta)}}{\varepsilon}\sqrt{K(df_{\bbm r}+\log T_g)}\max_{0\leq t<T_g}\max_{k\in[K]}\tau_{\mathscr E,k}^{(t)}\tau_{\mathscr X,k}
		\]
		implies $(\sqrt{d+m}+1)\mathsf{Rem}_{\mathscr A}^{\mathrm{fed},(4)} \leq R_{\mathscr A}/12$. Together with \eqref{eq:fed-rem1-small} and \eqref{eq:fed-rem23-small}, we have
		\[
			(\sqrt{d+m}+1)\sum_{\ell=1}^{4}\mathsf{Rem}_{\mathscr A}^{\mathrm{fed},(\ell)} \leq \frac14R_{\mathscr A}+\frac{1}{12}R_{\mathscr A}+\frac{1}{12}R_{\mathscr A}+\frac{1}{12}R_{\mathscr A} = \frac12 R_{\mathscr A}.
		\]
		Combining \eqref{eq:federated-clean-recursion}, the preceding remainder bound, and the induction hypothesis gives $\|\cm A_0^{(t+1)}-\cm A_0^*\|_\F\leq R_{\mathscr A}$. This concludes the induction step. By the initialization condition in the theorem, $\|\cm A_0^{(0)}-\cm A_0^*\|_\F\leq R_{\mathscr A}$. Therefore, by induction, all iterates remain in the local neighborhood $\|\cm A_0^{(t)}-\cm A_0^*\|_\F\leq R_{\mathscr A}$ for $0\leq t\leq T_g$.

		Iterating the recursion \eqref{eq:federated-clean-recursion}, we have
		\[
			\|\widehat{\cm A}_0-\cm A_0^*\|_\F = \|\cm A_0^{(T_g)}-\cm A_0^*\|_\F \leq 2^{-T_g}\|\cm A_0^{(0)}-\cm A_0^*\|_\F + 2(\sqrt{d+m}+1)\sum_{\ell=1}^{4}\mathsf{Rem}_{\mathscr A}^{\mathrm{fed},(\ell)}.
		\]
		Equivalently, substituting the four remainder terms gives
		\begin{align*}
			\|\widehat{\cm A}_0-\cm A_0^*\|_\F &\leq 2^{-T_g}\|\cm A_0^{(0)}-\cm A_0^*\|_\F + 2(\sqrt{d+m}+1)\eta_{\mathscr A}\sqrt{r_{\bbm q}+r_{\bbm p}}\,\bar\lambda_{\mathscr X}^{\max}\zeta \\
			&+ C(\sqrt{d+m}+1)\eta_{\mathscr A}\bar\sigma_{\mathscr X}^2\bar\lambda_{\mathscr X}^{\max}\bar h_{\mathscr B}\sqrt{\frac{df_{\bbm r}}{n}} + C(\sqrt{d+m}+1)\eta_{\mathscr A}\bar\sigma_{\mathscr X}\bar\sigma_{\mathscr E}\sqrt{\bar\lambda_{\mathscr X}^{\max}\bar\lambda_{\mathscr E}^{\max}}\sqrt{\frac{df_{\bbm r}}{n}} \\
			&+ C(\sqrt{d+m}+1)\eta_{\mathscr A}\frac{T_g\sqrt{\log(T_g/\delta)}}{\varepsilon n}\sqrt{K(df_{\bbm r}+\log T_g)}\max_{0\leq t<T_g}\max_{k\in[K]}\tau_{\mathscr E,k}^{(t)}\tau_{\mathscr X,k}.
		\end{align*}
		Taking $T_g\asymp\log n$ makes $2^{-T_g}\|\cm A_0^{(0)}-\cm A_0^*\|_\F$ negligible. Moreover, since
		\[
			\eta_{\mathscr A} \leq \frac{1+c_{d,m}}{\bar\lambda_{\mathscr X}^{\max}+\bar\delta_{\mathscr A}} \leq \frac{1+c_{d,m}}{\bar\lambda_{\mathscr X}^{\max}},
		\]
		it follows that, on the event $\mathcal E_{\mathscr A}^{\mathrm{fed}}$,
		\begin{align*}
			\|\widehat{\cm A}_0-\cm A_0^*\|_\F \lesssim & \left\{\bar\sigma_{\mathscr X}^2\bar h_{\mathscr B} \vee \bar\sigma_{\mathscr X}\bar\sigma_{\mathscr E}\left(\frac{\bar\lambda_{\mathscr E}^{\max}}{\bar\lambda_{\mathscr X}^{\max}}\right)^{1/2}\right\}\sqrt{\frac{df_{\bbm r}}{n}} + \sqrt{r_{\bbm q}+r_{\bbm p}}\zeta \\
			&\quad + \frac{\log n\sqrt{\log(\log n/\delta)}\sqrt{K(df_{\bbm r}+\log\log n)}\max_{0\leq t<T_g}\max_{k\in[K]}\tau_{\mathscr E,k}^{(t)}\tau_{\mathscr X,k}}{\bar\lambda_{\mathscr X}^{\max}\varepsilon n}.
		\end{align*}
		Since $T_g\asymp\log n$ and $df_{2\bbm r}\asymp df_{\bbm r}$, the event $\mathcal E_{\mathscr A}^{\mathrm{fed}}$ satisfies $\mathbb P(\mathcal E_{\mathscr A}^{\mathrm{fed}}) \geq 1 - C\exp(-Cdf_{2\bbm r}) - C\exp(-C(df_{\bbm r}+\log \log n)) - C\exp(-C\log n)$. The claimed high probability bound follows, completing the proof.
	\end{proof}

	\begin{proof}[\textbf{Proof of Theorem~\ref{thm:personalized_error}}]
		We prove the result for a fixed client $k$. Let $\overline{\cm A}_0$ denote a generic Stage-I estimator, representing either the federated estimator $\widehat{\cm A}_0$ or the single-client estimator $\widehat{\cm A}_{0,k}^{\mathrm{loc}}$. Given $\overline{\cm A}_0$, let $\overline{\cm B}_k^{\mathrm{opt}}(\overline{\cm A}_0)$ denote an exact minimizer of the Stage-II penalized problem in \eqref{eq:stage2-opt}, with $\widehat{\cm A}_0$ therein replaced by $\overline{\cm A}_0$, and let $\overline{\cm B}_k(\overline{\cm A}_0)$ denote the corresponding estimator returned by the FISTA algorithm after $T_l^{(k)}$ iterations. For notational simplicity, we suppress their dependence on $\overline{\cm A}_0$ throughout the proof and write $\overline{\cm B}_k^{\mathrm{opt}}$ and $\overline{\cm B}_k$, respectively. Define $\cm H_k=\overline{\cm B}_k^{\mathrm{opt}}-\cm B_k^*$. In Part~I, we establish the statistical error bound under weak sparsity $\nu\in(0,1)$ and then derive the sharper exact-sparse rate for $\nu=0$. In Part~II, we show that the FISTA optimization error $\|\overline{\cm B}_k-\overline{\cm B}_k^{\mathrm{opt}}\|_{\F}$ is negligible under the stated choice of $T_l^{(k)}$.

		\noindent
		\textbf{\emph{Part I: statistical error.}}
		We first treat the weakly sparse regime $\nu\in(0,1)$. The strictly sparse regime $\nu=0$ is handled afterward by replacing the entrywise deviation bound with the refined sparse deviation bound in Lemma~\ref{lem:deviation-condition-sparse}. By the optimality of $\overline{\cm B}_k^{\mathrm{opt}}$, and using the model $\cm Y_{k,i}=\langle\cm A_0^*+\cm B_k^*,\cm X_{k,i}\rangle+\cm E_{k,i}$, we have
		\[
			\frac{1}{2n_k}\sum_{i=1}^{n_k}\|\langle\cm H_k,\cm X_{k,i}\rangle\|_\F^2 \leq \omega_k\bigl(\|\cm B_k^*\|_1-\|\overline{\cm B}_k^{\mathrm{opt}}\|_1\bigr) + \frac{1}{n_k}\sum_{i=1}^{n_k}\left\langle\langle\cm H_k,\cm X_{k,i}\rangle,\cm E_{k,i}\right\rangle + \frac{1}{n_k}\sum_{i=1}^{n_k}\left\langle\langle\cm H_k,\cm X_{k,i}\rangle,\langle\overline{\cm A}_0-\cm A_0^*,\cm X_{k,i}\rangle\right\rangle.
		\]
		For the first empirical process term, H\"older's inequality gives
		\[
			\frac{1}{n_k}\sum_{i=1}^{n_k}\left\langle\langle\cm H_k,\cm X_{k,i}\rangle,\cm E_{k,i}\right\rangle \leq \|\cm H_k\|_1\left\|\frac{1}{n_k}\sum_{i=1}^{n_k}\cm E_{k,i}\circ\cm X_{k,i}\right\|_\infty.
		\]

		For the weakly sparse regime $\nu\in(0,1)$, Assumptions~\ref{assump:subg-design}--\ref{assump:subg-noise} and Lemma~\ref{lem:deviation-condition-infty} imply that, for the choice $\omega_k \asymp \vartheta_k\sqrt{\log(pq)/n_k}$ with $\vartheta_k = \sigma_{\mathscr X,k}\sigma_{\mathscr E,k}(\lambda_{\mathscr X,k}^{\max}\lambda_{\mathscr E,k}^{\max})^{1/2}$, with probability at least $1-C\exp(-C\log(pq))$,
		$\|n_k^{-1}\sum_{i=1}^{n_k}\cm E_{k,i}\circ\cm X_{k,i}\|_\infty \leq \omega_k/2$. Consequently, on this event,
		\[
			\frac{1}{n_k}\sum_{i=1}^{n_k}\left\langle\langle\cm H_k,\cm X_{k,i}\rangle,\cm E_{k,i}\right\rangle \leq \frac{1}{2}\omega_k\|\cm H_k\|_1.
		\]
		For the second inner-product term, let $\bm S_{X,k}=n_k^{-1}\sum_{i=1}^{n_k}\vect(\cm X_{k,i})\vect^\top(\cm X_{k,i})$. On the covariance event used in Lemma~\ref{lem:rsc-condition}, $\|\bm S_{X,k}-\bbm\Sigma_{\mathscr X,k}\|_{\op}\leq\lambda_{\mathscr X,k}^{\min}/2$, and hence $\|\bm S_{X,k}\|_{\op}\leq\lambda_{\mathscr X,k}^{\max}+\lambda_{\mathscr X,k}^{\min}/2\leq2\lambda_{\mathscr X,k}^{\max}$. Therefore,
		\[
			\left|\frac{1}{n_k}\sum_{i=1}^{n_k}\left\langle\langle\cm H_k,\cm X_{k,i}\rangle,\langle\overline{\cm A}_0-\cm A_0^*,\cm X_{k,i}\rangle\right\rangle\right|=\left|\tr\left((\cm H_k)_{[S_X]}\bm S_{X,k}(\overline{\cm A}_0-\cm A_0^*)_{[S_X]}^\top\right)\right|\leq 2\lambda_{\mathscr X,k}^{\max}\|\overline{\cm A}_0-\cm A_0^*\|_\F\|\cm H_k\|_\F.
		\]
		Combining the preceding displays gives
		\begin{align}\label{eq:upper-induction-weakly-sparse}
			\frac{1}{2n_k}\sum_{i=1}^{n_k}\|\langle\cm H_k,\cm X_{k,i}\rangle\|_\F^2 \leq \omega_k\bigl(\|\cm B_k^*\|_1-\|\overline{\cm B}_k^{\mathrm{opt}}\|_1\bigr) + \frac{1}{2}\omega_k\|\cm H_k\|_1 + 2\lambda_{\mathscr X,k}^{\max}\|\overline{\cm A}_0-\cm A_0^*\|_\F\|\cm H_k\|_\F.
		\end{align}

		We first define the coordinate subspace induced by the large entries of $\cm B_k^*$. Let $\mathcal I = [q_1]\times\cdots\times[q_m]\times[p_1]\times\cdots\times[p_d]$ be the index set of the ambient tensor space. For a threshold $\kappa>0$, define $\mathcal J_\kappa=\left\{\boldsymbol i\in\mathcal I:\left|(\cm B_k^*)_{\boldsymbol i}\right|\geq\kappa\right\}$. The coordinate subspace supported on $\mathcal J_\kappa$ is
		\begin{align}\label{eq:one-stage-S-space-cm-B}
			\mathcal S_\kappa = \left\{\cm S\in\mathbb R^{q_1\times\cdots\times q_m\times p_1\times\cdots\times p_d}:\cm S_{\boldsymbol i}=0\ \text{for all }\boldsymbol i\notin\mathcal J_\kappa\right\}.
		\end{align}
		Its closed orthogonal complement is 
		\begin{align}\label{eq:one-stage-S-complement-space-cm-B}
			\overline{\mathcal S}_\kappa^\perp = \left\{\cm S\in \mathbb R^{q_1\times\cdots\times q_m\times p_1\times\cdots\times p_d}:\cm S_{\boldsymbol i}=0 \ \text{for all }\boldsymbol i\in\mathcal J_\kappa\right\}.
		\end{align}
		For any tensor $\cm T$, let $\cm T_{\mathcal S_\kappa}$ and $\cm T_{\mathcal S_\kappa^\perp}$ denote its coordinate projections onto $\mathcal S_\kappa$ and $\mathcal S_\kappa^\perp$, respectively. Using decomposability of the $\ell_1$-norm with respect to $(\mathcal S_\kappa,\overline{\mathcal S}_\kappa^\perp)$, we have
		\begin{align}\label{eq:l1-decomp-stage2-tensor}
			\|\overline{\cm B}_k^{\mathrm{opt}}\|_1-\|\cm B_k^*\|_1 &= \|\cm H_k+\cm B_k^*\|_1-\|\cm B_k^*\|_1 \notag\\
			&= \|(\cm H_k)_{\mathcal S_\kappa}+(\cm H_k)_{\overline{\mathcal S}_\kappa^\perp} + (\cm B_k^*)_{\mathcal S_\kappa}+(\cm B_k^*)_{\overline{\mathcal S}_\kappa^\perp}\|_1 - \|(\cm B_k^*)_{\mathcal S_\kappa}+(\cm B_k^*)_{\overline{\mathcal S}_\kappa^\perp}\|_1 \notag\\
			&\geq \|(\cm H_k)_{\overline{\mathcal S}_\kappa^\perp}+(\cm B_k^*)_{\mathcal S_\kappa}\|_1 - \|(\cm H_k)_{\mathcal S_\kappa}+(\cm B_k^*)_{\overline{\mathcal S}_\kappa^\perp}\|_1 - \|(\cm B_k^*)_{\mathcal S_\kappa}\|_1 - \|(\cm B_k^*)_{\overline{\mathcal S}_\kappa^\perp}\|_1\notag \\
			&= \|(\cm H_k)_{\overline{\mathcal S}_\kappa^\perp}\|_1 - \|(\cm H_k)_{\mathcal S_\kappa}\|_1 - 2\|(\cm B_k^*)_{\overline{\mathcal S}_\kappa^\perp}\|_1.
		\end{align}
		Moreover, $\|\cm H_k\|_1=\|(\cm H_k)_{\mathcal S_\kappa}\|_1+\|(\cm H_k)_{\overline{\mathcal S}_\kappa^\perp}\|_1$. Substituting this identity and \eqref{eq:l1-decomp-stage2-tensor} into the preceding basic inequality gives
		\[
			0 \leq \frac{1}{2n_k}\sum_{i=1}^{n_k}\|\langle\cm H_k,\cm X_{k,i}\rangle\|_\F^2 \leq \frac{3}{2}\omega_k\|(\cm H_k)_{\mathcal S_\kappa}\|_1 + 2\omega_k\|(\cm B_k^*)_{\overline{\mathcal S}_\kappa^\perp}\|_1 - \frac{1}{2}\omega_k\|(\cm H_k)_{\overline{\mathcal S}_\kappa^\perp}\|_1 + 2\lambda_{\mathscr X,k}^{\max}\|\overline{\cm A}_0-\cm A_0^*\|_\F\|\cm H_k\|_\F.
		\]
		Since the quadratic loss term is nonnegative, moving $-\omega_k\|(\cm H_k)_{\overline{\mathcal S}_\kappa^\perp}\|_1/2$ to the left-hand side and multiplying both sides by $2/\omega_k$ yields the cone-type inequality
		\begin{align}\label{eq:cone-stage2-tensor}
			\|(\cm H_k)_{\overline{\mathcal S}_\kappa^\perp}\|_1 \leq 3\|(\cm H_k)_{\mathcal S_\kappa}\|_1 + 4\|(\cm B_k^*)_{\overline{\mathcal S}_\kappa^\perp}\|_1 + \frac{4\lambda_{\mathscr X,k}^{\max}}{\omega_k}\|\overline{\cm A}_0-\cm A_0^*\|_\F\|\cm H_k\|_\F.
		\end{align}
		Consequently, plugging \eqref{eq:cone-stage2-tensor} back into the preceding basic inequality gives
		\begin{align*}
			\frac{1}{2n_k}\sum_{i=1}^{n_k}\|\langle\cm H_k,\cm X_{k,i}\rangle\|_\F^2 &\leq \omega_k\bigl(\|\cm B_k^*\|_1-\|\overline{\cm B}_k^{\mathrm{opt}}\|_1\bigr) + \frac{1}{2}\omega_k\|\cm H_k\|_1 + 2\lambda_{\mathscr X,k}^{\max}\|\overline{\cm A}_0-\cm A_0^*\|_\F\|\cm H_k\|_\F \\
			&\leq \frac{3}{2}\omega_k\|\cm H_k\|_1 + 2\lambda_{\mathscr X,k}^{\max}\|\overline{\cm A}_0-\cm A_0^*\|_\F\|\cm H_k\|_\F \\
			&= \frac{3}{2}\omega_k\left(\|(\cm H_k)_{\mathcal S_\kappa}\|_1 + \|(\cm H_k)_{\overline{\mathcal S}_\kappa^\perp}\|_1\right) + 2\lambda_{\mathscr X,k}^{\max}\|\overline{\cm A}_0-\cm A_0^*\|_\F\|\cm H_k\|_\F \\
			&\leq \frac{3}{2}\omega_k\left(\|(\cm H_k)_{\mathcal S_\kappa}\|_1 + 3\|(\cm H_k)_{\mathcal S_\kappa}\|_1 + 4\|(\cm B_k^*)_{\overline{\mathcal S}_\kappa^\perp}\|_1 + \frac{4\lambda_{\mathscr X,k}^{\max}}{\omega_k}\|\overline{\cm A}_0-\cm A_0^*\|_\F\|\cm H_k\|_\F\right) \\
			&\quad+2\lambda_{\mathscr X,k}^{\max}\|\overline{\cm A}_0-\cm A_0^*\|_\F\|\cm H_k\|_\F \\
			&= 6\omega_k\left(\|(\cm H_k)_{\mathcal S_\kappa}\|_1 + \|(\cm B_k^*)_{\overline{\mathcal S}_\kappa^\perp}\|_1\right) + 8\lambda_{\mathscr X,k}^{\max}\|\overline{\cm A}_0-\cm A_0^*\|_\F\|\cm H_k\|_\F.
		\end{align*}

		Under the sample size condition, Lemma~\ref{lem:rsc-condition} gives, with probability at least $1-\exp(-Cp)$,
		\[
			\frac{1}{n_k}\sum_{i=1}^{n_k}\|\langle\cm H_k,\cm X_{k,i}\rangle\|_\F^2 \geq \frac12\lambda_{\mathscr X,k}^{\min}\|\cm H_k\|_\F^2.
		\]
		Combining the preceding two displays yields
		\[
			\lambda_{\mathscr X,k}^{\min}\|\cm H_k\|_\F^2 \lesssim \omega_k\|(\cm H_k)_{\mathcal S_{\kappa}}\|_1 + \omega_k\|(\cm B_k^*)_{\overline{\mathcal S}_\kappa^\perp}\|_1 + \lambda_{\mathscr X,k}^{\max}\|\overline{\cm A}_0-\cm A_0^*\|_\F\|\cm H_k\|_\F.
		\]
		Moreover, since $\cm B_k^*\in\mathbb B_\nu$ with $\mathbb B_\nu = \left\{\cm B \in \mathbb{R}^{q_1\times\cdots\times q_m\times p_1\times\cdots\times p_d}: \mathds 1_{\{\nu=0\}}\|\cm B\|_0/s_0 + \mathds 1_{\{0<\nu<1\}}\|\cm B\|_\nu^\nu/s_\nu\leq 1\right\}$, the set $\mathcal S_\kappa$ satisfies $s_\nu\geq |\mathcal S_\kappa|\kappa^\nu$, and hence $|\mathcal S_\kappa|\leq s_\nu\kappa^{-\nu}$. Therefore, $\|(\cm H_k)_{\mathcal S_\kappa}\|_1 \leq \sqrt{|S_\kappa|}\,\|(\cm H_k)_{\mathcal S_\kappa}\|_\F \leq \sqrt{s_\nu}\kappa^{-\nu/2}\|\cm H_k\|_\F$. For the weak-sparsity tail, by the definition of $\mathbb B_\nu$,
		\[
			\|(\cm B_k^*)_{\overline{\mathcal S}_\kappa^\perp}\|_1 =\sum_{\bbm i\in \overline{\mathcal S}_\kappa^\perp}|(\cm B_k^*)_{\bbm i}| \leq \sum_{\bbm i\in \overline{\mathcal S}_\kappa^\perp}|(\cm B_k^*)_{\bbm i}|^\nu\kappa^{1-\nu} \leq s_\nu\kappa^{1-\nu}.
		\]
		Therefore,
		\begin{align}\label{eq:lam-H-upper}
			\lambda_{\mathscr X,k}^{\min}\|\cm H_k\|_\F^2 \lesssim \omega_k\sqrt{s_\nu}\kappa^{-\nu/2}\|\cm H_k\|_\F + \omega_k s_\nu\kappa^{1-\nu} + \lambda_{\mathscr X,k}^{\max}\|\overline{\cm A}_0-\cm A_0^*\|_\F\|\cm H_k\|_\F.
		\end{align}
		By Young's inequality, we have
		\[
			\omega_k\sqrt{s_\nu}\kappa^{-\nu/2}\|\cm H_k\|_\F \leq \frac{1}{4}\lambda_{\mathscr X,k}^{\min}\|\cm H_k\|_\F^2 + \frac{\omega_k^2s_\nu\kappa^{-\nu}}{\lambda_{\mathscr X,k}^{\min}}, 
			\ \text{and} \ 
			\lambda_{\mathscr X,k}^{\max}\|\overline{\cm A}_0-\cm A_0^*\|_\F\|\cm H_k\|_\F \leq \frac{1}{4}\lambda_{\mathscr X,k}^{\min}\|\cm H_k\|_\F^2 + \frac{(\lambda_{\mathscr X,k}^{\max})^2}{\lambda_{\mathscr X,k}^{\min}}\|\overline{\cm A}_0-\cm A_0^*\|_\F^2.
		\]
		Substituting these two bounds into \eqref{eq:lam-H-upper} gives
		\[
			\|\cm H_k\|_\F^2 \lesssim s_\nu\left(\frac{\omega_k}{\lambda_{\mathscr X,k}^{\min}}\right)^2\kappa^{-\nu} + \frac{\omega_k}{\lambda_{\mathscr X,k}^{\min}}s_\nu\kappa^{1-\nu} + \left(\frac{\lambda_{\mathscr X,k}^{\max}}{\lambda_{\mathscr X,k}^{\min}}\right)^2\|\overline{\cm A}_0-\cm A_0^*\|_\F^2.
		\]
		Choosing $\kappa\asymp\omega_k/\lambda_{\mathscr X,k}^{\min}$ balances the two weak-sparsity terms. Taking square roots of the resulting squared-error bound and using $\sqrt{a+b}\leq\sqrt a+\sqrt b$, and abbreviating $\kappa_{\mathscr X,k}=\lambda_{\mathscr X,k}^{\max}/\lambda_{\mathscr X,k}^{\min}$, we obtain the weakly sparse statistical bound, for $\nu\in(0,1)$,
		\[
			\|\overline{\cm B}_k^{\mathrm{opt}}-\cm B_k^*\|_\F\lesssim \kappa_{\mathscr X,k}\|\overline{\cm A}_0-\cm A_0^*\|_\F + \sqrt{s_\nu}\left(\frac{\omega_k}{\lambda_{\mathscr X,k}^{\min}}\right)^{1-\nu/2}.
		\]

		We next treat the strictly sparse regime $\nu=0$. In this case, we replace the entrywise deviation bound in Lemma~\ref{lem:deviation-condition-infty} by the refined sparse deviation bound in Lemma~\ref{lem:deviation-condition-sparse}. By Lemma~\ref{lem:deviation-condition-sparse}, $|n_k^{-1}\sum_{i=1}^{n_k}\left\langle\langle\cm H_k,\cm X_{k,i}\rangle,\cm E_{k,i}\right\rangle|$is treated in two cases, determined by the relationship between $\|\cm H_k\|_1$ and $\sqrt{s_0}\|\cm H_k\|_\F$. We consider these two cases in turn. First, if $\|\cm H_k\|_1 > \sqrt{s_0}\|\cm H_k\|_\F$, then the same lemma gives $|n_k^{-1}\sum_{i=1}^{n_k}\left\langle\langle\cm H_k,\cm X_{k,i}\rangle,\cm E_{k,i}\right\rangle| \leq \omega_k\|\cm H_k\|_1/2$. In this case, the stochastic term is absorbed by the $\ell_1$-penalty, and the proof follows the same decomposability arguments as in the weakly sparse regime.

		Second, if $\|\cm H_k\|_1 \leq \sqrt{s_0}\|\cm H_k\|_\F$, then Lemma~\ref{lem:deviation-condition-sparse} gives $|n_k^{-1}\sum_{i=1}^{n_k}\left\langle\langle\cm H_k,\cm X_{k,i}\rangle,\cm E_{k,i}\right\rangle| \leq C_0\sqrt{s_0}\omega_k\|\cm H_k\|_\F$. Using the same basic inequality as in \eqref{eq:upper-induction-weakly-sparse}, but replacing the stochastic term by the preceding refined bound, we have
		\begin{align}\label{eq:lower-upper-H-X-quadratic}
			\frac{1}{2n_k}\sum_{i=1}^{n_k}\|\langle\cm H_k,\cm X_{k,i}\rangle\|_\F^2 \leq \omega_k\bigl(\|\cm B_k^*\|_1-\|\overline{\cm B}_k^{\mathrm{opt}}\|_1\bigr) + C_0\sqrt{s_0}\omega_k\|\cm H_k\|_\F + 2\lambda_{\mathscr X,k}^{\max}\|\overline{\cm A}_0-\cm A_0^*\|_\F\|\cm H_k\|_\F.
		\end{align}
		For the $\ell_1$-norm term in the strictly sparse regime, redefine $\mathcal S=\operatorname{supp}(\cm B_k^*)$, so that $\|(\cm B_k^*)_{\overline{\mathcal S}^\perp}\|_1=0$ and $|\mathcal S|\leq s_0$. By the same decomposability arguments as in \eqref{eq:l1-decomp-stage2-tensor}, we have $\|\overline{\cm B}_k^{\mathrm{opt}}\|_1-\|\cm B_k^*\|_1 \geq \|(\cm H_k)_{\overline{\mathcal S}^\perp}\|_1 - \|(\cm H_k)_{\mathcal S}\|_1$. Consequently, substituting this identity into \eqref{eq:lower-upper-H-X-quadratic} and rearranging the terms yields the cone-type inequality
		\begin{align}\label{eq:cone-stage2-tensor-strictly-sparse}
			\|(\cm H_k)_{\overline{\mathcal S}^\perp}\|_1 \leq \|(\cm H_k)_{\mathcal S}\|_1 + C_0\sqrt{s_0}\|\cm H_k\|_\F + \frac{2\lambda_{\mathscr X,k}^{\max}}{\omega_k}\|\overline{\cm A}_0-\cm A_0^*\|_\F\|\cm H_k\|_\F.
		\end{align}
		Consequently, plugging \eqref{eq:cone-stage2-tensor-strictly-sparse} back into \eqref{eq:lower-upper-H-X-quadratic} gives
		\begin{align*}
		\frac{1}{2n_k}\sum_{i=1}^{n_k}\|\langle\cm H_k,\cm X_{k,i}\rangle\|_\F^2 &\leq \omega_k\bigl(\|\cm B_k^*\|_1-\|\overline{\cm B}_k^{\mathrm{opt}}\|_1\bigr) + C_0\sqrt{s_0}\omega_k\|\cm H_k\|_\F + 2\lambda_{\mathscr X,k}^{\max}\|\overline{\cm A}_0-\cm A_0^*\|_\F\|\cm H_k\|_\F \\
		&\leq \omega_k\left(\|(\cm H_k)_{\mathcal S}\|_1 + \|(\cm H_k)_{\overline{\mathcal S}^\perp}\|_1\right) + C_0\sqrt{s_0}\omega_k\|\cm H_k\|_\F + 2\lambda_{\mathscr X,k}^{\max}\|\overline{\cm A}_0-\cm A_0^*\|_\F\|\cm H_k\|_\F \\
		&\leq \omega_k\left(2\|(\cm H_k)_{\mathcal S}\|_1+C_0\sqrt{s_0}\|\cm H_k\|_\F + \frac{2\lambda_{\mathscr X,k}^{\max}}{\omega_k}\|\overline{\cm A}_0-\cm A_0^*\|_\F\|\cm H_k\|_\F\right) \\
		&\quad+ C_0\sqrt{s_0}\omega_k\|\cm H_k\|_\F + 2\lambda_{\mathscr X,k}^{\max}\|\overline{\cm A}_0-\cm A_0^*\|_\F\|\cm H_k\|_\F \\
		&= 2\omega_k\|(\cm H_k)_{\mathcal S}\|_1 + 2C_0\sqrt{s_0}\omega_k\|\cm H_k\|_\F + 4\lambda_{\mathscr X,k}^{\max}\|\overline{\cm A}_0-\cm A_0^*\|_\F\|\cm H_k\|_\F \\
		&\lesssim \sqrt{s_0}\omega_k\|\cm H_k\|_\F + \lambda_{\mathscr X,k}^{\max}\|\overline{\cm A}_0-\cm A_0^*\|_\F\|\cm H_k\|_\F,
		\end{align*}
		where the last inequality uses $\|(\cm H_k)_{\mathcal S}\|_1\leq\sqrt{s_0}\|\cm H_k\|_\F$ in the exact sparse case. Under the exact sparsity tuning $\omega_k\asymp\vartheta_k\sqrt{\log(pq/s_0)/n_k}$, the preceding upper bound becomes
		\begin{align}\label{eq:exact-sparse-upper-bound-before-rsc}
			\frac{1}{2n_k}\sum_{i=1}^{n_k}\|\langle\cm H_k,\cm X_{k,i}\rangle\|_\F^2 \lesssim \vartheta_k\sqrt{\frac{s_0\log(pq/s_0)}{n_k}}\|\cm H_k\|_\F + \lambda_{\mathscr X,k}^{\max}\|\overline{\cm A}_0-\cm A_0^*\|_\F\|\cm H_k\|_\F.
		\end{align}
		Combining \eqref{eq:exact-sparse-upper-bound-before-rsc} with the RSC lower bound in Lemma~\ref{lem:rsc-condition}, we have
		\begin{align}\label{eq:lam-H-upper-refined}
			\lambda_{\mathscr X,k}^{\min}\|\cm H_k\|_\F^2 \lesssim \vartheta_k\sqrt{\frac{s_0\log(pq/s_0)}{n_k}}\|\cm H_k\|_\F + \lambda_{\mathscr X,k}^{\max}\|\overline{\cm A}_0-\cm A_0^*\|_\F\|\cm H_k\|_\F.
		\end{align}
		By Young's inequality, we have
		\begin{align*}
			&\vartheta_k\sqrt{\frac{s_0\log(pq/s_0)}{n_k}}\|\cm H_k\|_\F \leq \frac{1}{4}\lambda_{\mathscr X,k}^{\min}\|\cm H_k\|_\F^2 + C\frac{\vartheta_k^2}{\lambda_{\mathscr X,k}^{\min}}\frac{s_0\log(pq/s_0)}{n_k}.\\
			&\lambda_{\mathscr X,k}^{\max}\|\overline{\cm A}_0-\cm A_0^*\|_\F\|\cm H_k\|_\F \leq \frac{1}{4}\lambda_{\mathscr X,k}^{\min}\|\cm H_k\|_\F^2 + C\frac{(\lambda_{\mathscr X,k}^{\max})^2}{\lambda_{\mathscr X,k}^{\min}}\|\overline{\cm A}_0-\cm A_0^*\|_\F^2.
		\end{align*}

		Substituting these two inequalities into \eqref{eq:lam-H-upper-refined}, taking square roots and using $\sqrt{a+b}\leq\sqrt a+\sqrt b$, we have
		\[
			\|\overline{\cm B}_k^{\mathrm{opt}}-\cm B_k^*\|_\F \lesssim \kappa_{\mathscr X,k}\|\overline{\cm A}_0-\cm A_0^*\|_\F + \frac{\vartheta_k}{\lambda_{\mathscr X,k}^{\min}}\sqrt{\frac{s_0\log(pq/s_0)}{n_k}}.
		\]
		This gives the strictly sparse statistical bound. Combining this with the weakly sparse bound derived above yields the compact Part~I conclusion
		\[
			\|\overline{\cm B}_k^{\mathrm{opt}}-\cm B_k^*\|_\F \lesssim \mathsf{Error}_{\mathscr B}^{(k)}(\overline{\cm A}_0).
		\]

		\noindent
		\textbf{\emph{Part II: optimization error.}}
		We now control the optimization error $\|\overline{\cm B}_k-\overline{\cm B}_k^{\mathrm{opt}}\|_\F$ using the convergence guarantee of FISTA. Let $\bbm y_{k,i}=\vect(\cm Y_{k,i})$, $\bbm x_{k,i}=\vect(\cm X_{k,i})$, $\bm Y_k=(\bbm y_{k,1},\cdots,\bbm y_{k,n_k})^\top\in\mathbb R^{n_k\times q}$, and $\bm X_k=(\bbm x_{k,1},\cdots,\bbm x_{k,n_k})^\top\in\mathbb R^{n_k\times p}$. Write $\bar{\bm A}_{0}=(\overline{\cm A}_0)_{[S_X]}$, $\bm B=(\cm B)_{[S_X]}$, $\overline{\bm B}_k=(\overline{\cm B}_k)_{[S_X]}$ and $\overline{\bm B}_k^{\mathrm{opt}}=(\overline{\cm B}_k^{\mathrm{opt}})_{[S_X]}$. Under the $S_X$-matricization, the Stage-II objective is equivalent to $F_k(\bm B)=(2n_k)^{-1}\|\bm Y_k-\bm X_k(\bar{\bm A}_{0}+\bm B)^\top\|_\F^2+\omega_k\|\bm B\|_1$. Denote the smooth and nonsmooth parts by $f_k(\bm B)=(2n_k)^{-1}\|\bm Y_k-\bm X_k(\bar{\bm A}_{0}+\bm B)^\top\|_\F^2$ and $g_k(\bm B)=\omega_k\|\bm B\|_1$, respectively. We first verify the Lipschitz continuity of $\nabla f_k$. A direct calculation gives $\nabla f_k(\bm B)=n_k^{-1}(\bm B+\bar{\bm A}_{0})\bm X_k^\top\bm X_k-n_k^{-1}\bm Y_k^\top\bm X_k$. Thus, for any $\bm B_1,\bm B_2\in\mathbb R^{q\times p}$, $\|\nabla f_k(\bm B_1)-\nabla f_k(\bm B_2)\|_\F \leq \|n_k^{-1}\bm X_k^\top\bm X_k\|_{\op}\|\bm B_1-\bm B_2\|_\F$. Therefore, the gradient of $f_k$ is Lipschitz continuous with constant no larger than $L_k=2\|n_k^{-1}\bm X_k^\top\bm X_k\|_{\op}$, which is the constant used in the stepsize condition of the theorem.

		Although the original FISTA analysis in \citet{beck2009fast} is stated for vector-valued variables, Remark~2.1 therein shows that the arguments extends verbatim to any finite-dimensional Hilbert space. Since $\mathbb R^{q\times p}$ equipped with the Frobenius inner product is such a space, and since $f_k$ is convex and differentiable with an $L_k$-Lipschitz gradient while $g_k$ is convex but nonsmooth, the standard FISTA guarantee applies. Hence, for the iterates generated with stepsize $\eta_{\mathscr B,k}\leq1/L_k$, \citet[Theorem~4.4]{beck2009fast} implies that, for $\overline{\bm B}_k=\bm B_k^{(T_l^{(k)})}$, $F_k(\overline{\bm B}_k)-F_k(\overline{\bm B}_k^{\mathrm{opt}}) \leq 2L_k\|\bm B_k^{(0)}-\overline{\bm B}_k^{\mathrm{opt}}\|_\F^2/(T_l^{(k)}+1)^2$. Under the initialization $\bm B_k^{(0)}=\bm0$, this becomes
		\begin{align}\label{eq:fista-rate-stage2-B}
			F_k(\overline{\bm B}_k)-F_k(\overline{\bm B}_k^{\mathrm{opt}}) \leq \frac{2L_k\|\overline{\bm B}_k^{\mathrm{opt}}\|_\F^2}{(T_l^{(k)}+1)^2}.
		\end{align}

		Next, we relate this objective gap to the Frobenius optimization error. For any $\bm B\in\mathbb R^{q\times p}$, the quadratic structure of $f_k$ gives the exact expansion
		\[
			f_k(\bm B) = f_k(\overline{\bm B}_k^{\mathrm{opt}}) + \left\langle\nabla f_k(\overline{\bm B}_k^{\mathrm{opt}}),\bm B-\overline{\bm B}_k^{\mathrm{opt}}\right\rangle + \frac{1}{2n_k}\left\|\bm X_k(\bm B-\overline{\bm B}_k^{\mathrm{opt}})^\top\right\|_\F^2.
		\]
		Moreover, by the subgradient inequality for the entrywise $\ell_1$ norm, for any $\bm Z_k^{\mathrm{opt}}\in\partial\|\overline{\bm B}_k^{\mathrm{opt}}\|_1$, $g_k(\bm B) \geq g_k(\overline{\bm B}_k^{\mathrm{opt}}) + \langle\omega_k\bm Z_k^{\mathrm{opt}},\bm B-\overline{\bm B}_k^{\mathrm{opt}}\rangle$.
		Since $\overline{\bm B}_k^{\mathrm{opt}}$ minimizes the convex function $F_k=f_k+g_k$, there exists $\bm Z_k^{\mathrm{opt}}\in\partial\|\overline{\bm B}_k^{\mathrm{opt}}\|_1$ such that $\nabla f_k(\overline{\bm B}_k^{\mathrm{opt}}) + \omega_k\bm Z_k^{\mathrm{opt}} = \bm0$.
		Combining the preceding displays therefore yields, for any $\bm B\in\mathbb R^{q\times p}$,
		\begin{align}\label{eq:stage2-objective-gap-qg}
			\frac{1}{2n_k}\left\|\bm X_k(\bm B-\overline{\bm B}_k^{\mathrm{opt}})^\top\right\|_\F^2 \leq F_k(\bm B)-F_k(\overline{\bm B}_k^{\mathrm{opt}}).
		\end{align}
		Taking $\bm B=\overline{\bm B}_k$ in \eqref{eq:stage2-objective-gap-qg}, and applying Lemma~\ref{lem:rsc-condition} to the tensor counterpart of $\overline{\bm B}_k-\overline{\bm B}_k^{\mathrm{opt}}$, we obtain, on the RSC event, $\lambda_{\mathscr X,k}^{\min}\|\overline{\bm B}_k-\overline{\bm B}_k^{\mathrm{opt}}\|_\F^2 \lesssim F_k(\overline{\bm B}_k)-F_k(\overline{\bm B}_k^{\mathrm{opt}})$. Combining this inequality with \eqref{eq:fista-rate-stage2-B}, and using $L_k\lesssim\lambda_{\mathscr X,k}^{\max}$ on the same covariance event, gives
		\begin{align}\label{eq:stage2-opt-error-gap}
			\|\overline{\bm B}_k-\overline{\bm B}_k^{\mathrm{opt}}\|_\F^2 \lesssim \kappa_{\mathscr X,k}\frac{\|\overline{\bm B}_k^{\mathrm{opt}}\|_\F^2}{(T_l^{(k)}+1)^2}.
		\end{align}
		It remains to bound $\|\overline{\bm B}_k^{\mathrm{opt}}\|_\F$. By the triangle inequality and the compact statistical bound in Part~I, $\|\overline{\bm B}_k^{\mathrm{opt}}\|_\F^2 \leq 2\|\overline{\bm B}_k^{\mathrm{opt}}-\bm B_k^*\|_\F^2 + 2\|\bm B_k^*\|_\F^2 \lesssim (\mathsf{Error}_{\mathscr B}^{(k)}(\overline{\cm A}_0))^2+h_{\mathscr B,k}^2$, where we used $\|\bm B_k^*\|_\F\leq h_{\mathscr B,k}$. Substituting this bound into \eqref{eq:stage2-opt-error-gap} gives
		\[
			\|\overline{\bm B}_k-\overline{\bm B}_k^{\mathrm{opt}}\|_\F^2 \lesssim \kappa_{\mathscr X,k}\frac{(\mathsf{Error}_{\mathscr B}^{(k)}(\overline{\cm A}_0))^2+h_{\mathscr B,k}^2}{(T_l^{(k)}+1)^2}.
		\]
		Therefore, taking $T_l^{(k)}\asymp \sqrt{\kappa_{\mathscr X,k}}(1+h_{\mathscr B,k}/\mathsf{Error}_{\mathscr B}^{(k)}(\overline{\cm A}_0))$, or equivalently $T_l^{(k)}\asymp \sqrt{\kappa_{\mathscr X,k}}(1+h_{\mathscr B,k}/\mathsf{Error}_{\mathscr B}^{(k)}(\overline{\cm A}_0))$ up to constants, the optimization error is no larger than the statistical error, namely $\|\overline{\cm B}_k-\overline{\cm B}_k^{\mathrm{opt}}\|_\F \lesssim \mathsf{Error}_{\mathscr B}^{(k)}(\overline{\cm A}_0)$. Combining this optimization-error bound with the compact statistical bound in Part~I and using the triangle inequality yields
		\[
			\|\overline{\cm B}_k-\cm B_k^*\|_\F \leq \|\overline{\cm B}_k-\overline{\cm B}_k^{\mathrm{opt}}\|_\F + \|\overline{\cm B}_k^{\mathrm{opt}}-\cm B_k^*\|_\F \lesssim \mathsf{Error}_{\mathscr B}^{(k)}(\overline{\cm A}_0).
		\]
		Together with the definition of $\mathsf{Error}_{\mathscr B}^{(k)}(\overline{\cm A}_0)$, this also gives the stated joint bound involving $\|\overline{\cm A}_0-\cm A_0^*\|_\F+\|\overline{\cm B}_k-\cm B_k^*\|_\F$. This completes the proof.
	\end{proof}

	\begin{proof}[\textbf{Proof of Theorem \ref{thm:one-stage-initialization-error}}]
		Fix a client $k\in[K]$. We first pass to the equivalent matrix representation induced by the $S_X$-matricization. Let $\bbm y_{k,i}=\vect(\cm Y_{k,i})$, $\bbm x_{k,i}=\vect(\cm X_{k,i})$, and $\bbm e_{k,i}=\vect(\cm E_{k,i})$. Define $\bm A_0^*=(\cm A_0^*)_{[S_X]}$ and $\bm B_k^*=(\cm B_k^*)_{[S_X]}$. Then, the model can be written as $\bbm y_{k,i}=(\bm A_0^*+\bm B_k^*)\bbm x_{k,i}+\bbm e_{k,i}$. Under the same matricization, the tensor optimization problem in \eqref{eq:alg-ini} is equivalent to
		\begin{align}\label{eq:ini-opt-matricization}
			(\widetilde{\bm A}_{0,k},\widetilde{\bm B}_{k}) = \argmin_{\bm A,\bm B}\Biggl\{\frac{1}{2n_k}\sum_{i=1}^{n_k}\bigl\|\bbm y_{k,i}-(\bm A+\bm B)\bbm x_{k,i}\bigr\|_2^2 + \lambda_k\|\bm A\|_* + \varpi_k\|\bm B\|_1\Biggr\},
		\end{align}
		subject to $\|\bm B\|_{\op}\leq\zeta$, where $\widetilde{\bm A}_{0,k}=(\widetilde{\cm A}_{0,k})_{[S_X]}$ and $\widetilde{\bm B}_{k}=(\widetilde{\cm B}_{k})_{[S_X]}$. 

		Therefore, it suffices to prove the desired bound for these matrix estimators. Define $\bm H_A=\widetilde{\bm A}_{0,k}-\bm A_0^*$ and $\bm H_B=\widetilde{\bm B}_{k}-\bm B_k^*$.
		Since both $(\widetilde{\bm A}_{0,k},\widetilde{\bm B}_{k})$ and $(\bm A_0^*,\bm B_k^*)$ are feasible for the optimization~\eqref{eq:ini-opt-matricization}, optimality gives
		\[
			\frac{1}{2n_k}\sum_{i=1}^{n_k}\left\|\bbm y_{k,i}-(\widetilde{\bm A}_{0,k}+\widetilde{\bm B}_{k})\bbm x_{k,i}\right\|_2^2 + \lambda_k\|\widetilde{\bm A}_{0,k}\|_* + \varpi_k\|\widetilde{\bm B}_{k}\|_1 \leq \frac{1}{2n_k}\sum_{i=1}^{n_k}\left\|\bbm y_{k,i}-(\bm A_0^*+\bm B_k^*)\bbm x_{k,i}\right\|_2^2 + \lambda_k\|\bm A_0^*\|_* + \varpi_k\|\bm B_k^*\|_1.
		\]
		Rearranging the terms yields
		\begin{align}\label{eq:one-stage-loss-difference}
			\frac{1}{2n_k}\sum_{i=1}^{n_k}\left(\left\|\bbm y_{k,i}-(\widetilde{\bm A}_{0,k}+\widetilde{\bm B}_{k})\bbm x_{k,i}\right\|_2^2 - \left\|\bbm y_{k,i}-(\bm A_0^*+\bm B_k^*)\bbm x_{k,i}\right\|_2^2\right)\leq \lambda_k\bigl(\|\bm A_0^*\|_* - \|\widetilde{\bm A}_{0,k}\|_*\bigr) + \varpi_k\left(\|\bm B_k^*\|_1-\|\widetilde{\bm B}_{k}\|_1\right).
		\end{align}
		Using $\|\bbm a\|_2^2-\|\bbm b\|_2^2=\|\bbm a-\bbm b\|_2^2+2\langle\bbm a-\bbm b,\bbm b\rangle$ with $\bbm a=\bbm y_{k,i}-(\widetilde{\bm A}_{0,k}+\widetilde{\bm B}_{k})\bbm x_{k,i}$ and $\bbm b=\bbm y_{k,i}-(\bm A_0^*+\bm B_k^*)\bbm x_{k,i}=\bbm e_{k,i}$, we have
		\begin{align}\label{eq:one-stage-loss-expanded}
			\frac{1}{2n_k}\sum_{i=1}^{n_k}\left\|(\bm H_A+\bm H_B)\bbm x_{k,i}\right\|_2^2 - \frac{1}{n_k}\sum_{i=1}^{n_k}\left\langle(\bm H_A+\bm H_B)\bbm x_{k,i},\bbm e_{k,i}\right\rangle \leq \lambda_k\bigl(\|\bm A_0^*\|_* - \|\widetilde{\bm A}_{0,k}\|_*\bigr) + \varpi_k\bigl(\|\bm B_k^*\|_1-\|\widetilde{\bm B}_{k}\|_1\bigr).
		\end{align}
		For the inner product term, we have $\left\langle(\bm H_A+\bm H_B)\bbm x_{k,i},\bbm e_{k,i}\right\rangle = \operatorname{tr}\left(\bbm e_{k,i}^{\top}(\bm H_A+\bm H_B)\bbm x_{k,i}\right) = \left\langle\bm H_A+\bm H_B, \bbm e_{k,i}\bbm x_{k,i}^{\top}\right\rangle$.
		Then, rearranging \eqref{eq:one-stage-loss-expanded} gives
		\begin{align}\label{eq:one-stage-basic-ineq}
			\frac{1}{2n_k}\sum_{i=1}^{n_k}\left\|(\bm H_A+\bm H_B)\bbm x_{k,i}\right\|_2^2 \leq \lambda_k\bigl(\|\bm A_0^*\|_* - \|\widetilde{\bm A}_{0,k}\|_*\bigr) + \varpi_k\bigl(\|\bm B_k^*\|_1-\|\widetilde{\bm B}_{k}\|_1\bigr) + \left\langle\bm H_A+\bm H_B,\frac{1}{n_k}\sum_{i=1}^{n_k}\bbm e_{k,i}\bbm x_{k,i}^\top\right\rangle.
		\end{align}

		To bound the right-hand side of \eqref{eq:one-stage-basic-ineq}, we begin by deriving the constrained space for the estimators $\widetilde{\bm A}_{0,k}$ and $\widetilde{\bm B}_{k}$. We first define the coordinate subspace induced by the large entries of $\bm B_k^*$. Specifically, for a given threshold $\kappa>0$, let
		\[
			\mathcal S_\kappa = \Bigl\{\bm S\in\mathbb R^{q\times p}: \bm S_{ab}=0\ \text{for all }(a,b)\in[q]\times[p]\ \text{such that } |(\bm B_k^*)_{ab}|\geq\kappa\Bigr\},
		\]
		and let $\overline{\mathcal S}_\kappa^\perp$ denote its closed orthogonal complement, namely
		\[
			\overline{\mathcal S}_\kappa^\perp = \Bigl\{\bm S\in\mathbb R^{q\times p}: \bm S_{ab}=0\ \text{for all }(a,b)\in[q]\times[p]\ \text{such that } |(\bm B_k^*)_{ab}|<\kappa\Bigr\}.
		\]
		For any $\bm S\in\mathbb R^{q\times p}$, we write $\bm S_{\mathcal S_\kappa}$ for the restriction of $\bm S$ to the coordinates with $|(\bm B_k^*)_{ab}|\geq\kappa$, i.e., $(\bm S_{\mathcal S_\kappa})_{ab}=0$ if $|(\bm B_k^*)_{ab}|\geq\kappa$ and $(\bm S_{\mathcal S_\kappa})_{ab}=\bm S_{ab}$ otherwise. We define $\bm S_{\overline{\mathcal S}_\kappa^\perp}$ analogously as the restriction to the coordinates with $|(\bm B_k^*)_{ab}|<\kappa$.

		We next introduce the model subspace associated with the low-rank matrix $\bm A_0^*$. Let $r_A=\rank(\bm A_0^*)\leq\min(r_{\bbm q},r_{\bbm p})$, and write a SVD as $\bm A_0^*=\bm U_A^*\bm\Sigma_A^*\bm V_A^{*\top}$, where $\bm U_A^*\in\mathbb R^{q\times r_A}$ and $\bm V_A^*\in\mathbb R^{p\times r_A}$. Define
		\[
			\mathcal L_A = \Bigl\{\bm M\in\mathbb R^{q\times p}:\operatorname{col}(\bm M)\subseteq\operatorname{col}(\bm U_A^*),\ \operatorname{col}(\bm M^\top)\subseteq\operatorname{col}(\bm V_A^*)\Bigr\},
		\]
		and let $\overline{\mathcal L}_A^\perp$ be the orthogonal complement of the completion space associated with $\mathcal L_A$, namely
		\[
			\overline{\mathcal L}_A^\perp = \Bigl\{\bm M\in\mathbb R^{q\times p}:\operatorname{col}(\bm M)\perp\operatorname{col}(\bm U_A^*),\ \operatorname{col}(\bm M^\top)\perp\operatorname{col}(\bm V_A^*)\Bigr\}.
		\]
		For any $\bm M\in\mathbb R^{q\times p}$, write $\bm M_{\overline{\mathcal L}_A}$ and $\bm M_{\overline{\mathcal L}_A^\perp}$ for its orthogonal projections onto these two spaces. Using the decomposability of the nuclear norm with respect to $(\mathcal L_A,\overline{\mathcal L}_A^\perp)$, together with the decomposability of the entrywise $\ell_1$ norm with respect to $(\mathcal S_\kappa, \overline{\mathcal S}_\kappa^\perp)$, gives
		\begin{align}\label{eq:one-stage-decomposable-norms}
			&\lambda_k\bigl(\|\bm A_0^*\|_*-\|\widetilde{\bm A}_{0,k}\|_*\bigr)+\varpi_k\bigl(\|\bm B_k^*\|_1-\|\widetilde{\bm B}_{k}\|_1\bigr)=\lambda_k\bigl(\|\bm A_0^*\|_*-\|\bm A_0^*+\bm H_A\|_*\bigr)+\varpi_k\bigl(\|\bm B_k^*\|_1-\|\bm B_k^*+\bm H_B\|_1\bigr)\notag\\
			&\leq \lambda_k\Bigl[\|\bm A_0^*\|_*-\Bigl\{\|\bm A_0^*\|_*+\bigl\|(\bm H_A)_{\overline{\mathcal L}_A^\perp}\bigr\|_*-\bigl\|(\bm H_A)_{\mathcal L_A}\bigr\|_*\Bigr\}\Bigr]
			\notag\\
			&\quad+ \varpi_k\Bigl[\|\bm B_k^*\|_1-\Bigl\{\bigl\|(\bm B_k^*)_{\mathcal S_\kappa}\bigr\|_1-\bigl\|(\bm H_B)_{\mathcal S_\kappa}\bigr\|_1+\bigl\|(\bm H_B)_{\overline{\mathcal S}_\kappa^\perp}\bigr\|_1-\bigl\|(\bm B_k^*)_{\overline{\mathcal S}_\kappa^\perp}\bigr\|_1\Bigr\}\Bigr]\notag\\
			&= \lambda_k\Bigl(\bigl\|(\bm H_A)_{\mathcal L_A}\bigr\|_*-\bigl\|(\bm H_A)_{\overline{\mathcal L}_A^\perp}\bigr\|_*\Bigr)+\varpi_k\Bigl(\bigl\|(\bm H_B)_{\mathcal S_\kappa}\bigr\|_1-\bigl\|(\bm H_B)_{\overline{\mathcal S}_\kappa^\perp}\bigr\|_1+2\bigl\|(\bm B_k^*)_{\overline{\mathcal S}_\kappa^\perp}\bigr\|_1\Bigr).
		\end{align}

		By Lemmas~\ref{lem:deviation-condition-op} and~\ref{lem:deviation-condition-infty}, together with the choices of $\lambda_k$ and $\varpi_k$, with probability at least $1-C\exp(-C(p+q))-C\exp(-C\log(pq))$,
		\begin{align}\label{eq:deviation-bounds-e-x}
			\left\|\frac{1}{n_k}\sum_{i=1}^{n_k}\bbm e_{k,i}\bbm x_{k,i}^\top\right\|_{\op} \leq \frac{\lambda_k}{2}
			\quad \text{and} \quad 
			\left\|\frac{1}{n_k}\sum_{i=1}^{n_k}\bbm e_{k,i}\bbm x_{k,i}^\top\right\|_{\infty}\leq \frac{\varpi_k}{2}.
		\end{align}

		On these events, for the last term in \eqref{eq:one-stage-basic-ineq}, H\"older's inequality and the norm decomposability gives
		\begin{align}\label{eq:one-stage-deviation-decomposed}
			\left\langle\bm H_A+\bm H_B,\frac{1}{n_k}\sum_{i=1}^{n_k}\bbm e_{k,i}\bbm x_{k,i}^\top\right\rangle &\leq \|\bm H_A\|_*\left\|\frac{1}{n_k}\sum_{i=1}^{n_k}\bbm e_{k,i}\bbm x_{k,i}^\top\right\|_{\op}+\|\bm H_B\|_1\left\|\frac{1}{n_k}\sum_{i=1}^{n_k}\bbm e_{k,i}\bbm x_{k,i}^\top\right\|_{\infty}\notag\\
			&\leq \frac{\lambda_k}{2}\Bigl(\|(\bm H_A)_{\mathcal L_A}\|_*+\|(\bm H_A)_{\overline{\mathcal L}_A^\perp}\|_*\Bigr) + \frac{\varpi_k}{2}\Bigl(\|(\bm H_B)_{\mathcal S_\kappa}\|_1+\|(\bm H_B)_{\overline{\mathcal S}_\kappa^\perp}\|_1\Bigr).
		\end{align}

		Combining \eqref{eq:one-stage-basic-ineq}, \eqref{eq:one-stage-decomposable-norms}, and \eqref{eq:one-stage-deviation-decomposed}, and using the nonnegativity of the empirical quadratic term, yields
		\begin{align*}
			0 \leq \frac{1}{2n_k}\sum_{i=1}^{n_k}\left\|(\bm H_A+\bm H_B)\bbm x_{k,i}\right\|_2^2 \leq \lambda_k\Bigl(\frac{3}{2}\|(\bm H_A)_{\mathcal L_A}\|_*-\frac{1}{2}\|(\bm H_A)_{\overline{\mathcal L}_A^\perp}\|_*\Bigr) + \varpi_k\Bigl(\frac{3}{2}\|(\bm H_B)_{\mathcal S_\kappa}\|_1-\frac{1}{2}\|(\bm H_B)_{\overline{\mathcal S}_\kappa^\perp}\|_1+2\|(\bm B_k^*)_{\overline{\mathcal S}_\kappa^\perp}\|_1\Bigr).
		\end{align*}
		Rearranging the preceding display gives
		\begin{align}\label{eq:one-stage-cone}
			\lambda_k\|(\bm H_A)_{\overline{\mathcal L}_A^\perp}\|_*+\varpi_k\|(\bm H_B)_{\overline{\mathcal S}_\kappa^\perp}\|_1 \leq 3\lambda_k\|(\bm H_A)_{\mathcal L_A}\|_*+3\varpi_k\|(\bm H_B)_{\mathcal S_\kappa}\|_1+4\varpi_k\|(\bm B_k^*)_{\overline{\mathcal S}_\kappa^\perp}\|_1.
		\end{align}

		Next, we derive an upper bound for the empirical prediction error over the constrained space in \eqref{eq:one-stage-cone}. Plugging the deviation bounds \eqref{eq:deviation-bounds-e-x} into the right-hand side of \eqref{eq:one-stage-basic-ineq} and then applying the cone condition \eqref{eq:one-stage-cone}, we have
		\begin{align}\label{eq:one-stage-upper-bound-before-rsc}
			\frac{1}{2n_k}\sum_{i=1}^{n_k}\left\|(\bm H_A+\bm H_B)\bbm x_{k,i}\right\|_2^2 &\leq \lambda_k\bigl(\|\bm A_0^*\|_* - \|\widetilde{\bm A}_{0,k}\|_*\bigr) + \varpi_k\bigl(\|\bm B_k^*\|_1-\|\widetilde{\bm B}_{k}\|_1\bigr) + \frac{\lambda_k}{2}\|\bm H_A\|_* + \frac{\varpi_k}{2}\|\bm H_B\|_1 \notag\\
			&\leq \frac{3\lambda_k}{2}\|\bm H_A\|_* + \frac{3\varpi_k}{2}\|\bm H_B\|_1 \notag\\
			&= \frac{3\lambda_k}{2}\bigl(\|(\bm H_A)_{\mathcal L_A}\|_* + \|(\bm H_A)_{\overline{\mathcal L}_A^\perp}\|_*\bigr) + \frac{3\varpi_k}{2}\bigl(\|(\bm H_B)_{\mathcal S_\kappa}\|_1 + \|(\bm H_B)_{\overline{\mathcal S}_\kappa^\perp}\|_1\bigr)\\
			&\leq 6\left(\lambda_k\|(\bm H_A)_{\mathcal L_A}\|_* + \varpi_k\|(\bm H_B)_{\mathcal S_\kappa}\|_1 + \varpi_k\|(\bm B_k^*)_{\overline{\mathcal S}_\kappa^\perp}\|_1\right).\notag
		\end{align}

		We next derive a lower bound for the left-hand side of \eqref{eq:one-stage-upper-bound-before-rsc}. Let $\cm H_A$ and $\cm H_B$ denote the tensorization form of $\bm H_A$ and $\bm H_B$ under the inverse $S_X$-matricization, namely $(\cm H_A)_{[S_X]}=\bm H_A$ and $(\cm H_B)_{[S_X]}=\bm H_B$. Then $(\cm H_A+\cm H_B)_{[S_X]}=\bm H_A+\bm H_B$. Given $n_k \gtrsim\{\sigma_{\mathscr X,k}^{4}\kappa_{\mathscr X,k}^{2}\vee \sigma_{\mathscr X,k}^{2}\kappa_{\mathscr X,k}\vee 1\}p$, then applying Lemma~\ref{lem:rsc-condition} with $\cm T=\cm H_A+\cm H_B$, we have, with probability at least $1-\exp(-Cp)$,
		\begin{align}\label{eq:one-stage-rsc-lower-bound}
			\frac{1}{n_k}\sum_{i=1}^{n_k}\left\|(\bm H_A+\bm H_B)\bbm x_{k,i}\right\|_2^2 = \frac{1}{n_k}\sum_{i=1}^{n_k}\left\|\langle \cm H_A+\cm H_B,\cm X_{k,i}\rangle\right\|_\F^2 \geq \frac{\lambda_{\mathscr X,k}^{\min}}{2}\|\cm H_A+\cm H_B\|_\F^2 = \frac{\lambda_{\mathscr X,k}^{\min}}{2}\|\bm H_A+\bm H_B\|_\F^2.
		\end{align}
		Second, since $\|\bm M+\bm N\|_\F^2\geq \|\bm M\|_\F^2+\|\bm N\|_\F^2-2|\langle \bm M,\bm N\rangle|$, we have
		\begin{align}\label{eq:one-stage-F-decomposition}
			\lambda_{\mathscr X,k}^{\min}\|\bm H_A+\bm H_B\|_\F^2 \geq \lambda_{\mathscr X,k}^{\min}\left(\|\bm H_A\|_\F^2+\|\bm H_B\|_\F^2\right) - 2\lambda_{\mathscr X,k}^{\min}|\langle \bm H_A,\bm H_B\rangle|.
		\end{align}
		H\"older's inequality gives $|\langle \bm H_A,\bm H_B\rangle| \leq \|\bm H_B\|_{\op}\|\bm H_A\|_*$. Since both $\widetilde{\bm B}_{k}$ and $\bm B_k^*$ are feasible for \eqref{eq:ini-opt-matricization}, we have $\|\widetilde{\bm B}_{k}\|_{\op}\leq\zeta$ and $\|\bm B_k^*\|_{\op}\leq\zeta$, and hence $\|\bm H_B\|_{\op}=\|\widetilde{\bm B}_{k}-\bm B_k^*\|_{\op}\leq2\zeta$. Therefore, by the prescribed choice $\lambda_k\gtrsim\lambda_{\mathscr X,k}^{\min}\zeta$,
		\begin{align}\label{eq:one-stage-cross-term-bound}
			2\lambda_{\mathscr X,k}^{\min}|\langle \bm H_A,\bm H_B\rangle| \leq C\lambda_k\|\bm H_A\|_*.
		\end{align}
		Combining \eqref{eq:one-stage-rsc-lower-bound}--\eqref{eq:one-stage-cross-term-bound} yields
		\begin{align}\label{eq:one-stage-lower-bound-after-rsc}
			\frac{1}{2n_k}\sum_{i=1}^{n_k}\left\|(\bm H_A+\bm H_B)\bbm x_{k,i}\right\|_2^2 + C\lambda_k\|\bm H_A\|_* \gtrsim \lambda_{\mathscr X,k}^{\min}\left(\|\bm H_A\|_\F^2+\|\bm H_B\|_\F^2\right).
		\end{align}
		Finally, combining \eqref{eq:one-stage-upper-bound-before-rsc} and \eqref{eq:one-stage-lower-bound-after-rsc}, and using the cone condition \eqref{eq:one-stage-cone} to control $\lambda_k\|\bm H_A\|_*$, we have
		\begin{align}\label{eq:one-stage-pre-final}
			\lambda_{\mathscr X,k}^{\min}\left(\|\bm H_A\|_\F^2+\|\bm H_B\|_\F^2\right) \lesssim \lambda_k\|(\bm H_A)_{\mathcal L_A}\|_* + \varpi_k\|(\bm H_B)_{\mathcal S_\kappa}\|_1 + \varpi_k\|(\bm B_k^*)_{\overline{\mathcal S}_\kappa^\perp}\|_1.
		\end{align}

		We now translate \eqref{eq:one-stage-pre-final} into an explicit error bound for the weakly sparse regime $\nu\in(0,1)$. By the definition of $\mathcal L_A$, $(\bm H_A)_{\mathcal L_A}$ has rank at most $2r_A$. Since $2r_A\leq r_{\bbm q}+r_{\bbm p}$, we have
			\begin{align}\label{eq:one-stage-low-rank-projection-bound}
				\|(\bm H_A)_{\mathcal L_A}\|_* \leq \sqrt{2r_A}\,\|(\bm H_A)_{\mathcal L_A}\|_\F \leq \sqrt{r_{\bbm q}+r_{\bbm p}}\,\|\bm H_A\|_\F .
			\end{align}
		Moreover, since $\cm B_k^*\in\mathbb B_\nu$ with $\nu\in(0,1)$, the set $\mathcal S_\kappa$ satisfies $s_\nu\geq |\mathcal S_\kappa|\kappa^\nu$, and hence $|\mathcal S_\kappa|\leq s_\nu\kappa^{-\nu}$. Therefore,
		\begin{align}\label{eq:one-stage-sparse-projection-bound}
			\|(\bm H_B)_{\mathcal S_\kappa}\|_1 \leq \sqrt{|\mathcal S_\kappa|}\,\|(\bm H_B)_{\mathcal S_\kappa}\|_\F \leq \sqrt{s_\nu}\kappa^{-\nu/2}\|\bm H_B\|_\F.
		\end{align}
		By the definition of $\mathbb B_\nu=\left\{\cm B \in \mathbb{R}^{q_1\times\cdots\times q_m\times p_1\times\cdots\times p_d}: \mathds 1_{\{\nu=0\}}\|\cm B\|_0/s_0 + \mathds 1_{\{0<\nu<1\}}\|\cm B\|_\nu^\nu/s_\nu\leq 1\right\}$, we have
		\begin{align}\label{eq:one-stage-sparse-tail-bound}
			\|(\bm B_k^*)_{\overline{\mathcal S}_\kappa^\perp}\|_1 = \sum_{(a,b)\in \overline{\mathcal S}_\kappa^\perp}|(\bm B_k^*)_{ab}| \leq \sum_{(a,b)\in \overline{\mathcal S}_\kappa^\perp}|(\bm B_k^*)_{ab}|^\nu\kappa^{1-\nu} \leq s_\nu\kappa^{1-\nu}.
		\end{align}
		Substituting \eqref{eq:one-stage-low-rank-projection-bound}--\eqref{eq:one-stage-sparse-tail-bound} into \eqref{eq:one-stage-pre-final} gives
		\begin{align}\label{eq:one-stage-two-var-start}
			\|\bm H_A\|_\F^2+\|\bm H_B\|_\F^2 \lesssim \frac{\lambda_k\sqrt{r_{\bbm q}+r_{\bbm p}}}{\lambda_{\mathscr X,k}^{\min}}\|\bm H_A\|_\F + \frac{\varpi_k\sqrt{s_\nu}\kappa^{-\nu/2}}{\lambda_{\mathscr X,k}^{\min}}\|\bm H_B\|_\F + \frac{\varpi_k s_\nu\kappa^{1-\nu}}{\lambda_{\mathscr X,k}^{\min}}.
		\end{align}
		Using $ab\leq (a^2+b^2)/2$ for the first two terms on the right-hand side and absorbing constants, we have
		\begin{align}\label{eq:one-stage-two-var-absorbed}
			\|\bm H_A\|_\F^2+\|\bm H_B\|_\F^2 \lesssim (r_{\bbm q}+r_{\bbm p})\left(\frac{\lambda_k}{\lambda_{\mathscr X,k}^{\min}}\right)^2 + s_\nu\left(\frac{\varpi_k}{\lambda_{\mathscr X,k}^{\min}}\right)^2\kappa^{-\nu} + s_\nu\frac{\varpi_k}{\lambda_{\mathscr X,k}^{\min}}\kappa^{1-\nu}.
		\end{align}
		Choosing $\kappa\asymp\varpi_k/\lambda_{\mathscr X,k}^{\min}$ balances the last two terms and yields
		\begin{align}\label{eq:one-stage-squared-rate}
			\|\bm H_A\|_\F^2+\|\bm H_B\|_\F^2 \lesssim (r_{\bbm q}+r_{\bbm p})\left(\frac{\lambda_k}{\lambda_{\mathscr X,k}^{\min}}\right)^2 + s_\nu\left(\frac{\varpi_k}{\lambda_{\mathscr X,k}^{\min}}\right)^{2-\nu}.
		\end{align}
		Taking square roots and using the invariance of the Frobenius norm under matricization, we obtain the weakly sparse initialization bound
		\begin{align}\label{eq:one-stage-weak-sparse-rate}
			\|\widetilde{\cm A}_{0,k}-\cm A_0^*\|_{\F} + \|\widetilde{\cm B}_k-\cm B_k^*\|_{\F} \lesssim \sqrt{r_{\bbm q}+r_{\bbm p}}\frac{\lambda_k}{\lambda_{\mathscr X,k}^{\min}} + \sqrt{s_\nu}\left(\frac{\varpi_k}{\lambda_{\mathscr X,k}^{\min}}\right)^{1-\nu/2}, \quad \nu\in(0,1).
		\end{align}

		We next consider the strictly sparse regime $\nu=0$. Let $\mathcal S=\operatorname{supp}(\bm B_k^*)$, with $|\mathcal S|\leq s_0$. In this case, the weak-tail term related to $(\bm B_k^*)_{\mathcal S}=\bm0$ disappears, but the sparse stochastic term $\langle \bm H_B,n_k^{-1}\sum_{i=1}^{n_k}\bbm e_{k,i}\bbm x_{k,i}^\top\rangle$ needs to be treated by the refined sparse-deviation bound in Lemma~\ref{lem:deviation-condition-sparse}. The arguments below follows the strictly sparse part in the proof of Theorem~\ref{thm:personalized_error}, with the only difference that the present proof also contains the low-rank term $\bm H_A$. Let $\bm W_k=n_k^{-1}\sum_{i=1}^{n_k}\bbm e_{k,i}\bbm x_{k,i}^\top$. The nuclear-norm part for $\bm H_A$ and the operator-norm stochastic bound $\langle\bm H_A, \bm W_k\rangle$ are treated exactly as above, so it remains to handle $\langle \bm H_B,\bm W_k\rangle$. By Lemma~\ref{lem:deviation-condition-sparse}, there are two cases.

		\noindent
		\textbf{\emph{Case I.}}
		If $\|\bm H_B\|_1>\sqrt{s_0}\|\bm H_B\|_\F$, then $|\langle \bm H_B,\bm W_k\rangle|\leq \varpi_k\|\bm H_B\|_1/2$ under the stated choice of $\varpi_k$. This term is absorbed by the $\ell_1$-penalty. Following the same basic-inequality and cone arguments as in the weakly sparse regime, but with $(\bm B_k^*)_{\overline{\mathcal S}^\perp}=\bm0$, yields
		\[
			\lambda_k\| (\bm H_A)_{\overline{\mathcal L}_A^\perp}\|_* + \varpi_k\|(\bm H_B)_{\overline{\mathcal S}^\perp}\|_1 \lesssim \lambda_k\|(\bm H_A)_{\mathcal L_A}\|_* + \varpi_k\|(\bm H_B)_{\mathcal S}\|_1.
		\]
		Consequently, combining the cone relation with the empirical upper bound for $(2n_k)^{-1}\sum_{i=1}^{n_k}\left\|(\bm H_A+\bm H_B)\bbm x_{k,i}\right\|_2^2$, in the same manner as in \eqref{eq:one-stage-upper-bound-before-rsc}, and the RSC lower bound in \eqref{eq:one-stage-lower-bound-after-rsc}, we have
		\begin{align}\label{eq:one-stage-strict-sparse-pre-final-large-l1}
			\lambda_{\mathscr X,k}^{\min}\left(\|\bm H_A\|_\F^2+\|\bm H_B\|_\F^2\right) \lesssim \lambda_k\|(\bm H_A)_{\mathcal L_A}\|_* + \varpi_k\|(\bm H_B)_{\mathcal S}\|_1.
		\end{align}
		Applying \eqref{eq:one-stage-low-rank-projection-bound} and $\|(\bm H_B)_{\mathcal S}\|_1\leq\sqrt{s_0}\|\bm H_B\|_\F$ to \eqref{eq:one-stage-strict-sparse-pre-final-large-l1} further gives the Frobenius-norm upper bound
		\begin{align}\label{eq:one-stage-strict-sparse-frob-upper}
			\lambda_{\mathscr X,k}^{\min}\left(\|\bm H_A\|_\F^2+\|\bm H_B\|_\F^2\right) \lesssim \lambda_k\sqrt{r_{\bbm q}+r_{\bbm p}}\|\bm H_A\|_\F + \sqrt{s_0}\varpi_k\|\bm H_B\|_\F.
		\end{align}

		\noindent
		\textbf{\emph{Case II.}}
		If $\|\bm H_B\|_1\leq\sqrt{s_0}\|\bm H_B\|_\F$, then Lemma~\ref{lem:deviation-condition-sparse} gives $|\langle \bm H_B,\bm W_k\rangle|\lesssim \sqrt{s_0}\varpi_k\|\bm H_B\|_\F$. This together with \eqref{eq:one-stage-upper-bound-before-rsc}--\eqref{eq:one-stage-F-decomposition}, replacing $|\langle \bm H_B,\bm W_k\rangle|$ by the preceding Frobenius-order bound $\sqrt{s_0}\varpi_k\|\bm H_B\|_\F$, and the same nuclear-norm decomposition for $\bm H_A$, together with $\|\bm H_B\|_1\leq\sqrt{s_0}\|\bm H_B\|_\F$, directly yields \eqref{eq:one-stage-strict-sparse-frob-upper}. 
		
		Applying Young's inequality, we have
		\[
			\lambda_k\sqrt{r_{\bbm q}+r_{\bbm p}}\|\bm H_A\|_\F \leq \frac{1}{4}\lambda_{\mathscr X,k}^{\min}\|\bm H_A\|_\F^2 + C(r_{\bbm q}+r_{\bbm p})\frac{\lambda_k^2}{\lambda_{\mathscr X,k}^{\min}}
			\ \text{and} \
			\sqrt{s_0}\varpi_k\|\bm H_B\|_\F \leq \frac{1}{4}\lambda_{\mathscr X,k}^{\min}\|\bm H_B\|_\F^2 + C s_0\frac{\varpi_k^2}{\lambda_{\mathscr X,k}^{\min}}.
		\]
		Substituting these bounds into \eqref{eq:one-stage-strict-sparse-frob-upper}, moving the two quadratic terms to the left-hand side, and dividing both sides by $\lambda_{\mathscr X,k}^{\min}$, we obtain a squared-error bound. Taking square roots and using the invariance of the Frobenius norm under matricization yields
		\begin{align}\label{eq:one-stage-strict-sparse-rate}
			\|\widetilde{\cm A}_{0,k}-\cm A_0^*\|_{\F} + \|\widetilde{\cm B}_k-\cm B_k^*\|_{\F} \lesssim \sqrt{r_{\bbm q}+r_{\bbm p}}\frac{\lambda_k}{\lambda_{\mathscr X,k}^{\min}} + \sqrt{s_0}\frac{\varpi_k}{\lambda_{\mathscr X,k}^{\min}}, \quad \nu=0.
		\end{align}
		Combining the weakly sparse bound in \eqref{eq:one-stage-weak-sparse-rate} and the strictly sparse bound in \eqref{eq:one-stage-strict-sparse-rate}, together with the corresponding stochastic deviation events \eqref{eq:deviation-bounds-e-x} and the restricted empirical quadratic lower-bound event, gives the stated probability and completes the proof.
	\end{proof}

	\begin{proof}[\textbf{Proof of Theorem~\ref{thm:rank_selection}}]
		We prove the rank-selection consistency mode by mode. By Theorem~\ref{thm:one-stage-initialization-error}, under the stated conditions,
		\[
			\|\widetilde{\cm A}_{0,k}-\cm A_0^*\|_\F\lesssim \sqrt{r_{\bbm q}+r_{\bbm p}}\frac{\lambda_k}{\lambda_{\mathscr X,k}^{\min}} + \sqrt{s_\nu}\left(\frac{\varpi_k}{\lambda_{\mathscr X,k}^{\min}}\right)^{1-\nu/2} = o\bigl(c_s(\bbm p,\bbm q,n_k)\bigr)
		\]
		with probability tending to one. Since matricization preserves the Frobenius norm, $\|(\widetilde{\cm A}_{0,k}-\cm A_0^*)_{[s]}\|_\F = \|\widetilde{\cm A}_{0,k}-\cm A_0^*\|_\F$. Therefore, by Weyl's inequality,
		\[
			\max_j\left|\sigma_j\bigl((\widetilde{\cm A}_{0,k})_{[s]}\bigr)-\sigma_j\bigl((\cm A_0^*)_{[s]}\bigr)\right| \leq \|(\widetilde{\cm A}_{0,k}-\cm A_0^*)_{[s]}\|_{\op} \leq \|\widetilde{\cm A}_{0,k}-\cm A_0^*\|_\F = o_p\bigl(c_s(\bbm p,\bbm q,n_k)\bigr).
		\]
		We now analyze the ridge-type ratio criterion for a fixed mode $s\in[d+m]$. To simplify notation, write $c_{k,s}=c_s(\bbm p,\bbm q,n_k)$. Recall that
		\[
			\widehat r_{k,s} = \argmin_{1\leq r\leq \bar r_s-1}\widehat R_{k,s}(r),
			\ \text{with} \ 
			\widehat R_{k,s}(r) = \frac{\sigma_{r+1}\bigl((\widetilde{\cm A}_{0,k})_{[s]}\bigr)+c_{k,s}}{\sigma_{r}\bigl((\widetilde{\cm A}_{0,k})_{[s]}\bigr)+c_{k,s}}.
		\]

		Note that, for every $r$, $\sigma_r\bigl((\widetilde{\cm A}_{0,k})_{[s]}\bigr)+c_{k,s} = \sigma_r\bigl((\cm A_0^*)_{[s]}\bigr) + \{\sigma_r\bigl((\widetilde{\cm A}_{0,k})_{[s]}\bigr) - \sigma_r\bigl((\cm A_0^*)_{[s]}\bigr)\} + c_{k,s}$.
		Moreover, since $\sigma_{r_s^*}((\cm A_0^*)_{[s]})/\sigma_j((\cm A_0^*)_{[s]})\leq1$ for $1\leq j\leq r_s^*$, the assumed upper bound on $c_{k,s}$ implies
		\[
			c_{k,s} \ll \sigma_{r_s^*}\bigl((\cm A_0^*)_{[s]}\bigr)\min_{1\leq j\leq r_s^*-1}\frac{\sigma_{j+1}\bigl((\cm A_0^*)_{[s]}\bigr)}{\sigma_j\bigl((\cm A_0^*)_{[s]}\bigr)} \leq \min_{1\leq j\leq r_s^*-1}\sigma_{j+1}\bigl((\cm A_0^*)_{[s]}\bigr)\leq \sigma_r\bigl((\cm A_0^*)_{[s]}\bigr), \quad \forall r\leq r_s^*.
		\]
		Hence, $c_{k,s}=o(\sigma_r((\cm A_0^*)_{[s]}))$ for all $r\leq r_s^*$. Now write, for each candidate rank $r$,
		\[
			\sigma_r\bigl((\widetilde{\cm A}_{0,k})_{[s]}\bigr)+c_{k,s} = \sigma_r\bigl((\cm A_0^*)_{[s]}\bigr) + \left\{\sigma_r\bigl((\widetilde{\cm A}_{0,k})_{[s]}\bigr) - \sigma_r\bigl((\cm A_0^*)_{[s]}\bigr)\right\} + c_{k,s}.
		\]
		Therefore, we have the following two cases.
		\begin{itemize}
		\item[(a)] If $r>r_s^*$, then $\sigma_r((\cm A_0^*)_{[s]})=0$. Hence, $\sigma_r((\widetilde{\cm A}_{0,k})_{[s]})+c_{k,s} = o_p(c_{k,s}) + c_{k,s} = c_{k,s}(1+o_p(1))$.

		\item[(b)] If $r\leq r_s^*$, recall that $c_{k,s} = o(\sigma_r((\cm A_0^*)_{[s]}))$ and $\sigma_r((\widetilde{\cm A}_{0,k})_{[s]}) - \sigma_r\bigl((\cm A_0^*)_{[s]}\bigr) = o_p(c_{k,s}) = o_p(\sigma_r((\cm A_0^*)_{[s]}))$. Therefore, $\sigma_r\bigl((\widetilde{\cm A}_{0,k})_{[s]}\bigr)+c_{k,s} = \sigma_r\bigl((\cm A_0^*)_{[s]}\bigr) + o_p\left(\sigma_r\bigl((\cm A_0^*)_{[s]}\bigr)\right) = \sigma_r\bigl((\cm A_0^*)_{[s]}\bigr)(1+o_p(1))$.
		\end{itemize}
		We then analyze the ridge-ratio criterion in three cases.

		\noindent
		\textbf{\emph{Case I $(r>r_s^*)$.}}
		Both the numerator and denominator are dominated by $c_{k,s}$. Hence, $\widehat R_{k,s}(r) \xrightarrow{p}1$.

		\noindent
		\textbf{\emph{Case II $(r<r_s^*)$.}}
		Both terms are dominated by the corresponding nonzero singular values of $(\cm A_0^*)_{[s]}$. Thus,
		\[
			\widehat R_{k,s}(r)\xrightarrow{p}\frac{\sigma_{r+1}\bigl((\cm A_0^*)_{[s]}\bigr)}{\sigma_r\bigl((\cm A_0^*)_{[s]}\bigr)}.
		\]

		\noindent
		\textbf{\emph{Case III $(r=r_s^*)$.}}
		The numerator is dominated by $c_{k,s}$. In contrast, the denominator is dominated by $\sigma_{r_s^*}((\cm A_0^*)_{[s]})$. Since $c_{k,s}=o\{\sigma_{r_s^*}((\cm A_0^*)_{[s]})\}$, we have
		\[
			\widehat R_{k,s}(r_s^*) = \frac{c_{k,s}}{\sigma_{r_s^*}\bigl((\cm A_0^*)_{[s]}\bigr)}(1+o_p(1))\xrightarrow{p}0.
		\]

		Combining Cases I--III, with probability tending to one, $\widehat R_{k,s}(r_s^*)$ is the unique minimizer of $\{\widehat R_{k,s}(r):1\leq r\leq \bar r_s-1\}$. Therefore, $\mathbb P(\widehat r_{k,s}=r_s^*)\to1$ for all $s\in[d+m]$. By Bonferroni's inequality, and since the number of modes $d+m$ is fixed, we have
		\[
			\mathbb P\left(\widehat{\bbm r}_k^{loc}=\bbm r^*\right) = \mathbb P\left(\bigcap_{s=1}^{d+m}\{\widehat r_{k,s}=r_s^*\}\right)\geq 1- \sum_{s=1}^{d+m}\mathbb P\left(\widehat r_{k,s}\neq r_s^*\right)\to1.\qedhere
		\]
	\end{proof}

\section{Proof of Minimax Lower Bounds}\label{sec:minimax-lower bound}

	\begin{proof}[\textbf{Proof of Theorem~\ref{thm:joint-minimax-lower-bound}}]
		The proof for Theorem~\ref{thm:joint-minimax-lower-bound} follows by deriving Frobenius-norm lower bounds for estimating $\cm A_0^*$ and $\{\cm B_k^*\}_{k=1}^K$ in three complementary cases. The first concerns structural non-identifiability, as established in Theorem~\ref{thm:nonidentifiability-lower-bound}; the second concerns estimation of the low Tucker-rank shared component $\cm A_0^*$ in the single-client or federated settings, as established in Theorems~\ref{thm:single-client-minimax-lower} or~\ref{thm:federated-minimax-lower}; and the third concerns estimation of the sparse client-specific deviations $\{\cm B_k^*\}_{k=1}^K$, as established in Theorem~\ref{thm:single-client-sparse-deviation-minimax-lower}. 
		For the statistical lower bounds, we restrict attention to the normalized Gaussian submodel in which $\vect(\cm X_{k,i})\sim N(\bm0,\bm I_p)$ and $\vect(\cm E_{k,i})\sim N(\bm0,\bm I_q)$ independently across clients and observations and independently of one another. This submodel satisfies Assumptions~\ref{assump:subg-design} and~\ref{assump:subg-noise}; hence, a minimax lower bound established on it also applies to the full model class. 
		Combining the lower bounds obtained in these three cases yields the desired joint minimax lower bound.
	\end{proof}

	\begin{theorem}\label{thm:nonidentifiability-lower-bound}
		Suppose that $\tucrank(\cm A_0^*)\leq \bbm r$. Assume that $q_\ell\geq 2r_\ell$ for $\ell\in [m]$ and $p_\ell\geq 2r_{m+\ell}$ for $\ell\in[d]$.
		Furthermore, assume that $s_0\geq 2(r_{\bbm q}\wedge r_{\bbm p})$ for the exact sparse case, and $s_\nu\geq 2(r_{\bbm q}\wedge r_{\bbm p})\zeta^\nu$ for the weak sparse case. Then, for the federated setting,
		\[
			\inf_{\widehat{\boldsymbol{\mathscr A}}_0,\widehat{\boldsymbol{\mathscr B}}_k}\sup_{(\boldsymbol{\mathscr A}_0^*,\{\boldsymbol{\mathscr B}_k^*\}_{k=1}^K)\in\Theta^{\mathrm{fed}}}\mathbb E_{\boldsymbol{\mathscr A}_0^*,\{\boldsymbol{\mathscr B}_k^*\}_{k=1}^K}\left[\|\widehat{\cm A}_0-\cm A_0^*\|_{\F} + \|\widehat{\cm B}_k-\cm B_k^*\|_{\F}\right] \gtrsim \sqrt{r_{\bbm q}\wedge r_{\bbm p}}\,\zeta,
		\]
		and for the single-client setting,
		\[
			\inf_{\widehat{\boldsymbol{\mathscr A}}_0,\widehat{\boldsymbol{\mathscr B}}_k}\sup_{(\boldsymbol{\mathscr A}_0^*,\boldsymbol{\mathscr B}_k^*)\in\Theta_k^{\mathrm{single}}}\mathbb E_{\boldsymbol{\mathscr A}_0^*,\boldsymbol{\mathscr B}_k^*}\left[\|\widehat{\cm A}_0-\cm A_0^*\|_{\F} + \|\widehat{\cm B}_k-\cm B_k^*\|_{\F}\right]\gtrsim \sqrt{r_{\bbm q}\wedge r_{\bbm p}}\,\zeta, \quad \forall k\in[K].
		\]
	\end{theorem}

	\begin{proof}[\textbf{Proof of Theorem \ref{thm:nonidentifiability-lower-bound}}]
		Throughout this construction, $\bm e_j$ denotes the $j$th canonical basis vector in the corresponding Euclidean space. For $\ell\in[m]$, define two disjoint coordinate matrices $\bm U_\ell^{(0)}=(\bm e_1,\cdots,\bm e_{r_\ell})\in\mathbb R^{q_\ell\times r_\ell}$ and $\bm U_\ell^{(1)}=(\bm e_{r_\ell+1},\cdots,\bm e_{2r_\ell}) \in \mathbb R^{q_\ell\times r_\ell}$. For $\ell\in[d]$, similarly define two disjoint coordinate matrices $\bm V_\ell^{(0)}=(\bm e_{1},\cdots,\bm e_{r_{m+\ell}})\in\mathbb R^{p_\ell\times r_{m+\ell}}$ and $\bm V_\ell^{(1)}=(\bm e_{1+r_{m+\ell}},\cdots,\bm e_{2r_{m+\ell}})\in\mathbb R^{p_\ell\times r_{m+\ell}}$.
		Let $\cm C\in \mathbb R^{r_1\times\cdots\times r_m\times r_{m+1}\times\cdots\times r_{d+m}}$ be a core tensor whose $S_X$-matricization is
		\[
			\cm C_{[S_X]} = \zeta
			\begin{pmatrix}
				\bm I_{r_{\bbm q}\wedge r_{\bbm p}} & \bm 0 \\
				\bm 0 & \bm 0
			\end{pmatrix}
			\in\mathbb R^{r_{\bbm q}\times r_{\bbm p}}.
		\]
		For $s\in\{0,1\}$, define
		\begin{align*}
				\cm T^{(s)} = \cm C\times_{1}\bm U_1^{(s)}\cdots\times_{m}\bm U_m^{(s)}\times_{m+1}\bm V_1^{(s)}\cdots\times_{d+m}\bm V_d^{(s)}.
		\end{align*}
		We first verify that each $\cm T^{(s)}$ satisfies the required Tucker-rank constraint $\tucrank(\cm T^{(s)})\leq \bbm r$. For each $s\in\{0,1\}$, $\cm T^{(s)}$ is written in Tucker form with core $\cm C$, output-mode factor matrices $\bm U_1^{(s)},\cdots,\bm U_m^{(s)}$, and input-mode factor matrices $\bm V_1^{(s)},\cdots,\bm V_d^{(s)}$. Note that for $\cm M=\cm S\times_1\bm U_1\times_2\cdots\times_\ell\bm U_\ell=\cm S\times_{i=1}^{\ell}\bm U_i$ where $\cm S\in\mathbb R^{r_1\times\cdots\times r_\ell}$ is the core tensor and $\bm U_i\in\mathbb R^{p_i\times r_i}$ has orthonormal columns, $\cm M_{[s]}=\bm U_s\cm S_{[s]}(\otimes_{j\neq s}\bm U_j)^\top$.
		Hence, by the standard Tucker matricization formula, for $s \in \{0,1\}$,
		\[
			\rank\bigl((\cm T^{(s)})_{[\ell]}\bigr) \leq \rank(\bm U_\ell^{(s)})\leq r_\ell, \ \forall \ell\in[m]
			\ \text{and} \ 
			\rank\bigl((\cm T^{(s)})_{[m+\ell]}\bigr)\leq \rank(\bm V_\ell^{(s)})\leq r_{m+\ell},\ \forall \ell\in[d].
		\]
		Thus $\tucrank(\cm T^{(s)})\leq\bbm r$ for $s\in\{0,1\}$.
		Define the perturbation $\cm F=\cm T^{(1)}+\cm T^{(0)}$. We then set $\cm A_0^{(0)}=-\cm T^{(0)}$ and $\cm A_0^{(1)}=\cm T^{(1)}$. Thus $\cm A_0^{(1)}-\cm A_0^{(0)}=\cm F$. Hence $\tucrank(\cm A_0^{(0)})\leq\bbm r$ and $\tucrank(\cm A_0^{(1)})\leq\bbm r$. 
		
		Moreover, since the coordinate blocks used in $\cm T^{(0)}$ and $\cm T^{(1)}$ are disjoint in every mode, their $S_X$-matricizations have mutually orthogonal row spaces and mutually orthogonal column spaces. Note that, under the present construction, $(\cm T^{(s)})_{[S_X]} = \bm Q_s \cm C_{[S_X]} \bm P_s^\top$, where $\bm Q_s = \bm U_m^{(s)}\otimes\cdots\otimes\bm U_1^{(s)}$ and $\bm P_s = \bm V_d^{(s)}\otimes\cdots\otimes\bm V_1^{(s)}$. Since each $\bm U_\ell^{(s)}$ and $\bm V_\ell^{(s)}$ consists of canonical basis vectors, we have $(\bm U_\ell^{(s)})^\top \bm U_\ell^{(s)}=\bm I_{r_\ell}$ and $(\bm V_\ell^{(s)})^\top \bm V_\ell^{(s)}=\bm I_{r_{m+\ell}}$. Therefore, by the mixed-product property of Kronecker products,
		\[
			\begin{aligned}
				\bm Q_s^\top \bm Q_s
				&= \left(\bm U_m^{(s)}\otimes\cdots\otimes\bm U_1^{(s)}\right)^\top\left(\bm U_m^{(s)}\otimes\cdots\otimes\bm U_1^{(s)}\right) \\
				&= (\bm U_m^{(s)})^\top\bm U_m^{(s)}\otimes\cdots\otimes(\bm U_1^{(s)})^\top\bm U_1^{(s)} \\
				&= \bm I_{r_m}\otimes\cdots\otimes\bm I_{r_1} = \bm I_{r_{\bbm q}}.
			\end{aligned}
		\]
		Similarly,
		\[
			\begin{aligned}
				\bm P_s^\top \bm P_s
				&= \left(\bm V_d^{(s)}\otimes\cdots\otimes\bm V_1^{(s)}\right)^\top\left(\bm V_d^{(s)}\otimes\cdots\otimes\bm V_1^{(s)}\right) \\
				&= (\bm V_d^{(s)})^\top\bm V_d^{(s)}\otimes\cdots\otimes(\bm V_1^{(s)})^\top\bm V_1^{(s)} \\
				&= \bm I_{r_{d+m}}\otimes\cdots\otimes\bm I_{r_{m+1}} = \bm I_{r_{\bbm p}}.
			\end{aligned}
		\]
		Moreover, since the coordinate blocks used in the construction are disjoint, $(\bm U_\ell^{(0)})^\top\bm U_\ell^{(1)}=\bm 0$ for $\ell\in[m]$ and $(\bm V_\ell^{(0)})^\top\bm V_\ell^{(1)}=\bm 0$ for $\ell\in[d]$. Hence, again by the mixed-product property,
		\[
			\begin{aligned}
				\bm Q_0^\top\bm Q_1 &= \left(\bm U_m^{(0)}\otimes\cdots\otimes\bm U_1^{(0)}\right)^\top\left(\bm U_m^{(1)}\otimes\cdots\otimes\bm U_1^{(1)}\right) = (\bm U_m^{(0)})^\top\bm U_m^{(1)}\otimes\cdots\otimes(\bm U_1^{(0)})^\top\bm U_1^{(1)} = \bm 0,\\
				\bm P_0^\top\bm P_1 &= \left(\bm V_d^{(0)}\otimes\cdots\otimes\bm V_1^{(0)}\right)^\top\left(\bm V_d^{(1)}\otimes\cdots\otimes\bm V_1^{(1)}\right) = (\bm V_d^{(0)})^\top\bm V_d^{(1)}\otimes\cdots\otimes(\bm V_1^{(0)})^\top\bm V_1^{(1)} = \bm 0.
			\end{aligned}
		\]
		Hence $\cm T^{(0)}_{[S_X]}$ and $\cm T^{(1)}_{[S_X]}$ have mutually orthogonal column spaces and mutually orthogonal row spaces.

		Since $\cm F_{[S_X]} = \cm T^{(1)}_{[S_X]}+\cm T^{(0)}_{[S_X]}$, we can write
		\[
			\cm F_{[S_X]} =
				\begin{pmatrix}
					\bm Q_0 , \bm Q_1
				\end{pmatrix}
				\begin{pmatrix}
					\cm C_{[S_X]} & \bm 0 \\
					\bm 0 & \cm C_{[S_X]}
				\end{pmatrix}
				\begin{pmatrix}
					\bm P_0 , \bm P_1
				\end{pmatrix}^{\top}.
		\]
		Since $(\bm Q_0,\bm Q_1)$ and $(\bm P_0,\bm P_1)$ have orthonormal columns, multiplying by them preserves the nonzero singular values of the middle block-diagonal matrix. Hence the nonzero singular values of $\cm F_{[S_X]}$ are those of $\cm C_{[S_X]}$ repeated twice. Since $\cm C_{[S_X]}$ has $r_{\bbm q}\wedge r_{\bbm p}$ nonzero singular values, all equal to $\zeta$, we have $\|\cm F_{[S_X]}\|_{\op}=\zeta$, $\|\cm F\|_{\F}^2 = \|\cm F_{[S_X]}\|_{\F}^2 = 2(r_{\bbm q}\wedge r_{\bbm p})\zeta^2 \asymp(r_{\bbm q}\wedge r_{\bbm p})\zeta^2$.

		By construction, $\cm F$ has at most $2(r_{\bbm q}\wedge r_{\bbm p})$ nonzero entries, each with magnitude $\zeta$. Hence, in the exact sparse case, $\|\cm F\|_0\leq2(r_{\bbm q}\wedge r_{\bbm p})\leq s_0$, while in the weak sparse case, $\|\cm F\|_\nu^\nu \leq 2(r_{\bbm q}\wedge r_{\bbm p})\zeta^\nu \leq s_\nu$. Together with $\|\cm F_{[S_X]}\|_{\op}=\zeta$, this implies that both $\cm F$ and $-\cm F$ are admissible as client-specific deviations.

		We now construct two feasible decompositions:
		\[
			\begin{array}{lll}
				\text{Decomposition 0:}& \cm A_0^{(0)}=-\cm T^{(0)}, & \cm B_{k}^{(0)}=\bm 0,\quad k\in[K],\\[1mm]
				\text{Decomposition 1:}& \cm A_0^{(1)}=\cm T^{(1)}, & \cm B_{k}^{(1)}=-\cm F=-(\cm T^{(1)}+\cm T^{(0)}),\quad k\in[K].
			\end{array}
		\]
		Both decompositions are feasible by the preceding construction. Moreover, for every $k\in[K]$,
		\[
			\cm A_0^{(1)}+\cm B_{k}^{(1)} = \cm T^{(1)}-\cm F = -\cm T^{(0)} = \cm A_0^{(0)}+\cm B_{k}^{(0)}.
		\]
		Therefore, the two decompositions induce the same distribution of the observations. Then, for any estimators $(\widehat{\cm A}_0,\widehat{\cm B}_k)$, the federated risk satisfies
		\begin{align*}
			&\sup_{(\boldsymbol{\mathscr A}_0,\{\boldsymbol{\mathscr B}_k\}_{k=1}^K)\in\Theta^{\mathrm{fed}}}\mathbb E_{\boldsymbol{\mathscr A}_0,\{\boldsymbol{\mathscr B}_k\}_{k=1}^K}\left[\|\widehat{\cm A}_0-\cm A_0\|_{\F} + \|\widehat{\cm B}_k-\cm B_k\|_{\F}\right] \\
			&\quad\geq \max_{s\in\{0,1\}}\mathbb E\left[\|\widehat{\cm A}_0-\cm A_0^{(s)}\|_{\F} + \|\widehat{\cm B}_k-\cm B_k^{(s)}\|_{\F}\right] \\
			&\quad\geq\frac12\mathbb E\Bigl[\|\widehat{\cm A}_0-\cm A_0^{(0)}\|_{\F} + \|\widehat{\cm A}_0-\cm A_0^{(1)}\|_{\F} + \|\widehat{\cm B}_k-\cm B_k^{(0)}\|_{\F} + \|\widehat{\cm B}_k-\cm B_k^{(1)}\|_{\F}\Bigr] \\
			&\quad\geq \frac12\left(\|\cm A_0^{(1)}-\cm A_0^{(0)}\|_{\F} + \|\cm B_k^{(1)}-\cm B_k^{(0)}\|_{\F}\right) \\
			&\quad= \|\cm F\|_{\F} \asymp \sqrt{r_{\bbm q}\wedge r_{\bbm p}}\,\zeta,
		\end{align*}
		where the third inequality follows from the triangle inequality.
		Taking the infimum over all estimators proves the federated statement. The single-client statement follows from the same construction restricted to client $k$, namely by using only the two feasible pairs $(\cm A_0^{(0)},\cm B_k^{(0)})$ and $(\cm A_0^{(1)},\cm B_k^{(1)})$ in $\Theta_k^{\mathrm{single}}$.
	\end{proof}

	\begin{theorem}\label{thm:single-client-minimax-lower}
		Fix a client $k\in[K]$. Let $\mathcal F_{k}^{\mathrm{single}}(\bbm r) = \left\{\cm A_0:\tucrank(\cm A_0)\leq \bbm r\right\}$. Then,
		\[
			\inf_{\widehat{\boldsymbol{\mathscr A}}_{0,k}}\sup_{\boldsymbol{\mathscr A}_0^*\in\mathcal F_{k}^{\mathrm{single}}(\bbm r)}\mathbb E_{\boldsymbol{\mathscr A}_0^*}\left[\|\widehat{\cm A}_{0,k}-\cm A_0^*\|_\F\right] \gtrsim \sqrt{\frac{df_{\bbm r}}{n_k}}.
		\]
	\end{theorem}

	\begin{proof}[\textbf{Proof of Theorem \ref{thm:single-client-minimax-lower}}]
		It suffices to prove the lower bound over the least favorable submodel with $\cm B_k^*=\bm 0$. Since this submodel is contained in the full single-client parameter space, any minimax lower bound over it is also a valid lower bound for the original problem. Under this submodel, the response is generated from $\cm Y_{k,i}=\langle \cm A_0^*,\cm X_{k,i}\rangle+\cm E_{k,i}$ for $i\in[n_k]$ and the task reduces to estimating the common low Tucker-rank  tensor $\cm A_0^*$.

		Suppose that, for some radius $\delta>0$, there exists a finite packing set $\mathcal A=\{\cm A_0^{(1)},\cdots,\cm A_0^{(M)}\} \subseteq \mathcal F_k^{\mathrm{single}}(\bbm r)$ such that $\min_{j\neq \ell} \|\cm A_0^{(j)}-\cm A_0^{(\ell)}\|_\F \geq 2\delta$. Let $J$ be a random variable uniformly distributed over $[M]=\{1,\cdots,M\}$. Conditional on $J=j$, let the sample $\{(\cm X_{k,i},\cm Y_{k,i})\}_{i=1}^{n_k}$ be generated according to the parameter $\cm A_0^{(j)}$, and let $\mathbb P$ denote the joint law of $\left(\{(\cm X_{k,i},\cm Y_{k,i})\}_{i=1}^{n_k},J\right)$. By \citet[Proposition~15.1]{wainwright2019high}, we have
		\begin{align}\label{eq:single-estimation-testing-reduction}
			\inf_{\widehat{\boldsymbol{\mathscr A}}_{0,k}}\sup_{\boldsymbol{\mathscr A}_0^*\in\mathcal F_k^{\mathrm{single}}(\bbm r)}\mathbb E_{\boldsymbol{\mathscr A}_0^*}\left[\|\widehat{\cm A}_{0,k}-\cm A_0^*\|_\F\right]\geq\delta\inf_{\psi}\mathbb P\left(\psi\bigl(\{(\cm X_{k,i},\cm Y_{k,i})\}_{i=1}^{n_k}\bigr)\neq J\right),
		\end{align}
		where the infimum is over all testing rules $\psi$ based on the observed sample. Thus, it remains to lower bound the probability of error in this multiple testing problem.

		Let $\mathcal X_k=\{\cm X_{k,i}\}_{i=1}^{n_k}$ and $\mathcal Y_k=\{\cm Y_{k,i}\}_{i=1}^{n_k}$. For each $j\in[M]$, let $\mathbb P_k^{(j)}$ denote the joint law of $(\mathcal X_k,\mathcal Y_k)$ under $\cm A_0^*=\cm A_0^{(j)}$, and let $\mathbb Q_{Y,k}^{(j)}(\mathcal Y_k\mid\mathcal X_k)$ denote the corresponding conditional law of $\mathcal Y_k$ given $\mathcal X_k$. Write $\bbm x_{k,i}=\vect(\cm X_{k,i})$, $\bbm y_{k,i}=\vect(\cm Y_{k,i})$, and $\bm A^{(j)}=(\cm A_0^{(j)})_{[S_X]}$. Under the $j$-th hypothesis, $\bbm y_{k,i}\mid\bbm x_{k,i}\sim \mathscr N(\bm A^{(j)}\bbm x_{k,i},\bm I_q)$ for $i\in[n_k]$. Recall that $P\ll Q$ means that $P$ is absolutely continuous with respect to $Q$, namely, $Q(A)=0$ implies $P(A)=0$ for every measurable set $A$. For such probability measures $P$ and $Q$,
		\[
		\mathrm{KL}(P\|Q):=\mathbb E_P\left[\log\frac{dP}{dQ}\right].
		\]
		In particular, for two Gaussian laws with common covariance \(\bm I_q\),
		\[
		\mathrm{KL}\left(\mathscr N(\bbm\mu_1,\bm I_q)\|\mathscr N(\bbm\mu_2,\bm I_q)\right)
		=
		\mathbb E_{\bm Z\sim\mathscr N(\bbm\mu_1,\bm I_q)}
		\left[
		\log\frac{\phi_q(\bm Z;\bbm\mu_1,\bm I_q)}
		{\phi_q(\bm Z;\bbm\mu_2,\bm I_q)}
		\right]
		=
		\frac12\|\bbm\mu_1-\bbm\mu_2\|_2^2,
		\]
		where $\phi_q(\cdot;\bbm\mu,\bm I_q)$ denotes the density of $\mathscr N(\bbm\mu,\bm I_q)$. Because the observations are conditionally independent across $i$ given $\mathcal X_k$, the KL divergence between the two conditional product distributions is the sum of the sample-wise KL divergences. The preceding identity therefore gives, for any $j\neq\ell$,
		\begin{align}\label{eq:single-conditional-kl}
			\mathrm{KL}\left(\mathbb Q_{Y,k}^{(j)}(\mathcal Y_k\mid\mathcal X_k)\,\|\,\mathbb Q_{Y,k}^{(\ell)}(\mathcal Y_k\mid\mathcal X_k)\right) &=\sum_{i=1}^{n_k}\mathrm{KL}\left(\mathscr N\left(\bm A^{(j)}\bbm x_{k,i},\bm I_q\right)\,\middle\|\,\mathscr N\left(\bm A^{(\ell)}\bbm x_{k,i},\bm I_q\right)\right)\notag\\
			&= \frac12\sum_{i=1}^{n_k}\left\|\left(\bm A^{(j)}-\bm A^{(\ell)}\right)\bbm x_{k,i}\right\|_2^2.
		\end{align}
		Since the marginal distribution of $\mathcal X_k$ is the same under all hypotheses, the chain rule for KL divergence gives the expectation of the conditional KL divergence. Moreover, $\mathbb E\bbm x_{k,i}\bbm x_{k,i}^\top=\bm I_p$ under the isotropic Gaussian design. Consequently,
		\begin{align}\label{eq:single-joint-kl}
			\mathrm{KL}\left(\mathbb P_k^{(j)}\,\|\,\mathbb P_k^{(\ell)}\right) &= \mathbb E_{\mathcal X_k}\mathrm{KL}\left(\mathbb Q_{Y,k}^{(j)}(\mathcal Y_k\mid\mathcal X_k)\,\|\,\mathbb Q_{Y,k}^{(\ell)}(\mathcal Y_k\mid\mathcal X_k)\right)\notag\\
			&= \frac12\sum_{i=1}^{n_k}\mathbb E\left[\left\|\left(\bm A^{(j)}-\bm A^{(\ell)}\right)\bbm x_{k,i}\right\|_2^2\right]\notag\\
			&= \frac12\sum_{i=1}^{n_k}\tr\left[\left(\bm A^{(j)}-\bm A^{(\ell)}\right)\mathbb E\left(\bbm x_{k,i}\bbm x_{k,i}^\top\right)\left(\bm A^{(j)}-\bm A^{(\ell)}\right)^\top\right]\notag\\
			&= \frac{n_k}{2}\left\|\bm A^{(j)}-\bm A^{(\ell)}\right\|_\F^2\notag\\
			&= \frac{n_k}{2}\left\|\cm A_0^{(j)}-\cm A_0^{(\ell)}\right\|_\F^2.
		\end{align}
		Applying Fano's inequality \citep[Theorem~15.2]{wainwright2019high} directly to the joint laws $\{\mathbb P_k^{(j)}\}_{j=1}^M$ yields
		\begin{align}\label{eq:single-joint-fano}
			\inf_{\psi}\mathbb P\left(\psi\bigl(\{(\cm X_{k,i},\cm Y_{k,i})\}_{i=1}^{n_k}\bigr)\neq J\right)
			\geq 1-\frac{\max_{j\neq\ell}\mathrm{KL}\left(\mathbb P_k^{(j)}\,\|\,\mathbb P_k^{(\ell)}\right)+\log2}{\log M}.
		\end{align}
		If the packing set additionally satisfies $\max_{j\neq\ell}\|\cm A_0^{(j)}-\cm A_0^{(\ell)}\|_\F\leq c\delta$, then \eqref{eq:single-joint-kl} implies
		\begin{align}\label{eq:single-joint-kl-upper-delta}
			\max_{j\neq\ell}\mathrm{KL}\left(\mathbb P_k^{(j)}\,\|\,\mathbb P_k^{(\ell)}\right)\leq Cn_k\delta^2.
		\end{align}

		Consequently, combining \eqref{eq:single-estimation-testing-reduction}, \eqref{eq:single-joint-fano}, and \eqref{eq:single-joint-kl-upper-delta}, we have, if the packing set additionally satisfies $\max_{j\neq\ell}\|\cm A_0^{(j)}-\cm A_0^{(\ell)}\|_\F\leq c\delta$,
		\begin{align}\label{eq:single-fano-lower-delta}
			\inf_{\widehat{\boldsymbol{\mathscr A}}_{0,k}}\sup_{\boldsymbol{\mathscr A}_0^*\in\mathcal F_k^{\mathrm{single}}(\bbm r)}\mathbb E_{\boldsymbol{\mathscr A}_0^*}\left[\|\widehat{\cm A}_{0,k}-\cm A_0^*\|_\F\right] \geq\delta\left(1-\frac{C n_k\delta^2+\log 2}{\log M}\right).
		\end{align}
		We now construct the packing set by treating the core variation and the factor-space variations separately. Each construction gives a valid packing subset of $\mathcal F_k^{\mathrm{single}}(\bbm r)$, and hence gives a separate lower bound through \eqref{eq:single-fano-lower-delta}. Throughout the construction, for $d\geq r$, denote the Stiefel manifold of $d\times r$ matrices with orthonormal columns by
		\[
		\mathbb O_{d,r}:=\{\bm U\in\mathbb R^{d\times r}:\bm U^\top\bm U=\bm I_r\}.
		\]

		\textbf{\emph{Construction 1: core variation.}}
		For each $s\in[d+m]$, pick $\bm U_s\in\mathbb O_{d_s,r_s}$ with $d_s = q_s$ for $s\in[m]$ and $d_s = p_s$ for $s\in[d+m] \setminus [m]$. Given any $\delta>0$, similar to the standard packing arguments for the core tensor space as in the proof of \citet[Theorem~4]{luo2024tensor}, we can construct a set of Tucker-rank cores $\{\cm S^{(1)},\cdots,\cm S^{(M_0)}\} \subset \mathbb R^{r_1\times\cdots\times r_{d+m}}$ with cardinality $M_0\geq\exp(c_0\prod_{s=1}^{d+m}r_s)$ such that: (1) $\|\cm S^{(i)}\|_\F=\varepsilon_{\mathrm{core}}$ for $i\in[M_0]$, (2) $\|\cm S^{(i_1)}-\cm S^{(i_2)}\|_\F \geq c\varepsilon_{\mathrm{core}}$ for all $i_1, i_2 \in [M_0]$ and $i_1\neq i_2$. Let $\cm A_{\mathrm{core}}^{(i)} = \cm S^{(i)}\times_{s=1}^{d+m}\bm U_s$ for $i \in[M_0]$. Since all $\bm U_s$ have orthonormal columns, multilinear multiplication preserves the Frobenius norm. Therefore, for this set, we have
		\[
		\begin{aligned}
			\max_{i_1\neq i_2}\|\cm A_{\mathrm{core}}^{(i_1)}-\cm A_{\mathrm{core}}^{(i_2)}\|_\F &\leq \max_{i_1\neq i_2} (\|\cm A_{\mathrm{core}}^{(i_1)}\|_\F + \|\cm A_{\mathrm{core}}^{(i_2)}\|_\F) = 2\varepsilon_{\mathrm{core}},\\
			\min_{i_1\neq i_2}\|\cm A_{\mathrm{core}}^{(i_1)}-\cm A_{\mathrm{core}}^{(i_2)}\|_\F &= \min_{i_1\neq i_2}\|(\cm S^{(i_1)} - \cm S^{(i_2)}) \times_{s=1}^{d+m}\bm U_s\|_\F \geq c\varepsilon_{\mathrm{core}}.
		\end{aligned}
		\]
		Applying \eqref{eq:single-fano-lower-delta} with $M=M_0$ and choosing $\delta = \varepsilon_{\mathrm{core}} = c_1\sqrt{\prod_{s=1}^{d+m}r_s/n_k}$ gives
		\begin{align}\label{eq:single-core-lower-bound}
			\inf_{\widehat{\boldsymbol{\mathscr A}}_{0,k}}\sup_{\boldsymbol{\mathscr A}_0^*\in\mathcal F_k^{\mathrm{single}}(\bbm r)}\mathbb E_{\boldsymbol{\mathscr A}_0^*}\left[\|\widehat{\cm A}_{0,k}-\cm A_0^*\|_\F\right]\gtrsim\sqrt{\frac{\prod_{s=1}^{d+m}r_s}{n_k}}.
		\end{align}

		\textbf{\emph{Construction 2: factor-space variation.}} This construction follows the same idea as the factor-space packing arguments in the proof of \citet[Theorem~4]{luo2024tensor}. Let $\cm S\in\mathbb R^{r_1\times\cdots\times r_{d+m}}$ be a fixed full Tucker-rank core tensor such that, for each mode $s\in[d+m]$,
		\[
			c\alpha \leq \sigma_{r_s}(\cm S_{[s]}) \leq \sigma_1(\cm S_{[s]}) \leq C\alpha,
		\]
		for some constants $c,C>0$ and $\alpha>0$. For each $s\in[d+m]$, fix $\bm U_s\in\mathbb O_{d_s,r_s}$ with $d_s = q_s$ for $s\in[m]$ and $d_s = p_s$ for $s\in[d+m] \setminus [m]$. We focus on one mode $s\in[d+m]$ and vary only the corresponding factor space, while keeping the core tensor and all other factor matrices fixed.

		Consider the Grassmann manifold $\mathcal G_{d_s,r_s}$ equipped with the projection distance $d(\bm U,\bm V) = \|\bm U\bm U^\top-\bm V\bm V^\top\|_\F$. For a small radius $\varepsilon_s\in(0,1]$, define the local Grassmann ball $B(\bm U_s,\varepsilon_s) = \{\bm U\in\mathbb O_{d_s,r_s}: d(\bm U,\bm U_s)\leq \varepsilon_s\}$.
		By \citet[Lemma~1]{cai2013sparse}, there exist $\bm U_s^{\prime(1)},\cdots,\bm U_s^{\prime(M_s)} \in B(\bm U_s,\varepsilon_s)$ such that $M_s\geq \exp(c_0 r_s(d_s-r_s))$ and $\min_{i_1\neq i_2}d(\bm U_s^{\prime(i_1)},\bm U_s^{\prime(i_2)})\geq c\varepsilon_s$.

		For each $i\in[M_s]$, align $\bm U_s^{\prime(i)}$ with the reference basis $\bm U_s$ by setting $\bm O_i \in \argmin_{\bm O\in\mathbb O_{r_s}}\|\bm U_s^{\prime(i)}\bm O-\bm U_s\|_\F$ and taking $\widetilde{\bm U}_s^{(i)}=\bm U_s^{\prime(i)}\bm O_i$.
		Then $\widetilde{\bm U}_s^{(i)}\in\mathbb O_{d_s,r_s}$ and $\|\widetilde{\bm U}_s^{(i)}-\bm U_s\|_\F \leq C d(\bm U_s^{\prime(i)},\bm U_s) \leq C\varepsilon_s$ by the standard relation between projection distance and the Frobenius distance. We then construct $\cm A_s^{(i)} = \cm S\times_s \widetilde{\bm U}_s^{(i)}\times_{k\neq s}\bm U_k$ for $i\in[M_s]$. Each tensor $\cm A_s^{(i)}$ has Tucker rank at most $\bbm r$. For this set, we first control its diameter. Since all factor matrices other than the $s$-th one are fixed and have orthonormal columns,
		\begin{align*}
			\max_{i_1\neq i_2}\|\cm A_s^{(i_1)}-\cm A_s^{(i_2)}\|_\F &= \max_{i_1\neq i_2}\left\|\cm S\times_s\left(\widetilde{\bm U}_s^{(i_1)}-\widetilde{\bm U}_s^{(i_2)}\right)\times_{k\neq s}\bm U_k\right\|_\F\notag\\
			&\leq\max_{i_1\neq i_2}\Biggl(\left\|\cm S\times_s\left(\widetilde{\bm U}_s^{(i_1)}-\bm U_s\right)\times_{k\neq s}\bm U_k\right\|_\F + \left\|\cm S\times_s\left(\widetilde{\bm U}_s^{(i_2)}-\bm U_s\right)\times_{k\neq s}\bm U_k\right\|_\F\Biggr)\\
			&\leq 2\sigma_1(\cm S_{[s]})\max_{i\in[M_s]}\|\widetilde{\bm U}_s^{(i)}-\bm U_s\|_\F \leq C\alpha\varepsilon_s.\notag
		\end{align*}
		Next, we control the separation. For any $i_1\neq i_2$, since $\bm I_{r_s}\in\mathbb O_{r_s}$ is a feasible rotation, we have
		\begin{align*}
			\min_{i_1\neq i_2}\|\cm A_s^{(i_1)}-\cm A_s^{(i_2)}\|_\F&= \min_{i_1\neq i_2}\left\|\cm S\times_s\left(\widetilde{\bm U}_s^{(i_1)} - \widetilde{\bm U}_s^{(i_2)}\right)\times_{k\neq s}\bm U_k\right\|_\F\\
			&\geq \min_{i_1\neq i_2}\min_{\bm O\in\mathbb O_{r_s}}\left\|\cm S\times_s\left(\widetilde{\bm U}_s^{(i_1)} - \widetilde{\bm U}_s^{(i_2)}\bm O\right)\times_{k\neq s}\bm U_k\right\|_\F\\
			&\geq \sigma_{r_s}(\cm S_{[s]})\min_{i_1\neq i_2}\min_{\bm O\in\mathbb O_{r_s}}\left\|\widetilde{\bm U}_s^{(i_1)} - \widetilde{\bm U}_s^{(i_2)}\bm O\right\|_\F\\
			&\overset{(a)}{\geq}\frac{\sigma_{r_s}(\cm S_{[s]})}{\sqrt 2}\min_{i_1\neq i_2}d\left(\widetilde{\bm U}_s^{(i_1)},\widetilde{\bm U}_s^{(i_2)}\right)\\
			&= \frac{\sigma_{r_s}(\cm S_{[s]})}{\sqrt 2}\min_{i_1\neq i_2}d\left(\bm U_s^{\prime(i_1)},\bm U_s^{\prime(i_2)}\right)\\
			&\geq c\alpha\varepsilon_s.
		\end{align*}
		Here $(a)$ is because $\min_{\bm O\in\mathbb O_{r_s}}\|\widetilde{\bm U}_s^{(i_1)} - \widetilde{\bm U}_s^{(i_2)}\bm O\|_\F \geq d(\widetilde{\bm U}_s^{(i_1)},\widetilde{\bm U}_s^{(i_2)})/\sqrt{2}$ by \citet[Lemma~1]{cai2018rate}. Applying \eqref{eq:single-fano-lower-delta} with $M=M_s$ and choosing $\delta = \alpha\varepsilon_s = c_1\sqrt{r_s(d_s-r_s)/n_k}$ such that $\varepsilon_s=\delta_s/\alpha\leq 1$ gives
		\begin{align}\label{eq:single-factor-lower-bound-mode-s}
			\inf_{\widehat{\boldsymbol{\mathscr A}}_{0,k}}\sup_{\boldsymbol{\mathscr A}_0^*\in\mathcal F_k^{\mathrm{single}}(\bbm r)}\mathbb E_{\boldsymbol{\mathscr A}_0^*}\left[\|\widehat{\cm A}_{0,k}-\cm A_0^*\|_\F\right] \gtrsim \sqrt{\frac{r_s(d_s-r_s)}{n_k}}.
		\end{align}
		Since $s\in[d+m]$ is arbitrary, \eqref{eq:single-factor-lower-bound-mode-s} holds for every mode. Combining the core lower bound \eqref{eq:single-core-lower-bound} with the factor-space lower bounds over all modes yields
		\[
			\inf_{\widehat{\boldsymbol{\mathscr A}}_{0,k}}\sup_{\boldsymbol{\mathscr A}_0^*\in\mathcal F_k^{\mathrm{single}}(\bbm r)}\mathbb E_{\boldsymbol{\mathscr A}_0^*}\left[\|\widehat{\cm A}_{0,k}-\cm A_0^*\|_\F\right]\gtrsim \sqrt{\frac{\prod_{s=1}^{d+m}r_s+\sum_{s=1}^{d+m}r_s(d_s-r_s)}{n_k}} = \sqrt{\frac{df_{\bbm r}}{n_k}}. \qedhere
		\]
	\end{proof}

	\begin{theorem}\label{thm:federated-minimax-lower}
		Define the federated common-component parameter class $\mathcal F^{\mathrm{fed}}(\bbm r) = \left\{\cm A_0:\tucrank(\cm A_0)\leq \bbm r\right\}$. Then,
		\[
			\inf_{\widehat{\boldsymbol{\mathscr A}}_{0}}\sup_{\boldsymbol{\mathscr A}_0^*\in\mathcal F^{\mathrm{fed}}(\bbm r)}\mathbb E_{\boldsymbol{\mathscr A}_0^*}\left[\|\widehat{\cm A}_{0}-\cm A_0^*\|_\F\right] \gtrsim \sqrt{\frac{df_{\bbm r}}{n}}.
		\]
	\end{theorem}

	\begin{proof}[\textbf{Proof of Theorem \ref{thm:federated-minimax-lower}}]
		It suffices to prove the lower bound over the least favorable submodel $\cm B_k^*=\bm 0$ for all $k\in[K]$. Under this submodel, all clients share the common low Tucker-rank tensor $\cm A_0^*$, and the models reduce to $\cm Y_{k,i}=\langle \cm A_0^*,\cm X_{k,i}\rangle+\cm E_{k,i}$ for $k\in[K]$ and $i\in[n_k]$. Thus, the federated problem reduces to estimating the common tensor $\cm A_0^*$ from $n=\sum_{k=1}^K n_k$ independent samples under the isotropic Gaussian submodel with $\vect(\cm X_{k,i})\sim N(\bm 0,\bm I_p)$ and $\vect(\cm E_{k,i})\sim N(\bm 0,\bm I_q)$.

		We then follow the same estimation-to-testing reduction and Tucker-rank packing construction as in the proof of Theorem~\ref{thm:single-client-minimax-lower}. In particular, for two packing points $\cm A_0^{(j)}$ and $\cm A_0^{(\ell)}$, independence across clients and the calculation in \eqref{eq:single-joint-kl} gives
		\[
			\mathrm{KL}\left(\mathbb P_{\mathrm{fed}}^{(j)}\|\mathbb P_{\mathrm{fed}}^{(\ell)}\right)
			=\sum_{k=1}^K\frac{n_k}{2}\left\|\cm A_0^{(j)}-\cm A_0^{(\ell)}\right\|_\F^2
			=\frac{n}{2}\left\|\cm A_0^{(j)}-\cm A_0^{(\ell)}\right\|_\F^2,
		\]
		where $\mathbb P_{\mathrm{fed}}^{(j)}$ denotes the joint law of all client samples under $\cm A_0^{(j)}$. Thus, the core-variation packing is applied with radius $\varepsilon_{\mathrm{core}} \asymp \sqrt{\prod_{s=1}^{d+m}r_s/n}$, and, for each mode $s\in[d+m]$, the corresponding factor-space packing is applied with radius $\varepsilon_s \asymp \sqrt{r_s(d_s-r_s)/n}$. Combining the resulting core and factor-space lower bounds completes the proof.
	\end{proof}

	\begin{theorem}\label{thm:single-client-sparse-deviation-minimax-lower}
		Fix a client $k\in[K]$ and define
		\[
			\mathbb B_\nu(\zeta)=\left\{\cm B\in\mathbb B_\nu:\|(\cm B)_{[S_X]}\|_{\op}\leq\zeta\right\}.
		\]
		Let $\mathfrak s_{\nu,k}=\mathds 1_{\{\nu=0\}}\sqrt{s_0\log(pq/s_0)/n_k}+\mathds 1_{\{0<\nu<1\}}\{s_\nu(\log(pq)/n_k)^{1-\nu/2}\}^{1/2}$. Assume $\mathfrak s_{\nu,k}\leq\zeta$. Then,
		\[
			\inf_{\widehat{\boldsymbol{\mathscr B}}_{k}}\sup_{\boldsymbol{\mathscr B}_k^*\in\mathbb B_\nu(\zeta)}\mathbb E_{\boldsymbol{\mathscr B}_k^*}\left[\|\widehat{\cm B}_{k}-\cm B_k^*\|_\F\right]\gtrsim \mathfrak s_{\nu,k}
		\]
	\end{theorem}

	\begin{proof}[\textbf{Proof of Theorem~\ref{thm:single-client-sparse-deviation-minimax-lower}}]
		It suffices to prove the lower bound over the submodel in which $\cm A_0^*=\bm 0$. Since this submodel is contained in the full single-client parameter space, any minimax lower bound over it is also a valid lower bound for the original problem. Under this submodel, the response is generated from $\cm Y_{k,i}=\langle \cm B_k^*,\cm X_{k,i}\rangle+\cm E_{k,i}$ for $i\in[n_k]$, and the task reduces to estimating the sparse or weakly sparse client-specific deviation tensor $\cm B_k^*$.

		\noindent
		\textbf{\emph{Part I: exact sparse case with $\nu=0$.}}
		Fix an integer $s_0\in[pq]$ and define
		\[
			\mathcal H = \left\{\bbm z\in\{-1,0,1\}^{pq}:\|\bbm z\|_0=s_0\right\}.
		\]
		For any $\bbm z,\bbm z'\in\mathcal H$, let $\rho_H(\bbm z,\bbm z') = \left|\left\{j\in[pq]:z_j\neq z'_j\right\}\right|$ denote their Hamming distance. By \citet[Lemma~5]{raskutti2011minimax}, there exists a subset $\widetilde{\mathcal H}\subset\mathcal H$ with cardinality $|\widetilde{\mathcal H}| \geq \exp\{c_0s_0\log(pq/s_0)\}$ such that $\rho_H(\bbm z,\bbm z')\geq s_0/2$ for all $\bbm z\neq\bbm z'\in\widetilde{\mathcal H}$. Let $M=|\widetilde{\mathcal H}|$ and write $\widetilde{\mathcal H}=\{\bbm z^{(1)},\cdots,\bbm z^{(M)}\}$. For a radius $0<\delta\leq\zeta$ to be chosen later, define $\vect(\cm B^{(j)})=\bbm b^{(j)}\in\mathbb R^{pq}$ with $\bbm b^{(j)} = \delta\bbm z^{(j)}/\sqrt{s_0}$ for $j\in[M]$. Then $\cm B^{(j)}\in\mathbb B_0$ and
		\[
			\|(\cm B^{(j)})_{[S_X]}\|_{\op}
			\leq\|(\cm B^{(j)})_{[S_X]}\|_\F
			=\|\cm B^{(j)}\|_\F
			=\delta
			\leq\zeta.
		\]
		Hence, every packing point $\cm B^{(j)}$ belongs to $\mathbb B_0(\zeta)$. Moreover, for all $j\neq\ell$,
		\[
			\begin{aligned}
			\|\cm B^{(j)}-\cm B^{(\ell)}\|_\F^2 &= \|\bbm b^{(j)}-\bbm b^{(\ell)}\|_2^2 = \frac{\delta^2}{s_0}\|\bbm z^{(j)}-\bbm z^{(\ell)}\|_2^2 \leq \frac{\delta^2}{s_0}4s_0 = 4\delta^2,\\
			\|\cm B^{(j)}-\cm B^{(\ell)}\|_\F^2 &= \|\bbm b^{(j)}-\bbm b^{(\ell)}\|_2^2 = \frac{\delta^2}{s_0}\|\bbm z^{(j)}-\bbm z^{(\ell)}\|_2^2 \geq \frac{\delta^2}{s_0}\rho_H(\bbm z^{(j)},\bbm z^{(\ell)}) \geq \frac{\delta^2}{2}.
			\end{aligned}
		\]

		Let $\mathcal X_k=\{\cm X_{k,i}\}_{i=1}^{n_k}$ and $\mathcal Y_k=\{\cm Y_{k,i}\}_{i=1}^{n_k}$. For each $j\in[M]$, let $\mathbb P_{B,k}^{(j)}$ denote the joint law of $(\mathcal X_k,\mathcal Y_k)$ under $\cm B_k^*=\cm B^{(j)}$. Applying the same estimation-to-testing reduction and Fano's inequality as in the proof of Theorem~\ref{thm:single-client-minimax-lower} to the packing set $\{\cm B^{(1)},\cdots,\cm B^{(M)}\}$ gives
		\begin{align}\label{eq:sparse-l0-fano-lower-delta}
			\inf_{\widehat{\boldsymbol{\mathscr B}}_{k}}\sup_{\boldsymbol{\mathscr B}_k^*\in\mathbb B_0(\zeta)}\mathbb E_{\boldsymbol{\mathscr B}_k^*}\left[\|\widehat{\cm B}_{k}-\cm B_k^*\|_\F\right] \geq c\delta\left(1-\frac{\max_{j\neq\ell}\mathrm{KL}\left(\mathbb P_{B,k}^{(j)}\|\mathbb P_{B,k}^{(\ell)}\right)+\log2}{\log M}\right).
		\end{align}
		Let $\mathbb Q_{B,k}^{(j)}(\mathcal Y_k\mid\mathcal X_k)$ denote the corresponding conditional law of $\mathcal Y_k$ given $\mathcal X_k$. Write $\bbm x_{k,i}=\vect(\cm X_{k,i})$, $\bbm y_{k,i}=\vect(\cm Y_{k,i})$, and $\bm B^{(j)}=(\cm B^{(j)})_{[S_X]}$. Conditional on $\bbm x_{k,i}$, the response satisfies $\bbm y_{k,i}\sim \mathscr N(\bm B^{(j)}\bbm x_{k,i},\bm I_q)$ under the $j$-th hypothesis. Hence, following the same arguments in the proof of Theorem \ref{thm:single-client-minimax-lower}, by conditional independence across observations and the KL formula for Gaussian distributions with a common identity covariance, we have, for any $j\neq\ell$,
		\begin{align}\label{eq:sparse-l0-conditional-kl}
			\mathrm{KL}\left(\mathbb Q_{B,k}^{(j)}(\mathcal Y_k\mid\mathcal X_k)\|\mathbb Q_{B,k}^{(\ell)}(\mathcal Y_k\mid\mathcal X_k)\right) &=\sum_{i=1}^{n_k}\mathrm{KL}\left(\mathscr N\left(\bm B^{(j)}\bbm x_{k,i},\bm I_q\right)\,\middle\|\,\mathscr N\left(\bm B^{(\ell)}\bbm x_{k,i},\bm I_q\right)\right)\notag\\
			&= \frac12\sum_{i=1}^{n_k}\left\|\left(\bm B^{(j)}-\bm B^{(\ell)}\right)\bbm x_{k,i}\right\|_2^2.
		\end{align}
		Since the marginal law of $\mathcal X_k$ is identical under all hypotheses, then by the chain rule for KL divergence and $\mathbb E(\bbm x_{k,i}\bbm x_{k,i}^\top)=\bm I_p$, we have
		\begin{align}\label{eq:sparse-l0-joint-kl-upper-delta}
			\mathrm{KL}\left(\mathbb P_{B,k}^{(j)}\|\mathbb P_{B,k}^{(\ell)}\right)
			&=\mathbb E_{\mathcal X_k}\mathrm{KL}\left(\mathbb Q_{B,k}^{(j)}(\mathcal Y_k\mid\mathcal X_k)\|\mathbb Q_{B,k}^{(\ell)}(\mathcal Y_k\mid\mathcal X_k)\right)\notag\\
			&=\frac12\sum_{i=1}^{n_k}\mathbb E\left[\left\|\left(\bm B^{(j)}-\bm B^{(\ell)}\right)\bbm x_{k,i}\right\|_2^2\right]\notag\\
			&=\frac12\sum_{i=1}^{n_k}\tr\left[\left(\bm B^{(j)}-\bm B^{(\ell)}\right)\mathbb E\left(\bbm x_{k,i}\bbm x_{k,i}^\top\right)\left(\bm B^{(j)}-\bm B^{(\ell)}\right)^\top\right]\notag\\
			&=\frac{n_k}{2}\left\|\bm B^{(j)}-\bm B^{(\ell)}\right\|_\F^2\notag\\
			&=\frac{n_k}{2}\left\|\cm B^{(j)}-\cm B^{(\ell)}\right\|_\F^2\notag\\
			&\leq Cn_k\delta^2.
		\end{align}
		Combining \eqref{eq:sparse-l0-fano-lower-delta} and \eqref{eq:sparse-l0-joint-kl-upper-delta}, and using $\log M\geq c_0s_0\log(pq/s_0)$, we have
		\begin{align}\label{eq:sparse-l0-final-fano}
			\inf_{\widehat{\boldsymbol{\mathscr B}}_{k}}\sup_{\boldsymbol{\mathscr B}_k^*\in\mathbb B_0(\zeta)}\mathbb E_{\boldsymbol{\mathscr B}_k^*}\left[\|\widehat{\cm B}_{k}-\cm B_k^*\|_\F\right] \geq c\delta\left(1-\frac{C n_k\delta^2+\log2}{c_0s_0\log(pq/s_0)}\right).
		\end{align}
		Choosing $\delta^2 = c_1s_0\log(pq/s_0)/n_k$ with $0<c_1\leq1$ sufficiently small makes the bracket in \eqref{eq:sparse-l0-final-fano} bounded below by a positive universal constant. The assumption $\mathfrak s_{0,k}\leq\zeta$ also ensures $\delta=\sqrt{c_1}\mathfrak s_{0,k}\leq\zeta$, as required above. Therefore,
		\[
			\inf_{\widehat{\boldsymbol{\mathscr B}}_{k}}\sup_{\boldsymbol{\mathscr B}_k^*\in\mathbb B_0(\zeta)}\mathbb E_{\boldsymbol{\mathscr B}_k^*}\left[\|\widehat{\cm B}_{k}-\cm B_k^*\|_\F\right]\gtrsim \delta \asymp \sqrt{\frac{s_0\log(pq/s_0)}{n_k}}.
		\]

		\noindent
		\textbf{\emph{Part II: weak sparse case with $\nu\in(0,1)$.}}
		We use the localized metric-entropy construction for $\ell_\nu$ balls under the Euclidean metric underlying \citet[Lemma~3]{raskutti2011minimax}. Specifically, for any $\delta>0$ in the admissible range, the packing set $\mathcal B_\delta = \{\cm B^{(1)},\cdots,\cm B^{(M(\delta))}\}\subseteq \mathbb B_\nu$ can be chosen such that
		\begin{align}\label{eq:lq-packing-entropy-lower}
			\min_{j\neq\ell}\|\cm B^{(j)}-\cm B^{(\ell)}\|_\F\geq 2\delta,
			\
			\max_{j\in[M(\delta)]}\|\cm B^{(j)}\|_\F\leq2\delta,
			\ \text{and} \ 
			\log |M(\delta)| \geq cs_\nu^{\frac{2}{2-\nu}}\delta^{-\frac{2\nu}{2-\nu}}\log(pq).
		\end{align}
		Whenever $2\delta\leq\zeta$, every packing point satisfies $\|(\cm B^{(j)})_{[S_X]}\|_{\op}\leq\|\cm B^{(j)}\|_\F \leq2\delta\leq\zeta$, and hence $\mathcal B_\delta\subseteq\mathbb B_\nu(\zeta)$.

		We next reduce estimation to multiple hypothesis testing. Let $J$ be a random variable uniformly distributed over $[M(\delta)]$. Conditional on $J=j$, let the sample $\{(\cm X_{k,i},\cm Y_{k,i})\}_{i=1}^{n_k}$ be generated according to $\cm B_k^*=\cm B^{(j)}$, and let $\mathbb P$ denote the joint law of $\left(\{(\cm X_{k,i},\cm Y_{k,i})\}_{i=1}^{n_k},J\right)$. By the same estimation-to-testing reduction as in \eqref{eq:single-estimation-testing-reduction},
		\begin{align}\label{eq:sparse-estimation-testing-reduction}
			\inf_{\widehat{\boldsymbol{\mathscr B}}_{k}}\sup_{\boldsymbol{\mathscr B}_k^*\in\mathbb B_\nu(\zeta)}\mathbb E_{\boldsymbol{\mathscr B}_k^*}\left[\|\widehat{\cm B}_{k}-\cm B_k^*\|_\F\right] \geq \delta\inf_{\psi}\mathbb P\left(\psi\bigl(\{(\cm X_{k,i},\cm Y_{k,i})\}_{i=1}^{n_k}\bigr)\neq J\right),
		\end{align}
		where the infimum is over all testing rules $\psi$ based on the observed sample.

		Let $\mathcal X_k=\{\cm X_{k,i}\}_{i=1}^{n_k}$ and $\mathcal Y_k=\{\cm Y_{k,i}\}_{i=1}^{n_k}$. Since $J$ is independent of $\mathcal X_k$, Fano's inequality \citep[Theorem~15.2]{wainwright2019high} applied directly to the joint observation $(\mathcal X_k,\mathcal Y_k)$ gives
		\begin{align}\label{eq:sparse-joint-fano}
			\inf_{\psi}\mathbb P\left(\psi\bigl(\{(\cm X_{k,i},\cm Y_{k,i})\}_{i=1}^{n_k}\bigr)\neq J\right) \geq 1-\frac{I(J;\mathcal X_k,\mathcal Y_k)+\log2}{\log |M(\delta)|},
		\end{align}
		where $I(J;\mathcal X_k,\mathcal Y_k)$ is the mutual information between the random hypothesis index $J$ and the joint observation $(\mathcal X_k,\mathcal Y_k)$. It remains to control $I(J;\mathcal X_k,\mathcal Y_k)$. Let $\varepsilon>0$ be another radius to be chosen later, and let $\mathcal N_\varepsilon$ be an $\varepsilon$-net of $\mathbb B_\nu$ under the Frobenius norm. Following \citet[Lemma~3]{raskutti2011minimax}, the corresponding upper metric-entropy bound for $\ell_\nu$ balls under the Euclidean metric gives
		\begin{align}\label{eq:lq-net-entropy-upper}
			\log |\mathcal N_\varepsilon| \leq C s_\nu^{\frac{2}{2-\nu}}\varepsilon^{-\frac{2\nu}{2-\nu}}\log(pq).
		\end{align}
		For each $j\in[M(\delta)]$, choose $\widetilde{\cm B}^{(j)}\in\mathcal N_\varepsilon$ such that $\|\cm B^{(j)}-\widetilde{\cm B}^{(j)}\|_\F\leq\varepsilon$. Let $\mathbb P_{B,k}^{(j)}$ and $\widetilde{\mathbb P}_{B,k}^{(j)}$ denote the joint laws of $(\mathcal X_k,\mathcal Y_k)$ induced by $\cm B^{(j)}$ and $\widetilde{\cm B}^{(j)}$, respectively, and let $\mathbb Q_{B,k}^{(j)}(\mathcal Y_k\mid\mathcal X_k)$ and $\widetilde{\mathbb Q}_{B,k}^{(j)}(\mathcal Y_k\mid\mathcal X_k)$ denote the corresponding conditional response laws. By the information-radius arguments of \citet[Theorem~1]{yang1999information},
		\begin{align}\label{eq:sparse-joint-mutual-information-net-bound}
			I(J;\mathcal X_k,\mathcal Y_k) \leq \log |\mathcal N_\varepsilon| + \max_{j\in[M(\delta)]}\mathrm{KL}\left(\mathbb P_{B,k}^{(j)}\|\widetilde{\mathbb P}_{B,k}^{(j)}\right).
		\end{align}
		Write $\bbm x_{k,i}=\vect(\cm X_{k,i})$, $\bm B^{(j)}=(\cm B^{(j)})_{[S_X]}$, and $\widetilde{\bm B}^{(j)}=(\widetilde{\cm B}^{(j)})_{[S_X]}$. The marginal law of $\mathcal X_k$ is the same for these two joint distributions. Conditional on $\bbm x_{k,i}$, their response distributions are $\mathscr N(\bm B^{(j)}\bbm x_{k,i},\bm I_q)$ and $\mathscr N(\widetilde{\bm B}^{(j)}\bbm x_{k,i},\bm I_q)$, respectively. Thus, the KL chain rule, the Gaussian KL formula, and the isotropic design covariance give
		\begin{align}\label{eq:sparse-joint-net-kl-bound}
			\mathrm{KL}\left(\mathbb P_{B,k}^{(j)}\|\widetilde{\mathbb P}_{B,k}^{(j)}\right)
			&=\mathbb E_{\mathcal X_k}\mathrm{KL}\left(\mathbb Q_{B,k}^{(j)}(\mathcal Y_k\mid\mathcal X_k)\|\widetilde{\mathbb Q}_{B,k}^{(j)}(\mathcal Y_k\mid\mathcal X_k)\right)\notag\\
			&=\mathbb E_{\mathcal X_k}\sum_{i=1}^{n_k}\mathrm{KL}\left(\mathscr N\left(\bm B^{(j)}\bbm x_{k,i},\bm I_q\right)\,\middle\|\,\mathscr N\left(\widetilde{\bm B}^{(j)}\bbm x_{k,i},\bm I_q\right)\right)\notag\\
			&=\frac12\sum_{i=1}^{n_k}\mathbb E\left[\left\|\left(\bm B^{(j)}-\widetilde{\bm B}^{(j)}\right)\bbm x_{k,i}\right\|_2^2\right]\notag\\
			&=\frac12\sum_{i=1}^{n_k}\tr\left[\left(\bm B^{(j)}-\widetilde{\bm B}^{(j)}\right)\mathbb E\left(\bbm x_{k,i}\bbm x_{k,i}^\top\right)\left(\bm B^{(j)}-\widetilde{\bm B}^{(j)}\right)^\top\right]\notag\\
			&=\frac{n_k}{2}\left\|\bm B^{(j)}-\widetilde{\bm B}^{(j)}\right\|_\F^2=\frac{n_k}{2}\left\|\cm B^{(j)}-\widetilde{\cm B}^{(j)}\right\|_\F^2\leq\frac{n_k}{2}\varepsilon^2.
		\end{align}
		Combining \eqref{eq:lq-net-entropy-upper}, \eqref{eq:sparse-joint-mutual-information-net-bound}, and \eqref{eq:sparse-joint-net-kl-bound} yields
		\begin{align}\label{eq:sparse-joint-mutual-information-bound}
			I(J;\mathcal X_k,\mathcal Y_k) \leq C s_\nu^{\frac{2}{2-\nu}}\varepsilon^{-\frac{2\nu}{2-\nu}}\log(pq) + C n_k\varepsilon^2.
		\end{align}
		Consequently, combining \eqref{eq:sparse-estimation-testing-reduction}, \eqref{eq:sparse-joint-fano}, and \eqref{eq:sparse-joint-mutual-information-bound}, we have
		\begin{align}\label{eq:sparse-fano-lower-delta}
			\inf_{\widehat{\boldsymbol{\mathscr B}}_{k}}\sup_{\boldsymbol{\mathscr B}_k^*\in\mathbb B_\nu(\zeta)}\mathbb E_{\boldsymbol{\mathscr B}_k^*}\left[\|\widehat{\cm B}_{k}-\cm B_k^*\|_\F\right]
			&\geq \delta\left(1-\frac{C s_\nu^{\frac{2}{2-\nu}}\varepsilon^{-\frac{2\nu}{2-\nu}}\log(pq)+C n_k\varepsilon^2+\log2}{\log |M(\delta)|}\right).
		\end{align}
		We now choose the covering and packing radii so that the bracket in \eqref{eq:sparse-fano-lower-delta} is bounded below by a positive universal constant. First, choose the covering radius $\varepsilon$ such that
		\begin{align}\label{eq:lq-covering-radius-choice}
			C n_k\varepsilon^2 \leq C_0s_\nu^{\frac{2}{2-\nu}}\varepsilon^{-\frac{2\nu}{2-\nu}}\log(pq).
		\end{align}
		Equivalently, we may take $\varepsilon=c_\varepsilon\mathfrak s_{\nu,k}$ with $\mathfrak s_{\nu,k}=\mathds 1_{\{\nu=0\}}\sqrt{s_0\log(pq/s_0)/n_k}+\mathds 1_{\{0<\nu<1\}}\{s_\nu(\log(pq)/n_k)^{1-\nu/2}\}^{1/2}$ for a small constant $c_\varepsilon\in(0,1]$. With this choice, $C n_k\varepsilon^2$ is of the same order as, and can be made a sufficiently small constant multiple of, the covering entropy term in \eqref{eq:lq-net-entropy-upper}. Moreover, $\mathfrak s_{\nu,k}\leq\zeta$ implies $\varepsilon\leq\zeta$. We next verify that there exists a packing radius $\delta$ satisfying the required packing-covering comparison. Since
		\[
			\log |M(\delta)|\geq c s_\nu^{\frac{2}{2-\nu}}\delta^{-\frac{2\nu}{2-\nu}}\log(pq)
			\ \text{and} \
			\log |\mathcal N_\varepsilon| \leq C s_\nu^{\frac{2}{2-\nu}}\varepsilon^{-\frac{2\nu}{2-\nu}}\log(pq),
		\]
		it suffices to choose $\delta$ such that $C s_\nu^{\frac{2}{2-\nu}}\varepsilon^{-\frac{2\nu}{2-\nu}}\log(pq) \leq cs_\nu^{\frac{2}{2-\nu}}\delta^{-\frac{2\nu}{2-\nu}}\log(pq)/4$. Equivalently, it suffices to take $\delta \leq(c/4C)^{\frac{2-\nu}{2\nu}}\varepsilon$. Thus, choosing $\delta=a\varepsilon$ with
		\[
			0<a\leq\frac12\wedge\left(\frac{c}{4C}\right)^{\frac{2-\nu}{2\nu}}
		\]
		guarantees both $2\delta\leq\varepsilon\leq\zeta$, so that $\mathcal B_\delta\subseteq\mathbb B_\nu(\zeta)$, and
		\begin{align}\label{eq:lq-packing-covering-comparison}
			\log |M(\delta)| \geq 4\log |\mathcal N_\varepsilon|.
		\end{align}
		Such a choice exists because $a$ is a fixed positive constant independent of $n_k$, $pq$, and $s_\nu$, and the admissible packing range contains sufficiently small radii. With the above choices of $\varepsilon$ and $\delta$, the bracket in \eqref{eq:sparse-fano-lower-delta} is bounded below by a positive universal constant. Indeed, \eqref{eq:lq-covering-radius-choice} and \eqref{eq:lq-packing-covering-comparison} imply
		\[
			\frac{C s_\nu^{\frac{2}{2-\nu}}\varepsilon^{-\frac{2\nu}{2-\nu}}\log(pq) + C n_k\varepsilon^2 + \log2}{c s_\nu^{\frac{2}{2-\nu}}\delta^{-\frac{2\nu}{2-\nu}}\log(pq)} \leq \frac14+\frac14+\frac{\log2}{\log |M(\delta)|}.
		\]
		By taking the packing cardinality $|M(\delta)|$ sufficiently large, the last term is absorbed into the universal constant. Hence the bracket in \eqref{eq:sparse-fano-lower-delta} is bounded below by a positive universal constant, and therefore
		\[
			\inf_{\widehat{\boldsymbol{\mathscr B}}_{k}}\sup_{\boldsymbol{\mathscr B}_k^*\in\mathbb B_\nu(\zeta)}\mathbb E_{\boldsymbol{\mathscr B}_k^*}\left[\|\widehat{\cm B}_{k}-\cm B_k^*\|_\F\right] \gtrsim \delta.
		\]
		Since the above construction gives $\delta\asymp\varepsilon$ and $\varepsilon^2 \asymp s_\nu(\log(pq)/n_k)^{1-\nu/2}$, we have
		\[
			\inf_{\widehat{\boldsymbol{\mathscr B}}_{k}}\sup_{\boldsymbol{\mathscr B}_k^*\in\mathbb B_\nu(\zeta)}\mathbb E_{\boldsymbol{\mathscr B}_k^*}\left[\|\widehat{\cm B}_{k}-\cm B_k^*\|_\F\right] \gtrsim \left\{s_\nu\left(\frac{\log(pq)}{n_k}\right)^{1-\nu/2}\right\}^{1/2}.
		\]
		Combining Part~I and Part~II proves the stated minimax lower bounds for both exact sparse and weakly sparse deviations.
	\end{proof}

\section{Proofs of Primary Lemmas}\label{append:proofs of primary lemmas}
	\renewcommand{\theequation}{D.\arabic{equation}}
	\renewcommand{\theHequation}{D.\arabic{equation}}

	\renewcommand{\thelemma}{D.\arabic{lemma}}
	\renewcommand{\theHlemma}{D.\arabic{lemma}}

	\setcounter{lemma}{0}
	\setcounter{equation}{0}

	\begin{proof}[\textbf{Proof of Lemma~\ref{lem:rsc-condition}}]
		Let $\bbm x_{k,i}=\vect(\cm X_{k,i})\in\mathbb R^p$ and $\bm S_{X,k}=n_k^{-1}\sum_{i=1}^{n_k}\bbm x_{k,i}\bbm x_{k,i}^\top$. Under the matricization convention $\vect(\langle\cm T,\cm X_{k,i}\rangle)=\cm T_{[S_X]}\bbm x_{k,i}$, for any $\cm T$,
		\[
			\frac{1}{n_k}\sum_{i=1}^{n_k}\left\|\langle\cm T,\cm X_{k,i}\rangle\right\|_\F^2 = \tr\left(\cm T_{[S_X]}\bm S_{X,k}\cm T_{[S_X]}^\top\right) \geq \lambda_{\min}(\bm S_{X,k})\|\cm T\|_\F^2.
		\]
		Thus it suffices to show that $\lambda_{\min}(\bm S_{X,k})\geq\lambda_{\mathscr X,k}^{\min}/2$ with high probability.

		By Assumption~\ref{assump:subg-design}, write $\bbm x_{k,i}=\bbm\Sigma_{\mathscr X,k}^{1/2}\bbm\xi_{k,i}$, where $\mathbb E\bbm\xi_{k,i}=\bm0$, $\mathbb E(\bbm\xi_{k,i}\bbm\xi_{k,i}^\top)=\bm I_p$, and the coordinates of $\bbm\xi_{k,i}$ are independent sub-Gaussian random variables with variance proxy bounded by $\sigma_{\mathscr X,k}^2$. For any fixed $\bbm v\in\mathbb S^{p-1}$, define $R_n(\bbm v)=\bbm v^\top\bm S_{X,k}\bbm v=n_k^{-1}\sum_{i=1}^{n_k}(\bbm v^\top\bbm x_{k,i})^2$. Then $\mathbb E R_n(\bbm v)=\bbm v^\top\bbm\Sigma_{\mathscr X,k}\bbm v\geq\lambda_{\mathscr X,k}^{\min}$.
		Moreover, with $\bbm\xi_k=(\bbm\xi_{k,1}^\top,\cdots,\bbm\xi_{k,n_k}^\top)^\top$ and $\bm M_{\bbm v}=\bbm\Sigma_{\mathscr X,k}^{1/2}\bbm v\bbm v^\top\bbm\Sigma_{\mathscr X,k}^{1/2}$, we have
		\[
			R_n(\bbm v)-\mathbb E R_n(\bbm v) = \frac1{n_k}\left(\bbm\xi_k^\top(\bm I_{n_k}\otimes\bm M_{\bbm v})\bbm\xi_k - \mathbb E\bigl[\bbm\xi_k^\top(\bm I_{n_k}\otimes\bm M_{\bbm v})\bbm\xi_k\bigr]\right).
		\]
		Since $\|\bm M_{\bbm v}\|_{\op}\leq\lambda_{\mathscr X,k}^{\max}$ and $\|\bm M_{\bbm v}\|_\F\leq\lambda_{\mathscr X,k}^{\max}$, Lemma~\ref{lem:hanson-wright} gives, for any $t>0$,
		\[
			\mathbb P\left(\left|R_n(\bbm v)-\mathbb E R_n(\bbm v)\right|\geq t\right) \leq 2\exp\left(-Cn_k\min\left\{\frac{t^2}{\sigma_{\mathscr X,k}^4(\lambda_{\mathscr X,k}^{\max})^2},\frac{t}{\sigma_{\mathscr X,k}^2\lambda_{\mathscr X,k}^{\max}}\right\}\right).
		\]

		Let $\mathcal N$ be a $1/4$-net of $\mathbb S^{p-1}$ with $|\mathcal N|\leq9^p$. Taking $t=\lambda_{\mathscr X,k}^{\min}/4$ and applying the union bound over $\mathcal N$, we obtain
		\[
			\mathbb P\left(\sup_{\bbm v\in\mathcal N}\left|\bbm v^\top(\bm S_{X,k}-\bbm\Sigma_{\mathscr X,k})\bbm v\right| \geq \frac14\lambda_{\mathscr X,k}^{\min}\right) \leq 2\exp\left(Cp-Cn_k\min\left\{\sigma_{\mathscr X,k}^{-4}\kappa_{\mathscr X,k}^{-2},\sigma_{\mathscr X,k}^{-2}\kappa_{\mathscr X,k}^{-1}\right\}\right).
		\]
		Under the stated sample size condition, the last probability is bounded by $\exp(-Cp)$. Hence, with probability at least $1-\exp(-Cp)$,
		\[
			\sup_{\bbm v\in\mathcal N}\left|\bbm v^\top(\bm S_{X,k}-\bbm\Sigma_{\mathscr X,k})\bbm v\right| \leq \frac14\lambda_{\mathscr X,k}^{\min}.
		\]

		Let $\bm A_k=\bm S_{X,k}-\bbm\Sigma_{\mathscr X,k}$. By the standard $1/4$-net reduction for symmetric matrices,
		\[
			\|\bm A_k\|_{\op} = \sup_{\bbm v\in\mathbb S^{p-1}}|\bbm v^\top\bm A_k\bbm v| \leq 2\sup_{\bbm v\in\mathcal N}|\bbm v^\top\bm A_k\bbm v|.
		\]
		Indeed, for any $\bbm v\in\mathbb S^{p-1}$, choose $\bar{\bbm v}\in\mathcal N$ with $\|\bbm v-\bar{\bbm v}\|_2\leq1/4$. Then $|\bbm v^\top\bm A_k\bbm v|\leq |\bar{\bbm v}^\top\bm A_k\bar{\bbm v}|+2\|\bbm v-\bar{\bbm v}\|_2\|\bm A_k\|_{\op}\leq \sup_{\bbm u\in\mathcal N}|\bbm u^\top\bm A_k\bbm u|+\|\bm A_k\|_{\op}/2$. Taking the supremum over $\bbm v$ gives the displayed reduction. Consequently,
		\[
			\|\bm S_{X,k}-\bbm\Sigma_{\mathscr X,k}\|_{\op} \leq \frac12\lambda_{\mathscr X,k}^{\min}.
		\]
		By Weyl's inequality, $\lambda_{\min}(\bm S_{X,k})\geq\lambda_{\min}(\bbm\Sigma_{\mathscr X,k})-\|\bm S_{X,k}-\bbm\Sigma_{\mathscr X,k}\|_{\op}\geq\lambda_{\mathscr X,k}^{\min}/2$. Substituting this into the quadratic-form identity at the beginning of the proof yields the desired RSC bound simultaneously for all $\cm T$. This completes the proof.
	\end{proof}

	\begin{proof}[\textbf{Proof of Lemma \ref{lem:deviation-condition-op}}]
		Denote $\mathcal{M}(1;q,p)=\{\bm M\in \mathbb{R}^{q\times p}:\rank(\bm M)=1,\ \|\bm M\|_\F=1\}$. Since $(\cm E_{k,i}\circ\cm X_{k,i})_{[S_X]}=\vect(\cm E_{k,i})\vect^\top(\cm X_{k,i})$, the variational characterization of the operator norm gives
		\begin{align*}
			\left\|\frac{1}{n_k}\sum_{i=1}^{n_k}(\cm E_{k,i}\circ \cm X_{k,i})_{[S_X]}\right\|_\op
			&= \sup_{\bm M \in \mathcal{M}(1;q,p)}\frac{1}{n_k}\sum_{i=1}^{n_k}\left\langle\vect(\cm E_{k,i}),\bm M\vect(\cm X_{k,i})\right\rangle.
		\end{align*}
		
	For an arbitrary matrix $\bm M \in \mathcal{M}(1;q,p)$, set $\bm W=\bm M$, $A_u=\{n_k^{-1}\sum_{i=1}^{n_k}\langle \vect(\cm E_{k,i}), \allowbreak \bm W\vect(\cm X_{k,i}) \rangle \geq u\}$, and $B_v=\{n_k^{-1}\sum_{i=1}^{n_k}\|\bm W\vect(\cm X_{k,i})\|_2^2 \leq v\}$.
		Then, by a standard Chernoff arguments, for any $u>0$, $v>0$, and $\xi>0$, we have
		\begin{align}\label{eq:P-A-con-B}
			\mathbb P\left(A_u \mid B_v\right)
			&= \frac{\mathbb P(A_u\cap B_v)}{\mathbb P(B_v)} \notag\\
			&= \frac{1}{\mathbb P(B_v)}\mathbb P\left(\left\{\exp\left(\xi\sum_{i=1}^{n_k}\langle \vect(\cm E_{k,i}), \bm W\vect(\cm X_{k,i}) \rangle\right)\geq \exp(\xi n_k u)\right\}\cap B_v\right)\notag \\
			&\leq \frac{\exp(-\xi n_k u)}{\mathbb P(B_v)}\mathbb E\left(\exp\left(\xi\sum_{i=1}^{n_k}\langle \vect(\cm E_{k,i}), \bm W\vect(\cm X_{k,i}) \rangle\right)\mathds 1_{B_v}\right).
		\end{align}

		Since $B_v$ is measurable with respect to $\sigma(\cm X_{k,1},\cdots,\cm X_{k,n_k})$, and by Assumption~\ref{assump:subg-noise}, conditional on $\cm X_{k,i}$, the random vector $\vect(\cm E_{k,i})$ is mean-zero sub-Gaussian with covariance proxy $\bbm\Sigma_{\mathscr E,k}$ and parameter $\sigma_{\mathscr E,k}$, we have
		\begin{align}\label{eq:union-expectation}
			&\mathbb E\left(\exp\left(\xi\sum_{i=1}^{n_k}\langle \vect(\cm E_{k,i}), \bm W\vect(\cm X_{k,i}) \rangle\right)\mathds 1_{B_v}\right) \notag\\
			&=\mathbb E\left(\mathds 1_{B_v}\,\mathbb E\left(\exp\left(\xi\sum_{i=1}^{n_k}\langle \vect(\cm E_{k,i}), \bm W\vect(\cm X_{k,i}) \rangle\right)\,\middle|\, \cm X_{k,1},\cdots,\cm X_{k,n_k}\right)\right) \notag\\
			&\leq \mathbb E\left(\mathds 1_{B_v}\,\exp\left(\frac{\xi^2\sigma_{\mathscr E,k}^2}{2}\sum_{i=1}^{n_k}(\bm W\vect(\cm X_{k,i}))^\top\bbm\Sigma_{\mathscr E,k}(\bm W\vect(\cm X_{k,i}))\right)\right) \notag\\
			&\leq \mathbb E\left(\mathds 1_{B_v}\,\exp\left(\frac{\xi^2\sigma_{\mathscr E,k}^2\lambda_{\mathscr E,k}^{\max}}{2}\sum_{i=1}^{n_k}\|\bm W\vect(\cm X_{k,i})\|_2^2\right)\right) \leq \exp\left(\frac{\xi^2\sigma_{\mathscr E,k}^2\lambda_{\mathscr E,k}^{\max}}{2}\,n_k v\right)\mathbb P(B_v).
		\end{align}
			Therefore, combining \eqref{eq:P-A-con-B} with \eqref{eq:union-expectation} and taking $\xi=u/(\sigma_{\mathscr E,k}^2\lambda_{\mathscr E,k}^{\max}v)$ yields
		\[
			\mathbb P\left(A_u \mid B_v\right) \leq \exp\left(-\xi n_k u + \frac{\xi^2\sigma_{\mathscr E,k}^2\lambda_{\mathscr E,k}^{\max}}{2}n_k v\right) \leq \exp\left(-\frac{n_k u^2}{2\sigma_{\mathscr E,k}^2\lambda_{\mathscr E,k}^{\max}\,v}\right).
		\]
		Next, we control the probability with respect to event $B_v$. Note that
		\[
			\frac{1}{n_k}\sum_{i=1}^{n_k}\|\bm W\vect(\cm X_{k,i})\|_2^2
			=
			\frac{1}{n_k}\sum_{i=1}^{n_k}\vect^\top(\cm X_{k,i}) \bm W^\top \bm W \vect(\cm X_{k,i}).
		\]
		By Assumption~\ref{assump:subg-design}, we have $\vect(\cm X_{k,i})=\bbm\Sigma_{\mathscr X,k}^{1/2}\bbm\xi_{k,i}$, where $\bbm\xi_{k,i}$ has $i.i.d.$ mean-zero sub-Gaussian entries with variance proxy $\sigma_{\mathscr X,k}^2$. Hence $n_k^{-1}\sum_{i=1}^{n_k}\|\bm W\vect(\cm X_{k,i})\|_2^2 = n_k^{-1}\sum_{i=1}^{n_k}\bbm\xi_{k,i}^\top \bbm\Sigma_{\mathscr X,k}^{1/2}\bm W^\top\bm W\bbm\Sigma_{\mathscr X,k}^{1/2}\allowbreak\bbm\xi_{k,i} = n_k^{-1}\bbm\xi_k^\top (\bm I_{n_k}\otimes \bm H_k)\bbm\xi_k$ with $\bm H_k = \bbm\Sigma_{\mathscr X,k}^{1/2}\bm W^\top\bm W\bbm\Sigma_{\mathscr X,k}^{1/2}$ and $\bbm\xi_k = (\bbm\xi_{k,1}^\top,\cdots,\bbm\xi_{k,n_k}^\top)^\top$. 
		Consequently, by the Hanson-Wright inequality in Lemma~\ref{lem:hanson-wright}, we have, for any $v>0$,
		\begin{align*}
			&\mathbb P\left(\left|\frac{1}{n_k}\sum_{i=1}^{n_k}\|\bm W\vect(\cm X_{k,i})\|_2^2 - \mathbb E\left[\frac{1}{n_k}\sum_{i=1}^{n_k}\|\bm W\vect(\cm X_{k,i})\|_2^2\right]\right|\geq v\right)\\
			&\quad \leq 2\exp\left(-C\min\left\{\frac{n_k^2 v^2}{\max_j\|\xi_{k,i,j}\|_{\psi_2}^4\|\bm I_{n_k}\otimes \bm H_k\|_\F^2}, \frac{n_k v}{\max_j\|\xi_{k,i,j}\|_{\psi_2}^2\|\bm I_{n_k}\otimes \bm H_k\|_{\op}}\right\}\right).
		\end{align*}
		Since $\max_j\|\xi_{k,i,j}\|_{\psi_2}\lesssim \sigma_{\mathscr X,k}$, $\|\bm I_{n_k}\otimes \bm H_k\|_{\op} = \|\bm H_k\|_{\op} \leq \|\bbm\Sigma_{\mathscr X,k}\|_{\op}\,\|\bm W\|_{\op}^2 \leq \lambda_{\mathscr X,k}^{\max}\|\bm W\|_\F^2$, and $\|\bm I_{n_k}\otimes \bm H_k\|_\F^2 = n_k\|\bm H_k\|_\F^2 \leq n_k\|\bbm\Sigma_{\mathscr X,k}\|_{\op}^2\|\bm W^\top\bm W\|_\F^2 \leq n_k(\lambda_{\mathscr X,k}^{\max})^2\|\bm W\|_\F^4$, we have
		\begin{align*}
			&\mathbb P\left(\left|\frac{1}{n_k}\sum_{i=1}^{n_k}\|\bm W\vect(\cm X_{k,i})\|_2^2 - \mathbb E\left[\frac{1}{n_k}\sum_{i=1}^{n_k}\|\bm W\vect(\cm X_{k,i})\|_2^2\right]\right| \geq v\right)\\
			&\quad \leq 2\exp\left(-Cn_k\min\left\{\frac{v^2}{\sigma_{\mathscr X,k}^4(\lambda_{\mathscr X,k}^{\max})^2\|\bm W\|_\F^4}, \frac{v}{\sigma_{\mathscr X,k}^2\lambda_{\mathscr X,k}^{\max}\|\bm W\|_\F^2}\right\}\right).
		\end{align*}
		In addition, since $\bm W=\bm M$ with $\bm M\in\mathcal M(1;q,p)$, we have $\|\bm W\|_\F=1$ and
		\[
			\mathbb E\left[\frac{1}{n_k}\sum_{i=1}^{n_k}\|\bm W\vect(\cm X_{k,i})\|_2^2\right] = \tr(\bm W^\top\bm W\,\bbm\Sigma_{\mathscr X,k})\leq \lambda_{\mathscr X,k}^{\max}\|\bm W\|_\F^2 = \lambda_{\mathscr X,k}^{\max}.
		\]
		Taking $v = C\sigma_{\mathscr X,k}^2\lambda_{\mathscr X,k}^{\max}$, we have
		\begin{align}\label{eq:P-b-comple}
			\mathbb P(B_v^c) = \mathbb P\left(\frac{1}{n_k}\sum_{i=1}^{n_k}\|\bm W\vect(\cm X_{k,i})\|_2^2 \geq C\sigma_{\mathscr X,k}^2\lambda_{\mathscr X,k}^{\max}\right) \leq 2\exp(-Cn_k).
		\end{align}
		On the other hand, given the choice of $v$, taking $u = C\sigma_{\mathscr X,k}\sigma_{\mathscr E,k}(\lambda_{\mathscr X,k}^{\max}\lambda_{\mathscr E,k}^{\max})^{1/2}\sqrt{(p+q)/n_k}$, we have
		\begin{align}\label{eq:P-a-condi-b}
			\mathbb P\left(A_u\mid B_v\right) \leq \exp\left(-\frac{n_k u^2}{C\,\sigma_{\mathscr X,k}^2\sigma_{\mathscr E,k}^2\lambda_{\mathscr X,k}^{\max}\lambda_{\mathscr E,k}^{\max}}\right) \leq \exp(-C(p+q)).
		\end{align}
			Combining \eqref{eq:P-b-comple} with \eqref{eq:P-a-condi-b} and noting that $n_k\gtrsim p+q$ yields that, for any $\bm M\in\mathcal M(1;q,p)$,
			\begin{align*}
				&\mathbb P\left(\frac{1}{n_k}\sum_{i=1}^{n_k}\left\langle (\cm E_{k,i}\circ \cm X_{k,i})_{[S_X]}, \bm M \right\rangle \geq C\sigma_{\mathscr X,k}\sigma_{\mathscr E,k}(\lambda_{\mathscr X,k}^{\max}\lambda_{\mathscr E,k}^{\max})^{1/2}\sqrt{\frac{p+q}{n_k}}\right)\\ 
				&= \mathbb P(A_u) \leq \mathbb P(A_u\mid B_v)+\mathbb P(B_v^c) \leq \exp(-C(p+q))+2\exp(-Cn_k) = \exp(-C(p+q)).
			\end{align*}
			Applying the same arguments to $-\bm M$ gives the corresponding lower-tail bound and hence the same probability bound for the absolute value of the displayed empirical inner product.

			By the variational characterization of the operator norm,
			\begin{align*}
				\left\|\frac{1}{n_k}\sum_{i=1}^{n_k}(\cm E_{k,i}\circ \cm X_{k,i})_{[S_X]}\right\|_{\op}
				&= \sup_{\substack{\bbm u\in\mathbb S^{q-1}\\\bbm v\in\mathbb S^{p-1}}}\left|\bbm u^\top\left(\frac{1}{n_k}\sum_{i=1}^{n_k}(\cm E_{k,i}\circ \cm X_{k,i})_{[S_X]}\right)\bbm v\right| \\
				&= \sup_{\substack{\bbm u\in\mathbb S^{q-1}\\\bbm v\in\mathbb S^{p-1}}}\left|\left\langle\frac{1}{n_k}\sum_{i=1}^{n_k}(\cm E_{k,i}\circ \cm X_{k,i})_{[S_X]},\bbm u\bbm v^\top\right\rangle\right|.
			\end{align*}
		Since $\rank(\bbm u\bbm v^\top)=1$ and $\|\bbm u\bbm v^\top\|_\F=\|\bbm u\|_2\|\bbm v\|_2=1$, we have $\bbm u\bbm v^\top\in\mathcal M(1;q,p)$.
		Next, let $\mathcal N_q$ and $\mathcal N_p$ be $1/4$-nets of $\mathbb S^{q-1}$ and $\mathbb S^{p-1}$, respectively, such that $|\mathcal N_q|\leq 9^q$ and $|\mathcal N_p|\leq 9^p$. For notational simplicity, set $\bm Z_k=n_k^{-1}\sum_{i=1}^{n_k}(\cm E_{k,i}\circ \cm X_{k,i})_{[S_X]}$. Fix any $\bbm u\in\mathbb S^{q-1}$ and $\bbm v\in\mathbb S^{p-1}$. Choose $\bar{\bbm u}\in\mathcal N_q$ and $\bar{\bbm v}\in\mathcal N_p$ such that $\|\bbm u-\bar{\bbm u}\|_2\leq1/4$ and $\|\bbm v-\bar{\bbm v}\|_2\leq1/4$. Then
	        \begin{align*}
	        	|\bbm u^\top\bm Z_k\bbm v| &\leq |\bar{\bbm u}^{\top}\bm Z_k\bar{\bbm v}| + \|\bbm u-\bar{\bbm u}\|_2\|\bm Z_k\|_{\op}\|\bbm v\|_2 + \|\bar{\bbm u}\|_2\|\bm Z_k\|_{\op}\|\bbm v-\bar{\bbm v}\|_2\\
	        	&\leq \sup_{\bar{\bbm u}\in\mathcal N_q,\ \bar{\bbm v}\in\mathcal N_p}|\bar{\bbm u}^{\top}\bm Z_k\bar{\bbm v}| + \frac12\|\bm Z_k\|_{\op}.
	        \end{align*}
		Taking the supremum gives $\|\bm Z_k\|_{\op}\leq2\sup_{\bar{\bbm u}\in\mathcal N_q,\ \bar{\bbm v}\in\mathcal N_p}|\langle\bm Z_k,\bar{\bbm u}\bar{\bbm v}^\top\rangle|$. Therefore,
			\begin{align*}
				&\mathbb P\left(\left\|\frac{1}{n_k}\sum_{i=1}^{n_k}(\cm E_{k,i}\circ \cm X_{k,i})_{[S_X]}\right\|_{\op} \geq 2C\sigma_{\mathscr X,k}\sigma_{\mathscr E,k}(\lambda_{\mathscr X,k}^{\max}\lambda_{\mathscr E,k}^{\max})^{1/2}\sqrt{\frac{p+q}{n_k}}\right) \\
				&\leq \sum_{\bar{\bbm u}\in \mathcal N_q}\sum_{\bar{\bbm v}\in \mathcal N_p}\mathbb P\left(\left|\left\langle\frac{1}{n_k}\sum_{i=1}^{n_k}(\cm E_{k,i}\circ \cm X_{k,i})_{[S_X]},\bar{\bbm u}\bar{\bbm v}^\top\right\rangle\right| \geq C\sigma_{\mathscr X,k}\sigma_{\mathscr E,k}(\lambda_{\mathscr X,k}^{\max}\lambda_{\mathscr E,k}^{\max})^{1/2}\sqrt{\frac{p+q}{n_k}}\right)\\
				&\leq |\mathcal N_q|\,|\mathcal N_p|\,\exp(-C(p+q)) \leq 9^{p+q}\exp(-C(p+q)) = \exp\bigl(-C(p+q)\bigr).\qedhere
			\end{align*}
	\end{proof}

	\begin{proof}[\textbf{Proof of Lemma~\ref{lem:deviation-condition-infty}}]
		The proof follows the same conditioning arguments as Lemma~\ref{lem:deviation-condition-op}, but replaces the covering arguments for the operator norm by a union bound over all tensor entries. For any fixed entry indexed by $(a_1,\cdots,a_m,b_1,\cdots,b_d)$, define $Z_{k,i} = X_{k,i,b_1,\cdots,b_d}E_{k,i,a_1,\cdots,a_m}$ and $\bar Z_k=n_k^{-1}\sum_{i=1}^{n_k}Z_{k,i}$. It suffices to control $|\bar Z_k|$ uniformly over the $pq$ possible entries.

		For this fixed entry, let $\bm e_X\in\mathbb R^p$ and $\bm e_E\in\mathbb R^q$ be the corresponding coordinate vectors, so that $X_{k,i,b_1,\cdots,b_d}=\bm e_X^\top\vect(\cm X_{k,i})$ and $E_{k,i,a_1,\cdots,a_m}=\bm e_E^\top\vect(\cm E_{k,i})$. For any $u>0$ and $v>0$, define
		\[
			A_u = \left\{\frac1{n_k}\sum_{i=1}^{n_k}(\bm e_E^\top\vect(\cm E_{k,i}))(\bm e_X^\top\vect(\cm X_{k,i}))\geq u\right\}
			\ \text{and} \
			B_v = \left\{\frac1{n_k}\sum_{i=1}^{n_k}(\bm e_X^\top\vect(\cm X_{k,i}))^2\leq v\right\}.
		\]
		Conditioning on $\cm X_{k,1},\cdots,\cm X_{k,n_k}$, the event $B_v$ is measurable with respect to $\{\cm X_{k,i}\}_{i=1}^{n_k}$. By Assumption~\ref{assump:subg-noise}, conditional on the design, $\bm e_E^\top\vect(\cm E_{k,i})$ is mean-zero sub-Gaussian with variance proxy bounded by $\sigma_{\mathscr E,k}^2\bm e_E^\top\bbm\Sigma_{\mathscr E,k}\bm e_E\leq \sigma_{\mathscr E,k}^2\lambda_{\mathscr E,k}^{\max}$. Therefore, the same Chernoff arguments used in \eqref{eq:P-A-con-B}--\eqref{eq:union-expectation} gives
		\begin{align}\label{eq:infty-A-cond-B}
			\mathbb P(A_u\mid B_v) &\leq\exp\left(-\frac{n_ku^2}{2\sigma_{\mathscr E,k}^2\lambda_{\mathscr E,k}^{\max}v}\right).
		\end{align}
		We next control $B_v$. Similar to the arguments used to control the event $B_v$ in the proof of Lemma~\ref{lem:deviation-condition-op}, we apply the Hanson-Wright inequality to the empirical second moment of this design coordinate. By Assumption~\ref{assump:subg-design}, write $\vect(\cm X_{k,i})=\bbm\Sigma_{\mathscr X,k}^{1/2}\bbm\xi_{k,i}$, where the coordinates of $\bbm\xi_{k,i}$ are independent mean-zero sub-Gaussian random variables with sub-Gaussian norm bounded by $\sigma_{\mathscr X,k}$. Then
		\[
			\frac1{n_k}\sum_{i=1}^{n_k}(\bm e_X^\top\vect(\cm X_{k,i}))^2 = \frac1{n_k}\sum_{i=1}^{n_k}\bbm\xi_{k,i}^\top\bbm\Sigma_{\mathscr X,k}^{1/2}\bm e_X\bm e_X^\top\bbm\Sigma_{\mathscr X,k}^{1/2}\bbm\xi_{k,i}.
		\]
		Let $\bm H_{X,k} = \bbm\Sigma_{\mathscr X,k}^{1/2}\bm e_X\bm e_X^\top\bbm\Sigma_{\mathscr X,k}^{1/2}$. Since $\bm H_{X,k}$ is rank one, $\|\bm H_{X,k}\|_{\op} = \|\bm H_{X,k}\|_{\F} = \bm e_X^\top\bbm\Sigma_{\mathscr X,k}\bm e_X \leq \lambda_{\mathscr X,k}^{\max}$. Therefore, applying the Hanson-Wright inequality in Lemma~\ref{lem:hanson-wright}, for any $v>0$,
		\begin{align}\label{eq:element-wise-infty-bound}
			\mathbb P\left(\left|\frac1{n_k}\sum_{i=1}^{n_k}(\bm e_X^\top\vect(\cm X_{k,i}))^2 - \mathbb E(\bm e_X^\top\vect(\cm X_{k,i}))^2\right|\geq v\right) \leq 2\exp\left(-Cn_k\min\left\{\frac{v^2}{\sigma_{\mathscr X,k}^4(\lambda_{\mathscr X,k}^{\max})^2},\frac{v}{\sigma_{\mathscr X,k}^2\lambda_{\mathscr X,k}^{\max}}\right\}\right).
		\end{align}
		Moreover, $\mathbb E(\bm e_X^\top\vect(\cm X_{k,i}))^2 = \bm e_X^\top\bbm\Sigma_{\mathscr X,k}\bm e_X \leq \lambda_{\mathscr X,k}^{\max}$. Therefore, taking $v=C\sigma_{\mathscr X,k}^2\lambda_{\mathscr X,k}^{\max}$, we have
		\begin{align}\label{eq:infty-B-control}
			\mathbb P(B_v^{c}) = \mathbb P\left(\frac1{n_k}\sum_{i=1}^{n_k}(\bm e_X^\top\vect(\cm X_{k,i}))^2\geq C\sigma_{\mathscr X,k}^2\lambda_{\mathscr X,k}^{\max}\right)\leq 2\exp(-Cn_k),
		\end{align}
		With this choice of $v$, take $u = C\sigma_{\mathscr X,k}\sigma_{\mathscr E,k}(\lambda_{\mathscr X,k}^{\max}\lambda_{\mathscr E,k}^{\max})^{1/2}\sqrt{\log(pq)/n_k}$. Then \eqref{eq:infty-A-cond-B} gives
		\begin{align}\label{eq:infty-entry-upper-tail}
			\mathbb P(A_u\mid B_v) \leq \exp(-C\log(pq)).
		\end{align}
		Combining \eqref{eq:infty-B-control} and \eqref{eq:infty-entry-upper-tail}, and using $n_k\gtrsim\log(pq)$, we have
		\[
			\mathbb P\left(\bar Z_k \geq C\sigma_{\mathscr X,k}\sigma_{\mathscr E,k}(\lambda_{\mathscr X,k}^{\max}\lambda_{\mathscr E,k}^{\max})^{1/2}\sqrt{\frac{\log(pq)}{n_k}}\right) \leq \mathbb P(A_u\mid B_v) + \mathbb P(B_v^{c}) \leq \exp(-C\log(pq)).
		\]
		Applying the same arguments to $-\bar Z_k$ yields the two-sided bound
		\[
			\mathbb P\left(|\bar Z_k|\geq C\sigma_{\mathscr X,k}\sigma_{\mathscr E,k}(\lambda_{\mathscr X,k}^{\max}\lambda_{\mathscr E,k}^{\max})^{1/2}\sqrt{\frac{\log(pq)}{n_k}}\right) \leq 2\exp(-C\log(pq)).
		\]
		Finally, applying a union bound over all $pq$ entries of $n_k^{-1}\sum_{i=1}^{n_k}\cm E_{k,i}\circ\cm X_{k,i}$, and increasing the universal constant in the threshold if necessary, gives
		\[
			\mathbb P\left(\left\|\frac1{n_k}\sum_{i=1}^{n_k}\cm E_{k,i}\circ\cm X_{k,i}\right\|_\infty \geq C\sigma_{\mathscr X,k}\sigma_{\mathscr E,k}(\lambda_{\mathscr X,k}^{\max}\lambda_{\mathscr E,k}^{\max})^{1/2}\sqrt{\frac{\log(pq)}{n_k}}\right) \leq \exp(-C\log(pq)).\qedhere
		\]
	\end{proof}

	\begin{proof}[\textbf{Proof of Lemma~\ref{lem:deviation-condition-sparse}}]
		Let $\bm W_k=(n_k^{-1}\sum_{i=1}^{n_k}\cm E_{k,i}\circ\cm X_{k,i})_{[S_X]}$, $\bbm h=\cm H_{[S_X]}$, and write the Frobenius inner product after the $S_X$-matricization as $\langle\bm W_k,\bbm h\rangle$. The proof proceeds by separating the two possible effective sparsity regimes of $\bbm h$. For $t>0$, define the localized empirical process $Z(t)=\sup_{\|\bm U\|_1\leq t,\ \|\bm U\|_\F\leq1}|\langle\bm W_k,\bm U\rangle|$.

		\noindent
		\textbf{\emph{Case 1.}} Suppose $\|\bbm h\|_1\leq\sqrt{s}\|\bbm h\|_\F$. Then $|\langle\bm W_k,\bbm h\rangle|\leq\|\bbm h\|_\F Z(\sqrt{s})$, so it suffices to control $Z(\sqrt{s})$. This step uses the same conditioning arguments as in Lemma~\ref{lem:deviation-condition-infty}; the only difference is that the union bound over the $pq$ fixed entries is replaced by a sparse covering arguments. Let $\mathcal S_s=\{\bm U:\|\bm U\|_0\leq s,\ \|\bm U\|_\F\leq1\}$. For any fixed $\bm U\in\mathcal S_s$, define
		\[
			A_u(\bm U)=\left\{\left|\frac1{n_k}\sum_{i=1}^{n_k}\left\langle\vect(\cm E_{k,i}),\bm U\vect(\cm X_{k,i})\right\rangle\right|\geq u\right\},
		\]
		and let $B_v(\bm U)=\{n_k^{-1}\sum_{i=1}^{n_k}\|\bm U\vect(\cm X_{k,i})\|_2^2\leq v\}$. Conditional on the design, the same Chernoff arguments as in \eqref{eq:infty-A-cond-B} gives $\mathbb P(A_u(\bm U)\mid B_v(\bm U))\leq 2\exp(-n_ku^2/(C\sigma_{\mathscr E,k}^2\lambda_{\mathscr E,k}^{\max}v))$. Moreover, since $\|\bm U\|_\F\leq1$, the Hanson-Wright arguments used to prove \eqref{eq:infty-B-control} gives, for any fixed $\bm U\in\mathcal S_s$, $\mathbb P(B_v(\bm U)^c)\leq2\exp(-Cn_k)$ with $v=C\sigma_{\mathscr X,k}^2\lambda_{\mathscr X,k}^{\max}$. Taking $u=C\vartheta_k\sqrt{s\log(pq/s)/n_k}$ therefore yields $\mathbb P(A_u(\bm U))\leq2\exp(-Cs\log(pq/s))+2\exp(-Cn_k)$ for each fixed $\bm U\in\mathcal S_s$.

		We first extend this fixed-$\bm U$ tail bound to the exactly sparse unit-Frobenius class $\mathcal S_s$. For each support $S\subset[pq]$ with $|S|\leq s$, the subset of $\mathcal S_s$ supported on $S$ is isometric, under the Frobenius norm, to the Euclidean unit ball in $\mathbb R^{|S|}$, and hence admits an $\epsilon$-net of size at most $(3/\epsilon)^{|S|}$. Taking the union of these support-wise nets over all supports of size at most $s$ yields an $\epsilon$-net $\mathcal N_s\subset\mathcal S_s$ satisfying
		\[
			|\mathcal N_s|\leq\sum_{r=0}^{s}\binom{pq}{r}\left(\frac{3}{\epsilon}\right)^r.
		\]
		Since $r\leq s$ and $\epsilon\in(0,1)$ is fixed, $(3/\epsilon)^r\leq(3/\epsilon)^s$. Using $\sum_{r=0}^{s}\binom{pq}{r}\leq(epq/s)^s$, we have $\log|\mathcal N_s|\leq s\log(3e\,pq/(\epsilon s))\leq Cs\log(pq/s)$ under the usual sparsity regime $s\leq c pq$ for a fixed $c<1$; otherwise one may keep the bound as $Cs\log(epq/s)$. Applying the fixed-$\bm U$ tail bound over $\mathcal N_s$ and taking a union bound gives
		\[
			\sup_{\bm U\in\mathcal N_s}|\langle\bm W_k,\bm U\rangle| \leq C\vartheta_k\sqrt{\frac{s\log(pq/s)}{n_k}}
		\]
		with probability at least $1-\exp(-Cs\log(pq/s))$, where $\exp(-Cn_k)$ is absorbed by $n_k\gtrsim s\log(pq/s)$.

		Next, we pass the finite net to $\mathcal S_s$. Fix any $\bm U\in\mathcal S_s$, and choose $\bm U^{\mathcal N}\in\mathcal N_s$ supported on the same support as $\bm U$ such that $\|\bm U-\bm U^{\mathcal N}\|_\F\leq\epsilon$. Then $\bm U-\bm U^{\mathcal N}$ is supported on at most $s$ entries and has Frobenius norm at most $\epsilon$, so $(\bm U-\bm U^{\mathcal N})/\epsilon\in\mathcal S_s$ whenever $\bm U\neq\bm U^{\mathcal N}$. Hence $|\langle\bm W_k,\bm U\rangle|\leq \sup_{\bm V\in\mathcal N_s}|\langle\bm W_k,\bm V\rangle|+\epsilon\sup_{\bm V\in\mathcal S_s}|\langle\bm W_k,\bm V\rangle|$. Taking the supremum over $\bm U\in\mathcal S_s$ on the left-hand side and rearranging, and then taking $\epsilon=1/2$, gives
		\[
			\sup_{\bm U\in\mathcal S_s}|\langle\bm W_k,\bm U\rangle| \leq \frac{1}{1-\epsilon}\sup_{\bm U\in\mathcal N_s}|\langle\bm W_k,\bm U\rangle| \leq C\vartheta_k\sqrt{\frac{s\log(pq/s)}{n_k}}.
		\]

		We now transfer this bound from $\mathcal S_s$ to the localized $\ell_1$ class $\mathcal K_s=\{\bm U:\|\bm U\|_1\leq\sqrt{s},\ \|\bm U\|_\F\leq1\}$. For any fixed $\bm U\in\mathcal K_s$, order the entries of $\vect(\bm U)$ by decreasing absolute value and partition the indices into disjoint blocks $S_0,S_1,S_2,\ldots$, each of cardinality at most $s$. Write $\bm U=\sum_{j\geq0}\bm U_{S_j}$. For each nonzero block, $\bm U_{S_j}/\|\bm U_{S_j}\|_\F\in\mathcal S_s$. Moreover, because of the decreasing order, $\max_{a\in S_j}|U_a|\leq s^{-1}\|\bm U_{S_{j-1}}\|_1$ for $j\geq1$, and hence $\|\bm U_{S_j}\|_\F\leq\sqrt{s}\max_{a\in S_j}|U_a|\leq s^{-1/2}\|\bm U_{S_{j-1}}\|_1$. Therefore,
		\[
			\sum_{j\geq0}\|\bm U_{S_j}\|_\F \leq \|\bm U_{S_0}\|_\F+s^{-1/2}\sum_{j\geq1}\|\bm U_{S_{j-1}}\|_1 \leq \|\bm U\|_\F+s^{-1/2}\sum_{j\geq0}\|\bm U_{S_j}\|_1 = \|\bm U\|_\F+s^{-1/2}\|\bm U\|_1 \leq 2.
		\]
		Consequently, for this fixed $\bm U\in\mathcal K_s$, $|\langle\bm W_k,\bm U\rangle|\leq\sum_{j\geq0}\|\bm U_{S_j}\|_\F\sup_{\bm V\in\mathcal S_s}|\langle\bm W_k,\bm V\rangle|\leq2\sup_{\bm V\in\mathcal S_s}|\langle\bm W_k,\bm V\rangle|$. Taking the supremum over $\bm U\in\mathcal K_s$ gives $Z(\sqrt{s})\leq C\vartheta_k\sqrt{s\log(pq/s)/n_k}$. Consequently,
		\[
			|\langle\bm W_k,\bbm h\rangle| \leq \|\bbm h\|_\F Z(\sqrt{s}) \leq C\vartheta_k\sqrt{\frac{s\log(pq/s)}{n_k}}\|\cm H\|_\F.
		\]

		\noindent
		\textbf{\emph{Case 2.}} Suppose $\|\bbm h\|_1>\sqrt{s}\|\bbm h\|_\F$. Let $t_H=\|\bbm h\|_1/\|\bbm h\|_\F$. Then $t_H>\sqrt{s}$ and $|\langle\bm W_k,\bbm h\rangle|\leq\|\bbm h\|_\F Z(t_H)$. Therefore, it remains to control $Z(t)$ uniformly over $t>\sqrt{s}$. As in Case~1, we first control an exactly sparse class and then transfer the bound to the localized $\ell_1$ class by a sorted-block decomposition.

		Fix $t>\sqrt{s}$ and set $m_t=\lceil t^2\rceil$. Consider the exactly sparse class $\mathcal S_{m_t}=\{\bm U:\|\bm U\|_0\leq m_t,\ \|\bm U\|_\F\leq1\}$. By the support-wise covering arguments and the finite-net reduction used in Case~1, for this fixed $t$, with probability at least $1-\exp(-Cm_t\log(pq/m_t))$,
		\[
			\sup_{\bm U\in\mathcal S_{m_t}}|\langle\bm W_k,\bm U\rangle| \leq C\vartheta_k\sqrt{\frac{m_t\log(pq/m_t)}{n_k}} \leq C\vartheta_k t\sqrt{\frac{\log(pq/t^2)}{n_k}},
		\]
		where constants absorb the rounding from $m_t=\lceil t^2\rceil$.

		Next, take any $\bm U$ with $\|\bm U\|_1\leq t$ and $\|\bm U\|_\F\leq1$. Order the entries of $\vect(\bm U)$ by decreasing absolute value and partition them into disjoint blocks $S_0,S_1,S_2,\ldots$, each of cardinality at most $m_t$. The same sorted-coordinate arguments as in Case~1 gives $\sum_{j\geq0}\|\bm U_{S_j}\|_\F\leq\|\bm U\|_\F+m_t^{-1/2}\|\bm U\|_1\leq2$. Since each nonzero normalized block $\bm U_{S_j}/\|\bm U_{S_j}\|_\F$ belongs to $\mathcal S_{m_t}$, it follows that
		\[
			Z(t) \leq 2\sup_{\bm V\in\mathcal S_{m_t}}|\langle\bm W_k,\bm V\rangle| \leq C\vartheta_k t\sqrt{\frac{\log(pq/t^2)}{n_k}}.
		\]

		We now make this fixed-radius bound uniform over the possible value of $t_H$. Since $t_H=\|\bbm h\|_1/\|\bbm h\|_\F\leq\sqrt{pq}$ whenever $\bbm h\neq0$, it suffices to consider $t\in(\sqrt{s},\sqrt{pq}]$. Let $r_j=2^{j+1}\sqrt{s}$ and consider all indices $j$ such that $2^j\sqrt{s}<\sqrt{pq}$. Applying the fixed-radius bound at the endpoints $r_j$, with the threshold enlarged by the common peeling factor $\sqrt{\log(pq/s)/n_k}$, gives
		\[
			\mathbb P\left(Z(r_j)>C\vartheta_k r_j\left(\sqrt{\frac{\log(pq/r_j^2)}{n_k}}+\sqrt{\frac{\log(pq/s)}{n_k}}\right)\right) \leq \exp(-Cr_j^2\log(pq/s)).
		\]
		The additional term is not needed for a single fixed radius, but it makes the failure probabilities summable over the dyadic endpoints, since $\sum_j\exp(-Cr_j^2\log(pq/s))=\sum_j\exp(-C4^{j+1}s\log(pq/s))\leq\exp(-Cs\log(pq/s))$. Hence, with probability at least $1-\exp(-Cs\log(pq/s))$, the preceding endpoint bound holds simultaneously for all such $j$.

		On this event, fix any $t\in(\sqrt{s},\sqrt{pq}]$ and choose $j$ such that $t\in(2^j\sqrt{s},2^{j+1}\sqrt{s}]$. Since $Z(t)$ is nondecreasing in $t$ and $r_j=2^{j+1}\sqrt{s}$ satisfies $t\leq r_j\leq2t$, the endpoint bound implies
		\begin{align*}
		Z(t) \leq Z(r_j) \leq C\vartheta_k r_j\left(\sqrt{\frac{\log(pq/r_j^2)}{n_k}}+\sqrt{\frac{\log(pq/s)}{n_k}}\right) \leq C\vartheta_k t\left(\sqrt{\frac{\log(pq/t^2)}{n_k}}+\sqrt{\frac{\log(pq/s)}{n_k}}\right).
		\end{align*}
		Evaluating this bound at $t=t_H$ and using $t_H^2>s$, we obtain $Z(t_H)\leq C\vartheta_k t_H\sqrt{\log(pq/s)/n_k}$. Substituting this into $|\langle\bm W_k,\bbm h\rangle|\leq\|\bbm h\|_\F Z(t_H)$ and using $t_H\|\bbm h\|_\F=\|\bbm h\|_1=\|\cm H\|_1$, we get
		\[
			|\langle\bm W_k,\bbm h\rangle| \leq C\vartheta_k\sqrt{\frac{\log(pq/s)}{n_k}}\|\cm H\|_1.
		\]
		Combining Case~1 and Case~2 gives the two refined sparse deviation bounds. This completes the proof.
	\end{proof}

	\begin{proof}[\textbf{Proof of Lemma~\ref{lem:rsc-condition-tucker-rank-r}}]
		Since both sides of the desired empirical quadratic bound are homogeneous of degree two in $\cm T$, it suffices to prove the result on the unit-Frobenius Tucker-rank class. That is, for $\cm T \in \mathbb S_{\bbm r} = \{\cm T\in\mathbb R^{q_1\times\cdots\times q_m\times p_1\times\cdots\times p_d}:\tucrank(\cm T)\leq\bbm r, \|\cm T\|_\F = 1\}$, it suffices to show that, with probability at least $1-\exp(-Cdf_{\bbm r})$, 	
		\[
			\sup_{\boldsymbol{\mathscr T}\in\mathbb S_{\bbm r}}|L(\cm T)| = \sup_{\boldsymbol{\mathscr T}\in\mathbb S_{\bbm r}}\left|\frac{1}{n_k}\sum_{i=1}^{n_k}\left\|\left\langle\cm T,\cm X_{k,i}\right\rangle\right\|_\F^2 - \mathbb E\left\|\left\langle\cm T,\cm X_{k,i}\right\rangle\right\|_\F^2\right| \leq \delta_{\bbm r}.
		\]

		Define the centered empirical bilinear process and the centered empirical quadratic process
		\[
			B(\cm U,\cm V) = \frac{1}{n_k}\sum_{i=1}^{n_k}\left\langle\langle \cm U,\cm X_{k,i}\rangle,\langle \cm V,\cm X_{k,i}\rangle\right\rangle - \mathbb E\bigl[\left\langle\langle \cm U,\cm X_{k,i}\rangle,\langle \cm V,\cm X_{k,i}\rangle\right\rangle\bigr],
			\ \text{and} \ 
			L(\cm T)=B(\cm T,\cm T).
		\]
		For any $\cm U,\cm V$, we have $\vect(\langle\cm U,\cm X_{k,i}\rangle)=\cm U_{[S_X]}\vect(\cm X_{k,i})$ and $\vect(\langle\cm V,\cm X_{k,i}\rangle)=\cm V_{[S_X]}\vect(\cm X_{k,i})$. By Assumption~\ref{assump:subg-design}, $\vect(\cm X_{k,i}) = \bbm\Sigma_{\mathscr X,k}^{1/2}\bbm\xi_{k,i}$, where $\bbm\xi_{k,i}$ has independent mean-zero sub-Gaussian coordinates with variance proxy $\sigma_{\mathscr X,k}^2$. Let $\bbm\xi_k=(\bbm\xi_{k,1}^{\top},\cdots,\bbm\xi_{k,n_k}^{\top})^{\top}$. Then
		\begin{align*}
			\frac{1}{n_k}\sum_{i=1}^{n_k}\left\langle\langle \cm U,\cm X_{k,i}\rangle,\langle \cm V,\cm X_{k,i}\rangle\right\rangle = \frac{1}{n_k}\sum_{i=1}^{n_k}\bbm\xi_{k,i}^{\top}\bbm\Sigma_{\mathscr X,k}^{1/2}\cm U_{[S_X]}^\top\cm V_{[S_X]}\bbm\Sigma_{\mathscr X,k}^{1/2}\bbm\xi_{k,i}=\frac{1}{n_k}\bbm\xi_k^{\top}\left(\bm I_{n_k}\otimes\bm A_{\boldsymbol{\mathscr U},\boldsymbol{\mathscr V}}\right)\bbm\xi_k,
		\end{align*}
		where
		\[
			\bm A_{\boldsymbol{\mathscr U},\boldsymbol{\mathscr V}} = \bbm\Sigma_{\mathscr X,k}^{1/2}\left\{\frac{\cm U_{[S_X]}^\top\cm V_{[S_X]}+\cm V_{[S_X]}^\top\cm U_{[S_X]}}{2}\right\}\bbm\Sigma_{\mathscr X,k}^{1/2}.
		\]
		Consequently, applying the Hanson-Wright inequality in Lemma~\ref{lem:hanson-wright} yields that, for any $t>0$,
		\begin{align}\label{eq:pointwise-concentration}
			\mathbb P\left(|B(\cm U,\cm V)| \geq t\right) &\leq 2\exp\left(-C\min\left\{\frac{n_k^2t^2}{\sigma_{\mathscr X,k}^4\|\bm I_{n_k}\otimes\bm A_{\boldsymbol{\mathscr U},\boldsymbol{\mathscr V}}\|_{\F}^2},\frac{n_kt}{\sigma_{\mathscr X,k}^2\|\bm I_{n_k}\otimes\bm A_{\boldsymbol{\mathscr U},\boldsymbol{\mathscr V}}\|_{\op}}\right\}\right)\notag\\
			&\leq 2\exp\left(-Cn_k\min\left\{\frac{t^2}{\sigma_{\mathscr X,k}^4(\lambda_{\mathscr X,k}^{\max})^2},\frac{t}{\sigma_{\mathscr X,k}^2\lambda_{\mathscr X,k}^{\max}}\right\}\right).
		\end{align}
		The last inequality follows from $\|\bm I_{n_k}\otimes \bm A_{\boldsymbol{\mathscr U},\boldsymbol{\mathscr V}}\|_{\F}^{2} = n_k\|\bm A_{\boldsymbol{\mathscr U},\boldsymbol{\mathscr V}}\|_{\F}^{2}\leq n_k\lambda_{\max}^{2}(\bbm\Sigma_{\mathscr X,k})$ and $\|\bm I_{n_k}\otimes\bm A_{\boldsymbol{\mathscr U},\boldsymbol{\mathscr V}}\|_{\op} = \|\bm A_{\boldsymbol{\mathscr U},\boldsymbol{\mathscr V}}\|_{\op}\leq \lambda_{\mathscr X,k}^{\max}$, where we used $\|\cm U_{[S_X]}^\top\cm V_{[S_X]} + \cm V_{[S_X]}^\top\cm U_{[S_X]}\|_{\F}/2 \leq \|\cm U\|_\F\|\cm V\|_\F\leq1$, and the same bound holds for the operator norm.
		Next, we construct an $\varepsilon$-net to deduce the uniform bound over the Tucker-rank class $\mathbb S_{\bbm r}$ from the pointwise concentration bound. 
		
		Let $\mathcal N_{\mathrm{core}}$ be an $\varepsilon$-net of the Frobenius unit ball in $\mathbb R^{r_1\times\cdots\times r_{d+m}}$, and, for each mode $s\in[d+m]$, let $\mathcal N_s$ be an $\varepsilon$-net of the corresponding Stiefel manifold with respect to the projection distance $d(\bm U,\bm V)=\|\bm U\bm U^\top-\bm V\bm V^\top\|_{\op}$. By Lemma~7 in \cite{zhang2018tensor}, for any $0<\varepsilon<1$, we can construct an $\varepsilon$-net $\mathcal N_{\mathrm{core}}=\{\cm S^{(1)},\cdots,\cm S^{(N_S)}\}$ for the Frobenius unit ball $\{\cm S\in\mathbb R^{r_1\times\cdots\times r_{d+m}}: \|\cm S\|_\F\leq 1\}$ such that $\sup_{\boldsymbol{\mathscr S}:\|\boldsymbol{\mathscr S}\|_\F\leq 1}\min_{1\leq i\leq N_S}\|\cm S-\cm S^{(i)}\|_\F \leq \varepsilon$, with $|\mathcal N_{\mathrm{core}}|\leq(3/\varepsilon)^{\prod_{s=1}^{d+m}r_s}$.
		
		At the same time, for each $s\in[d+m]$, define the ambient mode dimension $d_s=q_s$ for $s\in [m]$, and $d_s=p_{s-m}$ for $s\in [d+m]\setminus[m]$. By Proposition~8 in \cite{szarek1982nets}, for each $s\in[d+m]$, we can construct an $\varepsilon$-net $\mathcal N_s=\{\bm U_s^{(1)},\cdots,\bm U_s^{(N_s)}\}$ on the Grassmann manifold of $r_s$-dimensional subspaces in $\mathbb R^{d_s}$ with respect to the metric $d(\bm U,\bm V) = \|\bm U\bm U^\top-\bm V\bm V^\top\|_{\op}$, such that $\sup_{\bm U_s\in\mathbb O_{d_s,r_s}} \min_{1\leq i\leq N_s} d(\bm U_s,\bm U_s^{(i)}) \leq \varepsilon$, with $|\mathcal N_s| \leq (c_0/\varepsilon)^{r_s(d_s-r_s)}$, where $c_0>0$ is an absolute constant. Consequently, the cardinality of the product net $\mathcal N_{\bbm r} = \{\cm S\times_{s=1}^{d+m}\bm U_s:\cm S\in\mathcal N_{\mathrm{core}},\bm U_s\in\mathcal N_s,\ \text{for}\ s\in[d+m]\}$
		satisfies $|\mathcal N_{\bbm r}| \leq (c/\varepsilon)^{df_{\bbm r}}$ with $df_{\bbm r} = {\prod_{s=1}^{d+m}r_s + \sum_{s=1}^{m} r_s(q_s-r_s) + \sum_{s=1}^{d} r_{m+s}(p_s-r_{m+s})}$.

		We now control the auxiliary centered bilinear process on the finite net. Since the net points have Frobenius norm at most one,
		\eqref{eq:pointwise-concentration} applies to every $\cm U,\cm T\in\mathcal N_{\bbm r}$. Moreover, $|\mathcal N_{\bbm r}|\leq(c/\varepsilon)^{df_{\bbm r}}$, taking $t=\delta_{\bbm r}/2$ and applying the union bound over $\mathcal N_{\bbm r}\times\mathcal N_{\bbm r}$ yields, after increasing the universal constant in the definition of $\delta_{\bbm r}$ if necessary,
		\begin{align}\label{eq:uniform-concentration-net}
			\mathbb P\left(\sup_{\boldsymbol{\mathscr U},\boldsymbol{\mathscr T}\in\mathcal N_{\bbm r}}|B(\cm U,\cm T)| \geq \frac{\delta_{\bbm r}}{2}\right) \leq \sum_{\boldsymbol{\mathscr U},\boldsymbol{\mathscr T}\in\mathcal N_{\bbm r}}\mathbb P\left(|B(\cm U,\cm T)| \geq \frac{\delta_{\bbm r}}{2}\right) \leq (c/\varepsilon)^{2df_{\bbm r}}\exp(-Cdf_{\bbm r}) \leq \exp(-Cdf_{\bbm r}),
		\end{align}
		where the last inequality holds because $\varepsilon$ is fixed and the universal constant in the concentration bound can be chosen sufficiently large.

		It remains to transfer this finite-net bound to $\mathbb S_{\bbm r}$. Fix any $\cm T\in\mathbb S_{\bbm r}$ and write its Tucker representation as $\cm T = \cm S^*\times_{s=1}^{d+m}\bm U_s^*$, with $\|\cm S^*\|_\F=1$ and $\bm U_s^*\in\mathbb O_{d_s,r_s}$. By the construction of the Grassmann nets, for each $s\in[d+m]$ there exist $\bm U_s^{(i_s)}\in\mathcal N_s$ and an orthogonal matrix $\bm O_s\in\mathbb R^{r_s\times r_s}$ such that $\|\bm U_s^*\bm O_s-\bm U_s^{(i_s)}\|_{\op}\leq\sqrt2\varepsilon$, see, e.g., Lemma~1 in \cite{cai2018rate}. Set $\bar{\cm S}=\cm S^*\times_{s=1}^{d+m}\bm O_s^\top$ and choose $\cm S^{(i_0)}\in\mathcal N_{\mathrm{core}}$ such that $\|\cm S^{(i_0)}-\bar{\cm S}\|_\F\leq\varepsilon$. Define the corresponding net tensor $\cm T^{\mathcal N} = \cm S^{(i_0)}\times_{s=1}^{d+m}\bm U_s^{(i_s)}\in\mathcal N_{\bbm r}$. Owing to the invariance of the Tucker representation under the core rotation, the original tensor can also be written as $\cm T = \bar{\cm S}\times_{s=1}^{d+m}\bm U_s^*\bm O_s$. Therefore,
		\begin{align}\label{eq:tucker-net-telescoping}
			\cm T-\cm T^{\mathcal N} = (\bar{\cm S}-\cm S^{(i_0)}) \times_{s=1}^{d+m}\bm U_s^*\bm O_s + \sum_{s=1}^{d+m}\cm S^{(i_0)} \times_{j<s}\bm U_j^{(i_j)} \times_s(\bm U_s^*\bm O_s-\bm U_s^{(i_s)}) \times_{j>s}\bm U_j^*\bm O_j.
		\end{align}

		We next bound the Frobenius norm of each term on the right-hand side of \eqref{eq:tucker-net-telescoping}. For the core-approximation term, since $\bm U_s^*\bm O_s$ has orthonormal columns for every $s\in[d+m]$, multiplication by these factor matrices preserves the Frobenius norm. Hence
		\[
			\left\|(\bar{\cm S}-\cm S^{(i_0)})\times_{s=1}^{d+m}\bm U_s^*\bm O_s\right\|_\F = \|\bar{\cm S}-\cm S^{(i_0)}\|_\F \leq \varepsilon.
		\]
		For the $s$-th factor-approximation term, using the fact that multiplying a tensor along one mode by a matrix increases the Frobenius norm by at most the operator norm of that matrix, and that all other factor matrices have orthonormal columns, we have, for $s\in[d+m]$,
		\[
			\|\cm S^{(i_0)}\times_{j<s}\bm U_j^{(i_j)}\times_s(\bm U_s^*\bm O_s-\bm U_s^{(i_s)})\times_{j>s}\bm U_j^*\bm O_j\|_\F \leq \|\cm S^{(i_0)}\|_\F\|\bm U_s^*\bm O_s-\bm U_s^{(i_s)}\|_{\op} \leq \sqrt{2}\varepsilon.
		\]
		In particular,
		\begin{align}\label{eq:tucker-net-approximation-error}
			\cm T-\cm T^{\mathcal N} = \sum_{s=0}^{d+m}\cm G_s,
			\ \text{with} \ 
			\tucrank(\cm G_s) \leq \bbm r
			\ \text{and} \ 
			\sum_{s=0}^{d+m}\|\cm G_s\|_\F \leq (\sqrt2(d+m)+1)\varepsilon.
		\end{align}
		Then, by the bilinearity of $B(\cdot,\cdot)$, we have, for any $\cm V\in\mathbb S_{\bbm r}$,
		\[
			|B(\cm T-\cm T^{\mathcal N},\cm V)| \leq \sum_{s=0}^{d+m}|B(\cm G_s,\cm V)| \leq \sum_{s=0}^{d+m}\|\cm G_s\|_\F\left|B\left(\frac{\cm G_s}{\|\cm G_s\|_\F},\cm V\right)\right| \leq (\sqrt2(d+m)+1)\varepsilon \sup_{\boldsymbol{\mathscr U},\boldsymbol{\mathscr V}\in\mathbb S_{\bbm r}}|B(\cm U,\cm V)|,
		\]
		where we used $\tucrank(\cm G_s)\leq\bbm r$. Since $\|\cm T^{\mathcal N}\|_\F\leq1$, the same arguments also gives
		\[
			|B(\cm T^{\mathcal N},\cm T-\cm T^{\mathcal N})| \leq \|\cm T^{\mathcal N}\|_\F\sum_{s=0}^{d+m}\|\cm G_s\|_\F \sup_{\boldsymbol{\mathscr U},\boldsymbol{\mathscr V}\in\mathbb S_{\bbm r}}|B(\cm U,\cm V)| \leq (\sqrt2(d+m)+1)\varepsilon \sup_{\boldsymbol{\mathscr U},\boldsymbol{\mathscr V}\in\mathbb S_{\bbm r}}|B(\cm U,\cm V)| .
		\]
		To close the bound, apply the same arguments to two arbitrary tensors $\cm U,\cm V\in\mathbb S_{\bbm r}$ and their net approximations. Since $B(\cm U,\cm V) = B(\cm U^{\mathcal N},\cm V^{\mathcal N}) + B(\cm U-\cm U^{\mathcal N},\cm V) + B(\cm U^{\mathcal N},\cm V-\cm V^{\mathcal N})$, we have
		\[
			\sup_{\boldsymbol{\mathscr U},\boldsymbol{\mathscr V}\in\mathbb S_{\bbm r}}|B(\cm U,\cm V)| \leq \sup_{\boldsymbol{\mathscr U},\boldsymbol{\mathscr V}\in\mathcal N_{\bbm r}}|B(\cm U,\cm V)| + 2(\sqrt2(d+m)+1)\varepsilon \sup_{\boldsymbol{\mathscr U},\boldsymbol{\mathscr V}\in\mathbb S_{\bbm r}}|B(\cm U,\cm V)|.
		\]
		Take $\varepsilon \leq 1/4(\sqrt{2}(d+m)+1)$ such that $2(\sqrt2(d+m)+1)\varepsilon\leq1/2$. Hence, on the event in \eqref{eq:uniform-concentration-net}, we have $\sup_{\boldsymbol{\mathscr U},\boldsymbol{\mathscr V}\in\mathbb S_{\bbm r}}|B(\cm U,\cm V)|\leq\delta_{\bbm r}$. In particular,
		\[
			\sup_{\boldsymbol{\mathscr T}\in\mathbb S_{\bbm r}}|L(\cm T)| \leq \delta_{\bbm r}. \qedhere
		\]
	\end{proof}

	\begin{proof}[\textbf{Proof of Lemma~\ref{lem:deviation-condition-tucker-rank-r}}]
		As in the proof of Lemma~\ref{lem:rsc-condition-tucker-rank-r}, construct an $\varepsilon$-net $\mathcal N_{\bbm r}$ for $\mathbb S_{\bbm r}$ with cardinality $|\mathcal N_{\bbm r}| \leq (c/\varepsilon)^{df_{\bbm r}}$. We first record a uniform variance bound for the design term. Define
		\[
			\mathcal E_{\mathscr X} = \left\{\sup_{\boldsymbol{\mathscr T}\in\mathbb S_{\bbm r}}\frac{1}{n_k}\sum_{i=1}^{n_k}\left\|\cm T_{[S_X]}\vect(\cm X_{k,i})\right\|_2^2 \leq C\sigma_{\mathscr X,k}^2\lambda_{\mathscr X,k}^{\max}\right\}.
		\]
		By Lemma~\ref{lem:rsc-condition-tucker-rank-r}, on an event with probability at least $1-\exp(-Cdf_{\bbm r})$, $\sup_{\boldsymbol{\mathscr T}\in\mathbb S_{\bbm r}}n_k^{-1}\sum_{i=1}^{n_k}\left\|\cm T_{[S_X]}\vect(\cm X_{k,i})\right\|_2^2 \leq \lambda_{\mathscr X,k}^{\max} + \delta_{\bbm r}$. Since $\delta_{\bbm r} = C\sigma_{\mathscr X,k}^2\lambda_{\mathscr X,k}^{\max}\sqrt{df_{\bbm r}/n_k}$, the sample size condition $n_k\gtrsim df_{\bbm r}$ implies $\lambda_{\mathscr X,k}^{\max}+\delta_{\bbm r} \leq C\sigma_{\mathscr X,k}^2\lambda_{\mathscr X,k}^{\max}$, after increasing the universal constant if necessary. Therefore, 
		\begin{align}\label{eq:variance-term-uniform-tucker}
			\mathbb P(\mathcal E_{\mathscr X})\geq1-\exp(-Cdf_{\bbm r}).
		\end{align}

		Next, fix any $\cm T\in\mathcal N_{\bbm r}$. We will show that, conditional on $\mathcal E_{\mathscr X}$, the desired deviation bound holds with high probability. Under the matricization convention, we have $\vect(\langle \cm T,\cm X_{k,i}\rangle) = \cm T_{[S_X]}\vect(\cm X_{k,i})$ and $(\cm E_{k,i}\circ\cm X_{k,i})_{[S_X]} = \vect(\cm E_{k,i})\vect^{\top}(\cm X_{k,i})$.
		Using the invariance of the Frobenius inner product under matricization, we have
		\begin{align*}
			\left\langle\frac{1}{n_k}\sum_{i=1}^{n_k}\cm E_{k,i}\circ \cm X_{k,i},\cm T\right\rangle &= \frac{1}{n_k}\sum_{i=1}^{n_k}\left\langle(\cm E_{k,i}\circ\cm X_{k,i})_{[S_X]},\cm T_{[S_X]}\right\rangle = \frac{1}{n_k}\sum_{i=1}^{n_k}\left\langle\vect(\cm E_{k,i})\vect^{\top}(\cm X_{k,i}),\cm T_{[S_X]}\right\rangle \\
			&= \frac{1}{n_k}\sum_{i=1}^{n_k}\operatorname{tr}\left(\vect(\cm X_{k,i})\vect^{\top}(\cm E_{k,i})\cm T_{[S_X]}\right) = \frac{1}{n_k}\sum_{i=1}^{n_k}\vect^{\top}(\cm E_{k,i})\cm T_{[S_X]}\vect(\cm X_{k,i}).
		\end{align*}
		Therefore, by Assumption~\ref{assump:subg-noise}, conditional on $\mathcal E_{\mathscr X}$, for any $\mu\in\mathbb R$, we have
		\begin{align*}
			&\mathbb E\left[\exp\left(\mu\sum_{i=1}^{n_k}\vect^{\top}(\cm E_{k,i})\cm T_{[S_X]}\vect(\cm X_{k,i})\right)\middle|\mathcal E_{\mathscr X}\right] \leq \exp\left(\frac{\mu^2\sigma_{\mathscr E,k}^2}{2}\sum_{i=1}^{n_k}\bigl(\cm T_{[S_X]}\vect(\cm X_{k,i})\bigr)^\top\bbm\Sigma_{\mathscr E,k}\bigl(\cm T_{[S_X]}\vect(\cm X_{k,i})\bigr)\right) \\
			&\quad\leq \exp\left(\frac{\mu^2\sigma_{\mathscr E,k}^2\lambda_{\mathscr E,k}^{\max}}{2}\sum_{i=1}^{n_k}\left\|\cm T_{[S_X]}\vect(\cm X_{k,i})\right\|_2^2\right) \leq \exp\left(C\mu^2\sigma_{\mathscr X,k}^2\sigma_{\mathscr E,k}^2\lambda_{\mathscr E,k}^{\max}\lambda_{\mathscr X,k}^{\max}n_k\right).
		\end{align*}

		It follows that, conditional on $\mathcal E_{\mathscr X}$, the random variable $\langle n_k^{-1}\sum_{i=1}^{n_k}\cm E_{k,i}\circ\cm X_{k,i},\cm T\rangle$ is mean-zero sub-Gaussian with conditional variance proxy bounded by $\sigma_{\mathscr X,k}^2\sigma_{\mathscr E,k}^2\lambda_{\mathscr E,k}^{\max}\lambda_{\mathscr X,k}^{\max}/n_k$. Hence, for any fixed $\cm T\in\mathcal N_{\bbm r}$ and $u>0$, we have
		\begin{align}\label{eq:conditional-tail-fixed-T}
			\mathbb P\left(\left|\left\langle\frac{1}{n_k}\sum_{i=1}^{n_k}\cm E_{k,i}\circ \cm X_{k,i},\cm T\right\rangle\right|\geq u \middle|\mathcal E_{\mathscr X}\right) \leq 2\exp\left(-\frac{Cn_ku^2}{\sigma_{\mathscr X,k}^2\sigma_{\mathscr E,k}^2\lambda_{\mathscr E,k}^{\max}\lambda_{\mathscr X,k}^{\max}}\right).
		\end{align}

		Then, combining \eqref{eq:variance-term-uniform-tucker} and \eqref{eq:conditional-tail-fixed-T} and taking $u = C\sigma_{\mathscr X,k}\sigma_{\mathscr E,k}(\lambda_{\mathscr X,k}^{\max}\lambda_{\mathscr E,k}^{\max})^{1/2}\sqrt{df_{\bbm r}/n_k}$, we have, for any $\cm T\in\mathcal N_{\bbm r}$,
		\begin{align*}
			&\mathbb P\left(\left|\left\langle\frac{1}{n_k}\sum_{i=1}^{n_k}\cm E_{k,i}\circ \cm X_{k,i},\cm T\right\rangle\right|\geq C\sigma_{\mathscr X,k}\sigma_{\mathscr E,k}(\lambda_{\mathscr X,k}^{\max}\lambda_{\mathscr E,k}^{\max})^{1/2}\sqrt{\frac{df_{\bbm r}}{n_k}}\right)\\
			&\quad\leq \mathbb P\left(\left|\left\langle\frac{1}{n_k}\sum_{i=1}^{n_k}\cm E_{k,i}\circ \cm X_{k,i},\cm T\right\rangle\right|\geq C\sigma_{\mathscr X,k}\sigma_{\mathscr E,k}(\lambda_{\mathscr X,k}^{\max}\lambda_{\mathscr E,k}^{\max})^{1/2}\sqrt{\frac{df_{\bbm r}}{n_k}} \middle|\mathcal E_{\mathscr X}\right) + \mathbb P\left(\mathcal E_{\mathscr X}^{c}\right)\\
			&\quad\leq 2\exp\left(-\frac{Cn_ku^2}{\sigma_{\mathscr X,k}^2\sigma_{\mathscr E,k}^2\lambda_{\mathscr E,k}^{\max}\lambda_{\mathscr X,k}^{\max}}\right) + \exp(-Cdf_{\bbm r}) \leq \exp(-Cdf_{\bbm r}).
		\end{align*}
		We now pass from a fixed net point to the whole finite net $\mathcal N_{\bbm r}$. Taking the supremum over $\mathcal N_{\bbm r}$, we have
		\begin{align}\label{eq:net-deviation-bound}
			&\mathbb P\left(\sup_{\boldsymbol{\mathscr T}\in\mathcal N_{\bbm r}}\left|\left\langle\frac{1}{n_k}\sum_{i=1}^{n_k}\cm E_{k,i}\circ \cm X_{k,i},\cm T\right\rangle\right|\geq C\sigma_{\mathscr X,k}\sigma_{\mathscr E,k}(\lambda_{\mathscr X,k}^{\max}\lambda_{\mathscr E,k}^{\max})^{1/2}\sqrt{\frac{df_{\bbm r}}{n_k}}\right) \notag\\
			&\quad\leq|\mathcal N_{\bbm r}|\exp(-Cdf_{\bbm r}) \leq (c/\varepsilon)^{df_{\bbm r}}\exp(-Cdf_{\bbm r}) \leq \exp(-Cdf_{\bbm r}),
		\end{align}
		where we used $|\mathcal N_{\bbm r}|\leq (c/\varepsilon)^{df_{\bbm r}}$ and that $\varepsilon$ is chosen as a fixed sufficiently small constant.

		It remains to pass from the finite net $\mathcal N_{\bbm r}$ to $\mathbb S_{\bbm r}$. Fix any $\cm T\in\mathbb S_{\bbm r}$ and let $\cm T^{\mathcal N}\in\mathcal N_{\bbm r}$ be the corresponding net tensor constructed as in the proof of Lemma~\ref{lem:rsc-condition-tucker-rank-r}. Similar to \eqref{eq:tucker-net-approximation-error}, we have
		\[
			\left|\left\langle\frac{1}{n_k}\sum_{i=1}^{n_k}\cm E_{k,i}\circ \cm X_{k,i},\cm T-\cm T^{\mathcal N}\right\rangle\right| \leq (\sqrt{2}(d+m)+1)\varepsilon\sup_{\boldsymbol{\mathscr T}\in\mathbb S_{\bbm r}}\left|\left\langle\frac{1}{n_k}\sum_{i=1}^{n_k}\cm E_{k,i}\circ \cm X_{k,i},\cm T\right\rangle\right|.
		\]
		Then, we have
		\begin{align*}
			\left|\left\langle\frac{1}{n_k}\sum_{i=1}^{n_k}\cm E_{k,i}\circ \cm X_{k,i},\cm T\right\rangle\right| &\leq \left|\left\langle\frac{1}{n_k}\sum_{i=1}^{n_k}\cm E_{k,i}\circ \cm X_{k,i},\cm T^{\mathcal N}\right\rangle\right| + \left|\left\langle\frac{1}{n_k}\sum_{i=1}^{n_k}\cm E_{k,i}\circ \cm X_{k,i},\cm T - \cm T^{\mathcal N}\right\rangle\right|\\
			&\leq \sup_{\boldsymbol{\mathscr T}\in \mathcal N_{\bbm r}}\left|\left\langle\frac{1}{n_k}\sum_{i=1}^{n_k}\cm E_{k,i}\circ \cm X_{k,i},\cm T^{\mathcal N}\right\rangle\right| + (\sqrt{2}(d+m)+1)\varepsilon\sup_{\boldsymbol{\mathscr T}\in\mathbb S_{\bbm r}}\left|\left\langle\frac{1}{n_k}\sum_{i=1}^{n_k}\cm E_{k,i}\circ \cm X_{k,i},\cm T\right\rangle\right|.
		\end{align*}

		Taking the supremum over $\cm T\in\mathbb S_{\bbm r}$ on the left-hand side and rearranging the inequality yields, provided that $(\sqrt{2}(d+m)+1)\varepsilon\leq 1/2$,
		\begin{align}\label{eq:deviation-net-reduction}
			\sup_{\boldsymbol{\mathscr T}\in\mathbb S_{\bbm r}}\left|\left\langle\frac{1}{n_k}\sum_{i=1}^{n_k}\cm E_{k,i}\circ \cm X_{k,i},\cm T\right\rangle\right| \leq 2\sup_{\boldsymbol{\mathscr T}\in\mathcal N_{\bbm r}}\left|\left\langle\frac{1}{n_k}\sum_{i=1}^{n_k}\cm E_{k,i}\circ \cm X_{k,i},\cm T\right\rangle\right|.
		\end{align}

		Finally, combining \eqref{eq:net-deviation-bound} and \eqref{eq:deviation-net-reduction}, we have
		\begin{align*}
			&\mathbb P\left(\sup_{\boldsymbol{\mathscr T}\in\mathbb S_{\bbm r}}\left|\left\langle\frac{1}{n_k}\sum_{i=1}^{n_k}\cm E_{k,i}\circ \cm X_{k,i},\cm T\right\rangle\right|\geq C\sigma_{\mathscr X,k}\sigma_{\mathscr E,k}(\lambda_{\mathscr X,k}^{\max}\lambda_{\mathscr E,k}^{\max})^{1/2}\sqrt{\frac{df_{\bbm r}}{n_k}}\right) \notag\\
			&\leq \mathbb P\left(\sup_{\boldsymbol{\mathscr T}\in\mathcal N_{\bbm r}}\left|\left\langle\frac{1}{n_k}\sum_{i=1}^{n_k}\cm E_{k,i}\circ \cm X_{k,i},\cm T\right\rangle\right|\geq \frac{C}{2}\sigma_{\mathscr X,k}\sigma_{\mathscr E,k}(\lambda_{\mathscr X,k}^{\max}\lambda_{\mathscr E,k}^{\max})^{1/2}\sqrt{\frac{df_{\bbm r}}{n_k}}\right) \leq \exp(-Cdf_{\bbm r}).\qedhere
		\end{align*}
	\end{proof}

	\begin{proof}[\textbf{Proof of Lemma~\ref{lem:restricted-bilinear-event-2r}}]
		For $\cm U,\cm V\in\mathbb S_{\bbm r}$, define the centered bilinear empirical process
		\[
			B(\cm U,\cm V) = \frac{1}{n_k}\sum_{i=1}^{n_k}\left\langle\langle\cm U,\cm X_{k,i}\rangle,\langle\cm V,\cm X_{k,i}\rangle\right\rangle - \mathbb E\!\left[\left\langle\langle\cm U,\cm X_{k,i}\rangle,\langle\cm V,\cm X_{k,i}\rangle\right\rangle\right].
		\]
		For fixed $\cm U,\cm V\in\mathbb S_{\bbm r}$, using the $S_X$-matricization convention, $\vect(\langle\cm U,\cm X_{k,i}\rangle) = \cm U_{[S_X]}\vect(\cm X_{k,i})$ and $\vect(\langle\cm V,\cm X_{k,i}\rangle) = \cm V_{[S_X]}\vect(\cm X_{k,i})$. By Assumption~\ref{assump:subg-design}, $\vect(\cm X_{k,i})=\bbm\Sigma_{\mathscr X,k}^{1/2}\bbm\xi_{k,i}$, where $\bbm\xi_{k,i}$ has independent mean-zero sub-Gaussian coordinates with variance proxy $\sigma_{\mathscr X,k}^2$. Let $\bbm\xi_k=(\bbm\xi_{k,1}^{\top},\cdots,\bbm\xi_{k,n_k}^{\top})^{\top}$. Then $n_k^{-1}\sum_{i=1}^{n_k}\left\langle\langle\cm U,\cm X_{k,i}\rangle,\langle\cm V,\cm X_{k,i}\rangle\right\rangle = n_k^{-1}\bbm\xi_k^{\top}\left(\bm I_{n_k}\otimes\bm A_{\boldsymbol{\mathscr U},\boldsymbol{\mathscr V}}\right)\bbm\xi_k$ where
		\[
			\bm A_{\boldsymbol{\mathscr U},\boldsymbol{\mathscr V}} = \bbm\Sigma_{\mathscr X,k}^{1/2}\left\{\frac{\cm U_{[S_X]}^{\top}\cm V_{[S_X]} + \cm V_{[S_X]}^{\top}\cm U_{[S_X]}}{2}\right\}\bbm\Sigma_{\mathscr X,k}^{1/2}.
		\]
		Since $\|\cm U\|_\F=\|\cm V\|_\F=1$, we have $\|\bm A_{\boldsymbol{\mathscr U},\boldsymbol{\mathscr V}}\|_{\op} \leq \lambda_{\mathscr X,k}^{\max}$ and $\|\bm A_{\boldsymbol{\mathscr U},\boldsymbol{\mathscr V}}\|_{\F} \leq \lambda_{\mathscr X,k}^{\max}$. Applying the Hanson-Wright inequality in Lemma~\ref{lem:hanson-wright} gives, for any $t>0$,
		\begin{align}\label{eq:bilinear-pointwise-concentration-2r}
			\mathbb P\left(|B(\cm U,\cm V)|\geq t\right) \leq 2\exp\left(-Cn_k\min\left\{\frac{t^2}{\sigma_{\mathscr X,k}^4(\lambda_{\mathscr X,k}^{\max})^2},\frac{t}{\sigma_{\mathscr X,k}^2\lambda_{\mathscr X,k}^{\max}}\right\}\right).
		\end{align}

		Next, construct an $\varepsilon$-net $\mathcal N_{\bbm r}$ for $\mathbb S_{\bbm r}$ as in the proof of Lemma~\ref{lem:rsc-condition-tucker-rank-r}. Its cardinality satisfies $|\mathcal N_{\bbm r}| \leq (c/\varepsilon)^{df_{\bbm r}}$. Taking $t=\delta_{\bbm r}/2$ in \eqref{eq:bilinear-pointwise-concentration-2r} and applying the union bound over $\mathcal N_{\bbm r}\times\mathcal N_{\bbm r}$ yield, after increasing the universal constant in $\delta_{\bbm r}$ if necessary,
		\begin{align}\label{eq:bilinear-net-concentration-2r}
			\mathbb P\left(\sup_{\boldsymbol{\mathscr U},\boldsymbol{\mathscr V}\in\mathcal N_{\bbm r}}|B(\cm U,\cm V)| \geq \frac{\delta_{\bbm r}}{2}\right) \leq \exp(-Cdf_{\bbm r}).
		\end{align}

		It remains to extend the finite-net bound to the full Tucker-rank class. Using the same Tucker net approximation argument as in \eqref{eq:tucker-net-approximation-error} in the proof of Lemma~\ref{lem:rsc-condition-tucker-rank-r}, for any $\cm T\in\mathbb S_{\bbm r}$, there exists $\cm T^{\mathcal N}\in\mathcal N_{\bbm r}$ such that
		\[
			\cm T-\cm T^{\mathcal N} = \sum_{s=0}^{d+m}\cm G_s,
			\quad
			\tucrank(\cm G_s)\leq\bbm r,
			\quad
			\sum_{s=0}^{d+m}\|\cm G_s\|_\F \leq (\sqrt2(d+m)+1)\varepsilon.
		\]
		Using the bilinearity of $B(\cdot,\cdot)$, for arbitrary $\cm U,\cm V\in\mathbb S_{\bbm r}$ and their net approximations $\cm U^{\mathcal N},\cm V^{\mathcal N}\in\mathcal N_{\bbm r}$, we have
		\[
			\sup_{\boldsymbol{\mathscr U},\boldsymbol{\mathscr V}\in\mathbb S_{\bbm r}}|B(\cm U,\cm V)| \leq \sup_{\boldsymbol{\mathscr U},\boldsymbol{\mathscr V}\in\mathcal N_{\bbm r}}|B(\cm U,\cm V)| + 2(\sqrt2(d+m)+1)\varepsilon\sup_{\boldsymbol{\mathscr U},\boldsymbol{\mathscr V}\in\mathbb S_{\bbm r}}|B(\cm U,\cm V)|.
		\]
		Choose $\varepsilon\leq1/[4(\sqrt2(d+m)+1)]$. Then, on the event in \eqref{eq:bilinear-net-concentration-2r},
		\[
			\sup_{\boldsymbol{\mathscr U},\boldsymbol{\mathscr V}\in\mathbb S_{\bbm r}}|B(\cm U,\cm V)| \leq \delta_{\bbm r}. \qedhere
		\]
	\end{proof}

	\begin{proof}[\textbf{Proof of Lemma~\ref{lem:tangent-deviation-bilinear-bound}}]
		Define the centered linear process
		\[
			L(\cm U) = \frac{1}{n_k}\sum_{i=1}^{n_k}\left\langle\langle \cm B_k^*,\cm X_{k,i}\rangle,\langle \cm U,\cm X_{k,i}\rangle\right\rangle - \mathbb E\bigl[\left\langle\langle \cm B_k^*,\cm X_{k,i}\rangle,\langle \cm U,\cm X_{k,i}\rangle\right\rangle\bigr].
		\]
		By Lemma~\ref{lem:tangent-space-tucker-rank-2r}, every unit-Frobenius element of $\mathcal T_{\bbm r}(\boldsymbol{\mathscr A}_{0,k}^{(t)})$ belongs to $\mathbb S_{2\bbm r}$. Therefore, $\sup_{\substack{\boldsymbol{\mathscr U}\in\mathcal T_{\bbm r}(\boldsymbol{\mathscr A}_{0,k}^{(t)})\\ \|\boldsymbol{\mathscr U}\|_\F=1}}|L(\cm U)| \leq \sup_{\boldsymbol{\mathscr U}\in\mathbb S_{2\bbm r}}|L(\cm U)|$. Construct a product $\varepsilon$-net $\mathcal N_{2\bbm r}$ for $\mathbb S_{2\bbm r}$ as in Lemma~\ref{lem:rsc-condition-tucker-rank-r}, satisfying $|\mathcal N_{2\bbm r}| \leq (c/\varepsilon)^{df_{2\bbm r}}$. By the same Tucker net approximation arguments, for any $\varepsilon\leq 1/[2(\sqrt2(d+m)+1)]$,
		\begin{align}\label{eq:tangent-deviation-net-reduction}
			\sup_{\substack{\boldsymbol{\mathscr U}\in\mathcal T_{\bbm r}(\boldsymbol{\mathscr A}_{0,k}^{(t)})\\ \|\boldsymbol{\mathscr U}\|_\F=1}}|L(\cm U)| \leq \sup_{\boldsymbol{\mathscr U}\in\mathbb S_{2\bbm r}}|L(\cm U)| \leq 2\sup_{\boldsymbol{\mathscr U}\in\mathcal N_{2\bbm r}}|L(\cm U)|.
		\end{align}
		Indeed, using the same Tucker net approximation argument as in \eqref{eq:tucker-net-approximation-error} in the proof of Lemma~\ref{lem:rsc-condition-tucker-rank-r}, for any $\cm T\in\mathbb S_{\bbm r}$, there exists $\cm T^{\mathcal N}\in\mathcal N_{\bbm r}$ such that
		\[
			\cm U-\cm U^{\mathcal N} = \sum_{s=0}^{d+m}\cm G_s,
			\quad
			\tucrank(\cm G_s)\leq2\bbm r,
			\quad
			\sum_{s=0}^{d+m}\|\cm G_s\|_\F \leq (\sqrt2(d+m)+1)\varepsilon.
		\]
		Since $L(\cdot)$ is linear, $|L(\cm U-\cm U^{\mathcal N})| \leq (\sqrt2(d+m)+1)\varepsilon \sup_{\boldsymbol{\mathscr V}\in\mathbb S_{2\bbm r}}|L(\cm V)|$. Taking the supremum over $\cm U\in\mathbb S_{2\bbm r}$ and choosing $\varepsilon$ sufficiently small gives \eqref{eq:tangent-deviation-net-reduction}.

		It remains to control the finite net. For fixed $\cm U\in\mathcal N_{2\bbm r}$, the $S_X$-matricization gives $\vect(\langle \cm B_k^*,\cm X_{k,i}\rangle)=(\cm B_k^*)_{[S_X]}\vect(\cm X_{k,i})$ and $\vect(\langle \cm U,\cm X_{k,i}\rangle)=\cm U_{[S_X]}\vect(\cm X_{k,i})$. Using Assumption~\ref{assump:subg-design}, write $\vect(\cm X_{k,i}) = \bbm\Sigma_{\mathscr X,k}^{1/2}\bbm\xi_{k,i}$, where $\bbm\xi_{k,i}$ has independent mean-zero sub-Gaussian coordinates with variance proxy $\sigma_{\mathscr X,k}^2$. Then $\left\langle\langle \cm B_k^*,\cm X_{k,i}\rangle,\langle \cm U,\cm X_{k,i}\rangle\right\rangle = \bbm\xi_{k,i}^{\top}\bm A_{\boldsymbol{\mathscr U}}\bbm\xi_{k,i}$, where
		\[
			\bm A_{\boldsymbol{\mathscr U}} = \bbm\Sigma_{\mathscr X,k}^{1/2}\left\{\frac{(\cm B_k^*)_{[S_X]}^\top\cm U_{[S_X]}+\cm U_{[S_X]}^\top(\cm B_k^*)_{[S_X]}}{2}\right\}\bbm\Sigma_{\mathscr X,k}^{1/2}.
		\]
		Since $\|\cm U\|_\F=1$, $\|\bm A_{\boldsymbol{\mathscr U}}\|_{\op} \leq \lambda_{\mathscr X,k}^{\max}\|(\cm B_k^*)_{[S_X]}\|_\F$ and $\|\bm A_{\boldsymbol{\mathscr U}}\|_{\F} \leq \lambda_{\mathscr X,k}^{\max}\|(\cm B_k^*)_{[S_X]}\|_\F$. Applying the Hanson--Wright inequality in Lemma~\ref{lem:hanson-wright}, for every $t>0$,
		\begin{align}\label{eq:fixed-u-deviation-bilinear-bound}
			\mathbb P\left(|L(\cm U)|\geq t\right) \leq 2\exp\left(-Cn_k\min\left\{\frac{t^2}{\sigma_{\mathscr X,k}^4(\lambda_{\mathscr X,k}^{\max})^2\|(\cm B_k^*)_{[S_X]}\|_\F^2},\frac{t}{\sigma_{\mathscr X,k}^2\lambda_{\mathscr X,k}^{\max}\|(\cm B_k^*)_{[S_X]}\|_\F}\right\}\right).
		\end{align}
		Taking $t = C\sigma_{\mathscr X,k}^2\lambda_{\mathscr X,k}^{\max}\|(\cm B_k^*)_{[S_X]}\|_\F\sqrt{df_{2\bbm r}/n_k}$, and applying a union bound over $\mathcal N_{2\bbm r}$ gives $\sup_{\boldsymbol{\mathscr U}\in\mathcal N_{2\bbm r}}|L(\cm U)|\leq C\sigma_{\mathscr X,k}^2\lambda_{\mathscr X,k}^{\max}\|(\cm B_k^*)_{[S_X]}\|_\F\sqrt{df_{2\bbm r}/n_k}$
		with probability at least $1-\exp(-Cdf_{2\bbm r})$. Combining this with \eqref{eq:tangent-deviation-net-reduction}, and using $\|(\cm B_k^*)_{[S_X]}\|_\F\leq h_{\mathscr B,k}$, we first obtain the intermediate bound
		\[
			\sup_{\substack{\boldsymbol{\mathscr U}\in\mathcal T_{\bbm r}(\boldsymbol{\mathscr A}_{0,k}^{(t)})\\ \|\boldsymbol{\mathscr U}\|_\F=1}}|L(\cm U)| \leq C\sigma_{\mathscr X,k}^2\lambda_{\mathscr X,k}^{\max}h_{\mathscr B,k}\sqrt{\frac{df_{2\bbm r}}{n_k}}.
		\]
		The intermediate bound holds with probability at least $1-\exp(-Cdf_{2\bbm r})$. Since $df_{2\bbm r}\asymp df_{\bbm r}$, the last display implies \eqref{eq:tangent-deviation-bilinear-bound} with probability at least $1-\exp(-Cdf_{\bbm r})$, after absorbing the constant factors into the universal constant $C$.
	\end{proof}

	\begin{proof}[\textbf{Proof of Lemma \ref{lem:truncation-level-choice}}]
		Fix a client $k\in[K]$. By Assumption~\ref{assump:subg-design}, $\vect(\cm X_{k,i}) = \bbm\Sigma_{\mathscr X,k}^{1/2}\bbm\xi_{k,i}$, where $\bbm\xi_{k,i}$ has independent standardized sub-Gaussian coordinates with sub-Gaussian parameter $\sigma_{\mathscr X,k}$. Since $\|\cm X_{k,i}\|_{\F}^2 = \bbm\xi_{k,i}^{\top}\bbm\Sigma_{\mathscr X,k}\bbm\xi_{k,i}$, Lemma~\ref{lem:hanson-wright} implies that, for any $t>0$,
		\[
			\mathbb P\left(\left|\|\cm X_{k,i}\|_{\F}^2 - \mathbb E\|\cm X_{k,i}\|_{\F}^2\right| \geq t\right) \leq 2\exp\left(-C\min\left\{\frac{t^2}{\sigma_{\mathscr X,k}^4\|\bbm\Sigma_{\mathscr X,k}\|_{\F}^2},\frac{t}{\sigma_{\mathscr X,k}^2\|\bbm\Sigma_{\mathscr X,k}\|_{\op}}\right\}\right).
		\]
		Moreover, $\mathbb E\|\cm X_{k,i}\|_{\F}^2 = \mathbb E(\bbm\xi_{k,i}^{\top}\bbm\Sigma_{\mathscr X,k}\bbm\xi_{k,i}) \leq C\sigma_{\mathscr X,k}^2\operatorname{tr}(\bbm\Sigma_{\mathscr X,k}) \leq C\sigma_{\mathscr X,k}^2\lambda_{\mathscr X,k}^{\max}p$, $\|\bbm\Sigma_{\mathscr X,k}\|_{\op}=\lambda_{\mathscr X,k}^{\max}$, and $\|\bbm\Sigma_{\mathscr X,k}\|_{\F}\leq \lambda_{\mathscr X,k}^{\max}\sqrt p$. Taking $t=C\sigma_{\mathscr X,k}^2\lambda_{\mathscr X,k}^{\max}(\sqrt{pu}+u)$, we have
		\[
			\frac{t^2}{\sigma_{\mathscr X,k}^4\|\bbm\Sigma_{\mathscr X,k}\|_{\F}^2} \geq C\frac{(\sqrt{pu}+u)^2}{p} \geq Cu
			\quad \text{and} \quad
			\frac{t}{\sigma_{\mathscr X,k}^2\|\bbm\Sigma_{\mathscr X,k}\|_{\op}} \geq C(\sqrt{pu}+u)\geq Cu.
		\]
		Therefore, with probability at least $1-2e^{-u}$,
		\[
			\|\cm X_{k,i}\|_{\F}^2 \leq \mathbb E\|\cm X_{k,i}\|_{\F}^2+t \leq C\sigma_{\mathscr X,k}^2\lambda_{\mathscr X,k}^{\max}p + C\sigma_{\mathscr X,k}^2\lambda_{\mathscr X,k}^{\max}(\sqrt{pu}+u) \leq C\sigma_{\mathscr X,k}^2\lambda_{\mathscr X,k}^{\max}(p+u),
		\]
		where the last inequality uses $\sqrt{pu}\leq(p+u)/2$. Hence
		\begin{align}\label{eq:X-norm-subg-bound}
			\mathbb P\left(\|\cm X_{k,i}\|_{\F} \geq C\sigma_{\mathscr X,k}\sqrt{\lambda_{\mathscr X,k}^{\max}}\sqrt{p+u}\right) \leq 2e^{-u}.
		\end{align}
		Taking $u=C_0\log n$ and applying a union bound over $i\in[n_k]$ gives
		\[
			\mathbb P\left(\max_{i\in[n_k]}\|\cm X_{k,i}\|_{\F}\geq C\sigma_{\mathscr X,k}\sqrt{\lambda_{\mathscr X,k}^{\max}}\sqrt{p+\log n}\right) \leq 2n_k\exp(-C_0\log n)\leq C\exp(-C\log n),
		\]
		where we used $n_k\leq n$ and absorbed constants into $C$. Hence,
		\begin{align}\label{eq:X-uniform-bound-proof}
			\max_{i\in[n_k]}\|\cm X_{k,i}\|_{\F}\leq C\sigma_{\mathscr X,k}\sqrt{\lambda_{\mathscr X,k}^{\max}}\sqrt{p+\log n}
		\end{align}
		with probability at least $1-C\exp(-C\log n)$.

		Next, we control the residual tensor. Under the decomposition \eqref{eq:decomposition}, $\langle \cm A,\cm X_{k,i}\rangle-\cm Y_{k,i} = \langle\cm A-\cm A_0^*-\cm B_k^*,\cm X_{k,i}\rangle - \cm E_{k,i}$. Therefore, for any $\cm A$ satisfying $\|\cm A-\cm A_0^*\|_{\F}\leq R_{\mathscr A,k}$, by the Cauchy--Schwarz inequality for the generalized tensor inner product,
		\begin{align}\label{eq:residual-basic-bound}
			\|\langle \cm A,\cm X_{k,i}\rangle-\cm Y_{k,i}\|_{\F} \leq \|\cm A-\cm A_0^*-\cm B_k^*\|_{\F}\|\cm X_{k,i}\|_{\F} + \|\cm E_{k,i}\|_{\F} \leq (R_{\mathscr A,k}+h_{\mathscr B,k})\|\cm X_{k,i}\|_{\F} + \|\cm E_{k,i}\|_{\F},
		\end{align}
		where the last inequality uses $\|\cm B_k^*\|_{\F}\leq h_{\mathscr B,k}$.
			By Assumption~\ref{assump:subg-noise}, conditional on $\cm X_{k,i}$, the vector $\bbm e_{k,i}=\vect(\cm E_{k,i})$ is mean-zero and satisfies
			\[
				\mathbb P\left(|\bbm u^\top\bbm e_{k,i}|\geq t\ \middle|\ \cm X_{k,i}\right)
				\leq 2\exp\left(-\frac{t^2}{2\sigma_{\mathscr E,k}^2\lambda_{\mathscr E,k}^{\max}}\right)
			\]
			for every $\bbm u\in\mathbb S^{q-1}$. Let $\mathcal N_q$ be a $1/2$-net of $\mathbb S^{q-1}$ with $|\mathcal N_q|\leq5^q$. The standard net reduction gives $\|\bbm e_{k,i}\|_2\leq2\max_{\bbm u\in\mathcal N_q}|\bbm u^\top\bbm e_{k,i}|$. Hence, a union bound over $\mathcal N_q$ yields, for every $u>0$,
			\begin{align}\label{eq:E-norm-subg-bound}
				\mathbb P\left(\|\cm E_{k,i}\|_{\F}\geq C\sigma_{\mathscr E,k}\sqrt{\lambda_{\mathscr E,k}^{\max}}\sqrt{q+u}\ \middle|\ \cm X_{k,i}\right)\leq2e^{-u}.
			\end{align}
		Let $\mathcal A_{k,i}^{E}(u) = \left\{\|\cm E_{k,i}\|_{\F} \geq C\sigma_{\mathscr E,k}\sqrt{\lambda_{\mathscr E,k}^{\max}}\sqrt{q+u}\right\}$. Since the right-hand side of \eqref{eq:E-norm-subg-bound} does not depend on $\cm X_{k,i}$, the tower property yields
		\[
			\mathbb P(\mathcal A_{k,i}^{E}(u)) = \mathbb E[\mathbb P(\mathcal A_{k,i}^{E}(u)\mid \cm X_{k,i})] \leq 2e^{-u}.
		\]
		Taking $u=C_0\log n$ and applying the union bound over $i\in[n_k]$, we have
		\begin{align*}
			\mathbb P\left(\max_{i\in[n_k]}\|\cm E_{k,i}\|_{\F} \geq C\sigma_{\mathscr E,k}\sqrt{\lambda_{\mathscr E,k}^{\max}}\sqrt{q+C_0\log n}\right) &\leq \sum_{i=1}^{n_k}\mathbb P\bigl(\mathcal A_{k,i}^{E}(C_0\log n)\bigr) \leq 2n_k\exp(-C_0\log n)\\
			&\leq 2n\exp(-C_0\log n)\leq C\exp(-C\log n),
		\end{align*}
		where the last inequality uses $n_k\leq n$ and $C_0$ is chosen sufficiently large. Hence,
		\begin{align}\label{eq:E-uniform-bound-proof}
			\max_{i\in[n_k]}\|\cm E_{k,i}\|_{\F} \leq C\sigma_{\mathscr E,k}\sqrt{\lambda_{\mathscr E,k}^{\max}}\sqrt{q+\log n}
		\end{align}
		with probability at least $1-C\exp(-C\log n)$.
		Combining \eqref{eq:X-uniform-bound-proof}, \eqref{eq:residual-basic-bound}, and \eqref{eq:E-uniform-bound-proof}, we have, uniformly over all $\cm A$ satisfying $\|\cm A-\cm A_0^*\|_{\F}\leq R_{\mathscr A,k}$,
		\[
				\max_{i\in[n_k]}\|\langle \cm A,\cm X_{k,i}\rangle-\cm Y_{k,i}\|_{\F} \leq C\left((R_{\mathscr A,k}+h_{\mathscr B,k})\sigma_{\mathscr X,k}\sqrt{\lambda_{\mathscr X,k}^{\max}}\sqrt{p+\log n}+\sigma_{\mathscr E,k}\sqrt{\lambda_{\mathscr E,k}^{\max}}\sqrt{q+\log n}\right).
		\]
		The two uniform bounds hold simultaneously with probability at least $1-C\exp(-C\log n)$ after adjusting constants. This completes the proof.
	\end{proof}

	\begin{proof}[\textbf{Proof of Lemma~\ref{lem:tangent-contraction-tensor-general}}]
		Since $\mathcal P_{\mathcal T_{\bbm r}(\boldsymbol{\mathscr A})}(\cdot)$ is the orthogonal projection onto $\mathcal T_{\bbm r}(\cm A)$, for any ambient coefficient tensor $\cm N$, $\cm N = \mathcal P_{\mathcal T_{\bbm r}(\boldsymbol{\mathscr A})}(\cm N) + \mathcal P_{\mathcal T_{\bbm r}^\perp(\boldsymbol{\mathscr A})}(\cm N)$. Hence, for every $\cm M_1\in \mathcal T_{\bbm r}(\cm A)$ and every ambient coefficient tensor $\cm N$, $\langle \cm M_1,\cm N\rangle = \langle \cm M_1,\mathcal P_{\mathcal T_{\bbm r}(\boldsymbol{\mathscr A})}(\cm N)\rangle$.

		Next, we show that $\mathcal H_n(\cdot)$ is self-adjoint on the ambient coefficient-tensor space with respect to the Frobenius inner product. Indeed, for any $\cm U,\cm V\in\mathbb R^{q_1\times\cdots\times q_m\times p_1\times\cdots\times p_d}$,
		\[
			\langle \cm U,\mathcal H_n(\cm V)\rangle = \left\langle\cm U,\frac1n\sum_{i=1}^n \langle \cm V,\cm X_i\rangle\circ \cm X_i\right\rangle = \frac1n\sum_{i=1}^n \left\langle\langle \cm V,\cm X_i\rangle,\langle \cm U,\cm X_i\rangle\right\rangle = \left\langle\cm V,\frac1n\sum_{i=1}^n \langle \cm U,\cm X_i\rangle\circ \cm X_i\right\rangle = \langle \cm V,\mathcal H_n(\cm U)\rangle = \langle \mathcal H_n(\cm U),\cm V\rangle.
		\]
		Therefore, for any $\cm M_1,\cm M_2\in \mathcal T_{\bbm r}(\cm A)$, we have
		\[
			\big\langle \cm M_1,\mathcal P_{\mathcal T_{\bbm r}(\boldsymbol{\mathscr A})}(\mathcal H_n(\cm M_2))\big\rangle = \langle \cm M_1,\mathcal H_n(\cm M_2)\rangle = \langle \mathcal H_n(\cm M_1),\cm M_2\rangle = \big\langle \mathcal P_{\mathcal T_{\bbm r}(\boldsymbol{\mathscr A})}(\mathcal H_n(\cm M_1)),\cm M_2\big\rangle.
		\]
		Thus the map $\cm Z\mapsto \mathcal P_{\mathcal T_{\bbm r}(\boldsymbol{\mathscr A})}(\mathcal H_n(\cm Z))$ is self-adjoint on the Hilbert space $\mathcal T_{\bbm r}(\cm A)$ equipped with the Frobenius inner product. Consequently, by the spectral theorem for self-adjoint operators on finite-dimensional real inner product spaces (see Spectral Theorem 2 in \citealt{ucsb_spectral_theorem_notes}), there exist an orthonormal basis $\{\cm U_j\}\subset \mathcal T_{\bbm r}(\cm A)$ and real eigenvalues $\{\lambda_j\}$ such that $\mathcal P_{\mathcal T_{\bbm r}(\boldsymbol{\mathscr A})}(\mathcal H_n(\cm U_j)) = \lambda_j\cm U_j$.
		Now write $\cm M=\sum_j a_j\cm U_j$. Since $\{\cm U_j\}$ is orthonormal under the Frobenius inner product, we have $\|\cm M\|_\F^2=\sum_j a_j^2$. Besides, $\mathcal P_{\mathcal T_{\bbm r}(\boldsymbol{\mathscr A})}(\mathcal H_n(\cm M)) = \sum_j a_j\lambda_j\cm U_j$.
		Therefore, $\cm M-\eta\mathcal P_{\mathcal T_{\bbm r}(\boldsymbol{\mathscr A})}\bigl(\mathcal H_n(\cm M)\bigr) = \sum_j a_j(1-\eta\lambda_j)\cm U_j$, and hence
		\begin{align}\label{eq:direct-spectral-expand}
			\Big\|\cm M-\eta\mathcal P_{\mathcal T_{\bbm r}(\boldsymbol{\mathscr A})}\bigl(\mathcal H_n(\cm M)\bigr)\Big\|_\F^2 = \sum_j a_j^2(1-\eta\lambda_j)^2.
		\end{align}

		Now fix any eigenvector $\cm U_j\in\mathcal T_{\bbm r}(\cm A)$ with $\|\cm U_j\|_\F=1$. Since the assumed restricted spectral bound holds for every tensor in $\mathcal T_{\bbm r}(\cm A)$, applying it to $\cm U_j$ gives $C_{\mathrm{tan}}^{\min} \leq n^{-1}\sum_{i=1}^n \|\langle \cm U_j,\cm X_i\rangle\|_\F^2 \leq C_{\mathrm{tan}}^{\max}$. On the other hand,
		\[
			\frac1n\sum_{i=1}^n \|\langle \cm U_j,\cm X_i\rangle\|_\F^2 = \langle \mathcal H_n(\cm U_j),\cm U_j\rangle = \left\langle\mathcal P_{\mathcal T_{\bbm r}(\boldsymbol{\mathscr A})}\bigl(\mathcal H_n(\cm U_j)\bigr),\cm U_j\right\rangle = \langle \lambda_j\cm U_j,\cm U_j\rangle = \lambda_j.
		\]
		Therefore, $C_{\mathrm{tan}}^{\min} \leq \lambda_j \leq C_{\mathrm{tan}}^{\max}$ for any $j$. Then, it follows that $|1-\eta\lambda_j| \leq \max\{|1-\eta C_{\mathrm{tan}}^{\min}|,|1-\eta C_{\mathrm{tan}}^{\max}|\}$.
		Substituting this bound into \eqref{eq:direct-spectral-expand} yields
		\begin{align}
			\Big\|\cm M-\eta\mathcal P_{\mathcal T_{\bbm r}(\boldsymbol{\mathscr A})}\bigl(\mathcal H_n(\cm M)\bigr)\Big\|_\F^2 &\leq \max\left\{|1-\eta C_{\mathrm{tan}}^{\min}|,|1-\eta C_{\mathrm{tan}}^{\max}|\right\}^2\sum_j a_j^2 \notag\\
			&= \max\left\{|1-\eta C_{\mathrm{tan}}^{\min}|,|1-\eta C_{\mathrm{tan}}^{\max}|\right\}^2\|\cm M\|_\F^2.
		\end{align}
		Taking square roots completes the proof.
	\end{proof}

	\begin{proof}[\textbf{Proof of Lemma \ref{lem:tangent-space-nuclear-bound}}]
		Since $\tucrank(\cm A)\leq\bbm r$, we may write $\cm A = \cm S\times_{a=1}^{m}\bm U_a\times_{b=1}^{d}\bm V_b$, where $\cm S\in\mathbb R^{r_1\times\cdots\times r_{d+m}}$, $\bm U_a\in\mathbb O_{q_a,r_a}$, and $\bm V_b\in\mathbb O_{p_b,r_{m+b}}$. Here $\bm U_a$ and $\bm V_b$ denote the Tucker factor matrices associated with the output and input modes, respectively. By the tangent-space characterization of the Tucker-rank manifold \citep[Lemma~1]{luo2023low}, every $\cm U\in\mathcal T_{\bbm r}(\cm A)$ admits the decomposition
		\begin{align*}
			\cm U = \underbrace{\cm F\times_{a=1}^{m}\bm U_a\times_{b=1}^{d}\bm V_b  + \sum_{a=1}^{m}\cm S\times_a \bm U_{a,\perp}\bm D_a\times_{\substack{1\leq j\leq m\\ j\neq a}}\bm U_j\times_{b=1}^{d}\bm V_b}_{\cm U_{\mathrm{out}}} + \underbrace{\sum_{b=1}^{d}\cm S\times_{m+b}\bm V_{b,\perp}\bm H_b\times_{a=1}^{m}\bm U_a\times_{\substack{1\leq j\leq d\\ j\neq b}}\bm V_j}_{\cm U_{\mathrm{in}}},
		\end{align*}
		where $\cm F\in\mathbb R^{r_1\times\cdots\times r_{d+m}}$ is the core-variation tensor, the matrices $\bm U_{a,\perp}$ and $\bm V_{b,\perp}$ have columns orthogonal to $\bm U_a$ and $\bm V_b$, respectively, and $\bm D_a\in\mathbb R^{(q_a-r_a)\times r_a}$, $\bm H_b\in\mathbb R^{(p_b-r_{m+b})\times r_{m+b}}$.
		We now bound the ranks of the $[S_X]$-matricizations of $\cm U_{\mathrm{out}}$ and $\cm U_{\mathrm{in}}$. Let $\bm V^{\otimes}=\bm V_d\otimes\cdots\otimes\bm V_1$. Under the convention of matricization $\cm T_{[S_X]}$, for any tensor of the form $\cm T = \cm G \times_{a=1}^{m}\bm A_a \times_{b=1}^{d}\bm B_b$, we have the matricization identity $\cm T_{[S_X]} = (\bm A_m\otimes\cdots\otimes\bm A_1)\cm G_{[S_X]}(\bm B_d\otimes\cdots\otimes\bm B_1)^{\top}$.

		For $\cm U_{\mathrm{out}}$, the input-mode factors are fixed as $\bm V_1,\cdots,\bm V_d$. Therefore, under the $[S_X]$-matricization, every term in $\cm U_{\mathrm{out}}$ shares the same right factor $(\bm V_d\otimes\cdots\otimes\bm V_1)^\top$. In particular, the core-variation term satisfies
		\[
			\left(\cm F\times_{a=1}^{m}\bm U_a\times_{b=1}^{d}\bm V_b\right)_{[S_X]} = (\bm U_m\otimes\cdots\otimes\bm U_1)\cm F_{[S_X]}(\bm V_d\otimes\cdots\otimes\bm V_1)^{\top}.
		\]
		Similarly, for the $a$th output-factor-variation term, the same matricization identity yields
		\[
			\left(\cm S\times_a \bm U_{a,\perp}\bm D_a\times_{\substack{1\leq j\leq m\\ j\neq a}}\bm U_j\times_{b=1}^{d}\bm V_b\right)_{[S_X]} = (\bm U_m\otimes\cdots\otimes\bm U_{a+1}\otimes \bm U_{a,\perp}\bm D_a\otimes \bm U_{a-1}\otimes\cdots\otimes\bm U_1)\cm S_{[S_X]}(\bm V_d\otimes\cdots\otimes\bm V_1)^{\top}.
		\]
		Therefore, after summing the core-variation term and all output-factor-variation terms, there exists a matrix $\bm M_{\mathrm{out}}$ such that $(\cm U_{\mathrm{out}})_{[S_X]} = \bm M_{\mathrm{out}}(\bm V^{\otimes})^{\top}$ with $\bm V^{\otimes}=\bm V_d\otimes\cdots\otimes\bm V_1$. Recall that each $\bm V_b\in\mathbb O_{p_b,r_{m+b}}$ has orthonormal columns, so $\bm V_b^\top\bm V_b=\bm I_{r_{m+b}}$. Hence,
		\[
			\begin{aligned}
				(\bm V^{\otimes})^\top\bm V^{\otimes}
				& = (\bm V_d\otimes\cdots\otimes\bm V_1)^\top(\bm V_d\otimes\cdots\otimes\bm V_1) \\
				& = (\bm V_d^\top\bm V_d)\otimes\cdots\otimes(\bm V_1^\top\bm V_1) \\
				& = \bm I_{r_{d+m}}\otimes\cdots\otimes\bm I_{r_{m+1}} \\
				& = \bm I_{\prod_{b=1}^{d}r_{m+b}}.
			\end{aligned}
		\]
		Here we used the standard Kronecker product identities $(\bm A_1\otimes\cdots\otimes\bm A_d)^\top=\bm A_1^\top\otimes\cdots\otimes\bm A_d^\top$ and $(\bm A_1\otimes\cdots\otimes\bm A_d)(\bm B_1\otimes\cdots\otimes\bm B_d)=(\bm A_1\bm B_1)\otimes\cdots\otimes(\bm A_d\bm B_d)$.

		Thus $\bm V^{\otimes}$ has orthonormal columns. Since each $\bm V_b\in\mathbb R^{p_b\times r_{m+b}}$ has $r_{m+b}$ columns, the Kronecker
		product $\bm V^{\otimes}=\bm V_d\otimes\cdots\otimes\bm V_1$ has $\prod_{b=1}^{d}r_{m+b}=r_{\bbm p}$ columns. Together with $(\bm V^{\otimes})^\top\bm V^{\otimes}=\bm I_{r_{\bbm p}}$, this implies that these $r_{\bbm p}$ columns are orthonormal and hence linearly independent. Therefore, $\dim\{\operatorname{span}(\bm V^{\otimes})\} = \rank(\bm V^{\otimes}) = r_{\bbm p}$. Moreover, since $(\cm U_{\mathrm{out}})_{[S_X]} = \bm M_{\mathrm{out}}(\bm V^{\otimes})^\top$, every row of $(\cm U_{\mathrm{out}})_{[S_X]}$ lies in the row space of $(\bm V^{\otimes})^\top$, equivalently in $\operatorname{span}(\bm V^{\otimes})$. Hence $\operatorname{row}((\cm U_{\mathrm{out}})_{[S_X]}) \subseteq \operatorname{span}(\bm V^{\otimes})$, and consequently
		\[
			\rank((\cm U_{\mathrm{out}})_{[S_X]}) = \dim\operatorname{row}((\cm U_{\mathrm{out}})_{[S_X]}) \leq \dim\{\operatorname{span}(\bm V^{\otimes})\} = r_{\bbm p}.
		\]

		Similarly, for $\cm U_{\mathrm{in}}$, every term in $\cm U_{\mathrm{in}}$ shares the same left factor $\bm U_m\otimes\cdots\otimes\bm U_1$ under the $[S_X]$-matricization. In particular, for the $b$th input-factor-variation term, we have
		\[
			\left(\cm S\times_{m+b}\bm V_{b,\perp}\bm H_b\times_{a=1}^{m}\bm U_a\times_{\substack{1\leq j\leq d\\ j\neq b}}\bm V_j\right)_{[S_X]} = (\bm U_m\otimes\cdots\otimes\bm U_1)\bm N_b,
		\]
		where $\bm N_b = \cm S_{[S_X]}(\bm V_d\otimes\cdots\otimes\bm V_{b+1} \otimes \bm V_{b,\perp}\bm H_b \otimes \bm V_{b-1}\otimes\cdots\otimes\bm V_1)^{\top}$. Hence, after summing over $b\in[d]$, there exists a matrix $\bm M_{\mathrm{in}}$ such that $(\cm U_{\mathrm{in}})_{[S_X]} = \bm U^{\otimes}\bm M_{\mathrm{in}}$ with $\bm U^{\otimes}=\bm U_m\otimes\cdots\otimes\bm U_1$. Recall that each $\bm U_a\in\mathbb O_{q_a,r_a}$ has orthonormal columns, so $\bm U_a^\top\bm U_a=\bm I_{r_a}$. Hence, $(\bm U^{\otimes})^\top\bm U^{\otimes} = \bm I_{\prod_{a=1}^{m}r_a}$.
		Here we again used the transpose and mixed-product identities of the Kronecker product. Thus $\bm U^{\otimes}$ has orthonormal columns. Since each $\bm U_a\in\mathbb R^{q_a\times r_a}$ has $r_a$ columns, the Kronecker product $\bm U^{\otimes}=\bm U_m\otimes\cdots\otimes\bm U_1$ has $\prod_{a=1}^{m}r_a=r_{\bbm q}$ columns. Together with $(\bm U^{\otimes})^\top\bm U^{\otimes}=\bm I_{r_{\bbm q}}$, this implies that these $r_{\bbm q}$ columns are orthonormal and hence linearly independent. Therefore, $\dim\{\operatorname{span}(\bm U^{\otimes})\} = \rank(\bm U^{\otimes}) = r_{\bbm q}$. Moreover, since $(\cm U_{\mathrm{in}})_{[S_X]} = \bm U^{\otimes}\bm M_{\mathrm{in}}$, every column of $(\cm U_{\mathrm{in}})_{[S_X]}$ lies in the column space of $\bm U^{\otimes}$, equivalently in $\operatorname{span}(\bm U^{\otimes})$. Hence $\operatorname{col}((\cm U_{\mathrm{in}})_{[S_X]}) \subseteq \operatorname{span}(\bm U^{\otimes})$, and consequently
		\[
			\rank\bigl((\cm U_{\mathrm{in}})_{[S_X]}\bigr) = \dim\operatorname{col}\bigl((\cm U_{\mathrm{in}})_{[S_X]}\bigr) \leq \dim\{\operatorname{span}(\bm U^{\otimes})\} = r_{\bbm q}.
		\]
		By the orthogonality of the tangent-space decomposition, $\cm U_{\mathrm{out}}$ and $\cm U_{\mathrm{in}}$ are orthogonal under the Frobenius inner product. Hence $\|\cm U\|_\F^2 = \|\cm U_{\mathrm{out}}\|_\F^2 + \|\cm U_{\mathrm{in}}\|_\F^2$. Combining the rank bounds for the two components with the inequality $\|\bm M\|_*\leq \sqrt{\rank(\bm M)}\|\bm M\|_\F$ and applying Cauchy's inequality yields
		\begin{align*}
			\|\cm U_{[S_X]}\|_* & \leq \|(\cm U_{\mathrm{out}})_{[S_X]}\|_* + \|(\cm U_{\mathrm{in}})_{[S_X]}\|_* \\
			& \leq \sqrt{r_{\bbm p}}\,\|(\cm U_{\mathrm{out}})_{[S_X]}\|_\F + \sqrt{r_{\bbm q}}\,\|(\cm U_{\mathrm{in}})_{[S_X]}\|_\F \\
			& = \sqrt{r_{\bbm p}}\,\|\cm U_{\mathrm{out}}\|_\F + \sqrt{r_{\bbm q}}\,\|\cm U_{\mathrm{in}}\|_\F \\
			& \leq \left(r_{\bbm p}+r_{\bbm q}\right)^{1/2}\left(\|\cm U_{\mathrm{out}}\|_\F^2 + \|\cm U_{\mathrm{in}}\|_\F^2\right)^{1/2} \\
			& = \sqrt{r_{\bbm q}+r_{\bbm p}}\,\|\cm U\|_\F. \qedhere
		\end{align*}
	\end{proof}

	\begin{proof}[\textbf{Proof of Lemma~\ref{lem:normal-projection-tucker-rank-2r}}]
		Since $\cm A\in\mathcal T_{\bbm r}(\cm A)$, we have $\mathcal P_{\mathcal T_{\bbm r}^{\perp}(\boldsymbol{\mathscr A})}(\cm A-\cm A^*)=-\mathcal P_{\mathcal T_{\bbm r}^{\perp}(\boldsymbol{\mathscr A})} (\cm A^*)$. Write $\cm A=\cm S\times_{s=1}^{d+m}\bm U_s$ and $\cm A^*=\cm S^*\times_{s=1}^{d+m}\bm U_s^*$, where $\bm U_s,\bm U_s^*\in\mathbb O_{d_s,r_s}$ with $d_s=q_s$ for $s\in[m]$ and $d_s=p_{s-m}$ for $s\in[d+m]\setminus[m]$. Let $\bm P_s=\bm U_s\bm U_s^\top$ and $\bm P_s^\perp=\bm I-\bm P_s$. Since $\bm I=\bm P_s+\bm P_s^\perp$ for each mode $s$, we have $\bm U_s^*=\bm P_s\bm U_s^*+\bm P_s^\perp\bm U_s^*$. Therefore, by the multilinearity of the Tucker product,
		\[
			\cm A^* = \cm S^*\times_{s=1}^{d+m}\left(\bm P_s\bm U_s^* + \bm P_s^\perp\bm U_s^*\right) = \sum_{J\subseteq[d+m]}\cm S^* \times_{s\in J}(\bm P_s^\perp\bm U_s^*)\times_{s\notin J}(\bm P_s\bm U_s^*).
		\]
		Here $J$ indexes the modes for which the complementary component $\bm P_s^\perp\bm U_s^*$ is selected in the multilinear expansion.
		
		For each mode $s\in[d+m]$, let $\bm V_s = \text{QR}(\cm S_{[s]}^\top)$, which corresponds to the row space of $\cm S_{[s]}$, and define $\bm W_\ell = (\bm U_{d+m}\otimes\cdots\otimes\bm U_{\ell+1}\otimes\bm U_{\ell-1}\otimes\cdots\otimes\bm U_1)\bm V_\ell$. Then the orthogonal projection of $\cm A^*$ onto $\mathcal T_{\bbm r}(\boldsymbol{\mathscr A})$ can be written as
		\[
			\mathcal P_{\mathcal T_{\bbm r}(\boldsymbol{\mathscr A})}(\cm A^*) = \cm A^*\times_{\ell=1}^{d+m}\bm P_\ell + \sum_{\ell=1}^{d+m}\mathrm{Fold}_{[\ell]}\bigl(\bm P_\ell^\perp\cm A^*_{[\ell]}\bm P_{\bm W_\ell}\bigr),
		\]
		where $\mathrm{Fold}_{[\ell]}(\cdot)$ denotes the inverse tensorization associated with the mode-$\ell$ unfolding, and $\bm P_{\bm W_\ell}=\bm W_\ell\bm W_\ell^\top$. We now inspect the mode-$s$ column space of each term in $\mathcal P_{\mathcal T_{\bbm r}(\boldsymbol{\mathscr A})}(\cm A^*)$. The first term is $\cm A^*\times_{j=1}^{d+m}\bm P_j$, and its mode-$s$ matricization has the form $(\cm A^*\times_{j=1}^{d+m}\bm P_j)_{[s]} = \bm P_s\cm A^*_{[s]}(\otimes_{j\neq s}\bm P_j)^{\top}$. Hence its mode-$s$ column space is contained in $\operatorname{span}(\bm P_s)=\operatorname{span}(\bm U_s)$.

		Next, consider the summand indexed by $s$, $\mathrm{Fold}_{[s]}(\bm P_s^\perp\cm A^*_{[s]}\bm P_{\bm W_s})$. Its mode-$s$ matricization is exactly $\bm P_s^\perp\cm A^*_{[s]}\bm P_{\bm W_s}$. Since $\cm A^*_{[s]} = \bm U_s^*\cm S^*_{[s]}(\otimes_{i\neq s}\bm U_i^*)^\top$, we have $\operatorname{col}(\cm A^*_{[s]})\subseteq\operatorname{span}(\bm U_s^*)$. Therefore, $\operatorname{col}(\bm P_s^\perp\cm A^*_{[s]}\bm P_{\bm W_s})\subseteq\operatorname{span}(\bm P_s^\perp\bm U_s^*)$. Thus the summand indexed by $s$ contributes only mode-$s$ directions lying in $\operatorname{span}(\bm P_s^\perp\bm U_s^*)$.

		It remains to examine a summand indexed by $\ell\neq s$. Define $\cm Z^{(\ell)} = \mathrm{Fold}_{[\ell]}(\bm P_\ell^\perp \cm A^*_{[\ell]}\bm P_{\bm W_\ell})$. Let $\bm U_{-\ell} = \bm U_{d+m}\otimes\cdots\otimes\bm U_{\ell+1}\otimes\bm U_{\ell-1}\otimes\cdots\otimes\bm U_1$. Since $\bm W_\ell=\bm U_{-\ell}\bm V_\ell$, we have $\bm P_{\bm W_\ell}=\bm U_{-\ell}\bm V_\ell\bm V_\ell^\top\bm U_{-\ell}^\top$. Consequently,
		\[
			\cm Z^{(\ell)}_{[\ell]} = \bm P_\ell^\perp \cm A^*_{[\ell]}\bm P_{\bm W_\ell} = \bm P_\ell^\perp \cm A^*_{[\ell]}\bm U_{-\ell}\bm V_\ell\bm V_\ell^\top\bm U_{-\ell}^\top = \bm C_\ell\bm U_{-\ell}^\top,
		\]
		where $\bm C_\ell=\bm P_\ell^\perp \cm A^*_{[\ell]} \bm U_{-\ell}\bm V_\ell\bm V_\ell^\top$. Equivalently, denote $\cm C^{(\ell)}=\mathrm{Fold}_{[\ell]}(\bm C_\ell)$; then $\cm Z^{(\ell)}=\cm C^{(\ell)}\times_{j\neq\ell}\bm U_j$. Therefore, for the fixed mode $s\neq\ell$, the mode-$s$ matricization is
		\[
			\cm Z^{(\ell)}_{[s]} = \bm U_s\cm C^{(\ell)}_{[s]}\left(\bm U_{d+m}\otimes\cdots\otimes\bm U_{s+1}\otimes\cdots\otimes\bm I_{d_\ell}\otimes\cdots\otimes\bm U_{s-1}\otimes\cdots\otimes\bm U_1\right)^\top.
		\]
		Hence, $\operatorname{col}(\cm Z^{(\ell)}_{[s]})\subseteq\operatorname{span}(\bm U_s)$ for $\ell\neq s$. Combining the three parts above, we have
		\[
			\operatorname{col}\bigl(\left[\mathcal P_{\mathcal T_{\bbm r}(\boldsymbol{\mathscr A})}(\cm A^*)\right]_{[s]}\bigr) \subseteq \operatorname{span}(\bm U_s) + \operatorname{span}(\bm P_s^\perp\bm U_s^*) = \mathcal S_s.
		\]
		The space $\mathcal S_s$ has dimension at most $2r_s$, because $\dim\operatorname{span}(\bm U_s)=r_s$ and $\dim\operatorname{span}(\bm P_s^\perp\bm U_s^*)\leq r_s$. Moreover, since $\bm U_s^*=\bm P_s\bm U_s^*+\bm P_s^\perp\bm U_s^*$ and $\operatorname{span}(\bm P_s\bm U_s^*)\subseteq\operatorname{span}(\bm U_s)$,
		we have
		\[
			\operatorname{span}(\bm U_s^*) \subseteq \operatorname{span}(\bm U_s) + \operatorname{span}(\bm P_s^\perp\bm U_s^*) = \mathcal S_s.
		\]
		Therefore, $\operatorname{col}(\cm A^*_{[s]})\subseteq\mathcal S_s$. Noting that $\mathcal P_{\mathcal T_{\bbm r}^{\perp}(\boldsymbol{\mathscr A})}(\cm A^*)=\cm A^*-\mathcal P_{\mathcal T_{\bbm r}(\boldsymbol{\mathscr A})}(\cm A^*)$, then $\operatorname{col}([\mathcal P_{\mathcal T_{\bbm r}(\boldsymbol{\mathscr A})}(\cm A^*)]_{[s]})\allowbreak \subseteq \mathcal S_s$. Hence $\rank([\mathcal P_{\mathcal T_{\bbm r}^{\perp}(\boldsymbol{\mathscr A})}(\cm A^*)]_{[s]})\leq\dim(\mathcal S_s)\leq2r_s$. Since $\mathcal P_{\mathcal T_{\bbm r}^{\perp}(\boldsymbol{\mathscr A})}(\cm A-\cm A^*)=-\mathcal P_{\mathcal T_{\bbm r}^{\perp}(\boldsymbol{\mathscr A})}(\cm A^*)$, the $s$-th Tucker rank of $\mathcal P_{\mathcal T_{\bbm r}^{\perp}(\boldsymbol{\mathscr A})}(\cm A-\cm A^*)$ is at most $2r_s$ for every $s\in[d+m]$.
	\end{proof}

	\begin{proof}[\textbf{Proof of Lemma~\ref{lem:tangent-space-tucker-rank-2r}}]
		Write $\cm A=\cm S\times_{a=1}^{d+m}\bm U_a$, where $\bm U_a\in\mathbb O_{d_a,r_a}$ with $d_a=q_a$ for $a\in[m]$ and $d_a=p_{a-m}$ for $a\in[d+m]\setminus[m]$. By the tangent-space characterization of the Tucker-rank manifold as in \cite{luo2024tensor}, any $\cm T\in\mathcal T_{\bbm r}(\cm A)$ can be written as
		\[
			\cm T = \cm F\times_{a=1}^{d+m}\bm U_a + \sum_{j=1}^{d+m}\cm S\times_{j}\bm U_{j,\perp}\bm D_{j}\times_{a\neq j}\bm U_a,
		\]
		where $\cm F\in \mathbb R^{r_1 \times \cdots \times r_{d+m}}$ and $\bm D_j\in\mathbb R^{(d_j-r_j)\times r_j}$. Fix a mode $\ell\in[d+m]$. In the mode-$\ell$ matricization, $\cm F\times_{a=1}^{d+m}\bm U_a$ and $\cm S\times_{j}\bm U_{j,\perp}\bm D_{j}\times_{a\neq j}\bm U_a$ with $j\neq \ell$ have their mode-$\ell$ column spaces contained in $\operatorname{span}(\bm U_\ell)$. $\cm S\times_{\ell}\bm U_{\ell,\perp}\bm D_{\ell}\times_{a\neq \ell}\bm U_a$ has mode-$\ell$ column space contained in $\operatorname{span}(\bm U_{\ell,\perp}\bm D_\ell)$, whose dimension is at most $r_\ell$. Therefore
		\[
			\operatorname{col}(\cm T_{[\ell]})\subseteq\operatorname{span}(\bm U_\ell) + \operatorname{span}(\bm U_{\ell,\perp}\bm D_\ell),
		\]
		and hence $\rank(\cm T_{[\ell]})\leq 2r_\ell$. Since this holds for $\ell\in[d+m]$, we have $\tucrank(\cm T)\leq 2\bbm r$.
	\end{proof}

\section{Technical Lemmas}\label{append:technical lemmas}
	\renewcommand{\theequation}{E.\arabic{equation}}
	\renewcommand{\theHequation}{E.\arabic{equation}}

	\renewcommand{\thelemma}{E.\arabic{lemma}}
	\renewcommand{\theHlemma}{E.\arabic{lemma}}

	\setcounter{lemma}{0}
	\setcounter{equation}{0}

	\begin{lemma}[Hanson-Wright inequality]\label{lem:hanson-wright}
		Let $\bbm \zeta = (\zeta_1,\cdots,\zeta_p)\in \mathbb{R}^p$ be a random vector with independent mean zero sub-Gaussian coordinates. Then, for any $\bm A\in \mathbb{R}^{p\times p}$ and $t\geq0$,
		\[
			\mathbb{P}\left(|\bbm \zeta^\top \bm A \bbm \zeta - \mathbb{E}\bbm \zeta^\top \bm A \bbm \zeta|\geq t\right) \leq 2\exp\left(-C\min\left\{\frac{t^2}{\max_j\|\zeta_j\|_{\psi_2}^4\|\bm A\|_\F^2},\frac{t}{\max_j\|\zeta_j\|_{\psi_2}^2\|\bm A\|_\op}\right\}\right).
		\]
	\end{lemma}

	\begin{lemma}[Bernstein's inequality]\label{lem:bernstein-inequality}
		Let $\zeta_1,\cdots,\zeta_p$ be independent mean-zero sub-exponential random variables. Then, for any $t\geq 0$,
		\[
			\mathbb P\left(\left|\sum_{i=1}^p \zeta_i\right|\geq t\right) \leq 2\exp\left(-C\min\left\{\frac{t^2}{\sum_{i=1}^p \|\zeta_i\|_{\psi_1}^2}, \frac{t}{\max_{1\leq i\leq p}\|\zeta_i\|_{\psi_1}}\right\}\right).
		\]
	\end{lemma}

		\begin{lemma}[{\citealp[Quasi-projection property of T-HOSVD and ST-HOSVD (Proposition 3)]{luo2024tensor}}]\label{lem:quasi-projection-thosvd-sthosvd}
		Let $\mathcal M_{\bbm r}$ denote the set of tensors with Tucker rank at most $\bbm r=(r_1,\cdots,r_{d+m})$. Let $\mathcal R_{\bbm r}(\cdot)$ be either the T-HOSVD or ST-HOSVD truncation operator. Then, for any $\cm T\in\mathbb R^{q_1\times\cdots\times q_m\times p_1\times\cdots\times p_d}$, we have
		\[
			\bigl\|\cm T-\mathcal R_{\bbm r}(\cm T)\bigr\|_{\F} \leq \sqrt{d+m}\,\bigl\|\cm T-\mathcal P_{\mathcal M_{\bbm r}}(\cm T)\bigr\|_{\F},
		\]
		where $\mathcal P_{\mathcal M_{\bbm r}}(\cm T)$ denotes a best Tucker-rank-$\bbm r$ approximation of $\cm T$ under the Frobenius norm; that is, for any $\cm T'\in\mathbb R^{q_1\times\cdots\times q_m\times p_1\times\cdots\times p_d}$ with $\tucrank(\cm T')\leq\bbm r$, one has $\bigl\|\cm T-\mathcal P_{\mathcal M_{\bbm r}}(\cm T)\bigr\|_{\F} \leq \bigl\|\cm T-\cm T'\bigr\|_{\F}$.
	\end{lemma}

	\begin{lemma}[{\citealp[Lemma~3]{luo2024tensor}}]\label{lem:normal-component-bound}
		Let $\cm T,\cm T' \in \mathbb R^{q_1\times\cdots\times q_m\times p_1\times\cdots\times p_d}$ be two order-$(d+m)$ tensors satisfying $\tucrank(\cm T)=\tucrank(\cm T')=\bbm r = (r_1,\cdots,r_{d+m})$.
		Then
		\[
			\Bigl\|\mathcal P_{\mathcal T_{\bbm r}^{\perp}(\boldsymbol{\mathscr T'})}\bigl(\cm T'-\cm T\bigr)\Bigr\|_\F \leq \frac{2(d+m)\|\cm T'-\cm T\|_\F^2}{\min_{s\in[d+m]}\sigma_{r_s}\bigl(\cm T_{[s]}\bigr)}.
		\]
	\end{lemma}

\section{Algorithmic Details}\label{sec:Algorithmic_Details}

\subsection{Two-stage Single Client Estimation}\label{sec:two-stage-single-est}

	\begin{breakablealgorithm}
	\caption{Two-stage single-client learning for personalized tensor regression}
	\label{alg:single_two_stage}
	\begin{algorithmic}[1]
	\State \textbf{Input:} Local dataset $\mathsf D_k=\{(\cm X_{k,i},\cm Y_{k,i})\}_{i=1}^{n_k}$; number of Stage-I iterations $T_g^{(k)}$; number of Stage-II iterations $T_l^{(k)}$; Tucker rank $\bbm r$; step sizes $\eta_{\mathscr A,k}$ and $\eta_{D,k}$; regularization parameter $\omega_k$; initialization $\cm A_{0,k}^{loc,(0)}$.

	\Statex \textbf{Stage I: Local Representation Learning}
	\For{$t=0,1,\cdots,T_g^{(k)}-1$}
		\State Compute the local gradient $\cm G_{\mathscr A,k}^{loc,(t)}$ and update $\cm A_{0,k}^{loc,(t+1)}$ by \eqref{eq:local-rgd-update}.
	\EndFor
	\State Set $\widehat{\cm A}_{0,k}^{loc}\gets \cm A_{0,k}^{loc,(T_g^{(k)})}$.

	\Statex \textbf{Stage II: Local Personalized Refinement}
	\State Initialize $\cm B_k^{loc,(0)}\gets \mathbf 0$, $q_{k,0}^{loc}\gets 1$, and set $\cm U_k^{loc,(0)}\gets \cm B_k^{loc,(0)}$.
	\For{$t=0,1,\cdots,T_l^{(k)}-1$}
		\State Update $\cm B_k^{loc,(t+1)}$ by the same FISTA step as in \eqref{eq:personalized-estimator}, with $\widehat{\cm A}_0$ replaced by $\widehat{\cm A}_{0,k}^{loc}$, $\cm B_k^{(t)}$ replaced by $\cm B_k^{loc,(t)}$, and $\cm U_k^{(t)}$ replaced by $\cm U_k^{loc,(t)}$.
		\State Update the momentum parameter $q_{k,t+1}^{loc}$ and extrapolated point $\cm U_k^{loc,(t+1)}$ by the same extrapolation rule as in \eqref{eq:momentum-param-extra-iter}, with $q_{k,t}$, $\cm B_k^{(t)}$, and $\cm U_k^{(t)}$ replaced by their local counterparts.
	\EndFor
	\State Set $\widehat{\cm B}_k^{loc}\gets \cm B_k^{loc,(T_l^{(k)})}$ and $\widehat{\cm A}_k^{loc}\gets \widehat{\cm A}_{0,k}^{loc}+\widehat{\cm B}_k^{loc}$.

	\State \textbf{Output:} Local Stage-I estimator $\widehat{\cm A}_{0,k}^{loc}$ and personalized estimator $\widehat{\cm A}_k^{loc}$.
	\end{algorithmic}
	\end{breakablealgorithm}

\subsection{T-HOSVD and ST-HOSVD}\label{sec:THOSVD_STHOSVD}
	Given an mode-$\ell$ tensor $\cm Y\in\mathbb R^{p_1\times\cdots\times p_\ell}$ and a target Tucker rank $\bbm r=(r_1,\cdots,r_\ell)$, a common retraction onto the low Tucker-rank manifold is based on truncated high-order singular value decomposition (T-HOSVD) or sequentially truncated high-order singular value decomposition (ST-HOSVD).

	\begin{breakablealgorithm}
	\caption{T-HOSVD}\label{alg:thosvd}
	\begin{algorithmic}[1]
	\State \textbf{Input:} tensor $\cm Y\in\mathbb R^{p_1\times\cdots\times p_\ell}$; Tucker rank $\bbm r=(r_1,\cdots,r_\ell)$.
	\For{$k=1,2,\cdots,\ell$}
		\State Compute $\bm U_k^{(0)} \gets \mathrm{SVD}_{r_k}(\cm Y_{[k]})\in\mathbb R^{p_k\times r_k}$.
	\EndFor
	\State \textbf{Output:} $\widehat{\cm Y}\gets \cm Y\times_{k=1}^{\ell}(\bm U_k^{(0)}\bm U_k^{(0)\top})$.
	\end{algorithmic}
	\end{breakablealgorithm}

	\begin{breakablealgorithm}
	\caption{ST-HOSVD}\label{alg:sthosvd}
	\begin{algorithmic}[1]
	\State \textbf{Input:} tensor $\cm Y\in\mathbb R^{p_1\times\cdots\times p_\ell}$; Tucker rank $\bbm r=(r_1,\cdots,r_\ell)$.
	\State Compute $\bm U_1^{(0)}\gets \mathrm{SVD}_{r_1}(\cm Y_{[1]})$.
	\For{$k=2,3,\cdots,\ell$}
		\State Compute
		\[
		\bm U_k^{(0)}\gets \mathrm{SVD}_{r_k}\Bigl(\bigl(\cm Y\times_{j=1}^{k-1}(\bm U_j^{(0)}\bm U_j^{(0)\top})\bigr)_{[k]}\Bigr).
		\]
	\EndFor
	\State \textbf{Output:} $\widehat{\cm Y}\gets \cm Y\times_{k=1}^{\ell}(\bm U_k^{(0)}\bm U_k^{(0)\top})$.
	\end{algorithmic}
	\end{breakablealgorithm}

\subsection{ADMM for Single Client Estimation}\label{sec:ADMM}
	For a fixed client $k\in[K]$, the single-client benchmark estimates the low-rank component $\cm A_{0,k}$ and the sparse deviation $\cm B_k$ by
	\begin{align}\label{eq:single-client-admm-objective}
		(\widetilde{\cm A}_{0,k},\widetilde{\cm B}_k) = \argmin_{\cm A_{0,k},\cm B_k}\Biggl\{\frac{1}{2n_k}\sum_{i=1}^{n_k}\bigl\|\cm Y_{k,i}-\langle \cm A_{0,k}+\cm B_k,\cm X_{k,i}\rangle\bigr\|_{\F}^2 + \lambda_k\|(\cm A_{0,k})_{[S_X]}\|_* + \varpi_k\|\cm B_k\|_1\Biggr\},
	\end{align}
	subject to $\|(\cm B_k)_{[S_X]}\|_{\op}\leq \zeta$, where $S_X$ denotes the matricization used for the low-rank structure and the deviation constraint. To obtain separable ADMM subproblems, introduce an auxiliary tensor $\cm F$ for the regression term, an auxiliary matrix $\bm Z$ for the nuclear-norm penalty, and an auxiliary matrix $\bm V$ for the operator-norm constraint:
	\begin{align}\label{eq:admm-splitting}
		\cm F = \cm A_{0,k}+\cm B_k,
		\qquad
		\bm Z=(\cm A_{0,k})_{[S_X]},
		\qquad
		\bm V=(\cm B_k)_{[S_X]},
		\qquad
		\|\bm V\|_{\op}\leq \zeta.
	\end{align}
	Then \eqref{eq:single-client-admm-objective} is equivalently written as
	\begin{align}\label{eq:single-client-admm-split-objective}
		\min_{\boldsymbol{\mathscr A}_{0,k},\boldsymbol{\mathscr B}_k,\boldsymbol{\mathscr F},\bm Z,\bm V}
		&\ \frac{1}{2n_k}\sum_{i=1}^{n_k}\bigl\|\cm Y_{k,i}-\langle \cm F,\cm X_{k,i}\rangle\bigr\|_{\F}^2
		+\lambda_k\|\bm Z\|_*+\varpi_k\|\cm B_k\|_1+\iota_{\op,\zeta}(\bm V) \\
		\text{s.t.}
		&\ \cm F=\cm A_{0,k}+\cm B_k,
		\qquad
		\bm Z=(\cm A_{0,k})_{[S_X]},
		\qquad
		\bm V=(\cm B_k)_{[S_X]},
	\end{align}
	where the indicator function of the operator-norm ball is defined by
	\[
		\iota_{\op,\zeta}(\bm V)
			=\begin{cases}
			0, & \|\bm V\|_{\op}\leq \zeta,\\
			+\infty, & \text{otherwise}.
			\end{cases}
	\]

	Let $\cm U$ be the scaled dual variable for $\cm F=\cm A_{0,k}+\cm B_k$, let $\bm U_{\mathscr Z}$ be the scaled dual variable for $\bm Z=(\cm A_{0,k})_{[S_X]}$, and let $\bm U_{\mathscr V}$ be the scaled dual variable for $\bm V=(\cm B_k)_{[S_X]}$. For a penalty parameter $\rho>0$, the scaled augmented Lagrangian is
	\begin{align}\label{eq:aug-lag-single}
		\mathcal L_{\rho}(\cm A_{0,k},\cm B_k,\cm F,\bm Z,\bm V;\cm U,\bm U_{\mathscr Z},\bm U_{\mathscr V}) &= \frac{1}{2n_k}\sum_{i=1}^{n_k}\bigl\|\cm Y_{k,i}-\langle \cm F,\cm X_{k,i}\rangle\bigr\|_{\F}^2+\lambda_k\|\bm Z\|_*+\varpi_k\|\cm B_k\|_1+\iota_{\op,\zeta}(\bm V) \notag\\
		&\quad + \frac{\rho}{2}\|\cm F-\cm A_{0,k}-\cm B_k+\cm U\|_{\F}^2 + \frac{\rho}{2}\|(\cm A_{0,k})_{[S_X]}-\bm Z+\bm U_{\mathscr Z}\|_{\F}^2 \notag\\
		&\quad + \frac{\rho}{2}\|(\cm B_k)_{[S_X]}-\bm V+\bm U_{\mathscr V}\|_{\F}^2.
	\end{align}
	We derive the ADMM updates by minimizing \eqref{eq:aug-lag-single} with respect to each primal variable while fixing the remaining variables; see Algorithm~\ref{alg:admm-single} for a summary.

	\noindent
	\textbf{(I) $\cm F$-update.}
	The $\cm F$-subproblem is
	\[
		\cm F^{(t+1)} = \argmin_{\boldsymbol{\mathscr F}}\ \frac{1}{2n_k}\sum_{i=1}^{n_k}\bigl\|\cm Y_{k,i}-\langle \cm F,\cm X_{k,i}\rangle\bigr\|_{\F}^2 + \frac{\rho}{2}\|\cm F-\cm A_{0,k}^{(t)}-\cm B_k^{(t)}+\cm U^{(t)}\|_{\F}^2.
	\]
	Let $S_0=[m]$ and denote $\bm F=\cm F_{[S_0]}\in\mathbb R^{q\times p}$, where $q=\prod_{\ell=1}^{m}q_\ell$ and $p=\prod_{\ell=1}^{d}p_\ell$. Define
	\[
		\bm X_k=[\vect(\cm X_{k,1}),\cdots,\vect(\cm X_{k,n_k})]\in\mathbb R^{p\times n_k},
		\qquad
		\bm Y_k=[\vect(\cm Y_{k,1}),\cdots,\vect(\cm Y_{k,n_k})]\in\mathbb R^{q\times n_k}.
	\]
	Let $\bm W^{(t)}=(\cm A_{0,k}^{(t)}+\cm B_k^{(t)}-\cm U^{(t)})_{[S_0]}$. Then
	\[
		\bm F^{(t+1)} = \Bigl(\frac{1}{n_k}\bm Y_k\bm X_k^{\top}+\rho\bm W^{(t)}\Bigr)\Bigl(\frac{1}{n_k}\bm X_k\bm X_k^{\top}+\rho\bm I_p\Bigr)^{-1},
	\]
	and $\cm F^{(t+1)}$ is obtained by folding $\bm F^{(t+1)}$ back to the original tensor shape.

	\noindent\textbf{(II) $\bm Z$-update.}
	The $\bm Z$-subproblem is
	\[
		\bm Z^{(t+1)} = \argmin_{\bm Z}\ \lambda_k\|\bm Z\|_*+\frac{\rho}{2}\bigl\|\bm Z-\bigl\{(\cm A_{0,k}^{(t)})_{[S_X]}+\bm U_{\mathscr Z}^{(t)}\bigr\}\bigr\|_{\F}^2.
	\]
	Therefore,
	\[
		\bm Z^{(t+1)}=\operatorname{SVT}_{\lambda_k/\rho}\bigl((\cm A_{0,k}^{(t)})_{[S_X]}+\bm U_{\mathscr Z}^{(t)}\bigr).
	\]

	\noindent\textbf{(III) $\bm V$-update.}
	The $\bm V$-subproblem is
	\[
		\bm V^{(t+1)} = \argmin_{\bm V}\ \iota_{\op,\zeta}(\bm V) + \frac{\rho}{2}\bigl\|\bm V-\bigl\{(\cm B_k^{(t)})_{[S_X]}+\bm U_{\mathscr V}^{(t)}\bigr\}\bigr\|_{\F}^2.
	\]
	Thus $\bm V^{(t+1)}$ is the Euclidean projection of $(\cm B_k^{(t)})_{[S_X]}+\bm U_{\mathscr V}^{(t)}$ onto the operator-norm ball. If
	\[
		(\cm B_k^{(t)})_{[S_X]}+\bm U_{\mathscr V}^{(t)} = \bm P\operatorname{diag}(\sigma_1,\cdots,\sigma_r)\bm Q^{\top}
	\]
	is its singular value decomposition, then
	\[
		\bm V^{(t+1)} = \bm P\operatorname{diag}\bigl(\min\{\sigma_1,\zeta\},\cdots,\min\{\sigma_r,\zeta\}\bigr)\bm Q^{\top}.
	\]

	\noindent\textbf{(IV) $\cm A_{0,k}$-update.}
	The $\cm A_{0,k}$-subproblem is
	\[
		\cm A_{0,k}^{(t+1)} = \argmin_{\boldsymbol{\mathscr A}_{0,k}}\ \frac{\rho}{2}\|\cm A_{0,k}-(\cm F^{(t+1)}-\cm B_k^{(t)}+\cm U^{(t)})\|_{\F}^2 + \frac{\rho}{2}\|\cm A_{0,k}-\mathrm{Fold}_{[S_X]}(\bm Z^{(t+1)}-\bm U_{\mathscr Z}^{(t)})\|_{\F}^2.
	\]
	Since matricization preserves the Frobenius norm, the minimizer is
	\[
		\cm A_{0,k}^{(t+1)} = \frac{\cm F^{(t+1)}-\cm B_k^{(t)}+\cm U^{(t)} + \mathrm{Fold}_{[S_X]}\bigl(\bm Z^{(t+1)}-\bm U_{\mathscr Z}^{(t)}\bigr)}{2}.
	\]

	\noindent\textbf{(V) $\cm B_k$-update.}
	The $\cm B_k$-subproblem is
	\[
		\cm B_k^{(t+1)} = \argmin_{\boldsymbol{\mathscr B}_k}\ \varpi_k\|\cm B_k\|_1 + \frac{\rho}{2}\|\cm B_k-(\cm F^{(t+1)}-\cm A_{0,k}^{(t+1)}+\cm U^{(t)})\|_{\F}^2 + \frac{\rho}{2}\|\cm B_k-\mathrm{Fold}_{[S_X]}(\bm V^{(t+1)}-\bm U_{\mathscr V}^{(t)})\|_{\F}^2.
	\]
	This is the proximal mapping of the $\ell_1$ norm around the averaged center
	\[
		\widetilde{\cm B}_k^{(t+1)} = \frac{\cm F^{(t+1)}-\cm A_{0,k}^{(t+1)}+\cm U^{(t)} + \mathrm{Fold}_{[S_X]}\bigl(\bm V^{(t+1)}-\bm U_{\mathscr V}^{(t)}\bigr)}{2}.
	\]
	Therefore,
	\[
		\cm B_k^{(t+1)} = \operatorname{Soft}_{\varpi_k/(2\rho)}\bigl(\widetilde{\cm B}_k^{(t+1)}\bigr).
	\]

	\noindent\textbf{(VI) Dual updates.}
	The scaled dual variables are updated by
	\[
		\cm U^{(t+1)}=\cm U^{(t)}+\cm F^{(t+1)}-\cm A_{0,k}^{(t+1)}-\cm B_k^{(t+1)},
	\]
	\[
		\bm U_{\mathscr Z}^{(t+1)}=\bm U_{\mathscr Z}^{(t)}+(\cm A_{0,k}^{(t+1)})_{[S_X]}-\bm Z^{(t+1)},
		\qquad
		\bm U_{\mathscr V}^{(t+1)}=\bm U_{\mathscr V}^{(t)}+(\cm B_k^{(t+1)})_{[S_X]}-\bm V^{(t+1)}.
	\]

	\begin{breakablealgorithm}
	\caption{ADMM updates for the single-client estimator}
	\label{alg:admm-single}
	\begin{algorithmic}[1]
	\Require data $\mathcal D_k=\{(\cm X_{k,i},\cm Y_{k,i})\}_{i=1}^{n_k}$; parameters $\lambda_k,\varpi_k,\zeta$; penalty $\rho>0$; tolerances $\varepsilon_{\rm pri},\varepsilon_{\rm dual}>0$.
	\State Initialize $(\cm A_{0,k}^{(0)},\cm B_k^{(0)},\cm F^{(0)},\bm Z^{(0)},\bm V^{(0)},\cm U^{(0)},\bm U_{\mathscr Z}^{(0)},\bm U_{\mathscr V}^{(0)})$.
	\State Form $\bm X_k\gets[\vect(\cm X_{k,1}),\cdots,\vect(\cm X_{k,n_k})]$ and $\bm Y_k\gets[\vect(\cm Y_{k,1}),\cdots,\vect(\cm Y_{k,n_k})]$.
	\For{$t=0,1,2,\cdots$}
	\State $\bm W^{(t)}\gets(\cm A_{0,k}^{(t)}+\cm B_k^{(t)}-\cm U^{(t)})_{[S_0]}$ with $S_0=[m]$.
	\State $\bm F^{(t+1)}\gets\Bigl(\frac{1}{n_k}\bm Y_k\bm X_k^{\top}+\rho\bm W^{(t)}\Bigr)\Bigl(\frac{1}{n_k}\bm X_k\bm X_k^{\top}+\rho\bm I_p\Bigr)^{-1}$.
	\State $\cm F^{(t+1)}\gets\mathrm{Fold}_{[S_0]}(\bm F^{(t+1)})$.
	\State $\bm Z^{(t+1)}\gets\operatorname{SVT}_{\lambda_k/\rho}\bigl((\cm A_{0,k}^{(t)})_{[S_X]}+\bm U_{\mathscr Z}^{(t)}\bigr)$.
	\State Compute the singular value decomposition
	$(\cm B_k^{(t)})_{[S_X]}+\bm U_{\mathscr V}^{(t)}=\bm P^{(t)}\operatorname{diag}(\sigma_1^{(t)},\cdots,\sigma_r^{(t)})(\bm Q^{(t)})^{\top}$, and set
	\[
		\bm V^{(t+1)}
		\gets
		\bm P^{(t)}\operatorname{diag}\bigl(\min\{\sigma_1^{(t)},\zeta\},\cdots,\min\{\sigma_r^{(t)},\zeta\}\bigr)(\bm Q^{(t)})^{\top}.
	\]
	\State $\cm A_{0,k}^{(t+1)}\gets\bigl[\cm F^{(t+1)}-\cm B_k^{(t)}+\cm U^{(t)}+\mathrm{Fold}_{[S_X]}(\bm Z^{(t+1)}-\bm U_{\mathscr Z}^{(t)})\bigr]/2$.
	\State $\widetilde{\cm B}_k^{(t+1)}\gets\bigl[\cm F^{(t+1)}-\cm A_{0,k}^{(t+1)}+\cm U^{(t)}+\mathrm{Fold}_{[S_X]}(\bm V^{(t+1)}-\bm U_{\mathscr V}^{(t)})\bigr]/2$.
	\State $\cm B_k^{(t+1)}\gets\operatorname{Soft}_{\varpi_k/(2\rho)}\bigl(\widetilde{\cm B}_k^{(t+1)}\bigr)$.
	\State $\cm U^{(t+1)}\gets\cm U^{(t)}+\cm F^{(t+1)}-\cm A_{0,k}^{(t+1)}-\cm B_k^{(t+1)}$.
	\State $\bm U_{\mathscr Z}^{(t+1)}\gets\bm U_{\mathscr Z}^{(t)}+(\cm A_{0,k}^{(t+1)})_{[S_X]}-\bm Z^{(t+1)}$.
	\State $\bm U_{\mathscr V}^{(t+1)}\gets\bm U_{\mathscr V}^{(t)}+(\cm B_k^{(t+1)})_{[S_X]}-\bm V^{(t+1)}$.
	\State $\cm R_{\mathscr B}^{(t+1)}\gets\cm F^{(t+1)}-\cm A_{0,k}^{(t+1)}-\cm B_k^{(t+1)}$.
	\State $\cm R_{\mathscr A}^{(t+1)}\gets(\cm A_{0,k}^{(t+1)})_{[S_X]}-\bm Z^{(t+1)}$, \quad $\bm R_{\mathscr V}^{(t+1)}\gets(\cm B_k^{(t+1)})_{[S_X]}-\bm V^{(t+1)}$.
	\State $\cm S_{\mathscr A}^{(t+1)}\gets\rho\{(\cm A_{0,k}^{(t+1)}-\cm A_{0,k}^{(t)})+(\cm B_k^{(t+1)}-\cm B_k^{(t)})\}$.
	\State $\bm S_{\mathscr Z}^{(t+1)}\gets\rho(\bm Z^{(t+1)}-\bm Z^{(t)})$, \quad $\bm S_{\mathscr V}^{(t+1)}\gets\rho(\bm V^{(t+1)}-\bm V^{(t)})$.
	\State \textbf{Stop} if
	\[
		\max\{\|\cm R_{\mathscr B}^{(t+1)}\|_{\F},\|\cm R_{\mathscr A}^{(t+1)}\|_{\F},\|\bm R_{\mathscr V}^{(t+1)}\|_{\F}\}
		\leq\varepsilon_{\rm pri} \ \text{and} \ \max\{\|\cm S_{\mathscr A}^{(t+1)}\|_{\F},\|\bm S_{\mathscr Z}^{(t+1)}\|_{\F},\|\bm S_{\mathscr V}^{(t+1)}\|_{\F}\}\leq\varepsilon_{\rm dual}.
	\]
	\EndFor
	\Ensure $(\widetilde{\cm A}_{0,k},\widetilde{\cm B}_k)\gets(\cm A_{0,k}^{(t)},\cm B_k^{(t)})$.
	\end{algorithmic}
	\end{breakablealgorithm}

	Here $\operatorname{Soft}_{\tau}(\cdot)$ denotes the entrywise soft-thresholding operator,
	\[
		\operatorname{Soft}_{\tau}(x)=\operatorname{sign}(x)(|x|-\tau)_+,
		\qquad
		(\operatorname{Soft}_{\tau}(\bm M))_{ij}=\operatorname{Soft}_{\tau}(M_{ij}),
	\]
	where $(u)_+=\max\{u,0\}$. Moreover, $\operatorname{SVT}_{\tau}(\cdot)$ denotes the singular-value soft-thresholding operator: for a matrix $\bm M$ with singular value decomposition $\bm M=\bm P\operatorname{diag}(\sigma_1,\cdots,\sigma_r)\bm Q^{\top}$,
	\[
		\operatorname{SVT}_{\tau}(\bm M) = \bm P\operatorname{diag}\bigl((\sigma_1-\tau)_+,\cdots,(\sigma_r-\tau)_+\bigr)\bm Q^{\top}.
	\]

\subsection{Tuning Parameter Selection for Single-Client Estimation}\label{sec:tuning-single}
	For the two-stage single-client procedure, we use the same validation-based principle as the federated two-stage procedure to select the tuning parameters.For a fixed client $k$, we randomly split its local dataset $\mathsf D_k$ into a training set $\mathsf D_k^{\mathrm{tr}}$ and a validation set $\mathsf D_k^{\mathrm{val}}$. After fixing the Tucker rank, the number of Stage~I and Stage~II iterations, the retraction method, and the Stage~II stepsize $\eta_{D,k}$, the remaining tuning parameters are the Stage~I stepsize $\eta_{\mathscr A,k}$ and the Stage~II regularization parameter $\omega_k$.

	For each candidate pair $(\eta_{\mathscr A,k},\omega_k)$ on a prescribed grid, all model fitting steps are carried out using only the training set $\mathsf D_k^{\mathrm{tr}}$. We first run the local Stage~I procedure to obtain the low-rank estimator $\widehat{\cm A}_{0,k}^{loc}(\eta_{\mathscr A,k})$, and then run the local Stage~II personalization step to obtain the sparse deviation estimator $\widehat{\cm B}_{k}^{loc}(\widehat{\cm A}_{0,k}^{loc}(\eta_{\mathscr A,k}),\omega_k)$. Each candidate pair is evaluated by the validation squared Frobenius prediction error
	\[
		\mathrm{Error}_{\mathrm{val},k}^{loc} = \frac{1}{|\mathsf D_k^{\mathrm{val}}|}\sum_{(\boldsymbol{\mathscr X}_{k,i},\boldsymbol{\mathscr Y}_{k,i})\in \mathsf D_k^{\mathrm{val}}}\left\|\cm Y_{k,i} - \left\langle\widehat{\cm A}_{0,k}^{loc}(\eta_{\mathscr A,k}) + \widehat{\cm B}_{k}^{loc}\bigl(\widehat{\cm A}_{0,k}^{loc}(\eta_{\mathscr A,k}),\omega_k\bigr),\cm X_{k,i}\right\rangle\right\|_{\F}^2.
	\]
	The selected tuning parameters are the candidate pair that minimizes $\mathrm{Error}_{\mathrm{val},k}^{loc}$ over the prescribed grid.

\linespread{1.54}
\selectfont{}

\setlength{\bibsep}{1pt}
\bibliography{mybib}

@inproceedings{dwork2006calibrating,
  title={{Calibrating Noise to Sensitivity in Private Data Analysis}},
  author={Dwork, Cynthia and McSherry, Frank and Nissim, Kobbi and Smith, Adam},
  booktitle={Theory of Cryptography Conference},
  volume={3876},
  pages={265--284},
  year={2006},
  publisher={Springer}
}

@article{lock2018tensor,
  title={{Tensor-on-tensor Regression}},
  author={Lock, Eric F},
  journal={Journal of Computational and Graphical Statistics},
  volume={27},
  pages={638--647},
  year={2018},
  publisher={Taylor \& Francis}
}

@article{zhou2013tensor,
  title={{Tensor Regression with Applications in Neuroimaging Data Analysis}},
  author={Zhou, Hua and Li, Lexin and Zhu, Hongtu},
  journal={Journal of the American Statistical Association},
  volume={108},
  pages={540--552},
  year={2013},
  publisher={Taylor \& Francis}
}

@article{li2018tucker,
  title={{Tucker Tensor Regression and Neuroimaging Analysis}},
  author={Li, Xiaoshan and Xu, Da and Zhou, Hua and Li, Lexin},
  journal={Statistics in Biosciences},
  volume={10},
  pages={520--545},
  year={2018},
  publisher={Springer}
}

@article{koltchinskii2011nuclear,
  title={{Nuclear-norm Penalization and Optimal Rates for Noisy Low-rank Matrix Completion}},
  author={Vladimir Koltchinskii and Karim Lounici and Alexandre B. Tsybakov},
  journal={The Annals of Statistics},
  volume={39},
  pages={2302--2329},
  year={2011},
  publisher={Institute of Mathematical Statistics}
}

@article{negahban2012restricted,
  title={{Restricted Strong Convexity and Weighted Matrix Completion: Optimal Bounds with Noise}},
  author={Negahban, Sahand and Wainwright, Martin J},
  journal={Journal of Machine Learning Research},
  volume={13},
  pages={1665--1697},
  year={2012},
  publisher={JMLR. org}
}

@article{li2017parsimonious,
  title={{Parsimonious Tensor Response Regression}},
  author={Li, Lexin and Zhang, Xin},
  journal={Journal of the American Statistical Association},
  volume={112},
  pages={1131--1146},
  year={2017},
  publisher={Taylor \& Francis}
}

@book{seber2003linear,
  title={{Linear Regression Analysis}},
  author={Seber, George AF and Lee, Alan J},
  year={2003},
  publisher={John Wiley \& Sons}
}

@article{pfeiffer2021least,
  title={{Least Squares and Maximum Likelihood Estimation of Sufficient Reductions in Regressions with Matrix-valued Predictors}},
  author={Pfeiffer, Ruth M and Kapla, Daniel B and Bura, Efstathia},
  journal={International Journal of Data Science and Analytics},
  volume={11},
  pages={11--26},
  year={2021},
  publisher={Springer}
}

@article{ding2018matrix,
  title={{Matrix Variate Regressions and Envelope Models}},
  author={Ding, Shanshan and Dennis Cook, R},
  journal={Journal of the Royal Statistical Society Series B: Statistical Methodology},
  volume={80},
  pages={387--408},
  year={2018},
  publisher={Oxford University Press}
}

@book{montgomery2021introduction,
  title={{Introduction to Linear Regression Analysis}},
  author={Montgomery, Douglas C and Peck, Elizabeth A and Vining, G Geoffrey},
  year={2021},
  publisher={John Wiley \& Sons}
}

@inproceedings{mu2014square,
  title={{Square Deal: Lower Bounds and Improved Relaxations for Tensor Recovery}},
  author={Mu, Cun and Huang, Bo and Wright, John and Goldfarb, Donald},
  booktitle={International Conference on Machine Learning},
  pages={73--81},
  year={2014},
  organization={PMLR}
}

@article{garvesh2019convex,
title = {{Convex regularization for high-dimensional multiresponse tensor regression}},
author = {Garvesh Raskutti and Ming Yuan and Han Chen},
journal = {The Annals of Statistics},
volume = {47},
pages = {1554--1584},
year = {2019},
publisher = {Institute of Mathematical Statistics}
}

@article{agarwal2012noisy,
  title={{Noisy Matrix Decomposition via Convex Relaxation: Optimal Rates in High Dimensions}},
  author={Agarwal, Alekh and Negahban, Sahand and Wainwright, Martin J},
  journal={The Annals of Statistics},
  pages={1171--1197},
  year={2012},
  publisher={JSTOR}
}

@article{luo2024tensor,
  title={{Tensor-on-Tensor Regression: Riemannian Optimization, Over-Parameterization, Statistical-Computational Gap and Their Interplay}},
  author={Luo, Yuetian and Zhang, Anru R},
  journal={The Annals of Statistics},
  volume={52},
  pages={2583--2612},
  year={2024},
  publisher={Institute of Mathematical Statistics}
}

@article{beck2009fast,
  title={{A Fast Iterative Shrinkage-thresholding Algorithm for Linear Inverse Problems}},
  author={Beck, Amir and Teboulle, Marc},
  journal={SIAM Journal on Imaging Sciences},
  volume={2},
  pages={183--202},
  year={2009},
  publisher={SIAM}
}

@book{wainwright2019high,
  title={{High-dimensional Statistics: A Non-asymptotic Viewpoint}},
  author={Wainwright, Martin J},
  volume={48},
  year={2019},
  publisher={Cambridge University Press}
}

@article{xia2015consistently,
  title = {{Consistently Determining the Number of Factors in Multivariate Volatility Modelling}},
  author={Xia, Qiang and Xu, Wangli and Zhu, Lixing},
  journal = {Statistica Sinica},
  volume={25},
  pages={1025--1044},
  year={2015},
  publisher={JSTOR}
}

@article{dwork2014algorithmic,
  title = {{The Algorithmic Foundations of Differential Privacy}},
  author = {Dwork, Cynthia and Roth, Aaron},
  journal = {Foundations and Trends in Theoretical Computer Science},
  volume = {9},
  pages = {211--487},
  year = {2014},
}

@article{neal2011distributed,
  title={{Distributed Optimization and Statistical Learning via the Alternating Direction Method of Multipliers}},
  author={Neal, Parikh and Eric, Chu and Borja, Peleato and Jonathan, Eckstein},
  journal={Foundations and Trends{\textregistered} in Machine learning},
  volume={3},
  pages={1--122},
  year={2011},
  publisher={Emerald Publishing Limited}
}

@article{negahban2009unified,
  title={{A Unified Framework for High-dimensional Analysis of $M$-estimators with Decomposable Regularizers}},
  author={Negahban, Sahand and Yu, Bin and Wainwright, Martin J and Ravikumar, Pradeep},
  journal={Advances in Neural Information Processing Systems},
  volume={22},
  year={2009}
}

@article{cai2020semisupervised,
  title={{Semisupervised Inference for Explained Variance in High Dimensional Linear Regression and its Applications}},
  author={Tony Cai, T. and Guo, Zijian},
  journal={Journal of the Royal Statistical Society Series B: Statistical Methodology},
  volume={82},
  pages={391--419},
  year={2020}
}

@misc{ucsb_spectral_theorem_notes,
  author={Helena Verrill},
  title={Math 108b: Notes on the Spectral Theorem},
  institution={University of California, Santa Barbara},
  year={2011},
  url={https://web.math.ucsb.edu/~helena/teaching/math108b/handouts/notes_spectral_theorems.pdf}
}

@inproceedings{szarek1982nets,
  title={{Nets of Grassmann Manifold and Orthogonal Group}},
  author={Szarek, Stanislaw J},
  booktitle={Proceedings of Research Workshop on Banach Space theory (Iowa City, Iowa, 1981)},
  volume={169},
  pages={185},
  year={1982},
  organization={University of Iowa Iowa City, IA}
}

@article{zhang2018tensor,
  title={{Tensor SVD: Statistical and Computational Limits}},
  author={Zhang, Anru and Xia, Dong},
  journal={IEEE Transactions on Information Theory},
  volume={64},
  pages={7311--7338},
  year={2018},
  publisher={IEEE}
}

@article{cai2018rate,
  title={{Rate-optimal Perturbation Bounds for Singular Subspaces with Applications to High-dimensional Statistics}},
  author={T. Tony Cai and Anru Zhang},
  journal={The Annals of Statistics},
  volume={46},
  pages={60--89},
  year={2018},
  publisher={Institute of Mathematical Statistics}
}

@book{james2013introduction,
  title={{An Introduction to Statistical Learning: with Applications in R}},
  author={James, Gareth and Witten, Daniela and Hastie, Trevor and Tibshirani, Robert and others},
  volume={103},
  year={2013},
  publisher={Springer}
}

@article{luo2023low,
  title={{Low-Rank Tensor Estimation via Riemannian Gauss-Newton: Statistical Optimality and Second-Order Convergence}},
  author={Luo, Yuetian and Zhang, Anru R},
  journal={Journal of Machine Learning Research},
  volume={24},
  pages={1--48},
  year={2023}
}

@article{cai2013sparse,
  title={{Sparse PCA: Optimal Rates and Adaptive Estimation}},
  author={T. Tony Cai and Zongming Ma and Yihong Wu},
  volume={41},
  journal={The Annals of Statistics},
  publisher={Institute of Mathematical Statistics},
  pages={3074--3110},
  year={2013}
}

@article{yang1999information,
  title={{Information-theoretic Determination of Minimax Rates of Convergence}},
  author={Yang, Yuhong and Barron, Andrew},
  journal={The Annals of Statistics},
  pages={1564--1599},
  year={1999},
  publisher={JSTOR}
}

@article{raskutti2011minimax,
  title={{Minimax Rates of Estimation for High-dimensional Linear Regression over Lq-balls}},
  author={Raskutti, Garvesh and Wainwright, Martin J and Yu, Bin},
  journal={IEEE Transactions on Information Theory},
  volume={57},
  pages={6976--6994},
  year={2011},
  publisher={IEEE}
}

@inproceedings{hou2015hierarchical,
  title={{Hierarchical Tucker Tensor Regression: Application to Brain Imaging Data Analysis}},
  author={Hou, Ming and Chaib-draa, Brahim},
  booktitle={2015 IEEE International Conference on Image Processing (ICIP)},
  pages={1344--1348},
  year={2015},
  publisher={IEEE}
}

@article{lu2026versatile,
  title={{Versatile Differentially Private Learning for General Loss Functions}},
  author={Qilong Lu and Song Xi Chen and Yumou Qiu},
  volume={54},
  journal={The Annals of Statistics},
  pages={692--717},
  year={2026},
  publisher={Institute of Mathematical Statistics}
}

@article{dong2022gaussian,
  title={{Gaussian Differential Privacy}},
  author={Dong, Jinshuo and Roth, Aaron and Su, Weijie J},
  journal={Journal of the Royal Statistical Society Series B: Statistical Methodology},
  volume={84},
  pages={3--37},
  year={2022},
  publisher={Oxford University Press}
}

@inproceedings{mcmahan2017communication,
  title={{Communication-efficient Learning of Deep Networks from Decentralized Data}},
  author={McMahan, Brendan and Moore, Eider and Ramage, Daniel and Hampson, Seth and y Arcas, Blaise Aguera},
  booktitle={Artificial Intelligence and Statistics},
  volume={54},
  pages={1273--1282},
  year={2017},
  organization={PMLR}
}

@inproceedings{konecny2015federated,
  title={{Federated Optimization: Distributed Optimization Beyond the Datacenter}},
  author={Kone{\v{c}}n{\'y}, Jakub and McMahan, H. Brendan and Ramage, Daniel},
  booktitle={NIPS Workshop on Optimization for Machine Learning},
  year={2015}
}

@article{kairouz2021advances,
  title={{Advances and Open Problems in Federated Learning}},
  author={Kairouz, Peter and McMahan, H Brendan},
  journal={Foundations and Trends in Machine Learning},
  volume={14},
  pages={1--210},
  year={2021},
  publisher={Emerald Publishing Limited}
}

@article{kaissis2020secure,
  title = {{Secure, Privacy-preserving and Federated Machine Learning in Medical Imaging}},
  author = {Kaissis, Georgios A and Makowski, Marcus R and R{\"u}ckert, Daniel and Braren, Rickmer F},
  journal = {Nature Machine Intelligence},
  volume = {2},
  pages = {305--311},
  year = {2020},
  publisher = {Nature Publishing Group UK London}
}

@inproceedings{abadi2016deep,
  title={{Deep Learning with Differential Privacy}},
  author={Abadi, Martin and Chu, Andy and Goodfellow, Ian and McMahan, H Brendan and Mironov, Ilya and Talwar, Kunal and Zhang, Li},
  booktitle={Proceedings of the 2016 ACM SIGSAC Conference on Computer and Communications Security},
  pages={308--318},
  year={2016}
}

@article{geyer2017differentially,
  title={{Differentially Private Federated Learning: A Client Level Perspective}},
  author={Geyer, Robin C and Klein, Tassilo and Nabi, Moin},
  journal={arXiv preprint arXiv:1712.07557},
  year={2017}
}

@article{yamashita2019harmonization,
  title={{Harmonization of Resting-state Functional MRI Data across Multiple Imaging Sites via the Separation of Site Differences into Sampling Bias and Measurement Bias}},
  author={Yamashita, Ayumu and Yahata, Noriaki and Itahashi, Takashi and Lisi, Giuseppe and Yamada, Takashi and Ichikawa, Naho and Takamura, Masahiro and Yoshihara, Yujiro and Kunimatsu, Akira and Okada, Naohiro and others},
  journal={PLoS Biology},
  volume={17},
  pages={e3000042},
  year={2019},
  publisher={Public Library of Science San Francisco, CA USA}
}

@article{fortin2018harmonization,
  title={{Harmonization of Cortical Thickness Measurements across Scanners and Sites}},
  author={Fortin, Jean-Philippe and Cullen, Nicholas C. and Sheline, Yvette I. and Taylor, Warren D. and Aselcioglu, Irem and Cook, Philip A. and Adams, Phil and Cooper, Crystal and Fava, Maurizio and McGrath, Patrick J. and McInnis, Melvin G. and Phillips, Mary L. and Trivedi, Madhukar H. and Weissman, Myrna M. and Shinohara, Russell T.},
  journal={NeuroImage},
  volume={167},
  pages={104--120},
  year={2018},
  publisher={Elsevier}
}

@article{smith2017federated,
  title={{Federated Multi-task Learning}},
  author={Smith, Virginia and Chiang, Chao-Kai and Sanjabi, Maziar and Talwalkar, Ameet S},
  journal={Advances in Neural Information Processing Systems},
  volume={30},
  year={2017}
}

@article{li2020federated,
  title={{Federated Optimization in Heterogeneous Networks}},
  author={Li, Tian and Sahu, Anit Kumar and Zaheer, Manzil and Sanjabi, Maziar and Talwalkar, Ameet and Smith, Virginia},
  journal={Proceedings of Machine Learning and Systems},
  volume={2},
  pages={429--450},
  year={2020}
}

@article{konyar2024federated,
  title={{Federated Generalized Scalar-on-Tensor Regression}},
  author={Konyar, Elif and Reisi Gahrooei, Mostafa},
  journal={Journal of Quality Technology},
  volume={56},
  pages={20--37},
  year={2024},
  publisher={Taylor \& Francis}
}

@article{zhang2024federated,
  title={{Federated Multiple Tensor-on-Tensor Regression (fedmtot) for Multimodal Data under Data-sharing Constraints}},
  author={Zhang, Zihan and Mou, Shancong and Reisi Gahrooei, Mostafa and Pacella, Massimo and Shi, Jianjun},
  journal={Technometrics},
  volume={66},
  pages={548--560},
  year={2024},
  publisher={Taylor \& Francis}
}

@article{guhaniyogi2017bayesian,
  title={{Bayesian Tensor Regression}},
  author={Guhaniyogi, Rajarshi and Qamar, Shaan and Dunson, David B.},
  journal={Journal of Machine Learning Research},
  volume={18},
  pages={1--31},
  year={2017}
}

@article{sun2017store,
  title={{STORE: Sparse Tensor Response Regression and Neuroimaging Analysis}},
  author={Sun, Will Wei and Li, Lexin},
  journal={Journal of Machine Learning Research},
  volume={18},
  pages={1--37},
  year={2017}
}

@inproceedings{yao1982protocols,
  title={{Protocols for Secure Computations}},
  author={Yao, Andrew C.},
  booktitle={23rd Annual Symposium on Foundations of Computer Science},
  year={1982},
  organization={IEEE}
}

@inproceedings{gentry2009fully,
  title={{Fully Homomorphic Encryption Using Ideal Lattices}},
  author={Gentry, Craig},
  booktitle={Proceedings of the 41st Annual ACM Symposium on Theory of Computing},
  year={2009},
  publisher={ACM}
}

@inproceedings{bonawitz2017practical,
  title={{Practical Secure Aggregation for Privacy-Preserving Machine Learning}},
  author={Bonawitz, Keith and Ivanov, Vladimir and Kreuter, Ben and Marcedone, Antonio and McMahan, H. Brendan and Patel, Sarvar and Ramage, Daniel and Segal, Aaron and Seth, Karn},
  booktitle={Proceedings of the 2017 ACM SIGSAC Conference on Computer and Communications Security},
  year={2017},
  publisher={ACM}
}

@article{castellanos1996quantitative,
  title={{Quantitative Brain Magnetic Resonance Imaging in Attention-Deficit Hyperactivity Disorder}},
  author={Castellanos, F. Xavier and Giedd, Jay N. and Marsh, Wendy L. and Hamburger, Susan D. and Vaituzis, A. Catherine and Dickstein, Daniel P. and Sarfatti, S. E. and Vauss, Y. C. and Snell, J. W. and Lange, Nicholas and Kaysen, D. and Krain, A. L. and Ritchie, G. F. and Rajapakse, J. C. and Rapoport, Judith L.},
  journal={Archives of General Psychiatry},
  volume={53},
  pages={607--616},
  year={1996}
}

@article{berquin1998cerebellum,
  title={{Cerebellum in Attention-Deficit Hyperactivity Disorder: A Morphometric MRI Study}},
  author={Berquin, Patrick C. and Giedd, Jay N. and Jacobsen, Leslie K. and Hamburger, Susan D. and Krain, A. L. and Rapoport, Judith L. and Castellanos, F. Xavier},
  journal={Neurology},
  volume={50},
  pages={1087--1093},
  year={1998}
}

@article{kobel2010temporal,
  title={{Structural and Functional Imaging Approaches in Attention Deficit/Hyperactivity Disorder: Does the Temporal Lobe Play a Key Role?}},
  author={Kobel, Maja and Bechtel, Nina and Specht, Karsten and Klarh{\"o}fer, Markus and Weber, Peter and Scheffler, Klaus and Opwis, Klaus and Penner, Iris-Katharina},
  journal={Psychiatry Research: Neuroimaging},
  volume={183},
  pages={230--236},
  year={2010}
}

\end{document}